%% file: ms.tex
\documentclass[fleqn,11pt,nodoublespace]{bhamthesis}
\pdfoutput=1

\usepackage[hidelinks]{hyperref}

\usepackage[compact]{titlesec}
\usepackage{url,multicol,fancyhdr,lastpage,setspace}
\usepackage[textsize=scriptsize]{todonotes}
\usepackage{amsmath,amsthm,stmaryrd,amsfonts,amssymb,mathtools}
\usepackage{pdfpages}
\usepackage{graphicx,pgf,tikz}
\usetikzlibrary{cd}
\usetikzlibrary{matrix}
\usepackage{mdwtab}
\usepackage{imakeidx}
\makeindex
\makeindex[name=symbol,title={Index of Symbols}]
\include{macros}

\usepackage{qtree,bussproofs}

\makeatletter
\def\operator@font{\sf}
\makeatother

\usepackage{enumitem}
\setlist{noitemsep}

\usepackage{framed}
\usepackage{soul}

\usepackage[a4paper]{geometry}
\usepackage{fullpage}
\usepackage[T1]{fontenc}
\usepackage{tgpagella}
\DeclareTextCommandDefault{\nobreakspace}{\leavevmode\nobreak\ }

\theoremstyle{plain}
\newtheorem{theorem}{Theorem}[chapter]

\newtheorem{corollary}[theorem]{Corollary}
\newtheorem{lemma}[theorem]{Lemma}

\newtheorem*{conjecture}{Conjecture}

\theoremstyle{definition}
\newtheorem{definition}[theorem]{Definition}
\newtheorem{axiom}[theorem]{Axiom}
\newtheorem{assumption}{Assumption}

\titleformat{\chapter}[display]
{\bfseries\Large}
{\filleft \MakeUppercase{\chaptertitlename} \Huge \thechapter}
{2ex}
{\titlerule \vspace{1ex} \LARGE \center \sc}
[\vspace{1ex} \titlerule]

\title{Partial Functions and Recursion in Univalent Type Theory}
\author{Cory Knapp}

\begin{document}
 
\maketitle
\pagenumbering{gobble}
\include{abstract}
\include{acknowledge}

\pagenumbering{arabic}
\tableofcontents

\include{chapters/introduction}
\part{Univalent Mathematics}\label{part:one}
\include{chapters/partoneprologue}
\include{chapters/type-theory}
\include{chapters/univalence}
\include{chapters/foundations}
\part{Partiality}\label{part:two}
\include{chapters/parttwoprologue}
\include{chapters/partiality}
\include{chapters/partial-functions}
\include{chapters/setoids}
\part{Computability theory}\label{part:three}
\include{chapters/computability}
\include{chapters/comp-as-structure}
\include{chapters/comp-as-prop}
\include{chapters/comp-and-partiality}
\include{chapters/future}
\include{chapters/conclusion}

\indexprologue{Named operators such as $\isProp$ point to the corresponding entry in
the main index.}
\printindex[symbol]
\printindex

{\footnotesize
\bibliographystyle{plain}
\bibliography{references}
}

\end{document}

%% file: macros.tex
\newcommand{\name}[1]{\emph{#1}\index{#1|textit}}
\newcommand{\nameas}[2]{\emph{#1}\index{#2|textit}}

\newcommand{\indexandabbr}[2]{\index{#1|textit}\index{#2|see {#1}}}

\newcommand{\definesymbol}[2]{{\index[symbol]{#2@{#1}|textit}}}
\newcommand{\definesymbolfor}[3]{\index[symbol]{#2@{#1}|seealso{#3}}\definesymbol{#1}{#2}}

\newcommand{\indexsymbolfor}[3]{\index[symbol]{#2@{#1}|see{#3}}}

\newcommand{\defineop}[2]{{\index[symbol]{#2@{#1}|textit}}}
\newcommand{\defineopfor}[3]{\index[symbol]{#2@{#1}|see{#3}}}

\newbox\gnBoxA
\newdimen\gnCornerHgt
\setbox\gnBoxA=\hbox{$\ulcorner$}
\global\gnCornerHgt=\ht\gnBoxA
\newdimen\gnArgHgt
\def\Godelnum #1{%
    \setbox\gnBoxA=\hbox{$#1$}%
    \gnArgHgt=\ht\gnBoxA%
    \ifnum     \gnArgHgt<\gnCornerHgt \gnArgHgt=0pt%
    \else \advance \gnArgHgt by -\gnCornerHgt%
    \fi \raise\gnArgHgt\hbox{$\ulcorner$} \box\gnBoxA %
    \raise\gnArgHgt\hbox{$\urcorner$}}

\newcommand{\quantif}[2]{{#1(#2)},}
\newcommand{\qPi}[1]{\quantif{\Pi}{#1}}
\newcommand{\qSig}[1]{\quantif{\Sigma}{#1}}
\newcommand{\qExists}[1]{\quantif{\exists}{#1}}
\newcommand{\qAll}[1]{\quantif{\forall}{#1}}
\newcommand{\qEmpty}[1]{\quantif{}{#1}}

\makeatletter
\def\prd#1{\@ifnextchar\bgroup{\prd@parens{#1}}{%
    \@ifnextchar\sm{\prd@parens{#1}\@eatsm}{%
    \@ifnextchar\prd{\prd@parens{#1}\@eatprd}{%
    \@ifnextchar\;{\prd@parens{#1}\@eatsemicolonspace}{%
    \@ifnextchar\\{\prd@parens{#1}\@eatlinebreak}{%
    \@ifnextchar\narrowbreak{\prd@parens{#1}\@eatnarrowbreak}{%
      \prd@noparens{#1}}}}}}}}
\def\prd@parens#1{\@ifnextchar\bgroup%
  {\mathchoice{\@dprd{#1}}{\@tprd{#1}}{\@tprd{#1}}{\@tprd{#1}}\prd@parens}%
  {\@ifnextchar\sm%
    {\mathchoice{\@dprd{#1}}{\@tprd{#1}}{\@tprd{#1}}{\@tprd{#1}}\@eatsm}%
    {\mathchoice{\@dprd{#1}}{\@tprd{#1}}{\@tprd{#1}}{\@tprd{#1}}}}}
\def\@eatsm\sm{\sm@parens}
\def\prd@noparens#1{\mathchoice{\@dprd@noparens{#1}}{\@tprd{#1}}{\@tprd{#1}}{\@tprd{#1}}}
\def\lprd#1{\@ifnextchar\bgroup{\@lprd{#1}\lprd}{\@@lprd{#1}}}
\def\@lprd#1{\mathchoice{{\textstyle\prod}}{\prod}{\prod}{\prod}({\textstyle #1})\;}
\def\@@lprd#1{\mathchoice{{\textstyle\prod}}{\prod}{\prod}{\prod}({\textstyle #1}),\ }
\def\tprd#1{\@tprd{#1}\@ifnextchar\bgroup{\tprd}{}}
\def\@tprd#1{\mathchoice{{\textstyle\prod_{(#1)}}}{\prod_{(#1)}}{\prod_{(#1)}}{\prod_{(#1)}}}
\def\dprd#1{\@dprd{#1}\@ifnextchar\bgroup{\dprd}{}}
\def\@dprd#1{\prod_{(#1)}\,}
\def\@dprd@noparens#1{\prod_{#1}\,}

\def\@eatnarrowbreak\narrowbreak{%
  \@ifnextchar\prd{\narrowbreak\@eatprd}{%
    \@ifnextchar\sm{\narrowbreak\@eatsm}{%
      \narrowbreak}}}
\def\@eatlinebreak\\{%
  \@ifnextchar\prd{\\\@eatprd}{%
    \@ifnextchar\sm{\\\@eatsm}{%
      \\}}}
\def\@eatsemicolonspace\;{%
  \@ifnextchar\prd{\;\@eatprd}{%
    \@ifnextchar\sm{\;\@eatsm}{%
      \;}}}

\def\lam#1{{\lambda}\@lamarg#1:\@endlamarg\@ifnextchar\bgroup{.\,\lam}{.\,}}
\def\@lamarg#1:#2\@endlamarg{\if\relax\detokenize{#2}\relax #1\else\@lamvar{\@lameatcolon#2},#1\@endlamvar\fi}
\def\@lamvar#1,#2\@endlamvar{(#2\,{:}\,#1)}
\def\@lameatcolon#1:{#1}

\def\lamu#1{{\lambda}\@lamuarg#1:\@endlamuarg\@ifnextchar\bgroup{.\,\lamu}{.\,}}
\def\@lamuarg#1:#2\@endlamuarg{#1}

\def\fall#1{\forall (#1)\@ifnextchar\bgroup{.\,\fall}{.\,}}

\def\exis#1{\exists (#1)\@ifnextchar\bgroup{.\,\exis}{.\,}}

\def\sm#1{\@ifnextchar\bgroup{\sm@parens{#1}}{%
    \@ifnextchar\prd{\sm@parens{#1}\@eatprd}{%
    \@ifnextchar\sm{\sm@parens{#1}\@eatsm}{%
    \@ifnextchar\;{\sm@parens{#1}\@eatsemicolonspace}{%
    \@ifnextchar\\{\sm@parens{#1}\@eatlinebreak}{%
    \@ifnextchar\narrowbreak{\sm@parens{#1}\@eatnarrowbreak}{%
        \sm@noparens{#1}}}}}}}}
\def\sm@parens#1{\@ifnextchar\bgroup%
  {\mathchoice{\@dsm{#1}}{\@tsm{#1}}{\@tsm{#1}}{\@tsm{#1}}\sm@parens}%
  {\@ifnextchar\prd%
    {\mathchoice{\@dsm{#1}}{\@tsm{#1}}{\@tsm{#1}}{\@tsm{#1}}\@eatprd}%
    {\mathchoice{\@dsm{#1}}{\@tsm{#1}}{\@tsm{#1}}{\@tsm{#1}}}}}
\def\@eatprd\prd{\prd@parens}
\def\sm@noparens#1{\mathchoice{\@dsm@noparens{#1}}{\@tsm{#1}}{\@tsm{#1}}{\@tsm{#1}}}
\def\lsm#1{\@ifnextchar\bgroup{\@lsm{#1}\lsm}{\@@lsm{#1}}}
\def\@lsm#1{\mathchoice{{\textstyle\sum}}{\sum}{\sum}{\sum}({\textstyle #1})\;}
\def\@@lsm#1{\mathchoice{{\textstyle\sum}}{\sum}{\sum}{\sum}({\textstyle #1}),\ }
\def\tsm#1{\@tsm{#1}\@ifnextchar\bgroup{\tsm}{}}
\def\@tsm#1{\mathchoice{{\textstyle\sum_{(#1)}}}{\sum_{(#1)}}{\sum_{(#1)}}{\sum_{(#1)}}}
\def\dsm#1{\@dsm{#1}\@ifnextchar\bgroup{\dsm}{}}
\def\@dsm#1{\sum_{(#1)}\,}
\def\@dsm@noparens#1{\sum_{#1}\,}

\newcommand{\ct}{%
  \mathchoice{\mathbin{\raisebox{0.5ex}{$\displaystyle\centerdot$}}}%
             {\mathbin{\raisebox{0.5ex}{$\centerdot$}}}%
             {\mathbin{\raisebox{0.25ex}{$\scriptstyle\,\centerdot\,$}}}%
             {\mathbin{\raisebox{0.1ex}{$\scriptscriptstyle\,\centerdot\,$}}}
}
\newcommand{\istype}[1]{\mathsf{is}\mbox{-}{#1}\mbox{-}\mathsf{type}}

\newcommand{\interpret}[1]{\llbracket #1 \rrbracket}
\newcommand{\defined}[1]{#1{\downarrow}}
\newcommand{\fix}{\operatorname{fix}}
\newcommand{\rel}{\operatorname{rel}}
\newcommand{\ele}{\operatorname{ele}}
\newcommand{\Elem}{\operatorname{Elem}}
\newcommand{\concat}{\mathbin{\dplus}}
\newcommand{\encode}{\mathsf{encode}}
\newcommand{\decode}{\mathsf{decode}}
\newcommand{\codefxn}{\mathsf{code}}
\newcommand{\correct}{\mathsf{C}}
\newcommand{\lar}[1]{\langle #1 \rangle}
\newcommand{\inl}{\operatorname{inl}}
\newcommand{\numeral}[1]{\overline{#1}}
\newcommand{\inr}{\operatorname{inr}}
\newcommand{\zerotype}{\emptyset}
\newcommand{\unittype}{1}

\newcommand{\zero}{\operatorname{zero}}
\newcommand{\succtt}{\operatorname{\mathsf{succ}}}
\newcommand{\ifz}{\operatorname{\mathsf{ifz}}}

\newcommand{\total}{\operatorname{total}}

\newcommand{\pr}{\operatorname{pr}}

\newcommand{\pred}{\operatorname{pred}}

\newcommand{\predtt}{\operatorname{\mathsf{pred}}}

\newcommand{\dd}{\operatorname{d}}
\newcommand{\DD}{\operatorname{D}}

\newcommand{\tame}{\operatorname{tame}}
\newcommand{\Dis}{\operatorname{Dis}}
\newcommand{\isDis}{\operatorname{isDis}}
\newcommand{\cadd}{\mathbin{+^*}}
\newcommand{\Sone}{{\operatorname{S}}^1}
\newcommand{\base}{{\mathsf{base}}}
\newcommand{\lp}{{\mathsf{loop}}}

\newcommand{\YY}{\mathbf{Y}}
\newcommand{\Lift}{\operatorname{\mathcal{L}}}
\newcommand{\Delay}{\operatorname{\mathsf{D}}}
\newcommand{\delay}{\operatorname{delay}}
\newcommand{\tick}{\operatorname{tick}}
\newcommand{\now}{\operatorname{now}}
\newcommand{\comp}{\mathop{\circ}}
\newcommand{\kcomp}{\mathop{\scalebox{0.65}{$\square$}}}
\newcommand{\pcomp}{\mathop{\diamond}}
\newcommand{\trunc}[1]{\left\lVert #1\right\rVert}
\newcommand{\univ}[1][]{
    \ifthenelse{\equal{#1}{}}{\mathcal{U}}{\mathcal{U}_{#1}}
    }
\newcommand{\vniv}[1][]{
    \ifthenelse{\equal{#1}{}}{\mathcal{V}}{\mathcal{V}_{#1}}
    }
\newcommand{\wniv}[1][]{
    \ifthenelse{\equal{#1}{}}{\mathcal{W}}{\mathcal{W}_{#1}}
    }
\newcommand{\Id}{\operatorname{Id}}
\newcommand{\id}{\operatorname{id}}
\newcommand{\idtoequiv}{\operatorname{idtoequiv}}
\newcommand{\idtopie}{\operatorname{idtopie}}

\newcommand{\happly}{\operatorname{happly}}
\newcommand{\isContr}{\operatorname{isContr}}
\newcommand{\isProp}{\operatorname{isProp}}
\newcommand{\Prop}{\operatorname{Prop}}
\newcommand{\isSet}{\operatorname{isSet}}
\newcommand{\Set}{\operatorname{Set}}
\newcommand{\isGpd}{\operatorname{isGroupoid}}
\newcommand{\Grp}{\operatorname{Grp}}
\newcommand{\GrpStruct}{\operatorname{GrpStruct}}
\newcommand{\MonStruct}{\operatorname{MonStruct}}
\newcommand{\Mon}{\operatorname{Mon}}
\newcommand{\isMonoid}{\operatorname{isMonoid}}
\newcommand{\isGrp}{\operatorname{isGrp}}
\newcommand{\cC}{\mathcal{C}}
\newcommand{\cE}{\mathcal{E}}
\newcommand{\isEquiv}{\operatorname{isEquiv}}
\newcommand{\isPIE}{\operatorname{PIE}}
\newcommand{\isPI}{\isPIE}
\newcommand{\linv}{\operatorname{linv}}
\newcommand{\rinv}{\operatorname{rinv}}
\newcommand{\inverse}{\operatorname{inverse}}
\newcommand{\qinv}{\operatorname{qinv}}
\newcommand{\funext}{\operatorname{funext}}
\newcommand{\FunExt}{\operatorname{FunExt}}
\newcommand{\PropExt}{\operatorname{PropExt}}
\newcommand{\ExpProp}{\operatorname{ExpProp}}
\newcommand{\ua}{\operatorname{ua}}
\newcommand{\UA}{\operatorname{UA}}

\newcommand{\ACtype}[2]{\big(\qPi{#1}{\trunc{#2}}\big)\to\trunc{\qPi{#1}{#2}}}
\newcommand{\propext}{\operatorname{propext}}

\newcommand{\refl}[1][]{
    \ifthenelse{\equal{#1}{}}{\mathsf{refl}}{\mathsf{refl}_{#1}}
    }
\newcommand{\transport}{\operatorname{transport}}
\newcommand{\idtofun}{\operatorname{idtofun}}
\newcommand{\coe}{\idtofun} 
\DeclareMathOperator{\apd}{apd}
\DeclareMathOperator{\ap}{ap}
\DeclareMathOperator{\app}{app}
\DeclareMathOperator{\fib}{fib}
\DeclareMathOperator{\im}{im}
\DeclareMathOperator{\dom}{dom}
\DeclareMathOperator{\ran}{range}
\DeclareMathOperator{\isEmbedding}{isEmbedding}
\newcommand{\defeq}{\stackrel{\operatorname{def}}{=}}
\newcommand{\eqby}[1]{\stackrel{{#1}}{=}}
\newcommand{\eqover}[4][]{{#2}=^{#4}_{#1}{#3}} 
\newcommand{\jeq}{\equiv}
\newcommand{\htpy}{\sim}
\newcommand{\pieq}{\simeq_{P}}
\newcommand{\bisim}{\approx}
\newcommand{\seq}{\approx}
\newcommand{\sle}{\sqsubseteq}
\newcommand{\lequiv}{\Leftrightarrow}
\DeclareMathOperator{\recop}{rec}
\DeclareMathOperator{\succop}{succ}
\newcommand{\relto}{\rightharpoonup^{\operatorname{R}}}

\newcommand{\pto}{\rightharpoonup}

\newcommand{\nat}{\mathbb{N}}
\newcommand{\Z}{\mathbb{Z}}
\newcommand{\rat}{\mathbb{Q}}
\newcommand{\real}{\mathbb{R}}
\newcommand{\bool}{2} 
\newcommand{\pos}{\operatorname{pos}}
\newcommand{\ff}{\mathsf{false}}
\newcommand{\sgn}{\operatorname{sn}}
\newcommand{\opsgn}{\overline{\operatorname{sn}}}
\newcommand{\N}{\nat}
\newcommand{\Ntwo}{\N_{-2}}
\newcommand{\NI}{\N_\infty}

\newcommand{\cantor}{2^\nat}
\newcommand{\R}{\operatorname{R}}
\newcommand{\RR}{\operatorname{R}}
\newcommand{\SD}{\operatorname{S}}
\newcommand{\powerset}{\mathcal{P}}
\newcommand{\defequiv}{\stackrel{\operatorname{def}}{\Leftrightarrow}}

\DeclareMathOperator{\KleeneT}{T}
\DeclareMathOperator{\KleeneU}{U}
\newcommand{\isRos}{\operatorname{isRosolini}}
\newcommand{\isReg}{\operatorname{isRegular}}
\newcommand{\isCauchy}{\operatorname{isCauchy}}
\newcommand{\RosStruct}{\operatorname{rosoliniStructure}}
\newcommand{\RS}{\operatorname{RS}}

\newcommand{\CT}{\operatorname{CT}}

\newcommand{\CC}{\operatorname{CC}}
\newcommand{\KS}{\operatorname{KS}}
\newcommand{\CompStruct}{\mathsf{CompStruct}}
\newcommand{\Comp}{\mathsf{Comp}}
\newcommand{\isComputable}{\mathsf{isComputable}}
\newcommand{\SemiDecision}{\mathsf{SemiDecision}}
\newcommand{\isSemiDecidable}{\mathsf{isSemiDecidable}}
\newcommand{\isDefined}{\extent}
\newcommand{\dplus}{\mathbin{+\mkern-10mu+}}
\newcommand{\eval}{\operatorname{eval}}
\newcommand{\source}{\operatorname{source}}
\newcommand{\extent}{\operatorname{defined}}
\newcommand{\val}{\operatorname{value}}
\newcommand{\isPR}{\operatorname{isPR}}
\newcommand{\PR}{\operatorname{PR}}

\newcommand{\RM}{\operatorname{RM}}
\newcommand{\PrimRec}{\operatorname{PrimRec}}
\newcommand{\isRec}{\operatorname{isRec}}
\newcommand{\succc}{\mathsf{s}}
\newcommand{\pc}{\mathsf{p}}
\newcommand{\cc}{\mathsf{c}}
\newcommand{\rc}{\mathsf{r}}

\newcommand{\run}{\operatorname{run}}
\newcommand{\U}{\univ}

\newcommand{\klext}[1]{#1^{\sharp}}
\newcommand{\smn}{\operatorname{s}}
\newcommand{\CH}{\mathrm{CH}}
\newcommand{\LPO}{\mathrm{LPO}}


%% file: abstract.tex
\begin{abstract}
    We investigate partial functions and computability theory from within a
    constructive, univalent type theory. The focus is on placing computability into a
    larger mathematical context, rather than on a complete development of computability
    theory. We begin with a treatment of partial functions, using the notion of
    dominance, which is used in synthetic domain theory to discuss classes of partial
    maps. We relate this and other ideas from synthetic domain theory to other approaches to partiality in
    type theory. We show that the notion of dominance is difficult to apply in our
    setting: the set of $\Sigma_0^1$ propositions investigated by Rosolini form a dominance
    precisely if a weak, but nevertheless unprovable, choice principle holds. To get
    around this problem, we suggest an alternative notion of partial
    function we call disciplined maps. In the presence of countable choice, this
    notion coincides with Rosolini's.

    Using a general notion of partial function, we take the first steps in
    constructive computability theory. We do this both with computability as
    structure, where we have direct access to programs; and with computability as
    property, where we must work in a program-invariant way. We demonstrate the
    difference between these two approaches by showing how these approaches relate
    to facts about computability theory arising from topos-theoretic and
    type-theoretic concerns. Finally, we tie the two threads together:
    assuming countable choice and that all total functions
    $\nat\to\nat$ are computable (both of which hold in the effective topos), the
    Rosolini partial functions, the disciplined maps, and the computable partial
    functions all coincide. We observe, however, that the class of all partial
    functions includes non-computable partial functions.
\end{abstract}

%% file: acknowledge.tex
\begin{acknowledgements}
    As always, there are far too many people to thank, and I'm bound to miss some, so
    I'll just leave it at the obvious, and academic: Thanks, especially, to my supervisor Mart\'in
    Escard\'o, for patient explanations, sage advice, and long working meetings. I
    cannot express how helpful of a supervisor you have been. Thanks to the staff in
    the Birmingham theory group---especially to Paul Levy, for teaching me category
    theory faster than I knew it could be taught, and for consistently playing
    devil's advocate; to Achim Jung, for explaining logical relations to me; and to
    Steve Vickers, for explaining geometric reasoning. Thanks also to the PhD
    Students in the group---especially to Auke for answering all my mathematical questions, and to
    everyone involved in CARGO for continuous motivation.

    Thank you to the examiners Steve Vickers and Bernhard Reus for taking the time to read this.

    Mike Shulman and Peter Lumsdaine may not know it, but they helped set me on the
    path that led here: one by being kind to a clueless undergrad, and the other by
    giving direction to a masters student in over his depth. Both have also given
    advice and encouragement since.

    Finally, of course, I couldn't have done this without endless support from
    friends and family, but such thanks are better said in person.
\end{acknowledgements}

%% file: chapters/introduction.tex
\chapter*{Introduction}
\addcontentsline{toc}{chapter}{Introduction}  

In \emph{Mathematics as a Numerical Language}~\cite{Bishop1970Numerical}, Bishop
outlines what he viewed as the most important
questions of constructive mathematics of his time. In it, he states ``that
mathematics concerns itself with the precise description of finitely performable
abstract operations.'' He also states that constructive mathematicians can use
classical mathematics ``as a guide'' and that much of classical mathematics ``will
raise fundamental questions which classically are trivial or perhaps do not even make
sense.'' The essay continues by sketching three problems that are classically easy, but when
viewed constructively demonstrate gaps in our conceptual understanding. He states
``Each of the three problems just discussed requires the development of new concepts
appropriate to the constructive point of view. None of them is likely to be given an
acceptable solution by the application of a general technique of
constructivization.'' In fact, for two of these problems (Birkhoff's ergodic theorem and the
construction of singular cohomology), there is no constructive solution, so a new
approach must be taken to develop the field. For the third (Hilbert's basis theorem),
the proof itself is constructive, but exhibiting objects to which the theorem can be
applied necessarily requires non-constructive arguments.

This thesis is focused on computability theory, but done from the point of view of
constructivization that we see in the previous paragraph. That
is,
\begin{itemize}
    \item our mathematical work should describe ``finitely performable abstract
        operations'';
    \item we will be guided by the classical development of computability theory, but
    \item we will pay special attention to fundamental questions which are
        classically trivial and on concepts that are appropriate for the constructive
        point of view.
\end{itemize}
The goal, then, is not to develop a full constructive account of computability theory
but to understand how computability theory fits into a constructive mathematical
world. In particular, the behavior of partial functions is more subtle in a
constructive world, so we spend a considerable amount of time examining the theory
of partial functions.

In order to proceed, we must fix a constructive framework for doing
mathematics. Formally, we use a modest extension of intensional
Martin-L\"{o}f type theory (MLTT). This connects directly to our motivating paragraph:
while introducing his type theory in~\cite{MartinLof1975Predicative}, Martin-L\"{o}f cites
Bishop~\cite{Bishop1970Numerical,bishop1967} as a motivation for MLTT. We will not
work formally: since we are communicating mathematics, and not describing a computer
formalization, we will describe and use this implicit framework in an informal way.
Nevertheless, it should be possible to formalize the material in this thesis
directly in a proof assistant such as Coq~\cite{coq} or Agda~\cite{AgdaWiki}.

Both the approach we take to mathematics and the language we use to communicate it
are not widely familiar. The traditional constructive approach is a
\emph{Bishop-style} mathematics: one gives sets via a collection of elements and an
equivalence relation saying how to identify these elements. Mathematics in MLTT
usually mimics this approach using setoids (Section~\ref{section:bishop-mathematics} and
Chapter~\ref{chapter:setoids}), but more recent work tends to be
done from a \emph{univalent} perspective (described in~\ref{section:utt}), which
we take here. Briefly, univalent mathematics is an approach to mathematics which
instead works with the identity types for equality, and which pays special attention
to the structure arising from this proof-relevant equality. This attention motivates
several new definitions, in particular a new interpretation of logic in type theory.
Moreover, this approach allows us to formally distinguish between notions seen as
structure and notions seen as property. Practical work in this approach requires the
addition of some extensionality principles to MLTT
(Chapter~\ref{chapter:univalence}).

Part I of this thesis develops the univalent approach to constructive mathematics.
The aim is to present a concise, self-contained and readable introduction to
univalent mathematics, going somewhat beyond what is required in Part II. Most of the
material in Part I can be found already in the HoTT Book~\cite{hottbook}, but we make
a number of different choices in terminology and presentation. In particular, we
simplify many proofs, insist less on homotopical intuition, and reorder the
priority and primacy of various concepts to better fit with the aims of Part II and
the changes in the field over the last several years.
We also spend quite some time discussing the relationship between extensionality
axioms. There is a technical contribution here: both the proof that univalence implies
function extensionality and the surrounding discussion in
Section~\ref{section:univalence} are new, although some of the ideas were already
implicit in Voevodsky's proof as presented in~\cite{gambino2011univalence}.

Part II begins the main technical contributions of the thesis by examining partial functions.
There are several approaches to partial functions
(Chapter~\ref{chapter:partiality} and Sections~\ref{section:single-valued-relations} and~\ref{section:ptl-elems}), but the
most technically useful is one based on a notion of \emph{partial element}
(Definition \ref{def:lifting}), which is a type-theoretic version of the notion of
subsingleton subset, as well as a constructive version of the lifting from
classical domain theory; indeed, all lifted sets are DCPOs
(Section~\ref{section:dcpos}).

As mentioned above, our broad aim is to better understand how the computable partial
functions fit into a constructive world. Although intuition tells us that
anything which can be done constructively can be done computably, we can define a
partial function $h:\nat\pto\nat$ which enumerates the complement of the halting set,
and so cannot be a computable partial function. It is worth asking then which partial
functions \emph{can be} computable.  Chapter~\ref{chapter:partial-functions} is
devoted to shedding light on this question by looking for a restricted notion
of partial function for which it is consistent that all such partial functions are
computable. Importantly, we expect such partial functions to be composable. The
obvious way to find a more restricted notion of partial function is via the
notion of \emph{dominance}---set of propositions satisfying certain closure
properties (Section~\ref{section:dominances}). Dominances were introduced and studied by
Rosolini~\cite{Rosolini1986} as part of a project of \emph{synthetic domain theory},
and modulo size issues (Section~\ref{section:size}), the theory of dominances transports
directly from a topos-theoretic to a type-theoretic setting. Importantly, we can give
a restricted notion of partial function for any set of propositions $S$, and the
partial functions arising this way are closed under composition precisely if $S$ is a
dominance. This fact suggests that we can use the dominance of \emph{Rosolini
propositions}, typically called $\Sigma$, as our restricted notion of partial function.
Indeed, this is the approach taken in work on both synthetic domain theory and
synthetic computability theory.

Unfortunately, a weak version of countable choice, which is not provable in our
system~\cite{coquand2017stack}, is required to show that the
Rosolini propositions form a dominance (Theorem~\ref{rosolini:choice-to-dominance}).
The Rosolini dominance arises by truncating a family of types---by replacing a
structure with the property that there is such a structure; for any set of
propositions $S$ arising this way, there is a weakening of
the axiom of choice equivalent to the claim that $S$ is a dominance
(Theorem~\ref{dominance:choice-to-dominance}). From the perspective of partial
functions, the issue is that intensional information is used to compose pairs of
$S$-partial functions, but this intensional information is hidden when we truncate
$S$. That is, $S$-partial functions come with no guarantee that
they respect the required intensional information. To resolve this issue we look at a
notion of \emph{disciplined map} (Section~\ref{section:disciplined-maps}), which can be seen as
the class of maps which respect the intensional information used to compose
$S$-partial functions. Indeed, disciplined maps can be shown to compose even without
choice (Theorem~\ref{disciplined-maps:compose}), while countable choice implies that
the disciplined maps out of $\nat$ are exactly the $S$-partial functions from $\nat$
(Theorem~\ref{thm:disc-factoring}). That is, in the settings examined by synthetic
domain theorists and synthetic computability theorists, the Rosolini disciplined maps
are the Rosolini partial functions, making them a natural candidate notion of
computable function.

In Part~\ref{part:three}, we finally turn to computation. Facts about computation
from the topos-theoretic perspective seem to conflict with facts about computation
from the type-theoretic perspective. From the topos-theoretic perspective,
\begin{enumerate}
    \item it is \emph{consistent} that all total functions $\nat\to \nat$ are computable: this is
        true of constructive higher-order logic, with the effective topos as a model.
    \item It cannot be proven that there is an embedding from the class of computable
        functions to the natural numbers: this embedding would tell us that excluded
        middle holds for equality between computable functions, which fails in the
        effective topos.
\end{enumerate}
On the other hand, from the type-theoretic perspective
\begin{enumerate}
    \item it is \emph{false} that all total functions $\nat\to\nat$ are computable:
        an argument by Troelstra shows that even in a constructive higher-order
        setting, this would allow us to solve the halting problem.
    \item There is an embedding from the class of computable functions to the
        natural numbers: this is a first theorem in computability theory.
\end{enumerate}
Univalent mathematics gives a general framework incorporating both of the above
perspectives. In particular, the topos-theoretic facts correspond to facts about
computability as \emph{property}, while the type-theoretic facts correspond to facts
about computability as \emph{structure}. We present computability via
a notion of \emph{recursive machine} (Section~\ref{section:recursive-machines}) which
abstracts away details of initialization, state and memory from the definition of
Turing machine. We say that a \emph{computation structure} for a partial function
$f:\nat\pto\nat$ is a recursive machine computing $f$, while \emph{being computable}
is the associated property arising by truncating the type of computation structures.
Much of a first course in computability theory goes through for both computation as
structure (Chapter~\ref{chapter:comp-as-struct}) and computation as
property (Chapter~\ref{chapter:comp-as-prop}). The computability theory we
present is not deep, and we reiterate that a complete development is not the goal.
Rather, these results should show that our approach to computability theory is sound.

The final technical Chapter~\ref{chapter:comp-and-partiality} contains a discussion of how
computability as structure and computability as property differ, and a look at how
computability theory fits with the view of partial functions developed in
Chapter~\ref{chapter:partial-functions}, bringing the two threads back together. The
key result is that countable choice and ``Church's thesis'' (the statement that all
\emph{total} functions $\nat\to\nat$ are computable) together imply that the
computable partial functions $\nat\pto\nat$ are exactly the Rosolini partial functions.
There is an unfortunate lacuna here: one of our goals is a notion of partial function
for which it is consistent that all such partial functions are computable. It is
indeed consistent with MLTT that our candidate notion of disciplined map is such a
notion: in the effective topos, the disciplined maps, the Rosolini partial
functions and the computable partial functions coincide. However, we use a univalent
version of proposition extensionality throughout, and this statement relies on
treating propositions as types satisfying a certain property, rather than as elements
of a sub-object classifier, as is done in topos logic. We would need to either give a
translation from our type theory to topos logic, or construct a model validating
countable choice and
Church's thesis to resolve this discrepancy. For reasons of scope, we do neither
here. Filling this gap is one of the obvious directions for
future work. Another direction for future work is to examine higher-type
computability from the perspective developed here. We discuss both of these
directions in more depth when discussing \hyperref[chapter:future]{future work}.

A few threads woven through this thesis determine its shape, despite being
difficult to express in a way that is both general and technical. I will try to
explain them here, with reference to Bishop's essay \emph{Schizophrenia in Contemporary
Mathematics}~\cite{bishop1973schizophrenia}.

Bishop's essay concerns itself primarily with
what Bishop calls the \emph{debasement of meaning}. Two particular trends leading to
this debasement (according to Bishop), were the tendency towards formalism---in which
meaning is simply ignored, in favor of formal manipulations---and an
``esotericism''---in which meaning becomes imprecise. Strangely, despite the
particularly formal subject matter, it seems computability suffers more from the
latter than the former. For example, many books on computability theory will fix some
formal notion of computation (or discuss and compare several), and
then after the first chapter, prove all results with an appeal to the Church-Turing
thesis and a rough description of an algorithm, leaving the reader left to wonder why
they went through the trouble of understanding (e.g.) Turing machines in the first
place. The appeal to intuition implicit in the reference to the Church-Turing thesis
would be more honest if done as in other areas of mathematics: rather than suggest we
are appealing to an empirical principle, simply describe informally how to construct
the mathematical object of interest.

For his part, Bishop responds to the identification of finitely-performable and
recursive by saying that the ``naive
concept`` of an algorithm is more basic than that of recursive function, and ``the
recursive concept derives whatever importance it has from some presumption that every
algorithm will turn out to be recursive.'' Bishop, then, seems to take the view that
computability theory is a misguided attempt to make more precise the concept of
algorithm. I take a different view: computability theory is the
mathematization of the real-world concept of program, and programmable functions.
This mathematization has proved fruitful in aiding our understanding of programming,
and one may suspect, will continue to do so. In other words, computability theory is
a part of mathematics which deserves treatment.

A more precise (and less polemical) example of the debasement of meaning in
computability theory is given by the notion of \emph{Kleene equality}. One often says
that $f(x)\simeq g(x)$ when either both $f(x)$ and $g(x)$ are undefined, or both are
defined and equal. This does not work in a constructive setting, since we cannot say
that a function is either defined or undefined on some input. Moreover, the treatment
of partial functions is such that the terms $f(x)$ and $g(x)$ are only meaningful when
defined. That is, this notation makes
reference to a non-existent thing, and we explain this notation by reference again to
this non-existent thing. The treatment given here differs: we have no need of Kleene equality,
because partial functions $f:X\pto Y$ are taken to be ordinary function into a special type of
partial values of $Y$. Then, $f(x)\simeq g(x)$ precisely if $f(x)$ and $g(x)$ are equal as partial
values.

In this case, an ad hoc notion of equivalence can be reduced to an honest equality by
revising our approach. This idea---replacing ad hoc notions of equivalence with the
general notion of equality---represents a crucial benefit of the univalent approach.
Working directly with setoids, as Bishop does and as is traditional in type theory,
brings with it byzantine bureaucracy, arising from the
ad-hoc treatment of equality. Martin Hoffman showed how to interpret MLTT with
equality in setoids~\cite{Hofmann1995}, so MLTT can be used to simplify this
bureaucracy. Identity types provide a uniform and general way to present the
machinery of setoids. However, some \emph{extensionality} principles are required to
treat identity types properly, and there has long been a question about how best to
approach extensionality in type theory. The univalence axiom provides one possible answer to
this question.

The last point concerns the meaning of existence. One of Bishops concluding remarks
is ``There seems to be no reason in principle that we should not be able to develop a viable terminology
that incorporates more than one meaning for some or all of the quantifiers and
connectives.'' A particular point of dispute between classical and constructive
mathematicians is the meaning of the existential quantifier. The univalent
perspective gives two versions, one which we denote $\exists$ is valued in
\emph{propositions}, the other is denoted $\Sigma$ and is valued in
\emph{structures}. While our logic of propositions is not classical, the meaning of
$\exists$ seems more in-line with the logicist interpretation used in classical
mathematics, while the meaning of $\Sigma$ aligns with the intuitionistic
interpretation used in constructive mathematics. These connectives are made possible
by the distinction between properties (families of propositions) and structures (general type
families) in a univalent setting. This distinction is certainly not unique to
univalent mathematics, but it seems to be the first to give a formal language for
exploring the distinction. The work reported here began as an attempt to study
the distinction between structure and property, not at a general abstract level, but
in the wild, with a particular example. It grew into its present form by returning
repeatedly to the central question: how does computability viewed as structure differ from
computability viewed as property?

\section*{Summary of contributions}
\addcontentsline{toc}{section}{Summary of contributions}
The main contributions are
Chapters~\ref{chapter:partial-functions},~\ref{chapter:comp-as-struct},~\ref{chapter:comp-as-prop},
and~\ref{chapter:comp-and-partiality}. Some of the material in these chapters has
appeared in~\cite{partialelems2017}.
Sections~\ref{section:dominance-choice}-\ref{section:disciplined-maps} contain the
truly novel
material of Chapter~\ref{chapter:partial-functions}; the rest of of the chapter is
largely an exercise in translation from one framework to another. The technical results in
Chapters~\ref{chapter:comp-as-struct} and~\ref{chapter:comp-as-prop} are standard,
well-known results, but they have not been proved in a univalent setting, and the key
contribution is the clear and rigorous distinction between computability as structure
and computability as property. This distinction is what leads to the denouement in
Chapter~\ref{chapter:comp-and-partiality}. Except for Theorem \ref{not-strong-ct} due
to Troelstra, this chapter is entirely new.
The discussion of partiality via setoids in Chapter~\ref{chapter:setoids} is new, but
is only in sketch, and serves mostly as comparison.

Part~\ref{part:one} is standard material. The main technical contribution is the new
proof that univalence implies function extensionality in
Section~\ref{section:univalence}. There are some minor contributions:  some minor simplifications in other proofs (the
proof of Theorem~\ref{thm:nat-is-set} comes to mind), and
Chapter~\ref{chapter:foundations}, which is mostly folklore, but which contains some
material that does not seem to be published anywhere. Nevertheless I feel there is a
larger contribution here: the organization of concepts in univalent mathematics has
changed in the 5 years since the HoTT book was first published.
Part~\ref{part:one} provides an introduction to the field which reflects (some of)
these changes.


%% file: chapters/partoneprologue.tex
Part I deals with univalent foundations. We begin with an overview chapter which
explains the underlying system (Section~\ref{section:mltt}), relates this system to
approaches to constructive math
(Sections~\ref{section:bishop-mathematics}-\ref{section:topos-logic}), and then
introduces the univalent perspective
(Sections~\ref{section:utt}-\ref{section:equiv-contr}). It is important to emphasize
that the way we make intuitive ideas rigorous in a univalent setting is different
from the classical approach. The first chapter is is chiefly designed to give language that
better fits the univalent approach, and to give intuition on how it differs from
other approaches to constructive mathematics. In particular, since univalent
mathematics is built on a type-theoretic framework, logic is not an underlying part
of the language, but something which we must interpret. Instead of the Curry-Howard
interpretation taken in traditional type-theoretic approaches, we take a
``propositions as subsingletons'' view of logic, isolating the propositions as a
distinguished class of types. Importantly, this distinction is made internally (not
using a judgement), using a defined operator $\isProp:\univ\to\univ$; except for
constructions made using $\isProp$, propositions are treated as any other type.

With this language in place, we develop basic tools for the univalent perspective in
Chapter~\ref{chapter:univalence}, focusing particularly on extensionality principles
(Sections~\ref{section:funext} and~\ref{section:univalence}) and on propositions
(Sections~\ref{section:unilogic}). We will see quickly that pure MLTT is insufficient
for developing a logic of propositions using $\isProp$. We must extend the type
theory with a \emph{truncation} operator, taking a type to its best representation as
a proposition. We will also need the propositions to be closed under $\Sigma$ and
$\Pi$; Closure under $\Pi$ happens to be equivalent to extensionality for functions
(pointwise equal functions are equal),
while closure under $\Sigma$ is true already in MLTT. Similarly, it is desirable to
have an extensionality principle for propositions (logically equivalent propositions
are equal). We will assume both of these principles in Parts~\ref{part:two}
and~\ref{part:three}, but here we discuss the \emph{univalence axiom}, which serves
as an extensionality principle for types, implying both proposition and function
extensionality. Here, and throughout, we use the term \emph{extensionality principle}
to mean a principle which allows us to prove equality between objects from a
seemingly weaker relationship. We do not consider the question of whether such
principles make the theory truly extensional (for this, see~\cite{Streicher2013How}
or~\cite{Ladyman218Foundation}).
Sections \ref{section:resizing},~\ref{section:hlevels}
and~\ref{section:hits} cover ideas that are new to univalent mathematics, via
applications to constructions that have been problematic in type theory.

Finally, we examine a few traditional notions from the univalent perspective in
Chapter~\ref{chapter:foundations}. Of particular importance are
Section~\ref{section:monads} on an internal notion of monad on the universe of types,
and Section~\ref{section:choice} on the axiom of choice. These are not of note
because the univalent approach is interesting (although with choice, it is), but
because we will be particularly concerned with monads and choice principles in
Part~\ref{part:two}.


%% file: chapters/type-theory.tex
\chapter{Type theory}
To begin, we explain the base formal system, which we will later extend in
Chapter~\ref{chapter:univalence}. We then move on to a discussion of the traditional
way to interpret logic in this system, and a discussion of informal constructive
mathematics. To end the chapter, we give a small taste of univalent foundations,
discussing the perspective taken in a univalent development and giving the basic
definitions needed to explore the univalent approach.

\section{Martin-L\"{o}f type theory}\label{section:mltt}
We work in an extension of MLTT similar to that in the HoTT Book~\cite{hottbook}, and
use some language from the HoTT Book. Our theory is a theory of objects called
\emph{types} which have members called \emph{elements} or \emph{inhabitants}. We write
$a:A$ to mean $a$ is an element of the type $A$. Judgmental equality--which
expresses a syntactic notion of equality---will be denoted by $\jeq$. Elements can
only be given as elements of a specified type. That is, we have no global universe of
discourse or global membership predicate. We have several basic types and type
operations, each of which come with \emph{formation rules} saying how to create the
type, \emph{introduction rules} saying what the type's \emph{canonical}
elements are, \emph{elimination rules} saying how to use elements of the type,
with a specification (the \emph{computation rules}) saying how the eliminator acts on
canonical elements. We treat types themselves as being elements of special types
called \nameas{universes}{universe}, representing types of types. To avoid circularity issues, we
have a hierarchy of universes, \definesymbolfor{$\univ$}{univ}{universe}
$\univ_0,\univ_1,\ldots$. Universes are cumulative (if $A:\univ[i]$ then
$A:\univ[i+1]$) but are not transitive (so $a:A$ and $A:\univ[i]$
does not mean $a:\univ[i]$), and we
have $\univ[i]:\univ[i+1]$ for each $\univ[i]$. We will work in an ambiguous universe
$\univ$, only indexing the universe when we are discussing size issues. As such, we
will write $A:\univ$ to mean that $A$ is a type. Type families indexed by a type $A$
will be treated as functions $B:A\to\univ$.

We will often use types as statements, in which case we mean to say that they have an
element. For example, we state many theorems by saying something of the form "We
have~$P$.", which means "There is some (unspecified) $p:P$". In fact, all statements
and constructions in our system reduce to the construction of some element $p:P$, and
we prove all of them by giving an explicit element. We discuss this in more depth in
Section~\ref{section:unilogic}.

The first type former of interest is that of \nameas{dependent
products}{type!dependent product}:
For any type $A$ and type family $B:A\to\univ$, we have a \emph{dependent
product} type, or \emph{type of dependent functions}, ${\qPi{a:A}{B(a)}:\univ}$.
Elements  are introduced via lambda abstraction: if for any $x:A$, we have
$t:B(x)$ then there is an element $\lam{x}t:\qPi{a:A}{B(a)}$. We will sometimes write definitions by
giving arguments to the function, so $f(x)\defeq y$ means $f\defeq \lambda x.y$. The
elimination rule is function application, and the computation rule says that
$(\lam{x} t)(a) \jeq t[x/a]$. The type $A\to B$ of \nameas{functions}{function} from $A$ to $B$ is the
product over the constant family $\lam{x:A}B$.

We treat a function of two 
arguments (of type $A$ and type $B$) as a \emph{curried} function of type $A\to (B\to C)$, or in the dependent
case, if $B:A\to\univ$ and $C:\qPi{a:A}(B(a)\to\univ)$, as a function
$\qPi{a:A}(\qPi{b:B(a)}C(a)(b))$. When we have $f:\qPi{a:A}(\qPi{b:B(a)}C(a)(b))$ we
may write $f(a)(b)$, or $f(a,b)$ or $f_a(b)$ for the application of $f$ to $a$ and
$b$, depending on focus and readability. To simplify notation, we treat quantifiers (such as
$\Pi$) as binding as far right as possible, and $\to$ as right-associative, so that
we can write the types above as $A\to B\to C$ and $\qPi{a:A}\qPi{b:B(b)}C(a)(b)$.
The most basic example of a function is the identity function,
\begin{flalign*}
    \id\;\;  &:\qPi{A:\univ}(A\to A),\\
    \id_A&\jeq \lambda x.x.
\end{flalign*}

The computation rules are sometimes know as \emph{$\eta$ rules}
\index{eta@{$\eta$}!rule|textit}, and the $\eta$ rule
for functions warrants some discussion. Consider the non-dependent composition
operator
\[-\comp-:(B\to C)\to (A\to B)\to (A\to C)\]
defined by \definesymbol{$\comp$}{composition}
\[g\comp f \defeq \lambda x.g(f(x))\]
Then the $\eta$ rules says that $h\comp(g\comp f) \jeq (h\comp g)\comp f$, as we have
\[h\comp(g\comp f) \defeq \lambda x.h(g\comp f(x)) \jeq \lambda x.h((\lambda
x.g(f(x)))(x)) \jeq \lambda x.h(g(f(x))) \jeq \lambda x. (h\comp g)(f x) \jeq (h\comp
g)\comp f.\]
In fact, in general, if two functions $f,g:A\to B$ are such that $f(x)\jeq g(x)$ for
every $x:A$, then we have $f\jeq g$. However, judgmental equality should be thought
of as part of the meta-system, and so we cannot express this fact in our system. We
will shortly introduce a type-level or \emph{propositional} version of equality, and the
corresponding statement (known as \emph{function extensionality}) for propositional
equality can be expressed internally, but does not hold in pure MLTT. We discuss
function extensionality in more detail in Section~\ref{section:funext}.

Most of the rest of our type formers will be \nameas{inductive types}{type!inductive}:
the elimination
rules are \emph{induction principles} saying that a (dependent) map out of the type
is defined uniquely by specifying its behavior on the canonical elements. We have,
\begin{itemize}
    \item An \nameas{empty type}{type!empty}, $\zerotype:\univ$ with no introduction rule, so that
        for any $C:\univ$ we have a unique function $!:\zerotype\to C$; more
        generally, if $C:\zerotype\to\univ$ then we have a unique dependent function
        $!:\qPi{z:\zerotype}C(z)$.
        We interpret statements of the form "Not $P$" as meaning
        $P\to\zerotype$.
    \item A \nameas{unit type}{type!unit}, $\unittype:\univ$ with a single canonical element
        $\star:\unittype$.
    \item A \emph{boolean type}, $\bool:\univ$ with canonical elements
        $0,1:\bool$.
    \item A type \definesymbolfor{$\nat$}{nat}{type of natural numbers}$\nat:\univ$ of
        \nameas{natural numbers}{type!of natural numbers}
        \index{natural numbers|see{type!of natural numbers}} with canonical elements $0:\nat$
        and $\succop(n):\nat$ for each $n:\nat$. As usual, we will often write $n+1$
        for $\succop(n)$.
    \item For each $A:\univ$ and $B:A\to\univ$ a \nameas{sum type}{type!sum}, or type of
        \emph{dependent pairs}, $\qSig{a:A}{B(a)}$ with canonical elements
        $(a,b):\qSig{a:A}{B(a)}$ for each $a:A$ and $b:B(a)$. The
        \nameas{cartesian product}{type!product} $A\times B$ of $A$ and $B$ is taken
        to be the sum over the constant family $\lam{x:A}B$
    \item For $A,B:\univ$, we have a \nameas{coproduct}{type!coproduct} or \emph{disjoint union} type
        $A+B$ with canonical elements $\inl(a):A+B$ for each $a:A$ and $\inr(b):A+B$ for
        each~$b:B$.
\end{itemize}
A member of $\qSig{x:A}B(x)$ is thought of as a pair, and this is expressed by the
elimination rules for $\Sigma$. In particular, we have
\begin{align*}
    \pr_0\qquad\,&:\big(\qSig{x:A}B(x)\big)\to A\\
    \pr_0(a,b)&\defeq a\\
    \pr_1\qquad\,&:\qPi{p:\qSig{x:A}B(x)}B(\pr_0(p))\\
    \pr_1(a,b)&\defeq b.
\end{align*}
These two functions are completely specified by giving their value on constructors.
We will sometimes abuse notation writing something of the form $\lambda(a,b).t(a,b)$ to
mean the function defined by induction with value $t(a,b)$ on the pair $(a,b)$.
As with $\Pi$, we take $\Sigma$ to bind as far right as can be made sense of. We also
treat $(-,-)$ as right associative, so $(a,b,c,d)$ means $(a,(b,(c,d)))$ and is an
element of the type
\[\qSig{a:A}\qSig{b:B(a)}\qSig{c:C(a,b)}D(a,b,c),\]
while $((a,b),c)$ is an element of the type
\[\qSig{p:\qSig{x:A}B(x)}C(\pr_0(p),\pr_1(p)).\]
or less formally, of the type
\[\qSig{(x,y):\qSig{x:A}B(x)}C(x,y),\]
However, as expected, pairing is associative (up to equivalence), and we will be
sloppy with our pairing notation when doing so aids readability.

Additionally, for each $A:\univ$, we have the inductive family
\defineopfor{$\Id$}{Id}{type, identity}$\Id_A:A\to A\to\univ$ of
\nameas{identity types.}{type!identity}
\index{equality|see{type!identity}}\index{path|see{type!indentity}}
\definesymbolfor{$=$}{equals}{type, identity}
We will usually write $a=b$ or $a=_{A}b$ for $\Id_A(a,b)$, and
we call elements of $a=b$ \emph{paths} or \emph{equalities}. The introduction rule
gives $\refl[a]:a=a$ and the elimination rule, called \emph{path induction}, is as
follows:

Given
\begin{itemize}
    \item $C:\qPi{a,b:A}{a=b\to \univ}$, and
    \item $c:\qPi{a:A}{C(a,a,\refl[a])}$;
\end{itemize}
we have an inhabitant 
\[f:\qPi{a,b:A}{\qPi{p:a=b}{C(a,b,p)}}.\]
such that
\[f(\refl[a]) \jeq c(a).\]
The element expressing the elimination rule is traditionally called $J$:
\[J_{A,C} : \Big(\qPi{a:A}C(a,a,\refl)\Big)\to \qPi{a,b:A}\qPi{p:a=b}C(a,b,p).\]

Note that here we have used a \emph{judgemental} equality, which is also assumed for
the inductive types given above. However, as we work informally, we will be vague
about the difference between judgemental equality and the type-level
equalities arising from identity types. In short, we only work up to type-level
equality.

There is also a \emph{based path induction} principle which says that for fixed $a:A$
and given
\begin{itemize}
    \item $C:\qPi{b:A}{a=b\to \univ}$, with
    \item $c:C(a,\refl[a])$;
\end{itemize}
we have an inhabitant
\[f:\qPi{b:A}{\qPi{p:a=b}{C(b,p)}}\]
such that
\[f(\refl[a]) \jeq c(a).\]
In practice, we usually present proofs by path induction by assuming $p$ is
$\refl$. We prove the following lemma first in a more rigorous style using $J$, and
then in the more intuitive style we use subsequently.
\begin{lemma}\label{lemma:path-concat}
    Fix a type $A:\univ$.
    \begin{enumerate}[label=(\roman*)]
        \item For any $x,y:A$ and $p:x=y$, there is $p^{-1}:y=x$.
        \item For any $x,y,z:A$ any $p:x=y$ and $q:y=z$, we have $p\ct q:x=z$ such
            that each of the following equations hold for all $p:x=y$:
            \begin{align*}
                p\ct\refl[y] &{}= p,\\
                \refl[x]\ct q &{}= q,\\
                p\ct p^{-1} &{}= \refl[x],\\
                p^{-1}\ct p &{}= \refl[y].\\
            \end{align*}
    \end{enumerate}
\end{lemma}
\begin{proof}
(Using $J$.)

    \begin{enumerate}[label=(\roman*)]
        \item Define $C:\qPi{x,y:A}(x=y)\to\univ$ by
            \[C(x,y,p) \defeq y=x,\]
            and note that for any $a:A$ we have $C(a,a,\refl[a]) \jeq (a = a)$, so then
            define $c:\qPi{a:A}C(a,a,\refl[a])$ by $c(a) = \refl[a]$. Then letting $f\defeq
            J_{A,C}(c)$ we have
            \[f : \qPi{x,y:A}(x=y)\to(x = y),\]
            satisfying the equation
            \[f(a,a,\refl[a]) \jeq \refl[a].\]
            Now define $p^{-1} = f(p)$.

        \item Define $C:\qPi{x,y:A}(x=y)\to\univ$ by
            \[C(x,y,p) \defeq \qPi{z:Z}(y=z) \to (x=z).\]
            Note that for any $a:A$ we have $C(a,a,\refl[a]) \jeq
            \qPi{z:Z}(a=z)\to(a=z)$, so define $c:\qPi{a:A}C(a,a,\refl[a])$ by $c(a)
            = \lambda z,p. p$. Let $f\defeq J(c)$ and for $p:x=y$ and $q:y=z$, define
            $p\ct q \defeq f(x,y,p)(z,q)$.

            That $\refl[x]\ct q = q$ follows by definition. For the other two
            equations, we need to again use the elimination principle.

            We need to check the equations. Define $C,C',D,D':\qPi{x,y:A}(x=y)\to\univ$ as
            \begin{align*}
                C(x,y,p)\,\, & \defeq p\ct\refl[y] = p\\
                C'(x,y,p) & \defeq \refl[x]\ct p = p\\
                D(x,y,p)\,\, & \defeq p \ct p^{-1} = \refl[x]\\
                D'(x,y,p) & \defeq p^{-1} \ct p = \refl[y]\\
            \end{align*}
            Then we have
            \begin{eqnarray*}
                c\,\,&:&\qPi{a:A}C(a,a,\refl[a])\\
                c_a &\defeq& \refl[{\refl[a]}],\\
                c' &:&\qPi{a:A}C'(a,a,\refl[{\refl[a]}])\\
                c_a' &\defeq& \refl[{\refl[a]}]\\
                d\,\, &: &\qPi{a:A}D(a,a,\refl[a])\\
                d_a &\defeq& \refl[{\refl[a]}]\\
                d_a'&:&\qPi{a:A}D'(a,a,\refl[a])\\
                d_a' &\defeq& \refl[{\refl[a]}]
            \end{eqnarray*}
            Then we have 
            \begin{align*}
                J_{A,C}(c)\,\,&: \qPi{x,y:A}\qPi{p:x=y}p\ct\refl[y] = p\\
                J_{A,C'}(c')&: \qPi{x,y:A}\qPi{p:x=y}\refl[x]\ct p = p\\
                J_{A,D}(d)\,\,&: \qPi{x,y:A}\qPi{p:x=y}p \ct p^{-1} = \refl[x]\\
                J_{A,D'}(d')&: \qPi{x,y:A}\qPi{p:x=y}p^{-1} \ct p = \refl[y]\qedhere\\
            \end{align*}
    \end{enumerate}
\end{proof}
\begin{proof} (By reduction to the case that $x$ and $y$ are the same and $p$ is $\refl$.)
    \begin{enumerate}[label=(\roman*)]
        \item Let $x\jeq y$ and $p$ be $\refl$. Then define $\refl^{-1} \defeq \refl$.
        \item Let $x\jeq y$ and $p$ be $\refl$. Then define $\refl\ct q \defeq q$,
            so that we have $\refl\ct q \jeq q$.

            We also have $\refl\ct \refl^{-1} \jeq \refl\ct\refl \jeq \refl$. Now
            suppose $q$ is $\refl$, so that we have a path $\refl[x]\ct \refl[y] \jeq
            \refl[x]$ and so for any $p$, we have $p\ct \refl[y] = p$. \qedhere
    \end{enumerate}
\end{proof}

\section{Bishop-style mathematics and the Curry-Howard
interpretation}\label{section:bishop-mathematics}
\index{Curry-Howard interpretation}
Type theory arose from an intuitionistic perspective, wherein mathematical
constructions are at least as basic as logical operations, in contrast to a
first-order theory such as ZFC, which builds mathematics on top of logical deduction.
In other words, to express logic in our language, logical notions must be encoded
using types.
This is traditionally done via the \emph{Curry-Howard} interpretation of logic: we
interpret propositions as types, and take a proof of the proposition $P$ to be the
construction of an element $p:P$.  Under this interpretation, the logical operations
are
\begin{align*}
    P\wedge Q &\jeq P\times Q\\
    P\vee Q &\jeq P + Q\\
    P\Rightarrow Q &\jeq P\to Q\\
    \neg P &\jeq P\to \zerotype \\
    \qAll{a:A}{P(a)} &\jeq \qPi{a:A}{P(a)}\\
    \qExists{a:A}{P(a)} &\jeq \qSig{a:A}{P(a)}.
\end{align*}
For example, the statement 
\[\forall (x\in \nat).x\neq 0\to \exists (y\in \nat).x=\succop (y)\]
becomes the type
\[ \qPi{x:\nat}\big(x= 0\to\zerotype)\to \qSig{y:\nat}x = \succop(y),\]
and this type has an element, by induction on $x$:

In the base case, we want to see
$(0=0\to\zerotype)\to \qSig{y:\nat}x=\succop(y)$. Let $f:0=0\to\zerotype$; as we have
$\refl:0=0$, we then get $f(\refl):\zerotype$, and so by the elimination principle
for $\zerotype$ (i.e., the unique function from $\zerotype$ into any other type), we have
$!(f(\refl)):\qSig{y:\nat}x=\succop(y)$.

For the inductive step, we know $x$ is of the form $\succop(n)$. Then 
\[\lambda f.(n,\refl[\succop(n)]) :\big(x= 0\to\zerotype)\to \qSig{y:\nat}x = \succop(y).\]
Using the induction principle then gives us an element of the desired type---a proof of
the fact that every non-zero natural number is a successor.

The Curry-Howard interpretation allows the interpretation of logical concepts which
matches the constructivist explanation of the logical connectives. However, it
quickly becomes apparent that many \emph{extensional
concepts}~\cite{Hofmann1995}---concepts specifying the behavior of equality---should
be added to MLTT in order to work with MLTT effectively. Four noteworthy extensional
concepts are \emph{function extensionality},
saying that pointwise equal functions are equal; proposition extensionality, saying
that logically equivalent propositions are equal; proof irrelevance, saying that
any two proofs of the same proposition are equal; and quotient types, identifying
related elements of a type. The traditional way to resolve these difficulties is to
equip types with equivalence relations on demand, instead of using the identity types
introduced earlier. This matches Bishop's view that to give a set is to give not only
its members, but also what we must show to show two members to be equal. A type equipped
with an equivalence relation is called a \emph{setoid}. Given two setoids $(A,\seq)$
and $(B,\seq)$, there is no reason to expect an arbitrary $f:A\to B$ to respect the
equivalence relation. It is typical to therefore call elements of $A\to B$
\emph{operations} from $A$ to $B$, and call an operation $f:A\to B$ a \emph{function}
when for all $a,a':A$ if $a\seq a'$ the $f(a)\seq f(a')$.

When working with setoids, we can force function extensionality to hold by equipping
the type of functions with the equivalence relation of pointwise equality. Similarly, 
we may introduce quotients by simply adjusting the equivalence relation on the
setoid: if $(A,\seq)$ is a setoid, and $E$ as an equivalence relation on $A$ which
is coarser than $\seq$, we can set $A/E$ to be the setoid $(A,E)$.
More subtly, we may consider a class of proof-irrelevant setoids: say that $(A,\seq)$
is \emph{proof-irrelevant} when $\seq$ is the chaotic relation---that is, when $a\seq
a'$ for any~$a,a':A$.

We will say more about setoids in Chapter~\ref{chapter:setoids}, after we examine the
extensional concepts we are interested in from the univalent perspective. For now, it
is worth remarking that it is some work to equip $\Sigma$ and
$\Pi$ types with setoids (a notion of substitution is required for this), and that it
is not clear whether there is a reasonable way to make the universe into a setoid. As
the universe is central to our approach to partial functions in
Chapter~\ref{chapter:partial-functions}, the setoid approach will not work for our
purposes.

\section{Informal mathematics and constructive taboos}\label{section:taboos}
\index{constructive!taboo}
The logic we get in a Bishop-style mathematics and the logic we get in the univalent
mathematics we will use are both constructive. As a result, we do not have the law of
the excluded middle. A consequence is that the possibility of an independent or
\emph{constructively undecided} proposition is always in the air. Discussion around these
principles can be misleading at times.  Consider the following two
sentences
\begin{itemize}
    \item ``The type of natural numbers is discrete.''
    \item ``The type of real numbers cannot be discrete.''
\end{itemize}
While we do not yet know the definition of discreteness, it is enough to know that
if LEM holds, then  all sets are discrete. Regardless, the first is a theorem in our
system (as a consequence of Theorem~\ref{thm:nat-is-set}) and the second is an
\emph{informal metatheorem}.
We could make this into an honest metatheorem by (e.g.) proving that there is no
proof that the real numbers are discrete, but we want to avoid a deep
detour into metamathematics. Such informal statements are
often justified in constructive math by appeal to \emph{constructive taboos}: statements
which we expect not to be provable in a constructive system. Statements of the second
form are explanations of results of the form $P\to Q$, where $Q$ is a constructive
taboo. In particular, if the type of real numbers has decidable
equality, then we can derive Bishop's \emph{limited principle of omniscience}, which
we can write in a logical style as
\[\LPO\defeq \forall(\alpha: \nat\to 2),\big(\qExists{n:\nat}\alpha(n)=1)\big)\vee
\neg\big(\qExists{n:\nat}\alpha(n) = 1\big).\]
As the limited principle of omniscience says that we can examine infinitely many
cases, it does not meet Bishop's requirement that we are describing ``finitely
performable abstract operations'', and so a system which allows us to prove $\LPO$ is
by its nature non-constructive. Then we cannot expect to have the type of real numbers
to be discrete in a constructive system.

Some constructive systems arise from considerations in analysis or computability
theory~\cite{Bridges1987Varieties}, and often contain axioms that are incompatible
with the law of the excluded middle. In this thesis, such principles will also be
called taboos. Below we list
several taboos that we will make mention of throughout. The list contains an informal
statement of each, a brief explanation of the reasons the principle might be
considered, and a formal logical statement, and the same statement in the
Curry-Howard interpretation.
\begin{description}
        \indexandabbr{Limited Principle of Omniscience}{LPO}
    \item[Limited principle of omniscience (LPO)] Every binary sequence either takes
        the value 1, or doesn't take the value 1.
        First introduced by Bishop~\cite{bishop1967} as a principle which should not
        be constructively valid.
        \[\qAll{\alpha: \nat\to 2}\big(\qExists{n:\nat}\alpha(n)=1)\big)\vee
        \neg\big(\qExists{n:\nat}\alpha(n) = 1\big).\]
        \[\qPi{\alpha: \nat\to 2}\big(\qSig{n:\nat}\alpha(n)=1)\big)+
        \neg\big(\qSig{n:\nat}\alpha(n) = 1\big).\]
        \indexandabbr{Limited Principle of Omniscience!Weak}{WLPO}
    \item[Weak limited principle of omniscience (WLPO)] Every binary sequence is
        either constantly 0 or not constantly 0. LPO weakened by replacing
        $Q$ by $\neg\neg Q$, where $Q$ is $\qExists{n:\nat}\alpha(n)=1$. Of interest
        when considering computability, as here it should fail
        (see Section~\ref{section:comp-functions}).
        \[\qAll{\alpha:\nat\to 2}\big(\qAll{n:\nat}\alpha(n)=0)\big)\vee
        \neg\big(\qAll{n:\nat}\alpha(n) = 0\big).\]
        \[\qPi{\alpha:\nat\to 2}\big(\qPi{n:\nat}\alpha(n)=0)\big)\vee
        \neg\big(\qPi{n:\nat}\alpha(n) = 0\big).\]
        \indexandabbr{Markov's Principle}{MP}
    \item[Markov's Principle (MP)] If a binary sequence is not constantly 0, then it
        takes the value 1. Introduced by Markov~\cite{markov1962} based on considerations in recursion
        theory.
        \[\qAll{\alpha:\nat\to 2}\big(\neg\qAll{n:\nat}\alpha(n)=0)\big)\to
        \big(\qExists{n:\nat}\alpha(n) = 1\big).\]
        \[\qPi{\alpha:\nat\to 2}\big(\neg\qPi{n:\nat}\alpha(n)=0)\big)\to
        \big(\qSig{n:\nat}\alpha(n) = 1\big).\]
        \indexandabbr{Kripke's Schema}{KS}
    \item[Kripke's Schema (KS)] For each proposition $P$, there is a binary sequence
        such that $P$ holds iff this sequence takes the value 1. Introduced
        to formalize Brouwer's concept of the \emph{creative subject}. See Section
        4.10 of~\cite{troelstra1988constructivism} for some history. Allowing
        quantification over propositions, this is written
        \[\qAll{P:\Prop}\qExists{\alpha:\nat\to 2}\big(P\leftrightarrow
        (\qExists{n:\nat}\alpha(n)=1)\big).\]
        Under Curry-Howard, we take $\Prop$ to be $\univ$:
        \[\qPi{P:\univ}\qSig{\alpha:\nat\to 2}\big(P\leftrightarrow
        (\qSig{n:\nat}\alpha(n)=1)\big).\]
        \indexandabbr{Church's Thesis}{CT}
    \item[Church's Thesis (CT)] Every total function $\nat\to\nat$ is computable.
        Arises from computability considerations, as all constructively definable
        functions ought to be computable.
        \[\qAll{f:\nat\to\nat}\qExists{e:\nat}\qAll{n:\nat}f(n)=\{e\}(n).\]
        \[\qPi{f:\nat\to\nat}\qSig{e:\nat}\qPi{n:\nat}f(n)=\{e\}(n).\]
        Here $\{e\}$ is the function coded by $e$.\label{def:ct}
\end{description}
The Curry-Howard statements of both Kripke's Schema and Church's thesis are false,
even though Kripke's Schema follows from LEM, and there are constructive settings in
which Church's thesis holds~\cite{hyland1982effective}. We will introduce a more
refined interpretation of logical principles in Section~\ref{section:unilogic}. Using
this interpretation, we consider a more reasonable version of Church's thesis in
Section~\ref{computability:church}. A version of Kripke's Schema which seems to
better capture the intent of the schema can be given using ideas from
Chapter~\ref{chapter:partial-functions}.

\section{Topos logic}\label{section:topos-logic}
It is possible to interpret Martin-L{\"o}f type theory in a topos $\cE$, following
Hoffman's modification of Seely's interpretation in any locally cartesian closed
category~\cite{seely1984lccc,hofmann1995lccc}. Briefly, we take the objects of $\cE$ to represent
contexts of free variables, and a type in context $\Gamma$ is a map into $\Gamma$ (an object of
$\cE/\Gamma$), while an element of that type is a section of this map. Any map
$f:\Delta\to\Gamma$ gives rise to a function $f^*:\cE/\Gamma\to\cE/\Delta$ by
pullback. The left-adjoint of this is used to interpret $\Sigma$, while the
right adjoint is used to interpret~$\Pi$.

The interpretation outlined above gives a way to interpret informal mathematics in a
topos, by appealing to the Curry-Howard interpretation in the type theory associated
to the topos. However, it is more common to interpret informal mathematics in a topos
by appealing to the \emph{Mitchell-Benabou} language of a topos. The Mitchell-Benabou
language for a topos $\cE$ is a higher-order logic which arises by considering the
subobject classifier $\Omega$ of $\cE$. We roughly follow the presentation
in~\cite{maclane1994sheaves}. A more detailed account can be found there, as well as
in \cite{lambek1986higher} and \cite{mclarty1992elementary}.

For each object $X$ of $\cE$, there is a \emph{type} $X$, and each type has a
set of \emph{terms}. Each term in the language additionally has
free variables of fixed types. A term of type $B$ with free
variables of types $A_1,\ldots,A_n$ will be written
\[x_1:A_1\ldots,x_n:A_n\vdash t:B.\]
A term $x_1:A_1\ldots,x_n:A_n\vdash t:B$ will be interpret by a morphism
$A_1\times\cdots\times A_n\to B$.
The basic data
of the LCCC structure of $\cE$ give terms of the language corresponding to the
morphisms given by the data. We omit these, except for
those given by exponentiation which are illustrative: Given a term $x:U\vdash
\sigma:X$ interpreted by $\chi:U\to X$
and a term $v:V\vdash f:Y^X$ whose interpretation is given by $\beta:V\to Y^X$, there
is a term $x:U,y:V\vdash f(v):Y$ whose
interpretation is given by 
\[V\times U\xrightarrow{\lar{\chi,\beta}} X\times Y^X \xrightarrow{e} Y,\]
where $e$ is the evaluation map.
Given a term $x:X,u:U\vdash Z$ which is interpreted by $\theta$, there is a term
$x:X\vdash \lambda x.\theta:Z^U$ whose interpretation is the transpose of $\theta$.

The terms of interest to us here are those whose type is $\Omega$, which we call
\emph{formulas}. As there are maps $\wedge,\vee,\Rightarrow:\Omega\to\Omega$ and a
map $\neg:\Omega\to\Omega$ corresponding to the logical connectives, we can interpret
the logical connectives on formulas. What remains is interpreting equality and
quantifiers. As the diagonal $\Delta:X\to X\times X$ is an embedding, it has a
characteristic function $\delta:X\times X\to\Omega$, and so given terms $u:U\vdash
\sigma:X$ and $v:V\vdash \tau:X$, we have a formula $u:U,v:V\vdash (\sigma=\tau)
:\Omega$ interpreted by composing $\lar{\sigma,\tau}$ with $\delta$.

The quantifiers are slightly more work. The formula $x:X,y:Y\vdash \varphi:\Omega$
is interpreted as a map $X\times Y\to\Omega$, which corresponds to an arrow
$\lambda x.\varphi:Y\to\Omega^X$. The unique map $X:X\to 1$ gives rise to a map $1^{\Omega}\to
X^{\Omega}$, which corresponds to the functor $(!_X)^{*}:\cE/1\to \cE/X$, and we know
that $1^{\Omega}=\Omega$. The left and right adjoints give us maps
$\forall_X,\exists_X:\Omega\to \Omega^X$. Then we have maps $\forall_X\comp \lambda
x.\varphi:Y\to\Omega$ and $\exists_X\comp\lambda.x\varphi:Y\to\Omega$, which we take
as the interpretations of $\qAll{x:X}\varphi(x,y)$ and $\qExists{x:X}\varphi(x,y)$.

As a formula is treated as a map $X\to\Omega$, any formula $\varphi$ corresponds to a
subobject of $X$, and we can write $\{x\mid\varphi(x)\}$ for the domain of this
subtype. In short, we have interpreted equality, the logical connectives,
quantifiers, and a form of comprehension. Then we can use this to interpret a great
deal of mathematical work. Moreover, the more structure our topos has, the more
mathematics we can interpret: for example, if the topos contains a natural numbers
object, we can interpret the natural numbers, and arithmetic in the topos.

The interpretation of MLTT and the Mitchell-Benabou language of a topos are not
entirely independent: give an object $X$ of $\cE$, there is the formula (with no free
variables) $\trunc{X}\defeq\exists x.\top$ which we can call the \emph{truncation} of
$X$. If $P:X\to\Omega$, we can form not only
the formula $\qExists{x:X}P(x)$, but also the type $\qSig{p:P}(X)$, and in fact we have
an isomorphism between the truncation $\trunc{\qSig{p:P}(X)}$ and
$\qExists{x:X}P(x)$. In other words, in a topos,
existential quantification can be defined from $\Sigma$ and truncation. In univalent
mathematics, we will do the same thing: since we are working in a type theory, we
have $\Sigma$, and we will then define a \emph{truncation operator} (See
Section~\ref{section:unilogic}), which we will use to define the existential
quantifier. This gives way to a third, more nuanced interpretation of mathematical
ideas where we are allowed to mix higher-order logic and type theory in a fluid way.
We call mathematics done in this style \emph{univalent mathematics}.

\section{Univalent foundations}\label{section:utt}
\index{univalent mathematics|(}

In Lemma~\ref{lemma:path-concat}, we proved that we can compose and invert the
elements of identity types. When we include the reflexivity path $\refl$, we can see
this result as expressing that identity is an equivalence relation; as the
elimination principle can be seen as saying that identity is the smallest
reflexive relation, this also makes identity into the smallest equivalence relation.
However, $x=y$ is not simply a proposition, but a \emph{type}---to assert $x=y$ is
not simply to give a truth value, but to give the data of \emph{which} inhabitant of
$x=y$ we have. This is to say, identity is not simply property in Martin-L\"{o}f type
theory, but \emph{structure}.

Univalent mathematics arises from the observation that the structure of identity
types ${x=_{A}y}$ is part of the structure of the type $A$. For example, the fact (which
we prove shortly) that $\qPi{n,m:\nat} (n=m) + \neg(n=m)$ is part of the description
of the structure of $\nat$. We can then
stratify the types into levels based on how deep the identity structure on a type
goes. These levels, called \emph{homotopy-levels} capture a notion of dimension.
Traditionally, these dimensions are indexed from $-2$, following the indexing in
homotopy theory; we use $\Ntwo$ for the type of integers at least $-2$. The most
important are the first four levels, those of
\emph{contractible types}\index{contractible!type|textit},
\index{type!contractible|see{contractible type}}
\emph{propositions}\index{proposition|textit}, \emph{sets}\index{set|textit} and \emph{groupoids}. In
order, these are types with exactly one element, types with at most one element,
types whose identity types are propositions, and types whose identity types are sets.
More explicitly:
\begin{definition}
    \defineopfor{$\isProp$}{isprop}{proposition}
    \defineopfor{$\Prop$}{prop}{proposition}
    \defineopfor{$\isContr$}{iscontr}{contractible type}
    \defineopfor{$\isSet$}{isset}{set}
    \defineopfor{$\Set$}{set}{set}
    We define the operations $\isProp:\univ\to\univ$, $\isContr:\univ\to\univ$, 
    $\isSet:\univ\to\univ$ and $\isGpd:\univ\to\univ$ by
    \begin{eqnarray*}
        \isContr(X) &\defeq& \qSig{a:X}{\qPi{x:X}{(a=x)}},\\
        \isProp(X) &\defeq& \qPi{x,y:X}{x=y}, \\
        \isSet(X) &\defeq& \qPi{x,y:X}{\isProp(x=y)},\\
        \isGpd(X) &\defeq& \qPi{x,y:X}{\isSet(x=y)}.
    \end{eqnarray*}
    We call a type $X$ a \emph{proposition} or a \emph{subsingleton} if $\isProp(X)$,
    \emph{contractible} or a \emph{singleton} if
    $\isContr(X)$ a \emph{set} if $\isSet(X)$ and a \emph{groupoid} if $\isGpd(X)$.
    We call a type family $P:X\to\univ$
    a \name{predicate} when $\qPi{x:X}{\isProp(P(x))}$. We also have the types
    $\Prop,\Set:\univ_1$,
    \[\Prop\defeq \qSig{P:\univ}{\isProp(P)}.\]
    \[\Set\defeq \qSig{X:\univ}{\isSet(X)}.\]
    For any $n:\Ntwo$, the \emph{homotopy $n$-types}, are given by
    \begin{eqnarray*}
        \istype{(-2)}(X) &\defeq& \isContr(X), \\
        \istype{(n+1)}(X) &\defeq& \qPi{x,y:X}{\istype{n}(x=y)}.
    \end{eqnarray*}
\end{definition}
We will be working primarily at the level of sets and propositions, but as far as
possible, we will state our results for general types. We will also make use of
examples and counterexamples which have higher levels. The question of whether a type
is a set is not a size issue, in the way that the class of all sets is too large to
be a set, but rather, it is an issue of depth of structure. Size issues are handled
by universe levels. For example, we gave the type $\Set$ as a type in $\univ[1]$.
More explicitly, we have for each $i$, a type $\Set_i:\univ[i+1]$ with $\Set_i \defeq
\qPi{X:\univ[i]}\isSet(X)$.

Note that the definition of $\isProp(P)$ and $\istype{(-1)}(P)$ are not quite the
same. We will show shortly that these two notions are in fact equivalent, and so also
that sets are $0$-types and groupoids are $1$-types. However, we do not yet have the
machinery to talk about equivalence properly.

In the meantime, we use a notion of \emph{logical equivalence}: Define
\defineop{$\lequiv$}{equivlogical}
\[A\lequiv B \defeq (A\to B) \times (B\to A)\]
and say that $A$ and $B$ are \emph{logically equivalent} if $A\lequiv B$. The name is
justified by Lemma~\ref{lemma:lequiv-equiv}, which says that when two propositions are
logically equivalent, they are in fact equivalent (in the sense defined in the next
section).

The canonical examples of propositions are $\zerotype$ and $\unittype$.
\begin{theorem}
    The empty type $\zerotype$ and unit type $\unittype$ are propositions. The unit type is
    contractible, while the empty type is not.
\end{theorem}
\begin{proof}
    The elimination principle for $\zerotype$ says directly that
    $\qPi{x,y:\zerotype}x=y$. Moreover, the first projection
    $\Big(\qSig{x:\zerotype}\qPi{y:\zerotype}x=y\Big)\to \zerotype$ gives a
    map $\isContr(\zerotype)\to\zerotype$.

    As we have $\refl:\star=\star$, by the elimination
    principle for the unit type we have some witness $w:\qPi{x:\unittype}\star=x$, and so
    $(\star,w):\isContr(\unittype)$. However, again by applying the elimination principle for
    $\unittype$ to $w$, we get $\qPi{x,y:\unittype}y=x$, by symmetry we have
    $\isProp(\unittype)$.
\end{proof}
The type $\nat$ is a set which is not a proposition.
The proof relies on the notion of equivalence, so we give
it in the next section.

When defining an algebraic structure, for example, a group, it is most natural to
work with underlying types which are \emph{sets}, since we want the group equations
to be property, with the structure given only by the operations. Regardless, we may define
$\GrpStruct:\univ\to\univ$ by
\begin{align*}
    \GrpStruct(X) \defeq {} & \qSig{-\cdot-:X\to X\to X}\qSig{(-)^{-1}:X\to
X}\qSig{e:X} \\
    & \qPi{x,y,z:X}\big(x\cdot( y\cdot z) = (x\cdot y)\cdot z\big)\\
    {}\times{} & \qPi{x:X}\big((x\cdot e = x) \times (e\cdot x = x)\times
(x\cdot x^{-1} = e)\times (x^{-1}\cdot x = e)\big),
\end{align*}
The first line says that we have the operations $-\cdot-$ and $(-)^{-1}$, as well as
an identity $e$; the second line ensures associativity, and the third the identity
and inverse rules. This is a direct formalization of the classical definition:

\vspace{15pt}
\hfill\begin{minipage}{0.95\textwidth}
    A \emph{group structure} on $X$ consists of a multiplication operation
    $-\cdot-:X^2\to X$, an inverse operation $(-)^{-1}:X\to X$ and an element
    $e\in X$ satisfying the equations
\begin{align*}
    (\forall x,y,z\in X)\quad x\cdot (y\cdot z) ={}& (x\cdot y)\cdot z\\
    (\forall x\in X)\qquad\quad\,  x\cdot e ={}& e \cdot x = x\\
    (\forall x\in X)\qquad  x\cdot x^{-1} ={}& e = x^{-1} \cdot x
\end{align*}
\end{minipage}\hfill
\vspace{15pt}

All of the operations on types are such formalizations. For example, a type $P$ is a
proposition if for all $x,y:P$, we have that $x=y$.

In the case that $X$ is a set, the required equations can be shown to be
propositions---that is, we can show that the equations are property.
Then we can define
$\Grp:\univ[1]$:
\[\Grp \defeq \qPi{X:\univ}\isSet(X)\times \GrpStruct(X).\]

If we assume Voevodsky's univalence axiom (see Section~\ref{section:univalence}),
we can show that equality in $\Grp$ is exactly
isomorphism~\cite{coquand2013isomorphism}. As we may have many non-trivial
isomorphisms between two groups, we know that we cannot have that $\Grp$ is a set. In
fact in this case we can show that $\Grp$ is a groupoid.

\section{Structure versus Property}\label{section:struct-vs-prop}
Before moving on, we should further discuss structures and properties. In a
traditional development of semi-formal mathematics, we establish the truth (or non-truth) of
statements. For example, one may express the first isomorphism theorem in such a development as

\vspace{5pt}
\hfill\begin{minipage}{0.95\textwidth}
    \textit{
    For any surjective homomorphism $\varphi : G\to H$, there is an isomorphism
    $G/\ker\varphi \simeq H$.
}
\end{minipage}\hfill
\vspace{5pt}

This is a perfectly honest theorem, but it asserts only a simple fact: it tells us that
we can identify $G/\ker\varphi$ and $H$. This fact gives us no information about what
happens to a given $g\in G$ under this identification, despite the fact that when
using the isomorphism theorem, we are often interested in a particular
isomorphism between $G/\ker\varphi$ and $H$. A better version of the theorem is
given by

\vspace{15pt}
\hfill\begin{minipage}{0.95\textwidth}
    \textit{
    For any surjective homomorphism $\varphi : G\to H$, the natural function
    $\nu:G/\ker\varphi \to H$
    defined by $\nu([g]) = \varphi(g)$
    is an isomorphism.
}
\end{minipage}\hfill
\vspace{15pt}

This version contains \emph{structure}: we not only know that $G/\ker\varphi$ can be
identified with $H$, but we know which map gives this identification, and how it
behaves. We can think of an isomorphism as a structured identification between
groups, so that in rephrasing the theorem, we have moved from a simple statement of
fact to a description of structure. On the other hand the property that $f$ is an
isomorphism between groups is something we usually do not wish to analyze for deeper
structure: the fact that $f$ is an isomorphism is all we are interested in.

In a formalization based on first order (or higher order) logic, we manipulate facts;
structure is not easily expressed by the formal manipulations. If we want the formal
framework underlying our mathematics to be able to manipulate structure directly,
then another sort of system is needed.
The obvious choice is to turn to some variant of MLTT. However, logic is a part of
mathematical reasoning, and MLTT does not natively have a notion of truth value,
so some convention on how to interpret logical operations must be taken. Propositions
can be represented by types, and then we say that a proposition $P$ is true if there
is some inhabitant $p:P$. The traditional convention is that of
Section~\ref{section:bishop-mathematics}. Namely, view \emph{all} types as
propositions and use a Curry-Howard interpretation of logic. But then we lose our
simple facts: even $3+4=7$ is (a priori) a type with structure which can be analyzed.

Our example above suggests that mathematics is concerned with both facts and 
structures. Traditional foundations often only directly express the former, while type
theoretic foundations often only directly express the later. Univalent definitions allow us
to talk about both, and to relate them. Arbitrary types represent structure or
data. Propositions represent simple facts---the definition tells us that the
structure of the type is trivial. In principle, we can give two versions
of a notion: a propositional version (expressing a simple fact), and a
structured version (giving data). For example, we can talk of a function $f$ having the
\emph{property} of being computable (which is a simple fact), or we can talk about
\emph{computation structures} for $f$: programs computing $f$.
Chapters~\ref{chapter:comp-as-struct} and~\ref{chapter:comp-as-prop} deal with these
notions.

When we switch to a univalent approach, we are not abandoning the Curry-Howard
interpretation, but augmenting it: we still have $\Sigma$ and $+$, allowing
us to express a logic of structures, but we also have a notion of proposition, for
which we can introduce a logic of propositions, as we do in
Section~\ref{section:unilogic}. This allows a richer expression of ideas, as we get
to choose in our definitions and theorems whether to use structured or propositional
notions. Often, one choice is more natural than the other.

As being a proposition is a property of types, we can introduce a type, and then
prove after the fact that it is a proposition. There are two benefits of this:
it gives us concise and clear language for expressing that structure is
trivial, and more usefully, it gives us tools for turning simple facts into
interesting structure. For an example of this latter situation, see
Lemma~\ref{lemma:untruncate-decidable-predicates}, where we show an example where the
truth of a proposition can allow us to compute a natural number.

The distinction between structure and property guides formalization: for example, we
can ask for a \emph{monoid structure} on a type:

\vspace{15pt}
\hfill\begin{minipage}{0.95\textwidth}
    A \emph{monoid structure} on $X$ consists of a multiplication operation
    $-\cdot-:X^2\to X$ and a unit $e:X$, satisfying the equations
\begin{align*}
    (\forall x,y,z\in X)\quad x\cdot (y\cdot z) ={}& (x\cdot y)\cdot z\\
    (\forall x\in X)\qquad\quad\,  x\cdot e ={}& e \cdot x = x\\
\end{align*}
\end{minipage}\hfill

Then we have the type family $\MonStruct:\univ\to\univ$.  Using ideas from
Chapter~\ref{chapter:univalence}, we can also discuss the property of \emph{being a
monoid} $\isMonoid:\univ\to\Prop$ using truncation, but this is of little use: there may be
many possible monoid structures on a type $M$, but a witness of $\isMonoid(M)$ cannot
distinguish between these: the property of being a monoid does not allow us to access
an actual monoid structure, in general. On the other hand, given a monoid
$(M,\cdot,e,-):\Mon$---namely, a set equipped with a monoid structure---and an
element $m:M$, we can define the type of inverses of an element
\[\inverse_M(m) : \qSig{n:M}(m\cdot n = e) \times (n\cdot m = e),\]
and while $\inverse:M\to\univ$ is a priori structure, we can in fact prove that
$\inverse(m)$ is a proposition for all $m:M$. That is, we can show that the type of
inverses is a proposition. In our case, this means that asserting the existence of an
inverse is the same as giving an explicit inverse; or even more evocatively: inverse
are uniquely specified. As we can also show that products over propositions are again
propositions, we have a \emph{predicate} $\isGrp:\Mon\to\Prop$, saying whether a
monoid is a group. So while a group is structure imposed on a set, the
structure is already imposed by the underlying monoid (in fact, already by the
underlying semigroup).

The ability to infer that a type is a proposition distinguishes univalent mathematics
from, for example, the calculus of (inductive) constructions~\cite{CoC1986,CoC1988},
which also has a universe of propositions. However, propositions are completely separate
from other types in CiC, and no structural content can be inferred from them.

Unfortunately, pure MLTT does not give us everything we need to make use of univalent
definitions properly. In particular, we cannot form a proposition-valued version of
$\vee$ or $\exists$, and we cannot show that $\neg$ and $\forall$ preserve propositionhood. We
will show how to extend MLTT with a truncation operator (resolving $\vee$ and
$\exists$) and a principle of \emph{function extensionality} (to resolve the problem
with $\neg$ and $\forall$) as we develop the univalent
approach in Chapter~\ref{chapter:univalence}.

Properties---proposition-valued type families---play a central role in any work in
a univalent framework. In Chapter~\ref{chapter:univalence} we will provide
additional technical justification for the claim that propositions should be viewed
as types with trivial structure, but we will make use of the fact in our exposition
from here onward. We will see in the next section that we have the
type-theoretic analog of the Leibniz law, that identical objects are indiscernible.
Then if $P$ is some proposition, all its elements will be indiscernible.

Consequently, once we know that a type $P$ is a proposition, we will often leave
elements of $P$ unnamed. Moreover, we will ignore the behavior of constructions on
propositions (see, for example, the definition of the lifting monad in
Chapter~\ref{chapter:partial-functions}). Finally, we will stick to a rigid naming
convention: If we give a type family $A\to \univ$ a name beginning with
$\operatorname{is}$, we will (at some point) prove that this type family is
proposition-valued, and we will also have capital letters somewhere in the name---for
example $\isContr,\isProp,\isSet,\isGpd$. Moreover, after proving
$\operatorname{isXYZ}:A\to\univ$ to be proposition-valued, when we need to name an
element of $\operatorname{isXYZ}(a)$, we will often use the
lowercase~$\operatorname{isxyz}:\operatorname{isXYZ}(a)$.
\index{univalent mathematics|)}

\section{Homotopies and equivalences}\label{section:homotopies}
The development of univalent mathematics relies on tools for comparing functions
and for comparing types which we explain here. The tools are those of \emph{homotopy}
between functions and \emph{equivalence} between types. The logical equivalence
defined earlier is sufficient when we are dealing with properties, but when comparing
structures, it is insufficient: for example, the unit type is logically equivalent to
the natural numbers. One of the first insights leading to univalent mathematics was
the definition of \emph{equivalence} in Voevodsky's Foundations
library~\cite{UniMath}, which refines the traditional notion of isomorphism. Since
then, several equivalent notions have been given, which are the focus of Chapter 4 of
the HoTT Book.  We are
interested in three notions, which they call \emph{quasi-invertibility} (Definition 2.4.6),
\emph{bi-invertibility} (Definition 4.3.1) and \emph{contractibility} (Definition
4.4.1). The notion of quasi-invertibility is problematic; we will discuss why after
formally introducing these notions. We diverge slightly from the HoTT Book in our
terminology. In particular we call bi-invertibility \emph{equivalence}, and a
\emph{quasi-inverse} of a function, simply an \emph{inverse}. To express these
notions, we need the notion of \emph{homotopy}, or \emph{pointwise equality}.

\begin{definition}
    \index{homotopy|textit}
    \definesymbolfor{$\htpy$}{eqhtpy}{homotopy}
    A \emph{homotopy} between two functions $f,g:\qPi{x:X}{B(x)}$ is a pointwise
    equality between $f(x)$ and $g(x)$:
    \[f\htpy g\defeq \qPi{x:X}{f(x) = g(x)}.\]
\end{definition}
For any $f,g:\qPi{x:X}{B(x)}$ we have a map
\defineop{$\happly$}{happly}
\[\happly : (f=g) \to (f\htpy g)\]
defined by path induction with $\happly(\refl)\defeq \lam{x}\refl$.

\begin{definition}
    For $f:A\to B$, and $b:B$, the \name{fiber} of $f$ over $b$ is the type
    \defineopfor{$\fib$}{fib}{fiber}
    \[\fib_f(b)\defeq \qSig{a:A}{f(a)=b}.\]
\end{definition}
Now we may define our notions of equivalence.
\begin{definition}
    A function $f:A\to B$ is \nameas{contractible}{contractible!function} when it \emph{has contractible
    fibers}:
    \[\isContr(f) \defeq \qPi{b:B}\isContr(\fib_f(b)).\]
    A \emph{right inverse} for $f$ is a map $g:B\to A$ such that $f\comp g\htpy
    \id_B$, and a \emph{left inverse} for $f$ is a map $g:B\to A$ such that $g\comp
    f\htpy \id_A$. That is,
    \defineopfor{$\rinv$}{rinv}{section}
    \defineopfor{$\linv$}{linv}{retraction}
    \[\rinv(f) \defeq \qSig{g:B\to A}(f\comp g \htpy \id_B).\]
    \[\linv(f) \defeq \qSig{g:B\to A}(g\comp f \htpy \id_A)\]
    The map $f$ is an \name{equivalence} when it has a left inverse and a right
    inverse:
    \[\isEquiv(f) \defeq \linv(f)\times\rinv(f).\]
    \defineopfor{$\isEquiv$}{isequiv}{equivalence}
    \definesymbolfor{$\simeq$}{equiv}{equivalence}
    We write $A\simeq B$ for the type of equivalence from $A$ to
    $B$:
    \[A\simeq B \defeq \qSig{f:A\to B}\isEquiv(f).\]
    When $f:A\to B$ has a left inverse, we say that $A$ is a \emph{retract} of $B$.
    We call a function which has a left inverse a \name{section} and a function which has
    a right inverse a \name{retraction}.

    An \name{inverse} for $f$ is a map $g:B\to A$ which is both a left and a right
    inverse for $f$: 
    \defineopfor{$\inverse$}{inverse}{inverse}
    \[\inverse(f) \defeq \qSig{g:B\to A}((f\comp g)\htpy \id) \times ((g\comp
    f)\htpy \id).\]
\end{definition}
The distinction between equivalence and having an inverse is the first example of the
difference between structure and property. We can show that $\isEquiv(f)$ is a
proposition---that is, that being an equivalence is property---but this is not true
in general for the type $\inverse(f)$---that is, having an inverse is structure.
Lemma 4.1.1 of the HoTT Book shows that if $\inverse(f)$ (which they call $\qinv(f)$) has an element, where $f:A\to B$, then
\[\inverse(f)\simeq \qPi{x:A}x=x.\]
The fact that having an inverse is not a proposition means we cannot naively use
$\inverse(f)$ as a replacement for $\isEquiv(f)$. In particular, we will shortly
define a map
$\idtoequiv:A=B\to A\simeq B$, taking a path in the universe to an equivalence,
and the univalence axiom
says that $\idtoequiv$ is an equivalence, but the corresponding statement, that the
map $(X=Y)\to \qSig{f:X\to Y}\inverse(f)$ has an inverse is actually inconsistent,
since a consequence of this statement is that $\inverse(f)$ is a proposition. We will
see this in detail when we talk about the univalence axiom on
Section~\ref{section:univalence}

We will show in Section~\ref{section:equiv-contr} that having an inverse, being contractible and being
an equivalence are all logically equivalent after we develop
some tools for working with paths. In the meantime, we fulfill our promise to
justify the phrase ``logical equivalence''. We aim to use propositions to encode
logic in our system, so the information encoded directly by propositions is logical
information. This means that when we call two propositions logically equivalent, we
should be able to conclude that they are equivalent. Indeed, we have the following.
\begin{lemma}\label{lemma:lequiv-equiv}
    If $A$ and $B$ are propositions,
    then $(A\lequiv B)\to (A\simeq B)$.
\end{lemma}
\begin{proof}
    Suppose $f:A\to B$ and $g:B\to A$. We have that $g(f(a)) = a$ and
    $f(g(b))= b$, as $A$ and $B$ are propositions.
\end{proof}

The words \emph{path}, \emph{homotopy}, \emph{contractible} and \emph{equivalence} come from
homotopy theory. The modern terminology for several other important notions come also
from homotopy theory. Nevertheless, no knowledge of homotopy theory is needed to
understand our type theory. The next name coming from homotopy is \nameas{application
on paths}{application!on paths}, which expresses that functions respect equality.
Given any function $f:A\to B$ and $x,y:A$, there is a function 
\defineopfor{$\ap$}{ap}{application on paths}
\[\ap_f:x=y \to f(x) = f(y)\]
defined by path induction with 
\[\ap_f(\refl[x]) \defeq \refl[f(x)].\]
In involved computations, we may abuse notation and write $f(p)$ instead
of $\ap_f(p)$.
Application on paths tells us that functions are \emph{congruences}, or act
\emph{functorially} with respect to equality. Formally, we mean the following.
\begin{lemma}\label{lemma:ap-functorial}
    For any $f:A\to B$ we have
    \begin{enumerate}
        \item $\ap_f(\refl[x]) = \refl[f(x)]$;
        \item $\ap_f(p\ct q) = \ap_f(p)\ct \ap_f(q)$;
        \item $\ap_f(p^{-1}) = \ap_f(p)^{-1}$.
    \end{enumerate}
    where $p:x=y$ and $q:y=z$ for some $x,y,z:A$.
\end{lemma}
\begin{proof}
    The first equation is definitional, the second and third follow by path
    induction: First, fix $p\jeq q\jeq \refl[x]$, so that
    \[\ap_f(\refl \ct \refl) = \ap_f(\refl) = \refl = \refl\ct\refl = \ap_f(p)\ct
    \ap_f(q).\]
    Then, fix $p\jeq \refl[x]$ so that we have
    \[\ap_f(\refl^{-1}) = \ap_f(\refl) = \refl = \refl^{-1} =
    \ap_f(\refl)^{-1}.\qedhere\]
\end{proof}

Note that $\ap$ immediately gives us that any retract of a proposition or contractible
type is again a proposition (or contractible type).
\begin{theorem}\label{thm:prop-retract-closed}
    If $f:A\to B$ and $(g,\eta):\linv(f)$ is a left inverse of $f$, then $\isContr(B)\to\isContr(A)$
    and $\isProp(B)\to \isProp(A)$.
\end{theorem}
\begin{proof}
    Letting $c:B$ be the center of contraction of $B$ with $w:\qPi{b:B}c=b$, we have for any
    $a:A$,
    \[ g(c) \eqby{\ap_g(w(f(a))} gf(a)\eqby{\eta(a)} a,\]
    so $g(c)$ is the center of contraction of $A$.

    Similarly, if $w:\isProp(B)$ and $a,b:A$, then
    \[a \eqby{\eta(a)^{-1}} gf(a) \eqby{\ap_{gf}(w(a,b))} gf(b) \eqby{\eta(b)} b,\]
    so $A$ is a proposition.
\end{proof}
A direct corollary is that propositions (and singletons) are closed under
equivalence.
\begin{corollary}
    If $A$ is a retract of $B$ and $B$ is a proposition, then $A$ and $B$ are
    equivalent.
\end{corollary}
\begin{proof}
    As $A$ is a retract of $B$, we have $A\lequiv B$ and $A$ is a proposition. Then
    $A\simeq B$ by Lemma~\ref{lemma:lequiv-equiv}.
\end{proof}

This allows us to give the following characterization of the contractible types as
propositions with elements.
$\unittype$.
\begin{lemma}\label{lemma:contr-char}
    The following are logically equivalent for any type $A$:
    \begin{itemize}
        \item $A\simeq \unittype$; (A is equivalent to $\unittype$)
        \item $\isContr(A)$; (A is contractible)
        \item $A\times \isProp(A)$; (A is a proposition with an inhabitant)
    \end{itemize}
\end{lemma}
\begin{proof}
    If $f:A\simeq\unittype$ then we know $c\defeq f^{-1}(\star):A$, where
    $f^{-1}:\unittype\to A$ is an inverse of $f$. For $a:A$, we
    have $f(a) = \star$, so $a = f^{-1}(f(a)) = f^{-1}(\star) = c$. So then $A$ is
    contractible.

    Let $A$ be contractible, with center of contraction $c:A$. For $a,b:A$, we
    have $a=c=b$, so $A$ is a proposition.

    If $a:A$, we have a map $c_a:\unittype\to A$ given by
    $c_a(\star) \defeq a$, and we have $c_{\star}:A\to\unittype$. If $A$ is
    additionally a proposition, then we have $A\lequiv \unittype$, and so
    $A\simeq\unittype$.
\end{proof}
\begin{corollary}\label{cor:iscontr-singleton}
    All contractible types are equivalent.
\end{corollary}

For dependent function $f:\qPi{x:A}{B(x)}$ the situation is more subtle. As $f(x)$
and $f(y)$ need not even have the same type, the expression $f(x)=f(y)$ is not even
well-formed. To resolve this issue, we notice that paths give rise to functions,
called \nameas{transport functions}{transport}. For each $x,y:A$, we define by path induction 
\defineopfor{$\transport$}{transport}{transport}
\begin{align*}
    \transport^B &: {(x=y)\to B(x)\to B(y)}\\
    \transport^B &(\refl,u) \defeq  u,
\end{align*}
and then given $f:\qPi{x:A}{B(x)}$ we have again by path induction, for each $x,y:A$,
a \nameas{dependent application}{application!on paths!dependent} function,
\defineopfor{$\apd$}{apd}{application, on paths, dependent}
\begin{align*}
    \apd_f&: \qPi{p:x=y}{\transport(p,f(x)) = f(y)}\\
    \apd_f&(\refl[x]) \defeq  \refl[f(x)].
\end{align*}
We will call a path $q:\transport^{B}(p,u)=v$ a path \emph{in $B$ lying over $p$},
and write 
\[(\eqover[B]{u}{v}{p})\defeq \transport^{B}(p,u)=v,\]
dropping the subscript when the type family is clear from context. Then for $p:x=y$ we may write 
\[\apd_f(p) : \eqover{f(x)}{f(y)}{p}.\]

The spatial picture giving rise to the homotopy-theoretic names is as follows: A type family
$B:A\to\univ$ can be seen as a space lying over $A$, with $B(a)$ the fiber over $A$.
A dependent function gives a section of $B$, and using this, we can lift any path
$p:x\rightsquigarrow y$ in $A$ to a path in the total space of $B$. This path can be
decomposed using a natural
family of paths $B(x)\rightsquigarrow B(y)$ and a path in $B(y)$. This natural family
of paths is captured by $\transport$, while the path in $B(y)$ is $\apd$.

There is a similar logical picture, which is captured more directly by the above
definitions: Via $\ap_B$, an equality $p:x=y$ in $A$ gives rise to an equality
$B(x)=B(y)$ in the universe. As an equality between types allows us to identify the
types, there should be a function $B(x)\to B(y)$, which we call the
function \emph{induced by~$p$}. Indeed, we
have for any $A,B:\univ$ the map
\defineop{$\idtofun$}{idtofun}
\[\coe:(A=B)\to (A\to B)\] defined by path induction with
\[\coe(\refl[A]) \defeq \id_A.\]
Then, applying this to $\ap_B:(x=y)\to (B(x)=B(y))$, and we
have
\[\transport(p,u) = \coe(\ap_B(p),u).\]
Conversely, we may define $\coe$ from $\transport$ as 
\[\coe(p,u) \defeq \transport^{\id_{\univ}}(p,u).\]

In any case, by path induction, we can show that the map $\coe(p)$ is an equivalence
for any $p:A=B$: for $\refl[A]:A=A$, we have that the identity function 
$\id_A:A\to A$ is an equivalence, as it is its own inverse:
\[((\id_A,\lambda x.\refl),(\id_A,\lambda x.\refl)) : \isEquiv(\coe(\refl)).\]
\defineop{$\idtoequiv$}{idtoequiv}
From this and path induction we get a map $\idtoequiv : (A = B) \to (A\simeq
B)$. In other words equal types are equivalent. Moreover, viewing type families as
encoding properties, the type of transport expresses Leibniz's law---the
indescernibility of identicals. The function $\apd$ tells us that Leibniz's law
respects choice of witness.

We now have the definitions we need to show that $\nat$ is a set.
\begin{theorem}\label{thm:nat-is-set}
    The type $\nat$ is a set, and is not a proposition.
\end{theorem}
We use a modification of a technique called \emph{encode-decode} which is typically
used to characterize the type of paths in a space.
Encode-decode proofs proceed by defining a \emph{coding family}
\[\codefxn:\nat\to\nat\to\univ,\]
together with a family of \emph{encoding functions},
\[\encode:\qPi{x,y:\nat}(x=y)\to\codefxn(x,y),\]
and \emph{decoding functions}
\[\decode:\qPi{x,y:\nat}\codefxn(x,y)\to(x=y),\]
such that $\encode_{x,y}$ and $\decode_{x,y}$ are inverses.

In our case, it suffices to show that $x=y$ is a retract of $\codefxn(x,y)$ and that
$\codefxn$ is a predicate.
\begin{proof}
    Define the coding family
    \begin{align*}
        \codefxn(0,0) &\defeq \unittype; \\
        \codefxn(0,x+1) &\defeq \zerotype;\\
        \codefxn(x+1,0) &\defeq \zerotype;\\
        \codefxn(x+1,y+1) &\defeq \codefxn(x,y).\\
    \end{align*}
    In order to define the encoding functions we define a function
    $r:\qPi{x:\nat}\codefxn(x,x)$ which gives the
    encoding of the reflexivity path,
    \begin{align*}
        r_0 &\defeq \star; \\
        r_{n+1} &\defeq r_{n}.
    \end{align*}
    Then we have encoding functions
    \[\encode_{m,n}(p) \defeq \transport^{\codefxn(m,-)}(p,r_m),\]
    and decoding functions
    \begin{align*}
        \decode_{0,0} &\defeq \lambda w.\refl[0]; \\
        \decode_{0,x+1} &\defeq !:\zerotype\to(0=x+1);\\
        \decode_{x+1,0} &\defeq !:\zerotype\to(x+1=0);\\
        \decode_{x+1,y+1} &\defeq \ap_{\succop}\comp \decode_{x,y}.
    \end{align*}
    To show that $\decode\comp\encode$ is the identity, we may use path induction. We
    have
    \[\decode_{x,x}(\encode_{x,x}(\refl)) =
    \decode(\transport^{\codefxn(x,-)}(\refl,r_x)) = \decode_{x,x}(r_x).\]
    By induction on $x$, we see that $\decode(r_x) = \refl[x]$: this is by
    definition at $0$. We have that $\ap_{\succop}(\refl[x]) = \refl[x+1]$ by
    definition of $\ap$, so then for the successor case we have
    \[\decode_{x+1,x+1}(r_{x+1}) = \ap_{\succop}(\decode_{x,x}(r_{x+1})) =
    \ap_{\succop}(\decode(r_{x})) = \ap_{\succop}(\refl[x]) = \refl[x+1].\]

    To show that $\codefxn(x,y)$ is a proposition, we proceed by induction on $x$
    and~$y$. We have that $\codefxn(0,0) \jeq \unittype$ and
    $\codefxn(x,y+1)=\codefxn(x+1,y)=\zerotype$. Finally we have
    \[\codefxn(x+1,y+1)\jeq \codefxn(x,y),\] and $\codefxn(x,y)$ is a proposition by
    the inductive hypothesis.

    We have that $x=y$ is a retract of a proposition, so by
    Theorem~\ref{thm:prop-retract-closed}, $x=y$ is a proposition for any $x,y:\nat$.
     Since we also have
    \[\lambda w.w(0,1):\big(\qPi{x,y:\nat}x=y\big)\to (0=1),\]
    and that $0\neq 1$, we have $\neg\isProp(\nat)$. 
\end{proof}
In the above proof we got lucky: we defined $\encode$ using $\transport$, but in the end
we didn't need to compute the behavior of transport on anything besides the
reflexivity paths. This will not always suffice, so in the next section we discuss
some key results concerning transport and identity types.

\section{Transport, identities and the basic types}\label{section:transport}
\index{transport|(}
Identities, and transport along them, are of core importance, so we need some basic
results telling us how to compute with them. The proofs of most of these results are
direct application of path induction which we do not try to motivate here. For our
purposes, the results in this section are technical lemmas characterizing the
behavior of identity types and $\transport$. More detailed discussion (and in some
cases, more detailed proofs) can be found in Chapter 2 of the HoTT Book.

We start with an important characterization of equality in dependent sum types.
\begin{theorem}\label{thm:patheq}
    Let $B:A \to \univ$ be a type family, and let $w,w':\qSig{a:A}B(a)$. Then there is
    an equivalence
    \[ (w=w') \simeq \big(\qSig{p:\pr_0(w)=\pr_0(w')}\transport(p,\pr_1(w)) =
    \pr_1(w')\big).\]
\end{theorem}
\begin{proof}
    Let $S\defeq \qSig{x:A}B(x)$ and let
    \[\codefxn_{w,w'} \defeq \qSig{p:\pr_0(w)=\pr_0(w')}\transport(p,\pr_1(w)) =
    \pr_1(w')\]
    and define the family of maps
    \[ \encode:\qPi{w,w':S}(w=w') \to \codefxn_{w,w'}\]
    by path induction with
    \[ \encode_{w,w}(\refl[w])\defeq (\refl[\pr_0(w)],\refl[\pr_1(w)]).\]
    Now consider $w,w':S$ and $r:\codefxn_{w,w'}$. By induction on all three pairs
    $w,w'$ and $r$,
    we may assume $w=(a,b)$ and $w'=(a',b')$ and $r=(p,q)$ with $p:a=a'$ and
    $q:\eqover[B]{b}{b'}{p}$. We can perform path induction on $p$, with $a\jeq a'$
    and $p\jeq \refl$, so that $q:b=_{B(a)}b'$. Again, by path induction, we may
    treat $q$ as $\refl[b]$. Then we know that $\refl[(a,b)]:w=w'$. So we have a
    family of functions
    \[\decode : \qPi{w,w':S}\codefxn_{w,w'}\to (w=w').\]
    such that
    \[\decode((a,b),(a,b),(\refl[a],\refl[b])) \jeq \refl[(a,b)].\]
    By induction, proving $\encode$ and $\decode$ to be inverse reduces to proving
    that
    \[\decode(\encode(\refl[(a,b)])) = \refl[(a,b)]\]
    and that
    \[\encode(\decode(\refl[a],\refl[b])) = (\refl[a],\refl[b]).\]
    Both of these equations hold definitionally.
\end{proof}
We will abuse notation write $(p,q)$ for the path $(a,b)=(a',b')$ arising from
$\decode(p,q)$.
\begin{corollary}
    For any $w:\qSig{a:A}B(a)$, we have $w = (\pr_0(w),\pr_1(w))$.
\end{corollary}
\begin{theorem}
    Let $B:A\to\univ$ and let $C:(\qSig{a:A}B(a))\to \univ$. Define $C':A\to\univ$ by
    \[C'(a) \defeq \qSig{b:B(a)}C(a,b).\]
    For any $p:a=a'$ and any pair $(b,c):\qSig{b:B(a)}C(a,b)$, we have
    \[\transport^{C'}(p,(b,c)) =
    (\transport^{B}(p,b),\transport^{C}((p,\refl),c)).\]
\end{theorem}
In other words, transport works pairwise.
\begin{proof}
    By path induction, it is enough to show
    \[\transport^{C'}(\refl,(b,c)) = 
    (\transport^{B}(\refl,b),\transport^{C}((\refl,\refl),c)).\]
    As we know that $(\refl,\refl)\jeq \refl$, the definition of
    transport tells us that this is asking for
    \[(b,c) = (b,c),\]
    which we have by reflexivity.
\end{proof}
\begin{lemma}
    Given a type $Y:\univ$, a path $p:X=X'$ between types $X$ and $X'$, and a
    function $f:X\to Y$, we have
    \[\transport^{\lambda X.X\to Y}(p,f) = \coe(p)\comp f.\]
\end{lemma}
\begin{proof}
    By path induction, we need to see that 
    \[\transport^{\lambda X.X\to Y}(\refl,f) = \coe(\refl)\comp f.\]
    Applying the definition of $\transport$ and $\coe$, it suffices to see that $f=
    \id\comp f$, which follows from the $\eta$ rule.
\end{proof}
\begin{theorem}\label{thm:transport-composition}
    Given a function $f:A\to B$, a family $C:B\to \univ$, a path $p:x=_{A} y$ and
    $c:C(f(x))$ we have
    \[\transport^{C\comp f}(p,c) = \transport^C(\ap_f(p),c).\]
\end{theorem}
\begin{proof}
    By path induction, we need to see that
    \[\transport^{C\comp f}(\refl,c) = \transport^C(\ap_f(\refl),c).\]
    This reduces to 
    \[c=c,\]
    which we have by reflexivity.
\end{proof}
\begin{theorem}\label{thm:path-transport}
    For any $A:\univ$ and $a:A$ and $p:x=_{A}y$, we have
    \begin{align*}
        \transport^{\lambda x. a=x}(p,q) =& q\ct p & \text{when $q:a=x$,}\\
        \transport^{\lambda x.x=a}(p,q) =& p^{-1}\ct q & \text{when $q:x=a$,}\\
        \transport^{\lambda x.x=x}(p,q) =& p^{-1}\ct q\ct p & \text{when $q:x=x$.}
    \end{align*}
\end{theorem}
This is a special case of the following theorem characterizing transport along
identity types in a type family, by taking $B=\lambda x.A$ and $f,g$ to be constant
functions.
\begin{theorem}\label{lemma:tp-path-functions}
    Let $f,g:\qPi{a:A}B(a)$ with $p:a=a'$ and $q:f(a) = g(a)$. Then,
    \[\transport^{\lambda x.f(x) = g(x)}(p,q) =
    (\apd_f(p))^{-1}\ct\ap_{\transport^B(p)}(q)\ct \apd_g(p).\]
\end{theorem}
\begin{proof}
    By path induction on $p$, we need to see
    \[\transport^{\lambda x.f(x) = g(x)}(\refl,q) =
    (\apd_f(\refl))^{-1}\ct\ap_{\transport^B(\refl)}(q)\ct \apd_g(\refl).\]
    This reduces to
    \[q = \ap_{\id}(q),\]
    and it is immediate that $\ap_{\id}$ is homotopic to the identity.
\end{proof}

\section{Equivalences and contractible fibers}\label{section:equiv-contr}
To conclude this chapter, we show that having contractible fibers, having an
inverse, and being an equivalence are logically equivalent notions. This section
therefore provides an important example of how to use the machinery in the previous
section.
\begin{theorem}\label{thm:equiv-char}
    For any $f:A\to B$, the following types are logically equivalent:
    \begin{enumerate}
        \item $\inverse(f)$, \label{equiv:inverse}
        \item $\isEquiv(f)$, \label{equiv:equiv}
        \item $\qPi{b:B}\isContr(\fib_f(b))$. \label{equiv:contr}
    \end{enumerate}
\end{theorem}
To relate
contractibility to equivalence, we need to spend some time examining fibers. As the
fiber is a sum over a family valued in path types, we first give two important lemmas
about families of paths.
\begin{lemma}\label{lemma:singleton-contractible}
    For any type $A$ and any $a:A$, the type $\qSig{x:A}x=a$ is contractible.
\end{lemma}
This type $\qSig{x:A}x=a$ is sometimes called the \emph{singleton (based) at $a$}.
\begin{proof}
    We wish to have $(a,\refl[a])$ as the center of contraction, so we need to see
    that for any $x:A$ and $p:x=a$ we have
    \[(x,p)=(a,\refl)\]
    which we can show by based path induction. We need only see that
    $(a,\refl)=(a,\refl)$, which we have by $\refl[(a,\refl)]$.
\end{proof}
We could also prove this by a direct computation using the machinery from the
previous section. For any $(x,p):\qSig{x:A}x=a$ we have
\[\big((x,p) = (a,\refl)\big) \simeq
\Big(\qSig{q:x=a}\transport^{-=a}(p,q)=\refl\Big).\]
By Theorem~\ref{thm:path-transport}, we have that $\transport^{-=a}(p,p) = p^{-1}\ct p = \refl$.
\index{transport|)}

Lemma~\ref{lemma:singleton-contractible} can be read as the (classically trivial) claim that $\{a\}$ is
equivalent to the type $\{x\in A \mid x=a\}$.
It will be used frequently and is relevant here
because it tells us that the identity function has contractible fibers:
\begin{lemma}\label{lemma:id-contractible}
    For any $A:\univ$, the identity function $\id_A:A\to A$ has contractible fibers.
\end{lemma}
\begin{proof}
    The fiber of $\id$ at $a:A$ is the singleton based at~$a$.
\end{proof}
In fact, the equivalence in Lemma~\ref{lemma:singleton-contractible} lifts to any
type family over $A$.
\begin{lemma}\label{lemma:path-families}
    For any type family $B:A\to\univ$ and any $a:A$,
    \[B(a) \simeq \qSig{(a',b'):\qSig{a':A}B(a')} a'=a.\]
\end{lemma}
\begin{proof}
    The right-hand side is equivalent, by a reshuffling map, to
    \[B\defeq \qSig{a':A}B(a')\times (a=a').\]
    There is a map $B\to B(a)$ given by $(a',b',p)\mapsto \transport^{B}(p,b')$, with
    a candidate inverse $b\mapsto (a,b,\refl)$. Since $\transport(\refl,b) = b$, the
    type $B(a)$ is a retract of $B$. Now, we need to see for any $(a',b',p):B$ that
    \[(a',b',p) = (a,\transport(p,b'),\refl).\]
    We have by assumption $p:a' = a$, and we know $\refl:\transport(p,b') = \transport(p,b')$.
    Finally, we need to see that $\transport^{-=a'}(p,p) = \refl$, which is an
    application of Theorem~\ref{thm:path-transport}.
\end{proof}
We can generalize this one step further.
\begin{lemma}\label{lemma:path-families-strong}
    For any $a:A$ and $B:\qPi{x:A}x=a\to\univ$ the map 
    \[e: B(a,\refl)\to \qSig{x:A}\qSig{p:x=a}B(x,p)\]
    given by
    \[b\mapsto (a,\refl,b)\]
    is an equivalence.
\end{lemma}
\begin{proof}
    The inverse is given by 
    \[f(x,p,b) \defeq \transport((\refl,w),b)\]
    where $w:\transport^{-=a}(p,p)=\refl$ comes from
    Theorem~\ref{thm:path-transport}.
\end{proof}
\index{equivalence|(}
The logical equivalence between having an inverse and being an equivalence is direct,
and by looking at the center of contraction of $\fib_f(b)$, when $f$ is contractible,
we can construct an inverse for $f$. The difficulty in Theorem~\ref{thm:equiv-char}
is in showing that a map with an inverse has contractible fibers. The key insight
is to relate the fibers of a section $s$ to
the fibers of the map $B\to B$ given by composing with its retraction.
\begin{lemma}\label{lemma:fib-retr}
    Given $r:B\to A$ and $s:A\to B$ together with a homotopy $\eta:r\comp s\htpy
    \id_{A}$, for any $b:B$ we have that $\fib_{s}(b)$ is a retract of~$\fib_{s\comp
    r}(b)$.
\end{lemma}
\begin{proof}
    The map 
    \[\overline{s}:\big(\qSig{a:A}s(a)=b\big)\to \big(\qSig{b':B}sr(b')= b\big)\]
    is given by
    \[\overline{s}(a,p) \defeq (s(a),\ap_s(\eta_a)\ct p),\]
    The second component is the path
    \[s(r(s(a))) \eqby{\ap_s(\eta_a)} s(a) \eqby{p} b.\]
    The candidate retraction
    \[\overline{r}:\big(\qSig{b':B}sr(b')=b\big)\to \big(\qSig{a:A}s(a)= b\big)\]
    is given by
    \[\overline{r}(b',q) \defeq (r(b'),q).\]
    We need to see that $(a,p) = \overline{r}(\overline{s}(a,p))$. By
    Theorem~\ref{thm:patheq} we need to see that there is some $w:a = r(s(a))$ such
    that $\transport^{\lambda a.s(a)=b}(w,p) = \ap_s(\eta_a)\ct p$. We have
    $\eta_a^{-1}:a = r(s(a))$, so we compute,
    \begin{align*}
        \transport^{\lambda a.s(a)=b}(\eta_a^{-1},p) &= \transport^{\lambda
        b'.b'=b}(\ap_s(\eta_a^{-1}),p) & \text{by Theorem~\ref{thm:transport-composition}}\\
        &= (\ap_s(\eta_a^{-1}))^{-1}\ct p &\text{by Theorem~\ref{thm:path-transport}}\\
        &= (\ap_s(\eta_a))\ct p &\text{by Lemma~\ref{lemma:ap-functorial}.\qedhere}
    \end{align*}
\end{proof}

Finally, we may prove Theorem~\ref{thm:equiv-char}.
\begin{proof}[Proof (of Theorem~\ref{thm:equiv-char})]
    Given $(g,(\eta,\epsilon)):\inverse(f)$ we have
    $((g,\eta),(g,\epsilon)):\isEquiv(f)$.
    Conversely, if $\eta:gf\htpy \id_{A}$ and $\epsilon:fh\htpy \id_{B}$, then for any $b:B$
    we have $g(b) = g(f(h(b))) = h(b)$, and this gives us a homotopy $\theta:fg\htpy
    \id_B$. Explicitly, $\theta_b$ is given by the composite path
    \[f(g(b)) \eqby{\ap_{f\comp g}(\epsilon_{b}^{-1})} f(g(f(h(b))))
    \eqby{\ap_f(\eta_{h(b)})} f(h(b)) \eqby{\eta_b} b, \]
    and so $(g,(\eta,\theta))$ gives an inverse of~$f$.

    The implication from~\ref{equiv:contr}.~to~\ref{equiv:inverse}.~is only slightly
    more involved: Define $g:B\to A$ to be the function giving the first component of
    the center of contraction of $\fib_f(b)$. In other words, for any $b:B$ we have by
    assumption an element
    \[((a,p),w):\isContr(\fib_f(b)),\]
    and we may define $g(b) \defeq a$.
    Note that $p:f(g(b)) = b$. Moreover, we know that we have $(a,\refl):\fib_f{f(a)}$, and
    so we have $a = g(f(a))$ by Theorem~\ref{thm:patheq} and the fact that
    $\fib_f(a)$ is contractible. Then we have that $g$ is an inverse of $f$.

    The implication from~\ref{equiv:inverse}.~to~\ref{equiv:contr}.~is less obvious.
    As $g$ is a section of $f$, we have that $\fib_g(b)$ is a retract of
    $\fib_{gf}(b)$, by Lemma~\ref{lemma:fib-retr}. Since a retract of a contractible
    type is again contractible, it suffices to show that $\fib_{gf}(b)$ is
    contractible. However, since we have $g\comp f \htpy \id_B$, we have an equivalence
    \[\fib_{g\comp f}(b) \simeq \fib_{\id_B}(b),\]
    and the latter is contractible by Lemma~\ref{lemma:id-contractible}.
\end{proof}
We will see later that being an equivalence and having contractible fibers are in
fact equivalent as types.
\index{equivalence|)}

Before we finish the chapter, we state two results about families of maps; these are
Theorems~4.7.6 and~4.7.7 of the HoTT Book. Let $P,Q:A\to\univ$ and
${f:\qPi{x:A}P(x)\to Q(x)}$. Define the \emph{total map} of $f$ by
\begin{align*}
    f_{\Sigma} &{}: (\qSig{x:A}P(x))\to (\qSig{x:A}Q(x))\\
    f_{\Sigma} &(x,p) \defeq (x,f(x,p)).
\end{align*}
\begin{theorem}\label{thm:fiberwise-sigma}
    For any $f:\qPi{x:A}P(x)\to Q(x)$, any $a:A$ and $q:Q(a)$, we have
    \[\fib_{f(a)}(q) \simeq \fib_{f_{\Sigma}}(a,q).\]
\end{theorem}
\begin{proof}
    Expanding the definition of $\fib_{f_\Sigma}$ and reshuffling gives an equivalence
    \[\fib_{f_\sigma}(a,q) \simeq \qSig{x:X}\qSig{u:x=a}\qSig{p:P(x)}\transport(u,f(x,p))=q.\]
    By Lemma~\ref{lemma:path-families-strong}, this is equivalent to
    \[\qSig{p:P(a)}\transport(\refl,f(a,p)) = q.\]
    In turn this is equivalent to
    \[\qSig{p:P(a)}f(a,p) = q,\]
    which is $\fib_{f(a)}(q)$ by definition.
\end{proof}
\begin{theorem}\label{thm:fiberwise-equivalence}
    For any $f:\qPi{x:A}P(x)\to Q(x)$ we have that each $f(a)$ is an equivalence iff
    the map $f_\Sigma$ is an equivalence.
\end{theorem}
\begin{proof}
    Suppose that $f(a)$ is an equivalence for each $a:A$. Let $(x,q):\qSig{x:A}Q(x)$.
    By the previous theorem, $\fib_{f_\Sigma}(x,q)\simeq \fib_{f(x)}(q)$, and the
    latter is contractible since $f(x)$ is an equivalence.

    Conversely, let $f_{\Sigma}$ be an equivalence and fix $a:X$ and let $q:Q(x)$.
    Then $\fib_{f(a)}(q)$ is equivalent to $\fib_{f_\sigma}(a,q)$, which is contractible since
    $f_\Sigma$ is an equivalence.
\end{proof}

\section{Discussion}
A key point in the univalent perspective is the careful distinction between structure
and property, while taking the former to subsume the latter. This distinction is
present in logical systems (such as CZF or the Mitchell-Benabou language), but
properties (formulas) are completely distinct from structures, which are captured by
certain objects of the system. The idea of treating property as being subsumed by structure 
is not obvious from this perspective---indeed it is not clear that there is a way to
give a purely logical system capturing the univalent perspective. Even in
type-theoretic contexts, where there has always been some idea that propositions are
types~\cite{MartinLof1975Predicative}, the idea that being a proposition can be an
internal statement (a predicate or structure), rather than a judgement is relatively
recent. This view arose in type theory in conjunction with the consideration of 
extensionality; \emph{bracket types}~\cite{pfenning2001}, which are used in the first
approach to this way of distinguishing propositions~\cite{AwodeyBauer2004Bracket},
arose as a way to handle proof-irrelevance and intensionality.

Univalent mathematics is not the only, or the first, type theoretic approach to
isolate a type of propositions. The calculus of inductive
constructions~\cite{CoC1986,CoC1986,Huet1987} (on which Coq
is based) is the canonical example. However, in CiC being a proposition is
a judgement, meaning there is no internal statement correspond to $\isProp(P)$.
Moreover, there is no computational content contained in the propositions of CiC. In
short, propositions fit somewhat more naturally in the univalent world than in other
type-theoretic foundational systems.



%% file: chapters/univalence.tex
\chapter{Univalent mathematics}\label{chapter:univalence}
We now develop the univalent perspective. Since we use
identity types rather than explicit equivalence relations (as we did in
Section~\ref{section:bishop-mathematics}), we need extensionality
principles---ways to prove that types are equal. Pure MLTT gives us no methods for
this, so we must extend our theory. We begin with an examination of \emph{function
extensionality}, which allows us to conclude that two functions are equal when they
are pointwise equal (Section~\ref{section:funext}). We then switch gears and examine
propositional logic in our system (Section~\ref{section:unilogic}). In order to
properly develop logic, we need a \emph{truncation} operator, taking a type to its
best representation as a proposition. This is analogous to the truncation operator we
mentioned for toposes, but in this case we actually use it as a part of our language.

Before moving on to types of higher homotopy level (Section~\ref{section:hlevels}), we examine an extensionality
principle for propositions and the \emph{univalence axiom} which generalizes this to
arbitrary types (Section~\ref{section:univalence}); in order to properly discuss univalence, we first need to take a
closer look at functions (Section~\ref{section:embeddings}). We end the chapter with some discussion of how to use
\emph{higher-inductive types}--types generated not only by
constructors giving elements of the type, but also by \emph{path constructors} giving
paths between them (Section~\ref{section:hits}). This final section and the
earlier Section~\ref{section:resizing} discuss principles which can be used, among
other things, to define the truncation operator we assume in
Section~\ref{section:unilogic}.

There is a subtlety with the extensionality principles we consider in this
chapter: all of them depend on a universe $\univ$. Until we explicitly introduce them
as assumptions, we will treat them as properties of a universe.

\section{Function extensionality}\label{section:funext}
\index{function extensionality|(}
One of the difficulties of working in MLTT is that MLTT gives almost no methods for
proving equalities. This first becomes apparent when we wish to show that two
functions $f,g:A\to B$ are equal. The only equalities between functions arise as
reflexivity out of judgmental equalities. For example consider the function
$f:\nat\to\nat$ defined by recursion with
\begin{align*}
    f(0) &{}= 0\\
    f(\succop(n)) &{}= \succop(f(n)).
\end{align*}
It is not possible to prove that $f=\id_{\nat}$, even though we can prove that
$\qPi{n:\nat}f(n) = n$, and moreover we can show that for any constructor of $\nat$
that this equality holds judgmentally. We expect $f$ to equal~$\id_{\nat}$---we
expect two functions to be equal when they are pointwise equal. In this section we
will posit an axiom of \emph{function extensionality} which tells us that this is the
case.

An \emph{axiom} in type theory is an element of a particular type
that is declared to exist by fiat, independent of the relevant introduction and elimination
rules. That axioms do not interact with introduction and elimination rules makes them
potentially dangerous. For example, we could extend MLTT with an axiom $a:\nat$.
Then, after defining the predecessor function ${\pred:\nat\to\nat}$, we would have
that $\pred(a):\nat$, but we cannot reduce this or compute with it in any way.

However, as we only work up to propositional equality, if we can
prove that the type $A$ is a proposition, then positing an element $a:A$ leads to no
problems: For any $b:A$ that we can construct, we have $a=b$, and so we
can use $\transport$ to interchange between $a$ and $b$.
So when we postulate function extensionality, we wish to ensure that the type
representing the axiom is a proposition.  We consider four logically equivalent
formulations of function extensionality below. All but one can be shown to be a
proposition.

\begin{theorem}\label{thm:funext-char}
    For any universe $\univ$, the following are logically equivalent:
\begin{enumerate}[label=\normalfont\textrm{F\arabic*}]

    \item For any $A:\univ$ and $B:A\to\univ$ and any $f,g:\qPi{a:A}B(a)$, the map
        $\happly_{f,g}:f= g\to f\htpy g$ is an equivalence.\label{funext:funext}
    \item For any $A:\univ$ and $B:A\to\univ$ and any $f,g:\qPi{a:A}B(a)$, there is a
        map $f\htpy g\to f=g$.\label{funext:htpy}
    \item The product of any family of propositions is a
        proposition: for any $A:\univ$ and $P:A\to\univ$, \label{funext:prop}
        \[\big(\qPi{a:A}\isProp(P(a))\big)\to
        \isProp(\qPi{a:A}P(a)).\]
    \item The product of any family of contractible types is contractible: for any
        $A:\univ$ and $P:A\to\univ$, \label{funext:contr}
        \[\big(\qPi{a:A}\isContr(P(a))\big)\to \isContr(\qPi{a:A}P(a)).\]
\end{enumerate}
\end{theorem}
\begin{proof}
    (\ref{funext:funext} $\Rightarrow$ \ref{funext:htpy}): If $\isEquiv(\happly_{f,g})$, then we have a map $f\htpy g\to f=g$, by definition.

    (\ref{funext:htpy} $\Rightarrow$ \ref{funext:prop}): Let $P:A\to\univ$ be a family of propositions, and let
    $f,g:\qPi{a:A}P(a)$. We wish to see $f=g$. By assumption we have $f\htpy g\to
    f=g$, so it is enough to show ${\qPi{a:A}f(a) = g(a)}$, but we know that $P(a)$ is
    a proposition for all~$a$.

    (\ref{funext:prop} $\Rightarrow$ \ref{funext:contr}): Let $P:A\to\univ$ be a family of contractible types.
    Then
    we have some function $f:\qPi{a:A}P(a)$ by taking the center of contraction of each
    $P(a)$. Moreover, we have that each $P(a)$ is a proposition, and so by assumption,
    $\qPi{a:A}P(a)$ is a proposition. As a proposition with an element, $\qPi{a:A}P(a)$
    is contractible.

    (\ref{funext:contr} $\Rightarrow$ \ref{funext:funext}): Let $B:A\to\univ$ and fix $f:\qPi{x:A}B(x)$. We need to
    see
    \[\qPi{g:\qPi{x:A}B(x)} \isEquiv(\happly_{f,g}),\]
    By Theorem~\ref{thm:fiberwise-equivalence}, this happens iff the map
    \[\lambda(g,p). \happly_{f,g}(p):\Big(\qSig{g:\qPi{x:A}B(x)}f=g\Big)\to
    \Big(\qSig{g:\qPi{x:A}B(x)}f\htpy g\Big)\]
    is an equivalence.
    As the type $\qSig{g:\qPi{x:A}B(x)}f=g$ is contractible, it suffices to show that
    \[\qSig{g:\qPi{x:A}B(x)}f\htpy g\]
    is contractible. We do this by showing that it is a retract of the type
    \[\qPi{x:A}\qSig{b:B(x)}f(x) = b,\]
    which is contractible by the assumption~\ref{funext:contr}.

    Define
    \begin{align*}
        H&:\Big(\qSig{g:\qPi{x:A}B(x)}\qPi{x:A}f(x)= g(x)\Big)\to
        \qPi{x:A}\qSig{b:B(x)}f(x) = b\\
        H(g,p) &\defeq \lambda x.(g(x),p(x))
    \end{align*}
    and 
    \begin{align*}
        J&:\Big(\qPi{x:A}\qSig{b:B(x)}f(x) = b\Big)\to
        \qSig{g:\qPi{x:A}B(x)}\qPi{x:A}f(x)= g(x) \\
        J(p) &\defeq (\lambda x. \pr_0(p(x)),\lambda x.\pr_1(p(x)))
    \end{align*}
    Then we have
    \[J(H(g,p)) = J(\lambda x.(g(x),p(x))) = (\lambda x.g(x), \lambda x.p(x)) =
    (g,p),\]
    so that $J$ is a left inverse of $H$.
\end{proof}

We wish to posit a form of function extensionality that we can show to be a
proposition without building too much theory.
We show that~\ref{funext:prop}\ above is a proposition as follows: We show
using~\ref{funext:prop}\ and~\ref{funext:funext}\ that the property of being a
proposition is a proposition. We then use~\ref{funext:prop} again to show
that~\ref{funext:prop} is a proposition. By Lemma~\ref{lemma:prop-to-prop} below, it
turns out that as a result we can show \ref{funext:prop}\ to be a proposition
without the additional assumptions. To aid readability, let us call a class
$T:\univ\to\univ$ of types an \emph{exponential ideal in $\univ$} when for any
$B:A\to\univ$ we have
\[\big(\qPi{x:A}T(B(x))\big) \to T\big(\qPi{x:A}B(x)\big),\]
so that~\ref{funext:prop}\ says that propositions are an exponential ideal,
and~\ref{funext:contr}\ says that singletons are an exponential ideal. Note that we
are generalizing the usual definition of exponential ideal to dependent functions.

In order to show that being a proposition is a proposition, we need two results.
First, we show that if $A$ can be shown to be a proposition under the assumption that
$A$ has an element, then $A$ is already a proposition.
\begin{lemma}\label{lemma:prop-to-prop}
    Let $A$ be a type with a function $A\to\isProp(A)$. Then $A$ is a proposition.
\end{lemma}
\begin{proof}
    Let $f:A\to\isProp(A)$, and let $x,y:A$. We have $f(x):\qPi{y,z:A}y=z$, so then
    $f(x)(x,y):x=y$.
\end{proof}

Second, we show that propositions are $(-1)$-types; from this, it follows that
propositions are sets.
\begin{lemma}\label{lemma:prop-is-one-type}
    A type $P$ is a proposition iff it is a $(-1)$-type. That is, $P$ is a
    proposition iff all path types in $P$ are contractible.
\end{lemma}
\begin{proof}
    Let $P$ be a $(-1)$-type, so that $\qPi{x,y:P}\isContr(x=y)$. Then taking the
    center of contraction of $x=y$ gives us $\qPi{x,y:P}x=y$.

    Conversely, let $w:\isProp(P)$, and fix $x:P$. Then $u\defeq \lambda y. w(x,y)$
    has type $\qPi{y:P}x=y$. Given any $p:y=z$, we have
    \[\apd_u(p) :\transport^{\lambda y.x=y}(p,u(y)) = u(z).\]
    By Theorem~\ref{thm:path-transport}, this gives us $u(y) \ct p= u(z)$, and so
    \[p=u(y)^{-1}\ct u(z),\]
    but $p$ is arbitrary. Then for $z\jeq x$, and any $p,q:y=x$ we have
    \[p = u(y)^{-1}\ct u(x) = q,\]
    and so $x=y$ is a proposition. As we also have $w(x,y):x=y$, we have that $x=y$
    is contractible.
\end{proof}
\begin{corollary}
    If $P$ is a proposition, then $P$ is a set.
\end{corollary}
\begin{proof}
    As $P$ is a proposition, its path types are contractible, and so path types in
    $P$ are propositions.
\end{proof}

\begin{theorem}\label{thm:prop-is-prop}
    If propositions form an exponential ideal in $\univ$, then for any type $X:\univ$, the type
    $\isProp(X)$ is a proposition.
\end{theorem}
\begin{proof}
    By Lemma~\ref{lemma:prop-to-prop}, it is enough to show
    \[\isProp(X)\to\isProp(\isProp(X)).\]
    Let $i:\isProp(X)$, so that for all $x,y:X$, we have that $x=y$
    is also a proposition by Lemma~\ref{lemma:prop-is-one-type}.
    Then $\qPi{x,y:X}x=y$ is a product of propositions, and then $\qPi{x,y:X}x=y$ is
    itself a proposition by assumption.
\end{proof}
By a similar argument we have
\begin{theorem}\label{thm:contr-is-prop}
    If propositions form an exponential ideal in $\univ$, then for any $X$, we have
    \[\isProp(\isContr(X)).\]
\end{theorem}
\begin{proof}
    Let $(c,w)(c',w'):\isContr(x)$. Then we have $c = c'$ by some path $p=w(c')$. Moreover, since
    $\isContr(X)$ implies $\isProp(X)$, we have that $\qPi{x:X}c'=x$ is a
    proposition, so $w$ transports over $p$ to $w'$.
\end{proof}

\begin{theorem}
    If propositions form an exponential ideal in $\univ$, then for any $A:\univ$ and
    any $P:A\to\univ$, the type
    \[\big(\qPi{a:A}\isProp(P(a))\big)\to \isProp(\qPi{a:A}P(a))\]
    is a proposition.
\end{theorem}
\begin{proof}
    Since propositions form an exponential ideal, the type $\isProp(X)$ is a
    proposition for any type $X$. Also because propositions form an exponential
    ideal, if $Y$ is a proposition, then $B\to Y$ is a proposition for any type $B$.
    Then in particular, we have that
    \[\big(\qPi{a:A}\isProp(P(a))\big)\to \isProp(\qPi{a:A}P(a))\]
    is a proposition.
\end{proof}
By Lemma~\ref{lemma:prop-to-prop} we have the following corollary.
\begin{corollary}\label{corllary:funext-isprop}
    The type
    \[\big(\qPi{a:A}\isProp(P(a))\big)\to \isProp(\qPi{a:A}P(a))\]
    is a proposition. I.e., the statement that propositions form an exponential ideal
    is a proposition.
\end{corollary}
Now that we know that whether propositions form an exponential ideal is a proposition,
we may assert this as an axiom. Here instead, we take it as a property of a universe
\begin{definition}[Function extensionality]
    A universe $\univ$ \nameas{satisfies function
    extensionality}{function extensionality} when propositions
    form an exponential ideal in $\univ$. That is, the type
    \[\ExpProp_{\univ}\defeq \qPi{A:\univ}\qPi{P:A\to\univ}\big(\qPi{a:A}\isProp(P(a))\big)\to \isProp(\qPi{a:A}P(a))\]
    has an element.
\end{definition}
The axiom of function extensionality is usually given via
statements~\ref{funext:funext} or~\ref{funext:htpy} of Theorem~\ref{thm:funext-char},
although Voevodsky preferred~\ref{funext:contr}.
We are most interested in statement~\ref{funext:funext}. Let
$\FunExt:\univ_1$ be the type
\[\FunExt\defeq \qPi{A:\univ}\qPi{B:A\to\univ}\qPi{f,g:\qPi{x:A}P(x)}\isEquiv(\happly_{f,g}),\]
so that function extensionality gives an element of $\FunExt$. From
this for any $f,g:\qPi{x:A}B(x)$ we get a map
\[\funext_{f,g}:(f\htpy g)\to (f=g),\]
which is an inverse to $\happly_{f,g}$.
We will see shortly that $\isEquiv(f)$ is a proposition for all $f$, and so $\FunExt$ is
itself a proposition, but already we know that the type $\ExpProp$ is a retract of $\FunExt$.

Function extensionality is the first extension to MLTT we need to develop the
univalent approach. We will assume function extensionality in Part II, and
often use it tacitly---showing that two functions are equal by showing a homotopy
between them. However, we will not assume function extensionality just yet, as we
show in Section~\ref{section:univalence} that function extensionality follows from
the univalence axiom. The other extensions we need are \emph{truncations} and \emph{proposition
extensionality}, which we cover shortly.

\index{function extensionality|)}
\section{Logic of propositions}\label{section:unilogic}
\index{proposition|(}
Recall that logic is traditionally interpreted in type theory according to the
Curry-Howard interpretation,
\begin{align*}
    P\wedge Q &\jeq P\times Q\\
    P\vee Q &\jeq P + Q\\
    P\Rightarrow Q &\jeq P\to Q\\
    \neg P &\jeq P\to \zerotype \\
    \qAll{a:A}{P(a)} &\jeq \qPi{a:A}{P(a)}\\
    \qExists{a:A}{P(a)} &\jeq \qSig{a:A}{P(a)}.
\end{align*}
This interpretation is not always well-behaved. Consider, for example, the image of a
function $f:A\to B$, which is the set of all $b:B$ such that there exists $a:A$ with
$f(a)=b$.  Using $\Sigma$ for exists, we have the \emph{Curry-Howard image},
$\im_{\CH}(f):\univ$ as the type
\[\im_{\CH}(f)\defeq \qSig{b:B}\qSig{a:A}f(a) = b.\]
However, this type is equivalent to $A$:
\begin{lemma}\label{lemma:domain-is-fiber}
    For any $f:A\to B$, the second projection $\pr_1:\im_{\CH}(f)\to A$ is an equivalence,
    and $f$ factors as $f(x) = \pr_0 (\pr_1^{-1}(x))$.
\end{lemma}
\begin{proof}
    Reshuffling the quantifiers in $\im_{\CH}(f)$ gives an equivalence
    \[\Big(\qSig{b:B}{\qSig{a:A}{f(a)=b}}\Big)
    \simeq\Big(\qSig{a:A}{\qSig{b:B}{f(a)=b}}\Big).\]
    As $\qSig{b:B}{f(a)=b}$ is the singleton at $f(a)$ we have 
    \[\Big(\qSig{a:A}{\qSig{b:B}{f(a)=b}}\Big)\simeq \big(\qSig{a:A}\unittype\big),\]
    and the first projection gives an equivalence. The second projection
    $\pr_0\im_{CH}(f)\to A$ is the composition of these two maps, and so is itself an
    equivalence. An inverse is then given by the map $e(a) = (f(a),\refl)$.
    Then for any $a:A$ we have that $f(a)=\pr_0(e(a))$.
\end{proof}
The point is that $\qSig{x:A}P(x)$ is not guaranteed to be a proposition, even if
$P(x)$ is a proposition for each~${x:A}$, so using $\Sigma$ carries extra structure.
In this case, the extra structure allows us to recover all of $A$. Instead, we want a
\emph{logic of propositions} in which $\qExists{x:A}P(x)$ is guaranteed to be a
proposition, so that to say that the statement $b$ is in the image of $f$ corresponds
to a propositional version of saying that the fiber of $b$ is inhabited. Then we can
define the image to be
\defineopfor{$\im$}{im}{image}
\[\im(f)\jeq \qSig{b:B}\qExists{a:A}f(a)=b.\]
We immediately see that we cannot in general expect an equivalence between $\im(f)$
and the type
\[\qExists{a:A}{\qSig{b:B}{f(a)=b}},\]
as this type is a proposition, but the image need not be.

Topos logic resolves this issue by restricting our attention to subsingletons.
We can instead use $\isProp$ in Section~\ref{section:utt}. However, in order to
encode a proposition-valued version of $\exists$, or even a proposition-valued
version of $\vee$, we must extend our type theory with \nameas{propositional
truncation}{truncation}\index{propositional truncation|see{truncation}}.

\begin{definition}
    \definesymbolfor{$\trunc{-}$}{truncation}{truncation}
    \defineopfor{${\mid}-{\mid}$}{truncate}{truncation}
    For any type $X:\univ$, a \emph{propositional truncation} of $X$ is a type
    $\trunc{X}$ which is a proposition, together with a map 
    \[|-|:X\to\trunc{X}\]
    such that if $P$ is any other proposition, any map $f:X\to P$ factors uniquely
    through the map $|-|$.

    \index{type!inhabited|textit}
    We say that a type $X$ \emph{is inhabited} when $\trunc{X}$ has an element.
\end{definition}
The truncation can be defined in several ways: by directly extending MLTT; as
a higher inductive type (Section~\ref{section:hits}); via resizing
(Section~\ref{section:resizing}), or by assuming the law of excluded middle.

In particular, if we assume the law of excluded middle:
we have $A\to(\neg \neg A)$ for any type $A$. Now take any proposition
$P$ such that $A\to P$. By contraposition we have $\neg P\to \neg A$, and then since
both $\neg A$ and $\neg P$ are propositions, using excluded middle we can get
$\neg\neg A\to \neg \neg P$. As $P$ is a proposition, we may (again using excluded
middle) eliminate the double negation in front of $P$ to get $\neg\neg A\to P$.

The higher-inductive type we give in Section~\ref{section:hits} directly
expresses the universal property of truncation as a rule of construction, while the
resizing and LEM approach encode the propositional truncation concretely. Formally, in the
absence of a general framework for higher-inductive types (see the discussion in
Section~\ref{section:uttdiscussion}), the higher-inductive type must be added directly, so the
higher-inductive approach corresponds to adding truncations directly. This was
already studied as \emph{squash types} before univalent type theory arose.

As we are working informally, it is enough for us to say that we are positing
directly that each type $X$ has a propositional truncation. Once we have the
truncation, we can begin to study how it behaves. In particular, truncation is
idempotent up to equivalence, which follows from a more general observation.
\begin{lemma}\label{lemma:prop-equals-truncation}
    If $P$ is a proposition, then $P\simeq \trunc{P}$.
\end{lemma}
\begin{proof}
    We need only see that $P\lequiv \trunc{P}$. Since
    $P$ is a proposition, we have that $\id:P\to P$ factors through the map
    $|-|:P\to\trunc{P}$.
\end{proof}
The utility of the above lemma is as follows: We will often define some operation
$T'$ on types such that $T'(X) = \trunc{T(X)}$. The type $T(X)$ may contain
useful information, but truncation hides this. Often, however, $T(X)$ is already a
proposition, so we have that $T'(X) \simeq T(X)$ so we can extract the information from
$T(X)$ already by knowing $\trunc{T(X)}$. More generally, we have
\begin{lemma}\label{leq-is-trunc}
    If $P$ is a proposition and $P\lequiv X$, then $P\simeq \trunc{X}$.
\end{lemma}
\begin{proof}
    As $X\to P$ and $P$ is a proposition, we have $\trunc{X}\to P$. Moreover, we have
    the composite map $P\to X\to\trunc{X}$. Hence, $P\lequiv \trunc{X}$.
\end{proof}

\begin{lemma}\label{lemma:trunc-func}
    Truncation is functorial: for any types $X$ and~$Y$,
    \[(X\to Y)\to (\trunc{X}\to\trunc{Y}).\]
\end{lemma}
\begin{proof}
    Suppose $f:X\to Y$. Then we have $\lambda x.|f(x)|:X\to \trunc{Y}$. As
    $\trunc{Y}$ is a proposition, we then have a map $\trunc{X}\to\trunc{Y}$.
\end{proof}
By truncating the Curry-Howard interpretation everywhere that it fails to give a
proposition, we get an interpretation of logic at the level of propositions.
\begin{align*}
    P\wedge Q &\defeq P\times Q\\
    P\vee Q &\defeq \trunc{P + Q}\\
    P\Rightarrow Q &\defeq P\to Q\\
    \neg P &\defeq P\to \zerotype \\
    \qAll{a:A}{P(a)} &\defeq \qPi{a:A}{P(a)}\\
    \qExists{a:A}{P(a)} &\defeq \trunc{\qSig{a:A}{P(a)}}\\
\end{align*}
Where $P$ and $Q$ are propositions on the first three lines, and $P:A\to\univ$ is
a predicate on a type $A$ in the last two lines. Each of these types is a
proposition: The type $P\wedge Q$ is a proposition because $P$ and $Q$ are
propositions; the types $P\vee Q$ and $\qExists{a:A}P(a)$ are propositions by
definition, while $P\to Q$ and $\qPi{a:A}P(a)$ are propositions
by function extensionality. In general, we will only use the
logical notation just defined when we are using $\exists$ or $\vee$, or trying to
make explicit the connection between what we are doing and some traditional logical
principle.

This interpretation allows us to define the image as we wanted to above.
\begin{definition}
    The \nameas{image}{function!image of} of $f$ is the type
    \[\im(f)\defeq \qSig{b:B}{\trunc{\fib_f(b)}}.\]
    That is,
    \[\im(f)\defeq \qSig{b:B}\qExists{a:A}(f(a)=b).\]
\end{definition}
That is, the image of $f$ is the type of all $b:B$ such that there exists (as
property) an $a:A$ with $f(a)=b$.

While in general $+$ is not guaranteed to give propositions, the coproduct of 
disjoint propositions is always again a proposition.
\begin{lemma}
    If $A$ and $B$ are disjoint propositions (I.e., such that $\neg (A\times B)$),
    then $A+B$ is a proposition; hence $A+B\simeq A\vee B$.
\end{lemma}
\begin{proof}
    Let $x,x':A+B$, and do a case analysis on $x$ and $x'$:
    \begin{itemize} 
        \item Case 1: $x=\inl a$ and $x'=\inl a'$. As $A$ is a proposition, $a=a'$,
            and so $\inl a = \inl a'$.
        \item Case 2: $x=\inr b$ and $x'=\inr b'$. Similar to Case 1.
        \item Case 3: $x=\inl a$ and $x' = \inr b'$. Then we have $(a,b'):A\times B$,
            but we assumed $\neg(A\times B)$, and so we get $x=x'$.
        \item Case 4: $x=\inr b$ and $x' = \inl a'$ Similar to Case 3.
    \end{itemize}
\end{proof}
\begin{definition}
    A type $P$ is \nameas{decidable}{type!decidable} if $P+\neg P$.\\
    A predicate $P:A\to\univ$ is \emph{decidable} if $\Pi(a:A),P(a)+\neg P(a)$.\\
    A type $A$ \emph{has decidable equality} or \nameas{is discrete}{type!discrete} when for all
    $a,a':A$, the type $a=a'$ is decidable.
\end{definition}
This terminology is standard in constructive mathematics, but since we also deal with
computability theory, this terminology creates a clash. We will always use
\emph{recursive} or \emph{computable} when discussing recursive decidability;
nevertheless, to avoid confusion we will sometimes say that $P$ is
\emph{complemented} when $P+\neg P$.
The standard example of a discrete type is the type of natural numbers, which is
decidable by the same argument as Theorem~\ref{thm:nat-is-set}.

\begin{lemma}
    For any family of types $P:X\to U$, we have
    \[\big\lVert\qSig{x:X}\trunc{P(x)}\big\rVert\simeq \trunc{\qSig{x:X}P(x)}\]
\end{lemma}
\begin{proof}
    We need only show logical equivalence. For the implication
    \[\big\lVert\qSig{x:X}\trunc{P(x)}\big\rVert\to \trunc{\qSig{x:X}P(x)},\]
    it is enough to define a map
    \[\big(\qSig{x:X}\trunc{P(x)}\big)\to \trunc{\qSig{x:X}P(x)},\]
    and by currying, such a map is the same thing as a map
    \[\qPi{x:X}\trunc{P(x)}\to \trunc{\qSig{x:X}P(x)}.\]
    Fixing $x:X$, a map $\trunc{P(x)}\to \trunc{\qSig{x:X}P(x)}$ arises by
    functoriality of truncation from any map $P(x)\to\qSig{x:X}P(x)$, and $\lambda
    p.(x,p)$
    suffices.

    For the other direction, it is enough to define a map
    \[\big(\qSig{x:X}P(x)\big)\to\qSig{x:X}\trunc{P(x)},\]
    and $\lambda(x,p).(x,|p|)$ suffices.
\end{proof}
\begin{lemma}\label{lemma:sigmatrunc-is-exists}
    If $X$ is a type and $P:X\to\univ$ is a family of types over $X$ such that
    $\isProp\left(\qSig{x:X}\trunc{P(x)}\right)$,
    then
    $ 
    \left(\qSig{x:X}\trunc{P(x)}\right) \simeq \trunc{\qSig{x:X}P(x)}.
    $ 
\end{lemma}
\begin{proof}
    As $\isProp\left(\qSig{x:X}\trunc{P(x)}\right)$, we have an equivalence
    \[\trunc{\qSig{x:X}\trunc{P(x)}}\simeq\qSig{x:X}\trunc{P(x)},\]
    and the latter is equivalent to $\trunc{\qSig{x:X}P(x)}$ by the previous lemma.
\end{proof}

\index{proposition|)}
\section{Surjections, embeddings and equivalences}\label{section:embeddings}
We can decompose $\isEquiv(f)$ into $\linv(f)\times \rinv(f)$, which
correspond to saying that $f$ is a section and $f$ is a retraction. We can similarly
decompose the notion of having contractible fibers: since
$\isContr(A)$ is equivalent to $A\times \isProp(A)$, for $f:A\to B$, we can restate 
$\isContr(f)$ as
\[\qPi{b:B}\isProp(\fib_f{b})\times \trunc{\fib_f(b)}.\]
We can then split these two notions to arrive at the following definitions.
\begin{definition}
    A function $f:A\to B$ is
    \begin{itemize}
        \item an \name{embedding} if it has propositional fibers:
            \[\qPi{b:B}\isProp(\fib_f{b}).\]
        \item a \name{surjection} if it has inhabited fibers:
            \[\qPi{b:B}{\trunc{\fib_f(b)}}.\]
    \end{itemize}
\end{definition}
Then this directly gives us that a function has contractible fibers iff it is both a
surjection and an embedding. By Theorem~\ref{thm:equiv-char}, we have that a function
is an equivalence iff it is both an embedding and a surjection.

Note that we follow the HoTT Book in highlighting the subtlety of the definition by
using the word \emph{embedding} rather than \emph{injection}. However, the HoTT Book
defines embedding by looking at $\ap_f$, which is sometimes more useful. Our
definition instead gives a symmetry between embedding and surjection. Nevertheless,
the above definition is indeed equivalent to the definition given in the HoTT Book.
\begin{lemma}\label{lemma:embedding-propositional-fibers}
    A function $f:A\to B$ is an embedding iff 
    $\ap_f:(x=y)\to(f (x) = f(y))$ is an equivalence for each $x,y:A$.
\end{lemma}
\begin{proof}
    We need to see that $\ap_f$ has contractible fibers iff $f$ has
    propositional fibers. Fix $b:B$, and $(x,p),(y,q):\fib_f(b)$. We have
    \begin{align*}
        \big((x,p)=(y,q)\big) & {}\simeq \qSig{r:x=y}p = \ap_f(r)\cdot q\\
                                & {}\simeq \qSig{r:x=y}\ap_f(r) = p\cdot q^{-1}\\
                                & {}\simeq \fib_{\ap_f}(p\cdot q^{-1}).
    \end{align*}
    So if $\ap_f$ has contractible fibers, then $f$ has propositional fibers.
    On the other hand, for any $p:f(x)=f(y)$ we have $(x,p)$ and $(y,\refl)$
    in $\fib_f(f(y))$, so that $\big((x,p)=(y,\refl)\big)\simeq \fib_{\ap_f}(p)$.
    Then if $f$ has propositional fibers, $\ap_f$ has contractible fibers.
\end{proof}
In the case where $B$ is a set, we can simplify the above characterization;
in this case, it is enough to know that there is a map $f(x)=f(y)\to x=y$.
We will give a somewhat indirect proof, in order to introduce an important
observation by Mart\'in Escard\'o concerning retracts of identity types.
\index{type!identity|(}
\begin{lemma}\label{lemma:identity-section}
    Suppose $R:X\to X \to \univ$ and that we have maps
    \begin{align*}
        r:{}&\qPi{x,y:X}x=y\to R(x,y)\\
        s:{}&\qPi{x,y:X}R(x,y)\to x=y,
    \end{align*}
    such that $r_{x,y}$ is a left inverse of $s_{x,y}$ for all~$x,y:X$. Then $r_{x,y}$ and
    $s_{x,y}$ are inverse. Hence, there is an equivalence $R(x,y)\simeq (x=y)$ for all~$x,y:X$.
\end{lemma}
\begin{proof}
    Since we already have that $r_{x,y}\comp s_{x,y}$ is homotopic to the identity,
    then we only need to see that $s(r(p)) = p$ for all $p:x=y$. Notice that $s\comp
    r$ is an idempotent function of type $\qPi{x,y:X}(x=y)\to (x=y)$, since $r$ is
    left inverse to $s$. Then it is enough to show that any family of functions
    \[f:\qPi{x,y:X}(x=y)\to (x=y)\]
    such that $(f\comp f) \htpy f$ is homotopic to the identity. Since $\refl$ is the
    identity with respect to composition, we have
    \[f(\refl) = f(\refl)\ct \refl\]
    for any function $f:\qPi{x,y:X}(x=y)\to(x=y)$. 
    Then by path induction, any function $f:\qPi{x,y:X}(x=y)\to (x=y)$
    satisfies
    \[\qPi{x,y:X}\qPi{p:x=y} f(p) = f(\refl[x])\ct p.\]
    In particular, if $f$ is idempotent we have
    \[f(p) = f(f(p)) = f(\refl) \ct f(p),\]
    where the second equality is the previous observation. Then composition with
    $f(p)^{-1}$, gives $f(\refl) = \refl$. Applying path induction, we then have that for any
    $p:x=y$ that $f(p) = p$ for any family of idempotent functions on~$x=y$. In
    particular, we have that $(s\comp r)(p) = p$, so $s$ is also a left inverse
    of~$r$.
\end{proof}
\index{type!identity|)}
\begin{theorem}
    If $B$ is a set, then $f:A\to B$ is an embedding iff for each $x,y:A$ there is
    a map $h:(f(x) = f(y))\to (x=y)$.
\end{theorem}
\begin{proof}
    As $B$ is a set, $f(x)=f(y)$ is a proposition, so $h$ tells us that $f(x)=f(y)$ is
    a retract of $x=y$. So, for every $x,y:A$ we have $f(x)=f(y)$ is a retract of
    $x=y$, and by Lemma~\ref{lemma:identity-section} we have an equivalence
    \[(f(x)=f(y))\simeq (x=y).\qedhere\]
\end{proof}

We have stated embedding and surjection as if they were property, not structure. Function
extensionality tells us that this is justified:
\begin{lemma}
    In the presence of function extensionality, being an embedding, being a surjection and
    having contractible fibers are all propositions.
\end{lemma}
\begin{proof}
    Function extensionality tells us that propositions form an exponential ideal
    and that $\isProp(X)$ is always a proposition, so being an embedding is a proposition;
    similarly since $\isContr(X)$ is a proposition for any $X:\univ$, so is $\isContr(f)$
    for any $f:A\to B$; again using function extensionality, since truncations are
    propositions by definition, we have that being a surjection is a proposition.
\end{proof}
We have alluded to the fact
that $\isContr(f)$ and $\isEquiv(f)$ are equivalent. In the presence of function
extensionality, it is enough to show that $\isEquiv(f)$ is a proposition, since
$\isContr(f)$  and $\isEquiv(f)$ are logically equivalent. First note that if
$f:A\to B$ has an inverse, then the composition maps $(f\comp -):(C\to A)\to (C\to B)$ and $(-\comp
f):(B\to C) \to (A\to C)$ do as well, by composition with the inverse of $f$. As a result, we
have the following.
\begin{theorem}
    If $\univ$ satisfies function extensionality, then for any $A,B:\univ$, if $f:A\to B$ has an inverse, then $\rinv(f)$ and $\linv(f)$ are contractible.
\end{theorem}
\begin{proof}
    Fix $f:A\to B$ with an inverse.
    By function extensionality, $\linv(f)$ is equivalent to the fiber of $(-\comp f)$
    over $\id_A$.  By the above observation, we know that $(-\comp f)$ has an inverse, and so has
    contractible fibers. Similarly, $\rinv(f)$ is equivalent to the fiber of $(f\comp -)$
    over $\id_B$.
\end{proof}
\begin{theorem}\label{thm:equiv-is-prop}
    If $\univ$ satisfies function extensionality, then for any $A,B:\univ$ and any $f:A\to B$, we have that $\isEquiv(f)$
    is a proposition.
\end{theorem}
\begin{proof}
    Note that $\isEquiv(f) \simeq \linv(f)\times \rinv(f)$. If $e:\isEquiv(f)$,
    then $f$ is invertible, so $\isEquiv(f)$ is a product of contractible types by
    function extensionality, and so is a proposition. Briefly, we have
    \[\isEquiv(f)\to \isProp(\isEquiv(f)).\]
    Now apply Lemma~\ref{lemma:prop-to-prop}.
\end{proof}
\begin{corollary}\label{cor:contr-is-isequiv}
    If $\univ$ satisfies function extensionality, then for any $A,B:\univ$ and any $f:A\to B$ we have
    \[\isContr(f)\simeq \isEquiv(f).\]
\end{corollary}
\begin{proof}
    Function extensionality tells us that the types are logically equivalent
    propositions.
\end{proof}

\section{Proposition extensionality and univalence}\label{section:univalence}
As with functions, MLTT provides no way to show that two \emph{types} are equal, so
if we want to prove types to be equal, we need an extensionality principle
for types. Before examining the situation for all types, we look at propositions.

In logical languages (including the Mitchell-Benabou language), two propositions are
considered equal if they are logically equivalent. This is not possible in MLTT with
our definition of proposition, so we suggest an extensionality principle for
propositions.
\begin{definition}\label{ax:propext}
    A universe $\univ$ \nameas{satisfies proposition extensionality}{proposition
    extensionality} if whenever
    $A:\univ$ and $B:\univ$ are
    propositions and $A\lequiv B$, then $A=B$:
    Let
    \[\PropExt_{\univ}\defeq\qPi{A,B:\univ}\isProp(A)\to \isProp(B)\to (A\lequiv B)\to (A=B).\]
    $\univ$ satisfies proposition extensionality when there is
    $\propext:\PropExt_{\univ}$.
\end{definition}
Notice that function extensionality tells us already that $A\simeq B$ is a
proposition when $A$ and $B$ are propositions, since $A\simeq B$ is a proposition
indexed sum of propositions. This gives us the following slight strengthenings of
Lemma~\ref{lemma:lequiv-equiv} and Lemma~\ref{lemma:contr-char}.
\begin{lemma}
    If $\univ$ satisfies function extensionality, and $A:\univ$ and $B:\univ$ are propositions, then
    \[(A\lequiv B)\simeq (A \simeq B).\]
\end{lemma}
\begin{proof}
    We have that 
    \[(A\lequiv B)\lequiv (A \simeq B),\]
    and both sides are propositions, so this implies
    \[(A\lequiv B)\simeq (A \simeq B).\qedhere\]
\end{proof}
So in the presence of function extensionality we can reformulate the type of
proposition extensionality as
    \[\qPi{A,B:\univ}\isProp(A)\to \isProp(B)\to (A\simeq B)\to (A=B).\]
Since we have that $A\simeq B$ is a proposition when both types are propositions,
this tells us that $A\simeq B$ is a
retract of $A=B$, and since $A\simeq B$ is a proposition, the retraction can be given
by any map $(A=B)\to(A\simeq B)$. We may as well use the map $\idtoequiv$. In short, we have
that proposition extensionality and function extensionality together imply that
${\idtoequiv:(A= B)\to (A\simeq B)}$ has a right inverse when $A$ and $B$ are
propositions.  As a corollary, we have
\begin{lemma}\label{lemma:propext-to-propuniv}
Proposition extensionality and function
extensionality together imply
\[\qPi{A,B:\univ}\isProp(A)\to\isProp(B)\to\isEquiv(\idtoequiv_{A,B}).\]
\end{lemma}
This principle is known as \emph{propositional univalence}.

In the section on function extensionality, we said we wanted axioms to be propositions,
but we have not yet shown proposition extensionality to be a proposition.
\begin{theorem}\label{thm:propext-isprop}
    If $\univ$ satisfies function extensionality, then propositional univalence implies
    $A=B$ is a proposition whenever $A$ and $B$ are propositions. Hence, in the
    presence of function extensionality, proposition extensionality is a proposition.
\end{theorem}
\begin{proof}
    When $A$ and $B$ are propositions, we already know $(A\simeq B) \simeq (A\lequiv B)$.
    As $A\lequiv B$ is a proposition, proposition extensionality implies that $A=B$
    is a proposition as well.

    As proposition extensionality implies propositional univalence, we have
    \[\PropExt\to \isProp(A=B),\]
    whenever $A$ and $B$ are propositions. The type of proposition
    extensionality is of the form $\qPi{x:X}A=B$, and since propositions
    form an exponential ideal, we have 
    \[\PropExt\to\isProp(\PropExt).\]
    Hence, $\PropExt$ is a proposition.
\end{proof}

In the statement of propositional univalence, we could drop the condition that $A$
and $B$ are propositions, to get the type
\[\UA_{\univ}\defeq \qPi{A,B:\univ}\isEquiv(\idtoequiv_{A,B}).\]
This type says that all equivalences arise from a path $A=B$. The axiom
that it is inhabited is known as the \emph{univalence axiom}.
\index{universe!univalent|see{univalence}}
\begin{definition}
    A universe $\univ$ is \nameas{univalent}{univalence} when there is an element of $\UA_{\univ}$. In
    particular, if $\univ$ is univalent, for $A,B:\univ$ we have a map
    \[\ua:(A\simeq B)\to (A=B)\]
    such that $\ua$ and $\idtoequiv$ are inverse.
\end{definition}
In fact, by Lemma~\ref{lemma:identity-section}, it is enough to say that $\idtoequiv$
has a left inverse.
\begin{theorem}\label{thm:identity-section}
    If $\idtoequiv:(A=B)\to (A\simeq B)$ has a right inverse for each $A,B:\univ$,
    then $\idtoequiv$ is an equivalence for each~$A,B:\univ$.
\end{theorem}

\index{function extensionality|(}
\index{univalence|(}
Function extensionality and proposition extensionality follow from
univalence. Proposition extensionality follows directly, but function extensionality
takes more work. The argument here is abstracted from Voevodsky's original
proof that univalence implies function extensionality, as presented
in~\cite{gambino2011univalence}. It hinges around the total
type of the identity relation, which we call the \emph{diagonal} of a type.

\begin{definition}
Given a type $A:\univ$, define the \emph{diagonal} of $A$ to be the type
\[\Delta A\defeq \qSig{a,a':A}a=a'.\]
    There is then a \emph{diagonal map} $\delta_A:A\to\Delta A$ given by
\[\delta(x) \defeq (x,x,\refl).\]
\end{definition}
It is straightforward to check that the diagonal map is an equivalence. In
particular, we have two candidate inverses given by $\pr_0$ and $\pr_1$. Both are
left inverse directly since we have the equalities $\pr_0(\delta(x)) =\pr_1(\delta(x)) = x$.
For the other direction, we have
\[\delta(\pr_0(x,y,p)) = (x,x,\refl[x])),\]
so we wish to see that $(x,x,\refl[x]) =_{\Delta A} (x,y,p)$. As $\refl:x=x$, we need to
see that we have $q:x=y$ such that $\transport^{x=-}(q,\refl[x]) = p$. Of course, 
Theorem~\ref{thm:path-transport} determines a witness $w:\transport(p,\refl[x]) =
\refl[x]\ct p$, so the given $p$ suffices. For $\pr_1$, the argument is similar.

Now consider the class of maps arising from paths, the
\nameas{path-induced equivalences}{equivalence!path-induced}:
\defineopfor{$\isPI$}{ispie}{equivalence!path-induced}
\[\isPI_{A,B}(f)\defeq \qSig{p:A=B}f=\coe(p).\]
Note that this type is the Curry-Howard image of $\coe$.
We have $r_A:\isPI_{A,A}(\id_A)$ given by $r_A\defeq (\refl[A],\refl[\id_A])$.
We can then define a version of $\idtoequiv$ for
$\isPI$ using the function $\epsilon_{A,B}:\Pi_{p:A=B}\isPI_{A,B}(\coe(p))$ defined by
path induction with
\[\epsilon(\refl[A]) \defeq r_A.\]
In particular, we have $\isPI_{A,B}(\coe(p))$ for any $p:A=B$.
Write $A\pieq B$ for
\[ \qSig{f:A\to B}\isPI_{A,B}(f).\]
Then we have 
\[\idtopie:A=B\to A\pieq B\]
given by
\[\idtopie(p)\defeq (\coe(p),\epsilon(p)).\]
Note that $\idtopie$ is one direction of an equivalence $(A\pieq B)\simeq (A=B)$
given by Lemma~\ref{lemma:domain-is-fiber}. Take $S_{A,B}:A\pieq B\to A= B$ to be a
section of $\idtopie$.
\begin{lemma}
    Let $P:\qPi{A,B:\univ}(A\to B) \to\univ$ such that $\qPi{A:\univ}P(\id_A)$. For
    every $A,B:\univ$ and $p:A=B$, the map $\idtofun(p)$ satisfies $P$. Formally,
    \[\qPi{A,B:\univ}\qPi{p:A=B}P(\coe(p)).\]
\end{lemma}
\begin{proof}
    By path induction, we need only see $P(\idtofun(\refl_A))$ for each $A$, and we
    have $P(\id_A)$ by assumption.
\end{proof}
That is, we can prove something for all path-induced equivalences by looking at the
identity maps
\begin{lemma}
    Let
    $P:\qPi{A,B:\univ}(A\to B) \to\univ$ such that $\qPi{A:\univ}P(\id_{A})$, then
    every path-induced equivalence in $\univ$ satisfies $P$. That is, for all
    $A,B:\univ$ and $f:A\to B$,
    $\isPI(f)\to P(f)$.
\end{lemma}
\begin{proof}
    Note that every $e:A\pieq B$ can be written as $\coe(S(e))$, and apply the
    previous lemma. The key point here is that $\idtopie$ has a section.
\end{proof}
\begin{corollary}\label{cor:pi-to-equiv}
   We have that $\isPI(f)\to \isEquiv(f)$.
\end{corollary}
\begin{proof}
    The identity map is an equivalence.
\end{proof}
In particular, if $f$ is a path-induced equivalence, then $f$ is both a section and a
retraction.

Ultimately, univalence tells us that equivalences are exactly path-induced
equivalences.
\begin{lemma}\label{lemma:ua-to-equiv-is-pi}
    If $\univ$ is univalent, then for any $A,B:\univ$ we have
    \[\qPi{f:A\to B}(\isEquiv(f)\simeq \isPI(f)).\]
\end{lemma}
\begin{proof}
    By Theorem~\ref{thm:fiberwise-equivalence}, we have that all maps
    $\isPI(f)\to\isEquiv(f)$ are equivalences iff the induced map
    \[e:(\qSig{f:A\to B}\isPI(f))\to (\qSig{f:A\to B}\isEquiv(f))\]
    is an equivalence. The induced map $e$ factors as the equivalence
    $(A\pieq B)\simeq (A=B)$ followed by $\idtoequiv$. Since univalence says that
    $\idtoequiv$ is an equivalence, it says that $e$ is the composite of
    equivalences.
\end{proof}
However, to show that univalence implies function extensionality, it is enough to
assume that the path-induced equivalences satisfy two properties that are
satisfied by equivalences, namely it is enough to assume
\begin{itemize}
    \item path-induced equivalences are closed under homotopy;
    \item all diagonal maps $\delta_X:X\to\Delta_X$ are path-induced equivalences.
\end{itemize}
We assume these explicitly as needed in the following lemmas.
\begin{lemma}\label{lemma:precomp-abstract-equiv}
    Suppose path-induced equivalences are closed under homotopies.
    If $f$ is a path-induced equivalence, then so is $(-\comp f)$.
\end{lemma}
\begin{proof}
    We need only see that $(-\comp f)$ is homotopic to a path-induced map when $\isPI(f)$.
    First note that for any $p:X=X'$ and $g:X'\to Y$, we have
    $\transport^{\lambda X.X\to Y}(p^{-1},g) = g\comp \coe(p)$, by path
    induction: when $p\jeq \refl$ both sides are equal to $g$. So then we have
    that 
    \[\transport^{\lambda X.X\to Y}(S(f,e)^{-1},g) = g\comp f,\]
    and by the discussion in Section~\ref{section:homotopies}, we know that every
    transport map is homotopic to a path-induced function.
\end{proof}
\begin{theorem}\label{thm:ua-to-funext-key}
    If path-induced equivalences are closed under homotopies, and for each $X$,
    we have $\isPI(\delta_X)$,
    then for any $f,g:A\to B$ we have $f\sim g \to f=g$.
\end{theorem}
\begin{proof}
    Note that for any $X$ we have $\pr_0\comp \delta_X = \pr_1\comp \delta_X$ by
    definition and the computation rule for function types. As $\delta_X$ is a
    path-induced equivalence, so is precomposition with $\delta_X$, but then we also
    have that $-\comp\delta_X$ is a section, and so we have that $\pr_0
    =_{\Delta(X)\to X} \pr_1$. Now take $f,g:A\to B$ with $\eta:f\htpy g$. Then we
    have $h:A\to \Delta B$ given by
    \[h(x)\defeq (f(x),g(x),\eta(x)).\]
    Finally, we have
    \[f = \pr_0\comp h = \pr_1\comp h = g.\qedhere\]
\end{proof}
\begin{corollary}
    If $\univ$ is univalent, then for any $A,B:\univ$ and any $f,g:A\to B$, we have
    ${f\htpy g\to f=g}$.
\end{corollary}
\begin{proof}
    We need to see that univalence implies that path-induced equivalences
    are closed under
    homotopies and that $\delta_X$ is path-induced. As univalence implies
    $\isPI(f)\simeq \isEquiv(f)$, it is enough to show that equivalences are closed
    under homotopies, which is immediate by transitivity of homotopy, and that
    $\isEquiv(\delta_X)$, which we showed above.
\end{proof}
In fact, the above non-dependent version of function extensionality---which we will
call \nameas{naive function extensionality}{function extensionality!naive}---implies full
function extensionality. We proceed in two steps: We first show that non-dependent function
extensionality implies that $f\comp -$ is an equivalence whenever $f$ is, and then
that this implies that contractible types form an exponential ideal.
\begin{lemma}\label{funext-equiv-lemma}
    Fix types $A,B,X:\univ$ and let $f:A\to B$ have inverse $g:B\to A$. If $\univ$
    satisfies naive function extensionality, then
    $g\comp-$ is an inverse of $f\comp-:(X\to A)\to (X\to B)$.
\end{lemma}
\begin{proof}
    Let $h:X\to A$ and $x:X$. We have that
    $g(f(h(x))) = h(x)$ since $g$ and $f$ are inverse. Then we have by naive function
    extensionality that $g\comp f \comp h = h$. Hence, $(g\comp -)\comp (f\comp
    -)\simeq \id_{X\to A}$. The other direction is similar.
\end{proof}
\begin{lemma}\label{funext-retract-lemma}
    If $B:A\to\univ$ and each $B(a)$ is contractible, then $\qPi{x:A}B(x)$ is a
    retract of $\qSig{h:A\to \qSig{x:A}B(x)}\pr_0\comp h= \id$.
\end{lemma}
\begin{proof}
    The section map is given by
    \[s(f) = (\lambda a.(a,f(a)), \refl),\]
    while the retraction map is given by
    \[r(g,p) = \lambda a. \transport^B(\happly(p,a),\pr_1(g(a))).\]
    For $f:\qPi{x:A}B(x)$ we compute
    \begin{align*}
        r(s(f)) &= \lambda a.\transport^B(\happly(\refl,a)),\pr_1(a,f(a))\\
        &= \lambda a.\transport^B(\refl_{a},f(a)) \\
        &= \lambda a.f(a)\\
        &= f.\qedhere
    \end{align*}
\end{proof}
\begin{theorem}
    If $\univ$ satisfies naive function extensionality, then $\univ$ satisfies function extensionality.
\end{theorem}
\begin{proof}
    Let $B:A\to\univ$ such that $B(a)$ is contractible for each $a:A$. We wish to see
    that $\qPi{x:A}B(x)$ is contractible.  By Lemma~\ref{funext-retract-lemma}, we know that
    this type is a retract of the type $\qSig{h:A\to \qSig{x:A}B(x)}\pr_0\comp h=
    \id$; since contractible types are closed under retracts, it suffices to show
    that $\qSig{h:A\to \qSig{x:A}B(x)}\pr_0\comp h= \id$ is contractible. Notice that
    this type is the fiber of $\id$ under the map
    \[\pr_0\comp -:(A\to \qSig{x:A}B(x))\to (A\to A).\]
    A map is an equivalence if and only if all its fibers are contractible, so it
    suffices to prove that $\pr_0\comp -$ is an equivalence and so by
    Lemma~\ref{funext-equiv-lemma}, it is enough to show that the first projection $\pr_0:\qSig{x:A}B(x)\to A$
    is an equivalence. A candidate inverse for $\pr_0$ is given by $g(a) = (a,c)$
    where $c$ is the center of contraction of $B(a)$. Then $\pr_0(g(a)) = a$ by
    $\refl$, and then for any pair $(a,b):\qSig{x:A}B(x)$ we have that $(a,b) =
    (a,c)$, since
    contractible types are propositions.
\end{proof}
\begin{corollary}
    Any univalent universe satisfies function extensionality.
\end{corollary}
\begin{proof}
    We know that any univalent universe satisfies naive function
    extensionality, and that any universe satisfying naive function
    extensionality satisfies function extensionality.
\end{proof}
\index{function extensionality|)}

In fact, the converse of Lemma~\ref{lemma:ua-to-equiv-is-pi} holds. That is, univalence
says exactly that the structure of a path-induced equivalence is the same as the
structure of an equivalence. Univalence is sometimes stated informally as ``all
equivalences arise as transport'', or ``all equivalences arise from paths''. We can
in fact prove univalence to be equivalent to a very weak interpretation of this informal
statement.
\begin{theorem}
    A universe $\univ$ is univalent if and only if, for all $A,B:\univ$ and $f:A\to
    B$ we have the implication
    \[\isEquiv(f)\to \isPI(f).\]
\end{theorem}
\begin{proof}
    Given a family of maps $\eta:\isEquiv(f)\to \isPIE(f)$, we have function
    extensionality, by Theorem~\ref{thm:ua-to-funext-key}. Fix some $f:A\to B$.
    Then, by function extensionality we have that $\isEquiv(f)$ is a proposition. Our
    assumption tells us that $\isEquiv(f)\lequiv \isPI(f)$, and since $\isEquiv(f)$
    is a proposition, we know that any map $\isPI(f)\to\isEquiv(f)$ is a retraction
    with section $\eta$. In particular, the map $\epsilon$ given by path induction in
    Corollary~\ref{cor:pi-to-equiv} has $\eta$ as a section. This then lifts to a
    section-retraction pair
    \begin{align*}
        s &{}: \big(\qSig{f:A\to B}\isEquiv(f)\big)\to\big(\qSig{f:A\to B}\isPI(f)\big)\\
        s(f,e) &{}= (f,\eta(e))\\
        r &{}: \big(\qSig{f:A\to B}\isPI(f)\big)\to\big(\qSig{f:A\to B}\isEquiv(f)\big)\\
        r(f,e) &{}= (f,\epsilon(e))
    \end{align*}
    By Lemma~\ref{lemma:identity-section}, we only need to see that $\idtoequiv$ has
    a section, but we defined $\epsilon$ so that
    $\idtoequiv(p) \htpy (\coe(p),\epsilon(\coe(p)))$, which means that
    $\idtoequiv$ factors as $r\comp \idtopie$. As the projection
    $\pr_{=}:(A\pieq B)\to (A=B)$ is a section of $\idtopie$, the composite
    $\pr_{=}\comp s$ is a section of $\idtoequiv$.
\end{proof}

We are finally in a place to show that univalence is a proposition.
\begin{theorem}
    The type $\UA_{\univ}$, is a proposition.
\end{theorem}
\begin{proof}
    By Lemma~\ref{lemma:prop-to-prop} it is enough to show that univalence
    implies that $\UA_{\univ}$ is a proposition. We know that univalence implies function
    extensionality. Since function extensionality implies that $\isEquiv(f)$ is a proposition and
    that propositions form an exponential ideal, function extensionality implies
    $\isProp(\UA_{\univ})$. Then univalence implies $\isProp(\UA_{\univ})$.
\end{proof}
\index{inverse|(}
The use of equivalences instead of functions with an inverse in the statement of
univalence is crucial. Most of the arguments in this section work when we
replace $\isEquiv(f)$ with $\inverse(f)$ everywhere. In particular,
If we replace $\idtoequiv$ with the corresponding function
\[i:A=B\to \qSig{f:A\to B}\inverse(f),\]
then the corresponding version of univalence tells us that $\inverse(f)$ is a section
of $\isPI(f)$, and so $\inverse(f)$ must be a proposition. However, we mentioned
before that if $f:A\to B$ has an inverse, then
\[\inverse(f)\simeq \qPi{x:A}x=x.\]
So, if $A$ is not a set, then $\inverse(f)$ is not a proposition. Moreover, both
univalence and the corresponding statement with $\inverse(f)$ imply that there
are types which are not sets. Consequently, the type
\[\qPi{A,B:\univ}(A=B)\simeq \big(\qSig{f:A\to B}\inverse(f)\big),\]
is not inhabited.
\index{inverse|)}
\index{univalence|)}

We will limit our use of univalence in the main development, using instead
function extensionality and proposition extensionality. The reason is twofold. First,
since we work primarily at the level of sets and propositions, it is usually enough
to use these two weaker principles; the full power of univalence is more than we
require. Secondly, the discussion of topos logic above is not simply for comparison;
while proposition extensionality as stated in this section is not expressible in
topos logic, a form of proposition extensionality does hold. As a result, the technical work
in Part II should translate with a little work to a topos-theoretic setting, as long
as we do not use univalence. In other words, we hope that the material in Part II will
still have technically interesting content to a mathematician who rejects or is
otherwise uninterested in univalent mathematics.

We will use function extensionality and proposition extensionality for all universes
throughout. This means in particular that the following all hold in all universe.
\begin{itemize}
    \item propositions form an exponential ideal (Theorem~\ref{thm:funext-char}.\ref{funext:prop});
    \item contractible types form an exponential ideal (Theorem~\ref{thm:funext-char}.\ref{funext:contr});
    \item all contractible types are equal (Corollary~\ref{cor:iscontr-singleton});
    \item propositional univalence (Lemma~\ref{lemma:propext-to-propuniv});
    \item proposition extensionality is a proposition
        (Theorem~\ref{thm:propext-isprop});
    \item for any function $f:A\to B$, we have that $\isEquiv(f)$ is a proposition
        (Theorem~\ref{thm:equiv-is-prop});
    \item For any $f:A\to B$, the types $\isEquiv(f)$ and $\isContr(f)$ are
        equivalent (Corollary~\ref{cor:contr-is-isequiv}).
\end{itemize}
Proposition extensionality and function extensionality allow us to strengthen
Lemma~\ref{lemma:contr-char}.
\begin{lemma}\label{lemma:contr-characterization}
    The following are equivalent (in fact, equal) for any type $A$:
    \begin{itemize}
        \item $A\simeq \unittype$; (A is equivalent to $\unittype$)
        \item $\isContr(A)$; (A is contractible)
        \item $A\times \isProp(A)$; (A is a proposition with an element)
    \end{itemize}
\end{lemma}
\begin{proof}
    Each of the above is a proposition: For $A\simeq\unittype$, this is a sum of
    propositions over a proposition; for $\isContr(A)$, we have already seen that
    this is a proposition in Theorem~\ref{thm:contr-is-prop}. Finally, we have
    $(A\times\isProp(A)) \to \isProp(A)$, and so then
    \[(A\times\isProp(A))\to \isProp(A\times\isProp(A)),\]
    hence by Lemma~\ref{lemma:prop-to-prop}, we have $A\times\isProp(A)$ is a
    proposition.

    As the types are logically equivalent propositions (by
    Lemma~\ref{lemma:contr-char}), they are all equivalent.
\end{proof}

\section{Resizing rules}\label{section:resizing}
\index{universe|(}
One difference between MLTT and Topos logic is that Martin-L\"{o}f type theory is a
predicative system which uses universe levels to
avoid paradoxes of self-reference. However, the stratification of types into universe
levels limits our ability to make certain arguments. For example, the type $A\to
\Prop$, which we view as the power set of $A$, is parameterized by the universe level of
$\Prop$. In particular, we have for each universe level $i$, a type
$\Prop_i:\univ[i+1]$, so that $A\to\Prop_{i}$ has type $\univ[j]$, where $j$ is the
maximum of $i+1$ and the universe level of $A$.
On the other hand any topos has an object $\Omega$ of truth values, so that
$\Omega^A$ is the power set of $A$. MLTT can be made impredicative by \emph{resizing}
propositions; we explain this below, and discuss it where relevant, but we do not
make use of resizing principles anywhere in the main development.

In his Bergen lecture~\cite{VoevodskyBergen}, Voevodsky introduced two resizing
rules, which say informally,
\begin{enumerate}
    \item if $P:\univ[i]$ is such that $\isProp(P)$, then in fact $P:\univ[0]$;
    \item for any $\univ[i]$, we have that $\qSig{A:\univ[i]}\isProp{A} : \univ[0]$.
\end{enumerate}
These allow us to define $\Omega\defeq \Prop_{0}$, and ensure that no new propositions
are added at higher levels. However, Voevodsky posited these as rules in the formal
theory, and the consistency of these rules is an open question.

The HoTT book gives a weaker version of resizing. For any universe level $i$, we have
an inclusion $\Prop_{i}\to\Prop_{i+1}$, where the type $\Prop_{i}\to\Prop_{i+1}$ is
a type in $\univ[i+2]$. The HoTT book assumes that all of these inclusion maps are
equivalences, so that all propositions can be replaced with a proposition in
$\univ[0]$. This allows us to define $\Omega:\univ[1]$ by $\Omega\defeq \Prop_{0}$,
and then we have an equivalence $\Omega\simeq \Prop_{n}$ in universe~$\univ[n+1]$.
Note that this weaker form of propositional resizing follows already from LEM, since
LEM tells us that $\Prop_{i}\simeq 2$ for any universe level~$i$.

If we have some type $\Omega:\univ$ of propositions, then we may define the
truncation for any $X:\univ$ as $\trunc{X}$ as
\[\trunc{X}\jeq \qPi{P:\Omega}(X\to P)\to P.\]
We have $|-|:X\to\trunc{X}$ given by
\[|x| = \lambda P.\lambda f. f(x).\]
Now if $P$ is any proposition, and we have $f:X\to P$, we can factor $f$ through
$|-|$ and a map $\trunc{X}\to P$ given by $w \mapsto f(P,w)$.

While we will sometimes discuss the consequences of resizing, we will not use it in
any of the main development.
\index{universe|)}

\section{Homotopy levels and sets}\label{section:hlevels}

\index{set|(}
Now that we've developed some of the theory of equivalences, we can direct our
attention more fully to homotopy $n$-types, defined in
Section~\ref{section:utt}. First, we note that being an $n$-type is
in fact a proposition, and then we continue onto some closure properties.
\begin{lemma}
    For all $X$, we have $\isProp(\isProp(X))$ and moreover, for any $n$, we have
    $\isProp(\istype{n}(X))$. That is,
    being a proposition or a homotopy $n$-type is a proposition.
\end{lemma}
\begin{proof}
    For $\isProp$, this is Theorem~\ref{thm:prop-is-prop} and function extensionality.

    For the $n$-types, we proceed by induction: For $n=-2$, this is
    Lemma~\ref{lemma:prop-to-prop} and function extensionality.

    For $n=k+1$, assume that $\istype{k}$ is proposition-valued. Then
    $\qPi{x,y:X}\istype{k}(x=y)$ is a proposition.
\end{proof}

\begin{lemma}
    If $e:A\to B$ is an embedding, and $B$ is an $n+1$-type, then so is $A$ for
    $n\geq -1$.
\end{lemma}
\begin{proof}
    We have that $(a=a')\simeq (e(a)=e(a'))$ and the later type is an $n$-type.
\end{proof}
\begin{theorem}\label{lemma:ntypes-are-retr-closed}
    If $r:A\to B$ is a retraction, and $A$ is an $n$-type, then $B$ is an $n$-type.
\end{theorem}
\begin{proof}
    For $n\geq -1$, this follows from the previous lemma, since sections are
    embeddings. For $n=-2$, this is Theorem~\ref{thm:prop-retract-closed}.
\end{proof}
\begin{theorem}
    If $e:A\to B$ is an embedding, and $B$ has decidable equality, then so does $A$.
\end{theorem}
\begin{proof}
    Let $a,a':A$. Since $e$ is an embedding, we have $(a=a')\simeq (e(a)=e(a'))$, and
    the latter has decidable equality.
\end{proof}
\begin{lemma}
    If $r:A\to B$ is a retraction and $A$ has decidable equality, then so does $B$.
\end{lemma}
\begin{theorem}
    The $n$-types are closed under equivalence: if $A\simeq B$ and $B$ is an
    $n$-type, then so is $A$.
\end{theorem}
\begin{proof}
    As equivalences are in particular retractions, this follows from
    Theorem~\ref{lemma:ntypes-are-retr-closed}.
\end{proof}
\begin{theorem}
    The $n$-types are closed under $\Sigma$, for all $n$.
\end{theorem}
\begin{proof}
    Induction on $n$. For $n=-2$, it suffices to show that $\qSig{a:\unittype}P(a)$
    is contractible when each $P(a)$ is, and in turn it is enough to show that
    $\unittype\times \unittype$ is contractible, which is immediate. So let $A$ be
    an $n+1$-type and let $B:A\to\univ$ such that $B(a)$ is an $n+1$-type for all $a:A$.
    We have that
    \[\big((a,b)=(a',b')\big)\simeq \Big(\qSig{p:a=a'}\eqover[B]{b}{b'}{p}\Big),\]
    and this is an $n$-type by the inductive hypothesis.
\end{proof}
\begin{theorem}\label{lemma:ntypes-exp-ideal}
    The $n$-types form an exponential ideal for each~$n$. In particular if $B$
    is an $n$-type, then $A\to B$ is an $n$-type.
\end{theorem}
\begin{proof}
    By induction on $n$. For $n=-2$, this is function extensionality. 
    Now let $B(a)$ be an $n+1$-type for each $a:A$, and let $f,g:\qPi{a:A}B(a)$. We have
    \[(f=g)\simeq \big(\qPi{x:A}f(x)=g(x)\big),\]
    and this is an $n$-type by the inductive hypothesis.
\end{proof}

In general, we will be interested in the level of sets and propositions. Importantly,
this includes all types with decidable equality.
\begin{lemma}\label{lemma:decidable-to-dne}
    If $A$ is a decidable type, then $(\neg\neg A)\to A$.
\end{lemma}
\begin{proof}
    By induction on the element $p:A+\neg A$. If $p=\inl a$, then $a:A$, and so
    $\lambda x.a:(\neg\neg A)\to A$. Otherwise, we have $p=\inr n$. But then if
    $q:\neg\neg A$, we have $q(n):\zerotype$, and so we have
    \[(\lambda q.!_{A}(q(n))) : (\neg\neg A)\to A.\qedhere\]
\end{proof}
\begin{theorem}[Hedberg]
    If $A$ has decidable equality, then $A$ is a set.
\end{theorem}
\begin{proof}
    If $A$ has decidable equality, then by Lemma~\ref{lemma:decidable-to-dne} we have
    \[\qPi{x,y:A}\neg\neg(x=y)\to (x=y).\]
    As $\neg\neg X$ is a proposition for any $X$ and $X\to\neg\neg X$ for any $X$, we
    have that $\neg\neg X$ is a retract of $X$ whenever $\neg\neg X\to X$. In
    particular, if $A$ has decidable equality, then for all $x,y:A$ the type
    $\neg\neg(x=y)$ is a retract of $x=y$. Then by
    Lemma~\ref{lemma:identity-section}, we know that $\neg\neg(x=y)\simeq (x=y)$.
    Since $\neg\neg(x=y)$ is a proposition, we then have that $A$ is a set.
\end{proof}
N.B.~the above proof is not Hedberg's original proof, and uses function
extensionality, while the original proof doesn't.

Since we are particularly interested in sets, we collect the following facts, which
we have already proved elsewhere.
\begin{theorem}
    The types $\zerotype$, $1$, $2$ and $\nat$ are sets. Sets are closed under all
    basic type formers: If $A$ is a set, and $B:A\to\univ$ is such that
    $\qPi{a:A}{\isSet(B(a))}$, then $\qSig{a:A}{B(a)}$ and $\qPi{a:A}{B(a)}$ are sets.
\end{theorem}
\begin{theorem}\label{thm:embedding-monomorphism}
    If $B$ is a set, then for any $f:A\to B$, the following are
    logically equivalent:
    \begin{enumerate}
        \item $f$ is an embedding, 
        \item $f$ is a monomorphism: For any $g,h:X\to A$ we have
            \[(f\comp g = f\comp h)\to (g = h).\]
    \end{enumerate}
\end{theorem}
\begin{proof}
    Let $f$ be an embedding and $g,h:X\to A$ such
    that $f\comp g= f\comp h$. Then for any $x:X$ we have $f(g(x)) = f(h(x))$, and
    since $f$ is an embedding $g=h$.

    Now let $f$ be a monomorphism, and fix $x,y:A$ such that $f(x)=f(y)$. Define
    $c_x , c_y :\unittype\to A$ as $c_x(u) \defeq x$ and $c_y(u) \defeq x$. As
    $f(x)=f(y)$, we have $f\comp c_x = f\comp c_y$, and so $c_x = c_y$, and $x=y$, so
    $f$ is an embedding.
\end{proof}
\begin{corollary}
    If $B$ is a set, then for any $f:A\to B$, $f$ being an embedding and $f$ being a
    monomorphism are equivalent as types.
\end{corollary}
\begin{proof}
    As $B$ is a sets, the type
    \[\qPi{X:\univ}\qPi{g,h:X\to A}(f\comp g=f\comp h)\to g = h,\]
    is a proposition, so we have that the type witnessing that $f$ is an embedding
    and the type witnessing that $f$ is a monomorphism are logically equivalent
    propositions.
\end{proof}
\index{set|)}

\section{Higher-inductive types}\label{section:hits}
There are two sorts of higher inductive types
we will look at: \emph{quotient inductive types}, which we will use for comparison of
our developments with other approaches to partiality, and the \emph{homotopy circle},
which we use as the prototypical example
of a type which is not a set to show when assumptions on types are necessary. This
appears, for example, in Section~\ref{section:sdt-comparison}, when we discuss the
difference between our notion of dominance and the usual one.

Quotient inductive types are given by specifying simultaneously an inductive type
together with a quotient of that type by some relation. The motivating example is
that of \emph{set quotients}: given a type $A$ and a relation $R:A\to A\to\Prop$, we
define the smallest \emph{set} $A/R$ for which there is a map $A\to A/R$ respecting $R$.
Explicitly, the type $A/R$ is defined inductively by
\begin{itemize}
    \item a \emph{quotient map} $[-]:A\to A/R$;
    \item for each $x,y:A$ with $R(x,y)$, a path $[x]=[y]$.
    \item the \emph{set truncation}: For each $x,y:A/R$, and $p,q:x=y$, a path $p=q$.
\end{itemize}
Then the recursion principle for $A/R$ says that if $B$ is any set, and we have a
function $f:A\to B$ such that $R(x,y)\to f(x)=f(y)$, then there is a unique map
$f/R:A/R\to B$ such that $f(a) = (f/R)([a])$ for all~$a:A$.
\begin{lemma}
    For any type $A$ and any relation $R:A\to A\to\Prop$, the quotient map $[-]:A\to A/R$ is surjective.
\end{lemma}
\begin{proof}
    For any element of the form $[a]:A/R$, we know have that $\refl:[a]=[a]$. For the
    higher constructors, we need to see that $|(a,\refl)| = |(b,\refl)|$ for any
    $a,b:A$ with $R(a,b)$, and that being a surjection is a set. Since being a
    surjection is a proposition, both of these facts are immediate.
\end{proof}

The set truncation is necessary for this notion to be well-behaved: suppose we
defined the quotient of $A$ by $R$ as above, but without the truncation constructor,
and consider the quotient of $\unittype$ by the relation $R$ with $R(\star,\star)$.
This has as constructors,
\begin{itemize}
    \item a point $b:1/R$,
    \item a path $b=b$.
\end{itemize}
This type is exactly the homotopy circle $\Sone$ given by Definition~\ref{def:sone}
below, which is not a set if univalence holds. Then we would have that the quotient
of a set by a relation is no longer a set.

On the other hand, we can form the propositional truncation as the set-quotient by
the total relation $\nabla \defeq \lambda x,y.\unittype$. We can take the truncation
to be $X/\nabla$ with the quotient map $X\to X/\nabla$ as the truncating map. Then any
proposition $P$ is in particular a set, and any map $X\to P$ must respect the
equivalence relation, so it factors through $X/\nabla$.

It is more typical, however, to give the propositional truncation of $X$ as a higher
inductive type by directly expressing the universal property: the type $\trunc{X}$ is
generated by
\begin{itemize}
    \item a truncation constructor $|-|:X\to\trunc{X}$, and
    \item a family of path constructors $\qPi{x,y:\trunc{X}}(x=y)$.
\end{itemize}
There is a difference in the elimination principle for $\trunc{X}$ as a HIT and
$X/\nabla$. We can only eliminate $\trunc{X}$ into propositions, while we can
eliminate $X/\nabla$ into any set $A$, so long as we have a constant function $X\to
A$. However, we will see in Section~\ref{section:constancy} that the elimination
principle for $X/\nabla$ also holds for $\trunc{X}$.

The quotient by an equivalence relation can be defined concretely in a
higher universe; with resizing, we can perform this truncation without raising
universe levels. Define an \emph{equivalence class of $R$} to be a predicate
$P:A\to\Prop$ such that $P$ is equivalent to $R(a,-)$ for an anonymous $a:A$:
\[\qExists{a:A}\qPi{b:B}\left(P(b)\simeq R(a,b)\right).\]
Then define 
\[A\sslash R\defeq \qSig{P:A\to \univ}\text{$P$ is an equivalence class of $R$}.\]
Note that $A\sslash R$ lives in $\univ[1]$. Nevertheless, we then have a map
$q:A\to A\sslash R$ given by
\[q(a) = \lambda b. R(a,b).\]
This map $q$ respects $R$, and so there is a map $A/R\to A\sslash R$. 
\begin{theorem}
    For any type $A:\univ$ and equivalence relation $R:A\to A\to \Prop$, the type
    $A\sslash R$ is a set, the map $q:A\to A\sslash R$ is a surjection which respects
    $R$, and moreover the induced map $s:A/R\to A\sslash R$ is an equivalence.
\end{theorem}
\begin{proof}
    We have that $A\to\Prop$ is a set, and that being an equivalence class is a
    property, so $A\sslash R$ is a set, as a subtype of a set.

    Let $P:A\sslash R$. We wish to see $\qExists{a:A} P = q(a)$. As $q(a) = R(a,-)$,
    this is the definition of being an equivalence class.

    Next we need to see that if $R(a,a')$, then $q(a)=q(a')$. For a fixed $b$, we
    have $q(a)(b) = R(a,b)$ and $q(a')(b) = R(a',b)$, so we need to see that $R(a,b)=
    R(a',b)$. As both are propositions, we need $R(a,b) \lequiv R(a',b)$. We have
    maps in both directions by transitivity and symmetry of~$R$.

    Finally, we need to see that the induced map $s:A/R\to A\sslash R$ is an
    equivalence. First, $s$ is surjective: as $q$ is surjective, for any equivalence
    relation $P$, we have
    \[\qExists{a:A}q(a) = P,\]
    and $q$ factors as $s\comp {[-]}$, so we have
    \[\qExists{a:A}s[a] = P.\]
    Finally, $s$ is an embedding: let $x,y:A/R$ and let $s(x)= s(y)$. As $x=y$ is a
    proposition, and $[-]$ is surjective, we may assume that $x=[a]$ and $y=[b]$ for
    some $a,b:A$. Then we have for any $c$ that 
    \[R(a,c) = s([a])(c) = s(x)(c) = s(y)(c) = s([b])(c) = R(b,c),\]
    so that $R(a,b) = R(b,b)$, and the latter is inhabited by reflexivity.
\end{proof}

We turn now to the homotopy circle, which we use to show deficiencies in proposed
notions. The basic idea is that an obvious way of stating a particular notion may not
interact as expected with general types, and we use $\Sone$ to show that we need to
be more subtle when defining a particular notion. This appears most clearly in
Section~\ref{section:sdt-comparison}, where we see why we cannot translate
definitions from synthetic domain theory directly.

\begin{definition}\label{def:sone}
    \index{circle|see{homotopy circle}}
    \definesymbolfor{$\Sone$}{Sone}{homotopy circle}
    The \name{homotopy circle}, $\Sone:\univ$, is inductively generated by
\begin{itemize}
    \item an element $\base:\Sone$,
    \item a path $\lp:\base = \base$.
\end{itemize}
\end{definition}
The recursion principle is easy to state:
\begin{itemize}
    \item For any type $C:\univ$, with $b:C$ and $l:b=b$, we have a unique map
    $f:\Sone\to C$ such that
    \[f(\base) \jeq b\]
    and
    \[\ap_f(\lp) = l.\]
\end{itemize}
With ordinary inductive types, it is safe to treat the computation rules associated
with recursion and induction principles as judgemental equalities. For higher
inductive types, there are philosophical and technical issues with treating the
computation rules as judgemental equalities. More information on the status of
higher-inductive types can be found in the discussion at the end of the chapter
(Section~\ref{section:uttdiscussion}).

The induction principle has the same shape, but is unexpectedly trickier to state
explicitly. The key point is that the image of $\lp$ is no longer a simple path, but
a path living over $\lp$:
\begin{itemize}
    \item For any $C:\Sone\to \univ$ with $b:C(\base)$ and $l:\transport^{C}(\lp,b)=b$,
        we have a unique $f:\qPi{x:\Sone}C(x)$ such that
        \[f(\base)\jeq b\]
        and
        \[\apd_f(\lp) = l.\]
\end{itemize}
Note that here $\apd_f(\lp)$ must have type $\transport^C(\lp,b) = b$, so the
required equation for the image of $\lp$ does type-check.
The induction principle can be stated another way, which we give as a lemma.
\begin{lemma}\label{lemma:sone-induction}
Given any $C:\Sone\to\univ$, we have that
    \[\big(\qPi{x:\Sone}C(x)\big)\simeq
    \Big(\qSig{b:C(\base)}\transport^{B}(\lp,b)=b\Big).\]
\end{lemma}
That is, dependent functions from $\Sone$ to a type family $C$ are exactly given by
an element of $C(\base)$ and a path over $\lp$ in $C$.
\begin{proof}
    Given $f:\qPi{x:\Sone}C(x)$ define $b \defeq f(\base)$ and $p \defeq
    \apd_f(\lp)$, so that we have the pair $(b,p):\qSig{b:C(\base)}\transport^{C}(\lp,b)=b$. This gives
    us a map
    \[c:\big(\qPi{x:\Sone}C(x)\big)\to \qSig{b:C(\base)}\transport^{C}(\lp,b)=b.\]
    Now if $(b,p)$ is any element of $\qSig{b:C(\base)}\transport^{C}(\lp,b)=b$, we
    have that the dependent function $f:\qPi{x:\Sone}C(x)$ defined by $f(\base)=b$
    and $\apd_f(\lp) = p$ is the unique function $f$ such that $c(f) = (b,p)$. Hence,
    $c$ has contractible fibers.
\end{proof}

The utility of $\Sone$ as a counterexample stems from the fact that it is not a set:
\begin{lemma}
        Assuming univalence, $\Sone$ is not a set.
\end{lemma}
There are several ways to prove this; in the HoTT Book, they show that $\Sone$ is
equivalent to a type (in a higher universe) which is not a set. We instead follow the
proof that $\pi_1(\Sone)=\Z$. The type-theoretical proof is an encode-decode proof
(See Section~\ref{section:homotopies}).
\begin{proof}
    We show that $\base=\base$ is equivalent to $\Z$, by giving a type
    family \[\codefxn:\Sone\to \univ,\] along with a family of
    functions \[\encode:\qPi{x:\Sone}(\base=x)\to \codefxn(x)\] and a family of
    inverses \[\decode:\qPi{x:\Sone}\codefxn (x)\to (\base=x).\]

    We define
    \begin{align*}
        \codefxn(\base) &\defeq \Z,\\
        \apd_\codefxn(\lp) &= \ua(\succop).
    \end{align*}
    So that if $\codefxn(\base) \simeq (\base=\base)$ we have that $\base=\base$ is
    equivalent to a set which is not a proposition. Observe that we have
    $\transport^{\codefxn}(\lp,z) = z+1$ as follows:
    \[\transport^{\codefxn}(\lp,z) = \coe(\apd_{\codefxn}(\lp),z) =
    \coe(\ua(\succop),z) = z+1.\]
    Similarly, we have $\transport^{\codefxn}(\lp^{-1},z) = z-1$.

    We define $\encode$ directly, and $\decode$ by induction:
    \begin{align*}
        \encode(p) &\defeq \transport^{\codefxn}(p,0)\\
        \decode_{\base}(z) &= \lp^{z}.
    \end{align*}
    For the loop case of $\decode$, we need an element of
    \[\transport^{\lambda x.\codefxn(x) \to (\base = x)}(\lp,\lp^{-}) =
    \lp^{-}.\]
    We have
    \begin{align*}
        \transport^{\lambda x.\codefxn(x) \to (\base = x)}(\lp,\lp^{-}) 
        &= \transport^{\lambda x.\to (\base = x)}(\lp) \comp \lp^{-}\comp
        \transport^{\codefxn}(\lp^{-1})  \\
        &= (\lambda p.p\ct \lp) \comp \lp^{-}\comp
        \transport^{\codefxn}(\lp^{-1})  \\
        &= (\lambda p.p\ct \lp) \comp \lp^{-}\comp
        \pred \\
        &= \lambda n . \lp^{n-1}\ct \lp\\
        &= \lambda n.\lp^n.
    \end{align*}

    We need to see that $\encode$ and $\decode$ are inverse. To show that
    $\decode_x(\encode_x(p)) = p$ for all $x:\Sone$ and $p:\base=x$
    we use based path induction:
    \[\decode_{\base}(\encode_{\base}(\refl)) = \decode(\transport^{\codefxn}(\refl,0))
    = \decode(0) = \lp^{0} = \refl.\]
    For the other direction, we use the induction principle for $\Sone$. The claim is
    that for all $z:\Z$ we have ${\encode_{\base}(\decode_{\base}(z))=z}$. That is, we
    want to see
    $\transport^{\codefxn}(\lp^z,0) = z,$
    but we remarked earlier that $\transport^{\codefxn}(\lp,x) = \succop(x)$, and so
    we have
    \[\transport^{\codefxn}(\lp^z,0) = 0 + z = z.\qedhere\]
\end{proof}
Note that even without univalence, we cannot have that $\Sone$ is a set, since our
type theory is still consistent with univalence. This means we can still use $\Sone$
to see whether a notion behaves well with general types.

\begin{lemma}
    We have $\qPi{x,y:\Sone}\trunc{x=y}$.
\end{lemma}
\begin{proof}
    By induction on both $x$ and $y$: fix $x\jeq y \jeq \base$. Then we have
    $\refl:x=y$. Truncating this element gives $ |\refl |:\trunc{x=y}$. We need to
    give a path over $\lp$ in both coordinates from $ |\refl |$ to $ |\refl |$, but
    since $\trunc{x=y}$ is a proposition, there is nothing to check.
\end{proof}
\begin{corollary}\label{lemma:sone-constant}
    If $Y$ is a set and $f:\Sone\to Y$, then for any $x,y:\Sone$ we have that
    $f(x)=f(y)$.
\end{corollary}
\begin{proof}
    Fix $x,y:\Sone$. By the previous
    lemma, we have $\trunc{x=y}$. As $Y$ is a set, $f(x)=f(y)$ is a
    proposition, so we may assume we have $p:x=y$. Then $\ap_f(p):f(x)=f(y)$.
\end{proof}

\section{Discussion}\label{section:uttdiscussion}
Truncation was already being studied \cite{mendler,pfenning2001} as \emph{squash
types} and \emph{bracket types} before the univalent notion of proposition was
introduced. In fact, Awodey and Bauer~\cite{AwodeyBauer2004Bracket} proposed
propositions as bracket types already in 2004, in the context of extensional type
theory. When the general notion of higher-inductive type was introduced, they quickly
became a standard example. Higher inductive types are, however,  poorly understood.
The schema is intuitively clear, but a mathematically satisfying account of their
syntax and semantics is surprisingly difficult. Lumsdaine and Shulman recently
released an awaited approach to the semantics of HITs via \emph{cell-monads with
parameters}~\cite{ls2017semantics}, but these only work for HITs of a certain form,
and it is not clear if there is a syntactic scheme that corresponds precisely to the
semantic notions they give. A semantic view of HITs from the cubical approach is
given in \cite{chm2018cubical}, and a view of HITs, focusing on those that exist at
the level of groupoids (which is enough for the HITs we are interested in here) is
in~\cite{DybjerMoeneclaey2017}. On the syntactic side, recent progress can be seen in
\cite{kaposi2018hiit} and in the more sweeping~\cite{computationalttIV}, which also
contains a more thorough summary of the state of the art.

Higher-inductive types represent one standard approach to truncation, while
propositional resize represents another. It is open whether Voevodsky's resizing rules
are consistent with MLTT~\cite{Coquand2018VV}. The weaker form of
resizing discussed in the HoTT Book holds in the simplicial set model, and in fact in
any model which validates the law of excluded middle.

The univalence axiom was first proposed by Vladimir Voevodsky, motivated by his
simplicial set model~\cite{ssetmodel}. However, it was noticed (under the name
\emph{universe extensionality}) in the groupoid model already by Hofmann and
Streicher~\cite{gpdmodel} in 1995, although all types in the universe are sets in the
groupoid model.  The notion of $h$-level is also due to Voevodsky.  The first proof
of function extensionality from univalence was given by Voevodsky in his Foundations
library~\cite{UniMath}, and it has been been modified and reworked in various ways
since. The proof of function extensionality from naive function
extensionality was first done by Voevdosky in the Foundations library, but seems to
have been overlooked until Mart\'in Escard\'o found it and translated it to Agda in
his TypeTopology library~\cite{typetopology}.

Extensionality axioms, such as univalence, have always been treated with some
trepidation in type theory, since axioms break computational properties like
canonicity---the metatheorem that all terms of type $\nat$ are judgementally equal to
a numeral. Since the outset, there has been hope that the univalence axiom can be
given a computational interpretation. Following the introduction of the first model
in cubical sets~\cite{huber2016,bezem2014model}, a great deal of work has been done
on developing \emph{cubical type theories}~\cite{cubicaltt} which take paths as
primitive, and in which the univalence axiom is a theorem. Importantly, canonicity
holds in cubical type theory~\cite{huber2016canonicity}.


%% file: chapters/foundations.tex
\chapter{Mathematical foundations}\label{chapter:foundations}
In the previous chapter, we developed the core ideas of univalent mathematics. Here
we turn to the more traditional mathematical notions we will need in part 2. Our
approach is still heavily informed by the univalent perspective. In particular,
choice principles in Section~\ref{section:choice} are stated in a form which is
unique to univalent mathematics. We begin by examining relations
and a predicative notion of power set (Section~\ref{section:predicates}).
Importantly, we will see (in Theorem~\ref{thm:functions-are-relations}) that
assuming the univalence axiom, functions are single-valued relations, which will be a
prototype of similar results
about partial maps in Chapter~\ref{chapter:partial-functions}. We then discuss a
notion of constancy introduced in~\cite{keca2016} that will appear in a few places,
before turning to the natural numbers. The natural numbers satisfy a property that
looks superficially like Markov's principle, which allows us to remove truncations in
certain cases. This is an interesting example, and will be crucial when working with
computability in Part~\ref{part:three}; we also introduce here a
representation of a predicate on $\nat$ as a function $\nat\to\nat$ (which we give
the ad hoc name \emph{characteristic function}) that will be used to define
computability of predicates. Before turning to choice, we discuss an internal notion
of monad via Kleisli extension.

\section{Assumptions}\label{section:assumptions}
In the previous chapter, we introduced several extensions of MLTT which are used in
univalent approaches to mathematics. Here we lay out explicitly the assumptions we use in the
rest of the thesis.

Each of the following assumptions is parameterized by a universe $\univ$, and we
assume them for all universes.
\begin{assumption}[Function extensionality] Propositions form a strong exponential
    ideal in every universe.
\end{assumption}
\begin{assumption}[Proposition extensionality] For every universe $\univ$ and for
    every propositions $P,Q:\univ$, if $P$ and $Q$ are logically equivalent, then
    $P=Q$.
\end{assumption}
\begin{assumption}[Propositional truncation] For every universe $\univ$ and any type $X:\univ$ we have a
    proposition
    $\trunc{X}:\univ$ and a map ${|{-}|}:X\to\trunc{X}$ universal
    among all propositions which $X$ maps into.
\end{assumption}
When we need univalence, we will assume it explicitly. We will not assume
resizing rules, but we occasionally discuss their consequences. Similarly, we will
not use any higher inductive types besides truncations, but we will
sometimes discuss their properties. In particular, we will compare our work to work
by others using truncations (see Sections~\ref{section:delay} and~\ref{section:qiit},
as well as Section~\ref{section:rosolini}), and we will look at the circle for
motivation in Section~\ref{section:sdt-comparison}.

\section{Predicates and relations}\label{section:predicates}
\index{predicate|(}
In category theory, a \emph{subobject} of an object $X$ is an equivalence class of
monos into $X$. Replacing mono with embedding (Section~\ref{section:embeddings}), we
could take a subobject to be an equivalence class of embeddings into $X$. Since we
should write the type of embeddings into $X$ as
\[\qSig{A:\univ}\qSig{f:A\to X}\isEmbedding(f),\]
an equality between embeddings $(A,f,-)$ and $(B,g,-)$ would consist of a path
$p:A=B$ such that $g\comp \coe(p) = f$. If we have univalence, equality here is
exactly the equivalence relation we need; if not we should replace equality in the
first component with equivalence as types.

With univalence, the equivalence relation of interest is simply equality, so there is
no quotienting to be done. However, even without univalence, for any
embedding $f:A\to X$ we can find a canonical representative of its equivalence class, by looking
at the total space of the fibers over $f$. In fact, this suggests a definition of
subtype which works better in our framework: predicates on $X$. Explicitly, given any
embedding $f:A\to X$, we have the predicate $\fib_f:X\to\univ$, which is proposition
valued precisely because $f$ is an embedding. Conversely, given any predicate
$P:X\to\univ$ the
first projection $\big(\qSig{x:X}P(x)\big)\to X$ is an embedding.
\begin{lemma}
    For any $X:\univ$, the type of predicates $X\to\Prop$ is a retract of the type
    \[\qSig{A:\univ}\qSig{f:A\to X}\isEmbedding(f)\]
    of embeddings into $X$. Moreover, assuming univalence, the retraction is an
    equivalence.
\end{lemma}
\begin{proof}
    Let $f:A\to X$ be an embedding and define $R:X\to\univ$ by
    \[R(x) \defeq \fib_f(x).\]
    As $f$ is an embedding, $R$ is proposition valued. Conversely, if $R$ is a
    predicate, define 
    \[A\defeq \qSig{x:X}R(x),\]
    with $f:A\to X$ the first projection. The fiber of the first projection at $x$ is
    \[\qSig{(x',r'):\qSig{x':X}R(x)}x'=x,\]
    which is equivalent to $R(x)$ by Lemma~\ref{lemma:path-families}. Note that by
    proposition extensionality and function extensionality, we have that $\fib_{\pr_0}
    = R$, so that $(A,f)\mapsto \fib_f(x)$ is a retraction. Moreover, by
    Lemma~\ref{lemma:domain-is-fiber}, we know that $A\simeq \qSig{x:X}\fib_f(x)$ and
    that this equivalence composes with the projection to give $f$, so that by
    univalence, $(A,f) = (\qSig{x:X}\fib_f(x),\pr_0)$.
\end{proof}
With this in mind, we make the following definition.
\begin{definition}
    A \emph{subtype} of a type $X$ is a map $X\to\Prop$. If $X$ is a set, we will use
    the term \emph{subset} interchangeably with \emph{subtype}. We will call the type
    of subtypes of $X$ the \emph{power set} of $X$ and write
    \[\powerset{X}\defeq X\to \Prop.\]
    If $A:\powerset{X}$ we sometimes write $x\in A$ for the predicate $A(x)$.
    Following our terminology that a type $Y$ is inhabited when $\trunc{Y}$, we say
    that $A$ \emph{is an inhabited} subtype of $X$ when
    \[\trunc{\qSig{x:X}x\in A}.\]
\end{definition}
The term power \emph{set} is indeed justified:
\begin{lemma}
    The power set of $X$ is always a set.
\end{lemma}
\begin{proof}
    We know that $\Prop$ is a set, so by Lemma~\ref{lemma:ntypes-exp-ideal},
    $X\to\Prop$ is a set.
\end{proof}

In a traditional set-theoretic formalization, relations are given either as formulas,
or as subsets of a cartesian product, and then functions $f:X\to Y$ are defined as relations for
which there is a unique $y:Y$ such that $f(x,y)$. While functions are primitive
in type theory, we can show that this correspondence holds, but to account for
general types, we must replace relations with arbitrary type families.
\begin{definition}
    A type family $R:X\to Y\to\univ$ is \emph{functional} when for each $x:X$, there is
    a unique $y:Y$ such that $R(x,y)$. That is, when
    \[\qPi{x:X}\isContr(\qSig{y:Y}R(x,y)).\]
\end{definition}
This definition demonstrates a peculiar feature of univalent definitions. The
traditional way to define a \emph{functional relation} would be to say that $R$ is
functional when for each $x:X$ there is $y:Y$ such that $R(x,y)$ and whenever
$R(x,y)$ and $R(x,y')$ then $y=y'$. This becomes,
\[
\qPi{x:X}\qSig{y:Y}R(x,y)\times \qPi{y,y':Y}R(x,y)\to R(x,y')\to y=y'.
\]
This definition is inadequate in a univalent framework. In the case where $Y$
is not a set, and in the case where $R$ is not valued in propositions, the type above
may have many elements, but we would like to use being functional as a property: we
would like not only that there is a unique $y:Y$ such that $R(x,y)$, but that there
is a unique pair $(y,r):\qSig{y:Y}R(x,y)$. This situation is captured by contractibility.

In the case where $Y$ is a set, any functional type family is already a relation, so
the two types above are equivalent.
\begin{lemma}\label{lemma:propositional-families-have-propositional-sums}
    If $Y$ is a set, and $A:Y\to\univ$ such that $\isProp(\qSig{y:Y}{A(y)})$, then
    $Y$ is valued in propositions. That is,
    \[\qPi{y;Y}{\isProp(A(y))}.\]
\end{lemma}
\begin{proof}
    Fix $y$ so that we have
    \[A(y)\simeq \qSig{(y',a):\qSig{y':Y}{A(y')}} y=y',\]
    by Lemma~\ref{lemma:path-families}, and this is the sum of propositions indexed
    by a proposition.
\end{proof}

Given a type family $R:X\to Y\to \univ$ with a witness $w$ that $R$ is functional,
there is a function $f_R:X\to Y$ defined by
\[f_R(x) \defeq \pi_0(w(x)).\]
Conversely, given any $f:X\to Y$ there is a type family $R_f:X\to Y\to\univ$ given by
\[R_f(x,y) \defeq f(x) = y.\]
Then $\qSig{y:Y}f(x) = y$ is the singleton based at $f(x)$, so is contractible. In
the presence of univalence, this gives an equivalence between the type of functional
relations and the type of functions. And
using Lemma~\ref{lemma:propositional-families-have-propositional-sums}, if $Y$ is a set,
then univalence is not needed.
\begin{theorem}[Rijke, \cite{rijke2012homotopy}]\label{thm:functions-are-relations}
    For any $X,Y:\univ$, the map taking a functional $R:X\to Y\to\univ$ to $f_R$ is a
    left-inverse of the map taking $f:X\to Y$ to $R_f$. Moreover, for any $x:X$ and
    $y:Y$, the type $R_{f_R}(x,y)$ is equivalent to $R(x,y)$.

    Hence, assuming univalence, the map $f\mapsto R_f$ gives rise to an equivalence
    \[(X\to Y)\simeq \qSig{R:X\to Y\to\univ}\qPi{x:X}\isContr(\qSig{y:Y}R(x,y)).\]

    Moreover, if $Y$ is a set, univalence is not needed in the proof.
\end{theorem}
\begin{proof}
    For any $f:X\to Y$ we have $f_{R_f}(x) = f(x)$ by definition, since
    $(f(x),\refl)$ is the center of contraction of $\qSig{y:Y}f(x)=y$.

    For the other direction, we wish to see that for any $x$ and $y$,
    \[R_{f_R}(x,y)\simeq R(x,y).\]
    By Theorem~\ref{thm:fiberwise-equivalence}, it is enough to find an equivalence
    \[\big(\qSig{x:X}R(x,y)\big)\simeq \big(\qSig{x:X}R_{f_R}(x,y)\big),\]
    which is the identity on the first component. 

    Now, suppose $Y$ is a set.
    Let us first verify that not only is $\qSig{y:Y}R(x,y)$ a proposition, but also
    that $R(x,y)$ is a proposition for all~$y:Y$. With $x$ fixed, define
    $A:Y\to\univ$ to be $A(y)\defeq R(x,y)$, so that by
    Lemma~\ref{lemma:propositional-families-have-propositional-sums} we know
    $\isProp(A(y))$ for all~$y:Y$. That is, $\isProp(R(x,y))$ for all~$y:Y$ as
    desired.

    Letting $x$ vary again, this means that $\qPi{x:X}\qPi{y:Y}\isProp(R(x,y))$.
    Similarly, since $Y$ is a set, for any $f:X\to Y$, and elements $x:X$ and $y:Y$
    the type $f(x)=y$ is a proposition, and so both $R_{f_R}(x,y)$ and $R(x,y)$ are
    propositions.  Then the result follows from
    Theorem~\ref{thm:functions-are-relations} by proposition extensionality, since
    $R_{f_R}(x,y)\simeq R(x,y)$.
\end{proof}
\index{predicate|)}

\section{Constancy}\label{section:constancy}
\index{function!constant|(}
This section covers material from~\cite{keca2016,keca2013}. Our goal is to show
that $\trunc{X}$ has the universal property of the quotient of $X$ by the total
relation $\lambda x,y.\unittype:X\to X\to\univ$. That is, we can map out of
$\trunc{X}$ to sets, even though the universal property of truncation only tells us
how to define maps to propositions.
\begin{definition}
    A function $f:A\to B$ is \nameas{constant}{function!constant} if for all $x,y:A$, we have $f(x)
    = f(y)$.
\end{definition}
Constant functions are called \emph{weakly constant} in~\cite{keca2016,keca2013}.

\begin{theorem}\label{thm:keca-factorization}
    Let $X,Y:\univ$ and let $P:\univ$ be a type for which $P\to Y$. If $X$ implies that $P$
    is contractible, then $\trunc{X}\to Y$. In particular, if $f:X\to Y$ is a
    function which factors through $P$, then $f$ factors through $\trunc{X}$.
\end{theorem}
\begin{proof}
    Consider the following diagram.

    \hspace{\mathindent}\begin{tikzcd}
        X \arrow[d,swap, "|-|"] \arrow[r] & \isContr(P) \arrow[d, "\pr_0"] & \\
        \trunc{X} \arrow[ur, dashrightarrow] & P \arrow[r] & Y
    \end{tikzcd}

    As $\isContr(P)$ is a proposition, we have the dashed function $\trunc{X}\to
    \isContr(P)$; composition gives $\trunc{X}\to Y$.

    Now let $f:X\to Y$ factor through $P$ as $g\comp h$. As
    $\trunc{X}\to\isContr(P)$, we have that
    $\trunc{X}\times P$ is a proposition, and moreover we have $X\to \trunc{X}\times
    P$, given by $x\mapsto (|x|,h(x))$. Since $\trunc{X}\times P$ is a proposition,
    we have a factorization of $f$ as

    \hspace{\mathindent}\begin{tikzcd}
        X \ar[r, "|-|"] \ar[rrr, bend right=20, swap,  "h"] & \trunc{X} \ar[r] &
        \trunc{X}\times P \ar[r,"\pr_1"] & P \ar[r, "g"] & Y.
    \end{tikzcd}
\end{proof}
Note that if $P$ is a proposition and a function $f:X\to Y$ factors through $P$, then
$X$ implies that $P$ is contractible. Then if $f:X\to Y$ factors through any
proposition $P$, then $f$ factors through~$\trunc{X}$.
\begin{theorem}\label{thm:constant-functions-factor}
    If $f:X\to Y$ is constant and $Y$ is a set, then $\trunc{X}\to Y$.
\end{theorem}
\begin{proof}
    We know that $f$ factors through $\im(f)$, so it suffices to show that when $Y$
    is a set and $f$ is constant, then $\im(f)$ is a proposition.

    Let $(y,u),(y',u') : \qSig{y:Y}\trunc{\qSig{x:X}f(x) = y}$. Since the second
    component inhabits a proposition, we need only a path $y=y'$. As $y$ is a set,
    we know that $y=y'$ is a proposition, so we may untruncate our assumptions
    $u:\trunc{\fib_f(y)}$ and $u':\trunc{\fib_f(y')}$ and use explicit witnesses
    $(x,p):\qSig{x:X}f(x) = y$ and $(x',p'):\qSig{x:X}f(x) = y'$. Then we have
    \[y \eqby{p} f(x) = f(x') \eqby{p'} y',\]
    where the middle equality comes from constancy of~$f$.

    Explicitly, the above argument above gives
    \[\qPi{y,y':Y}\fib_f(y) \to \fib_f(y') \to y = y',\]
    and since $y=y'$ is a proposition, we conclude
    \[\qPi{y,y':Y}\trunc{\fib_f(y)} \to \trunc{\fib_f(y')} \to y = y'.\]
    By reshuffling and uncurrying, we get
    \[\qPi{u,v:\im(f)} u = v.\]
    Then $f$ factors through a proposition, so by
    Theorem~\ref{thm:keca-factorization}, we are done.
\end{proof}

The factorization of constant functions through propositions is related to
another fact about constancy: if $f$ is constant, then so is $\ap_f$.
\begin{lemma}\label{lemma:const-ap}
    If $f:X\to Y$ is constant, then 
    $\ap_f:x=x'\to f(x)=f(x')$ is constant for all $x,x':X$. In particular, $\ap_f(p) =
    \refl_{f(x)}$
    whenever $p:x=x$.
\end{lemma}
\begin{proof}
    Let $k$ be a witness that $f$ is constant. By path induction, for any
    $p:x=x'$ we have that $\ap_f(p) = k_{x,x}^{-1} \ct k_{x,x'}.$
    Indeed, when $p$ is reflexivity, we have
    \[\refl_{f(x)} = k_{x,x}^{-1} \ct k_{x,x}.\]
    Now if $p,q:x=x'$, then
    \[p = (k_{x,x}^{-1} \ct k_{x,x'}) = q.\qedhere\]
\end{proof}
\index{function!constant|)}

\section{Predicates on $\nat$}\label{section:nat}
\index{type!of natural numbers|(}
Since we will be focused on computability theory, we now turn to predicates on
$\nat$.
\begin{lemma}
    The standard ordering on $\nat$ is decidable and is a total order.
\end{lemma}
\begin{lemma}[Bounded search]
    Let $P$ be a decidable predicate on $\nat$. Then for all $n:\nat$, the types
    \[\qSig{k:\nat}{P(k) \times (k\le n)}\]
    and
    \[\qSig{k:\nat}{P(k) \times (k< n)}\]
    are decidable.
\end{lemma}
\begin{proof}
    The proof is straightforward by induction: for $0$ it holds immediately, since
    \[\left(\qSig{k:\nat}{P(k)\times (k\le 0)}\right)\simeq P(0).\]
    For $n+1$, we have
    \[\left(\qSig{k:\nat}{P(k)\times (k\le n+1)}\right) \simeq
    P(n+1)+\left(\qSig{k:\nat}{P(k)\times (k\le n)}\right)\]
    If $P(n+1)$ holds, we're done. Otherwise, check $\qSig{k:\nat}{P(k)\times (k\le
    n)}$, which is decidable by
    the inductive hypothesis.
\end{proof}
\begin{lemma}\label{lemma:untruncate-decidable-predicates}
    If $P$ is a decidable predicate on $\nat$, then
    \[(\qExists{n:\nat}{P(n)})\simeq\qSig{n:\nat}{\left(P(n)\times
        \qPi{k:\nat}{P(k)\rightarrow n\le k}\right)}\]

    In particular, we have $\trunc{\sum_{n:\nat} P(n)}\to \sum_{n:\nat}P(n)$.
\end{lemma}
\begin{proof}
    Let $Q\defeq\qSig{n:\nat}{\left(P(n)\times\qPi{k:\nat}{P(k)\rightarrow n\le
    k}\right)}$.

    We proceed in 3 steps: first we show $\left(\qSig{n:\nat}{P(n)}\right)\to Q$
    using bounded search. Next, we show that $Q$ is a proposition, so that we can apply
    the universal property of truncation to get a map
     \[\trunc{\qSig{n:\nat}{P(n)}}\to Q.\]
    Finally, we project out of $Q$ to see that $Q\rightarrow \qSig{n:\nat}{P(n)}$.

    First, observe that under our assumptions, we have that for all $k$,
    \[\left(\neg\qSig{j:\nat}{P(j)\times j< k}\right)\simeq
    \left(\qPi{j:\nat}{P(j)\rightarrow k\le j}\right).\]
    By bounded search, the predicate
    \[ P'(n) \defeq \qSig{k:\nat}{P(k)\times (k\le n)\times
    \left(\neg\qSig{j:\nat}{P(j)\times j< k}\right)}\]
    is decidable. Let $n$ be such that $P(n)$. If $P'(n)$ fails, then a quick
    argument shows that $P(n)$ must also fail, so we must have $P'(n)$. Let $k$
    be such that $P(k)\times (k\le n)\times \left(\neg\qSig{j:\nat}{P(j)\times j< k}\right)$.
    Then we also have $m:\qPi{j:\nat}{P(j)\rightarrow k\le j}$, and $(k,(p,m)):Q$, where
    $p:P(k)$.

    If $(k,w):Q$ and $(k',w'):Q$, we must have $k\le k'\le k$, so we have $p:k=k'$.
    Moreover, $P(k)\times\qPi{j:\nat}{P(j)\rightarrow k\le j}$ is a proposition,
    so we get $\transport(p,w)=w'$. That is, $Q$ is a proposition, and $\qSig{n:\nat}{P(n)}$
    implies $Q$, so then does $\exists(n:\nat),P(n)$.

    Finally, we have the map $Q\rightarrow\qSig{n:\nat}{P(n)}$ defined by
    \[(k,(p,m))\mapsto (k,p).\qedhere\]
\end{proof}
It is tempting to think of propositions as types which contain no more
information than their inhabitedness. The above shows that this is misleading,
as the type \[\qSig{n:\nat}{P(n)\times \qPi{k:\nat}{P(k)\rightarrow n\le k}}\]
contains useful computational content. Moreover, we see here and in the
definition of equivalence that there are predicates for which
$(\qExists{x:A}{P(x)})\to \qSig{x:A}{P(x)}$. In other words, if there exists an $x$
such that $P(x)$, then we can find one explicitly.
The proof of Lemma~\ref{lemma:untruncate-decidable-predicates} can be found
in~\cite{escardoxu2015}, where they explicitly assume function extensionality in the
proof. However, they sketch another way to untruncate
$\trunc{\qSig{n:\nat}P(n)}$, without using function extensionality.
The other argument uses ideas from Section~\ref{section:constancy}. For an arbitrary
type $X$, if there is a constant function $f:X\to X$, then the type of fixed points
of $f$ has the universal property of the truncation.

For any $f:X\to X$, define the type
\[\fix(f)\defeq \qSig{x:X}x = f(x).\]
\begin{lemma}[Fixed Point Lemma\cite{keca2013}]\label{lemma:fix-trunc}
    Let $X$ be a type and $f:X\to X$ be constant. Then
    $\fix(f)\simeq\trunc{X}$.
\end{lemma}
\begin{proof}
    Fix $k:\qPi{x,y:X}f(x)=f(y)$.

    First, we need to see that $\fix(f)$ is a proposition. Given
    $(x,p),(y,q):\fix(f)$ we have the path $r:x=y$ given by
    \[x \eqby{p} f(x) \eqby{k_{x,y}} f(y) \eqby{q^{-1}} y.\]
    We want a path $t:x=y$ such that $\transport(t,p) = q$, and $r$ does not suffice,
    as we can compute $\transport(r,p) = r^{-1}\ct p$ (in a moment, we give this
    computation for our path of interest). 
    By Theorem~\ref{lemma:tp-path-functions}, for any $t:x=y$ and $p:x=f(x)$,
    $\transport^{\lambda x.x=f(x)}(t,p) = t^{-1}\ct p\ct \ap_f(t)$. So defining
    $t\defeq p\ct \ap_f(r) \ct q^{-1}$, we have
    \[\transport(t,p) = (q\ct \ap_f(r^{-1})\ct p^{-1}) \ct p \ct
    \ap_f(p\ct\ap_f(r)\ct q^{-1}).\]
    By Lemma~\ref{lemma:const-ap}, we can say that $\ap_f(p\ct\ap_f(r)\ct q^{-1}) =
    \ap_f(r)$, and so we have
    \[q\ct \ap_f(r^{-1})\ct p^{-1} \ct p \ct
    \ap_f(r) = q \ct \ap_f(r^{-1}) \ct \ap_f(r) = q.\]
    Then, we have $(x,p) = (y,q)$ with witness given by our path $t$ and the
    above argument. That is, $\fix(f)$ is a proposition.

    We have the projection map $\fix(f)\to X$, as well as a map $X\to\fix(f)$ given
    by 
    \[x\mapsto (f(x),k(f(x),x)).\] That is, $\fix(f)$ is a proposition which is
    logically equivalent to $X$. The theorem follows by Lemma~\ref{leq-is-trunc}.
\end{proof}
Now fix a decidable predicate $P:\nat\to\univ$. This factors through a map
$f:\nat\to\nat$, which takes value $0$ at $x$ if $\neg R(x)$ and takes value $1$ at
$x$ if $P(x)$. Consider the type
\[X\defeq \qSig{n:\nat}f(n)=1.\]
We can define a map $f:X\to X$ as follows: for
$(n,w):\qSig{n:\nat}f(n)=1$, we find the least $k\le n$ with $v:f(k)=1$, using a
bounded search. Take $f(n,w)=(k,v)$. Since there is at most one least $k$ such that
$f(k)=1$, this function $f$ is constant. Then by Lemma~\ref{lemma:fix-trunc}, we have
that 
\[\fix(f)\simeq \trunc{\qSig{n:\nat}f(n)=1}.\]
Because $\fix(f)\lequiv X$, we then get a map
$\trunc{\qSig{n:\nat}f(n)=1}\to \qSig{n:\nat}f(n)=1$.

In the above, we used that a decidable predicate $P:\nat\to\univ$ factors as $\bool\circ p$
where $\bool$ is the inclusion $\bool\to\Prop$, and some function $p:\nat\to 2$. The standard
inclusion $i:2\to\nat$ is a section of the function $\pos:\nat\to 2$ taking $0$ to
$0$ and all positive numbers to $1$. Then $P$ factors through a map
$\chi_P:\nat\to\nat$ which we will call a \emph{characteristic function} for $P$.
By function extensionality, there is a unique characteristic function for $P$. We can
put this together as follows
\begin{definition}\label{def:characteristic-function}
    A \emph{characteristic function} for a predicate $P:\nat\to\Prop$ is a function
    $\chi_P:\nat\to\nat$ such that
    \[(\forall (n), f(n) = 0 \vee f(n)=1)
        \wedge (\chi_P(n) = 0 \Leftrightarrow \neg R(n))
        \wedge (\chi_P(n) = 1 \Leftrightarrow R(n)).\]
\end{definition}
\begin{theorem}
    The type of characteristic functions for a predicate $P$ is a proposition, and is
    inhabited precisely if $P$ is decidable.
\end{theorem}
\begin{proof}
    Fix characteristic functions $f$ and $g$ for $R$. We know $\big(f(n) = 0\big) +
    \big(f(n)=1\big)$.
    In the first case, we have $\neg P(n)$ and so $g(n) = 0$. Otherwise, $f(n)=1$ and so
    $P(n)$; hence, $g(n)=1$.

    If $P$ has a characteristic function, then since $\nat$ has decidable equality,
    $P$ is decidable. Conversely, if $P$ is decidable, then $i\circ p$ is a
    characteristic function, where $p:\nat\to 2$ factors $P$.
\end{proof}
\index{type!of natural numbers|)}

\section{Monads}\label{section:monads}
\index{monad|(}

A monad on a category $\cC$ is a functor $T:\cC\to\cC$ equipped with natural
transformations $\eta:\id_{\cC}\to T$ and $\mu:T^2\to T$ such that the following
diagrams commute.

\begin{center}
\begin{tikzcd}
    T^3(X) \arrow[r,"\mu_{TX}"] \arrow[d,swap,"T(\mu_{X})"] & T^2(X) \arrow[d,"\mu_{X}"] \\
    T^2(X) \arrow[r,swap,"\mu_{X}"]                         & T(X)
\end{tikzcd}
    \hspace{0.05\textwidth}
\begin{tikzcd}
    T^2(X) \arrow[dr,swap,"\mu_{X}"] &
        T(X) \arrow[l,swap,"\eta_{T(X)}"] \arrow[r,"T(\eta_X)"] \arrow[d,equal] &
        T^2(X) \arrow[dl,"\mu_{X}"] \\
        & T(X) &
\end{tikzcd}
\end{center}
A monad $T$ can equivalently be presented as a \emph{Kleisli triple}: an operation
$T$ on objects, together with a family of maps $\eta_X:X\to T(X)$ and a \emph{Kleisli
extension} operator which takes a map $f:X\to T(Y)$ to a map $\klext{f}:T(X)\to
T(Y)$, satisfying the \emph{Kleisli laws}:
\begin{align*}
    \klext{\eta_X} &= \id_{T(X)} \\
    f &= \klext{f} \comp \eta \\
    \klext{(\klext{g}\comp f)} &= (\klext{g})\comp (\klext{f}).
\end{align*}
Given $\eta$ and $\klext{(-)}$ we get the application on morphisms $f:X\to Y$ as 
\[T(f) \defeq \klext{(\eta_{Y}\comp f)}\]
and  $\mu:T^2 \to T$ as
\[\mu_{X} \defeq \klext{\id_{T(X)}}.\]
Functoriality of $T$, naturality of $\eta$ and $\mu$ and the monad laws follow from
the Kleisli equations. Conversely, given a monad $(T,\eta,\mu)$ and $f:X\to Y$ in
$\cC$, define $\klext{f} \defeq \mu_{Y} \comp T(f)$. Again, the Kleisli laws follow from
the monad laws and naturality of $\eta$ and $\mu$.

Given a Kleisli triple $(T,\eta,\klext{(-)})$, we can define \emph{Kleisli
composition} of $f:X\to T(Y)$ and $g:Y\to T(Z)$ by
\[g \kcomp f \defeq (\klext{g})\comp f.\]
We may internalize the definition of a monad by looking at type operators.
    \index[symbol]{klext@{$\klext{(-)}$}|see {Kleisli extension}}
\begin{definition}
    A \name{monad} is a map $T:\univ \to \univ$ together with a family of
    maps \index{eta@{$\eta$}!unit of a monad|textit}
    \[\eta:\qPi{X:\univ} X\to T(X)\]
    and a Kleisli extension operator
    \index[symbol]{klext@{$\klext{(-)}$}}
    \index[symbol]{compkleisli@{$\kcomp$}|see{Kleisli extension}}
    \index{Kleisli extension!for a monad|textit}
    \[\klext{(-)} : \qPi{X,Y:\univ}(X\to T Y) \to (T X \to TY)\]
    such that
    \begin{align*}
        \klext{\eta_X} &= \id_{T(X)} \\
        f &= \klext{f} \comp \eta \\
        \klext{(\klext{g}\comp f)} &= (\klext{g})\comp (\klext{f})
    \end{align*}
\end{definition}

\begin{definition}\label{def:submonad}
    If $\eta$ and $\klext{(-)}$ give $T:\univ\to\univ$ the structure of a monad,
    a \nameas{submonad}{monad!submonad} of $T$ is an operation
    $S:\univ\to\univ$ such that for every $X$ there is an embedding $i_X:SX\to TX$
    such that
    \begin{itemize}
        \item for every $X:\univ$, $\eta_X$ factors through $i_X$, and
        \item the embedding $i$ commutes with Kleisli extension.
    \end{itemize}

    That is, we have $\eta' : X \to SX$ and for every $f:X\to S Y$, we have
    ${\klext{f}}':S X\to S Y$
    such that the following diagram commutes:

    \begin{center}
    \begin{tikzcd}
        X \arrow[r,"f"] \arrow[d,"\eta'"] \arrow[dd,bend right, "\eta",swap] & S Y
        \arrow[d,"i"] \\
        S X \arrow[d,"i"] \arrow[ur,"{\klext{f}}'",swap]  & T Y \\
        T X \arrow[ur, "\klext{(i\comp f)}",swap] &
    \end{tikzcd}
    \end{center}
\end{definition}
We will be interested in the \emph{lifting} monads $\Lift$ of
Chapter~\ref{chapter:partial-functions}. However, they are formed via quantification
over a universe $\univ$, and so can raise the universe level; naively, we have
the lifting as a map $\Lift:\univ[0]\to\univ[1]$; We discuss this discrepancy in Section~\ref{section:size},
where we see that in fact lifting can be applied at higher universes, and there has
type $\univ[i]\to\univ[i]$, so indeed forms a monad as stated above. However, size
issues can also be resolved in another way, by making the definition of monad universe
polymorphic: we can take a monad to be a map $T:\univ\to\vniv$, for two universes $\univ$
and $\vniv$, and then we have $\eta:\qPi{X:\univ}X\to T(X)$, where then the type of
$\eta$ lives in universe $\max(\univ',\vniv)$, where $\univ'$ is the universe above
$\univ$, as does the type of the Kleisli extension operator
$\klext{(-)}:\qPi{X,Y:\univ}:(X\to T(Y))\to T(X)\to T(Y)$.
\index{monad|)}

\section{Choice principles}\label{section:choice}
\index{choice!axiom of|(}
The forward direction of the following equivalence is often called the \emph{type
theoretic axiom of choice}.
\begin{theorem}\label{thm:ttchoice}
    For any $A,B:\univ$ and $R:A\to B\to \univ$, we have an equivalence,
    \[\Big(\qPi{x:A}\qSig{y:B}{R(x,y)}\Big)
    \simeq\Big(\qSig{f:A\to B}{\qPi{x:A}{R(x,f(x))}}\Big)\]
\end{theorem}
\begin{proof}
    Let
    \[ P\defeq\qPi{x:A}\qSig{y:B}{R(x,y)},\]
    and
    \[ Q\defeq\qSig{f:A\to B}{\qPi{x:A}{R(x,f(x))}},\]
    so that we are looking for an equivalence $P\simeq Q$. Define $h:P\to Q$ by
    \[h(\psi) \defeq \big(\lambda x.\pr_0(\psi\, x),\lambda x.\pr_1(\psi\, x)\big)\]
    and define $g:Q\to P$ by
    \[g(f,\varphi)\defeq \lambda x. \big(f(x),\varphi(x)\big).\]
    Then we have 
    \[g(h (\psi)) = g\big(\lambda x.\pr_0(\psi\, x),\lambda x.\pr_1(\psi\, x)\big) = \lambda
    x.\big(\pr_0(\psi\, x), \pr_1(\psi\, x)\big) = \lambda x.\psi(x) = \psi,\]
    and
    \[h(g(f,\varphi)) = h(\lambda x. (f(x),\varphi(x))) = (\lambda x.f(x),\lambda
    x.\varphi(x)) = (f,\varphi).\]
    So $h$ and $g$ are inverse.
\end{proof}
It is somewhat misleading to call this ``the axiom of choice'' as the notion of
existence given by $\Sigma$ is too strong for this to
properly interpret the axiom of choice: $\Sigma$ asserts a given witness, which we
can then extract in a functional way. It is more in line with the usual uses of
choice to instead weaken $\Sigma$ to exists $\exists$.
However, asserting this for all types and type families is too strong
\cite[Lemma 3.8.5]{hottbook}, so we can only assert the axiom of choice for sets.
\begin{axiom}
    The \nameas{axiom of choice}{choice!axiom of} posits that
    given sets $A$ and $B$ and a relation $R:A\to B\to\univ$ such that $R(x,y)$ is a proposition for
    all $x:A$ and $y:B$, there is an inhabitant of the type
\[\Big(\qPi{x:A}{\qExists{y:B}{R(x,y)}}\Big)
    \to\Big(\qExists{f:A\to B}{\qPi{x:A}{R(x,f(x))}}\Big).\]
\end{axiom}
We will be interested in various weakenings of the axiom of choice. For simplicity of
notation, we will use an alternate characterization of the axiom of choice for the
statement of such principles.
\begin{lemma}
    The axiom of choice is equivalent to the following:
    
    For all sets $X$ and families $Y:X\to\univ$ such that $Y(x)$ is a set for all
    $x$, we have
    \[\ACtype{x:X}{Y(x)}\]
\end{lemma}
\begin{proof} Fix a type $A:\univ$ and take $X\defeq A$ so that for any $B:\univ$ the map
    \[(f,\varphi)\mapsto \lambda x.(f(x),\varphi(x))\]
    defines an equivalence
    \[\Big(\qSig{f:A\to
    B}\qPi{x:A}R(x,f(x))\Big)\simeq\Big(\qPi{x:A}\qSig{y:B}R(x,y)\Big).\]
    So then setting $Y(x)\defeq \qSig{y:B}R(x,y)$, we have the backward direction.

    For the forward, direction, assume we have $Y:X\to \univ$ and define
    $B=\qSig{x:X}Y(x)$ and $R:X\to B\to\univ$ by
    \[R(x,y)\defeq \pr_0(y)=x.\]
    Then we know for any $x:X$ that $\qSig{y:B}R(x,y)$ is equivalent to $Y(x)$, and
    so the type $\qPi{x:X}\qExists{y:B}R(x,y)$ is equivalent to $\qSig{x:X}\trunc{Y(x)}$. So if
    we assume the latter is inhabited then so is the former. Then by the axiom of choice, we have
    \[\trunc{\qSig{f:X\to B}\qPi{x:X}\pr_0(f(x))=x}.\]
    Then we have $\trunc{\qPi{x:X},Y(x)}$ by functoriality of truncation
    (Lemma~\ref{lemma:trunc-func}) using the function 
    \[\lambda (f,\varphi). \lambda x.\transport(\varphi(x),\pr_2(f(x)).\qedhere\]
\end{proof}
We will call statement of the form
    \[\ACtype{x:X}{Y(x)} \]
\emph{choice from $X$ to $Y$}, or \emph{choice of $Y$ over $X$}, and similarly if
$P,Q:\Set\to\univ$ are property of sets, we will call
    \[P(X)\to\Big(\qPi{x:X}{Q(Yx)}\Big)\to  \ACtype{x:X}{Y(x)} \]
\emph{choice from $P$ to $Q$}. In this language, the full axiom of choice is choice
from sets to sets. Two more
crucial examples are \nameas{countable choice}{choice!countable}, (choice from $\nat$),
    \[\ACtype{n:\nat}{Y(n)};\] 
and \emph{propositional choice}, (choice from propositions),
    \[\isProp(X)\to \ACtype{x:X}{Y(x)}.\]
Neither of these principles are provable in univalent type theory. Coquand et
al.~\cite{CMR2017Stacks} recently gave a model (which validates univalence) but in which countable
choice fails, and a similar model which invalidates a consequence of propositional
choice. However, the \nameas{principle of unique choice}{choice!unique} (choice from propositions to
propositions) is provable in our type theory.
\begin{theorem}\label{thm:unique-choice}
    For any proposition $X$ and any $P:X\to\univ$ such that $\qPi{x:X}{\isProp(Px)}$
    we have
    \[\ACtype{x:X}{P(x)}.\]
\end{theorem}
\begin{proof}
    In fact, the assumptions give 
    \[\big(\qPi{x:X}{\trunc{P(x)}}\big)\to \qPi{x:X}{P(x)},\]
    by Lemma~\ref{lemma:prop-equals-truncation}.
\end{proof}
We also have choice from decidable propositions:
\begin{theorem}\label{thm:decidable-choice}
    If $P$ is a decidable proposition, then
    \[\ACtype{p:P}{Q(p)}.\]
\end{theorem}
\begin{proof}
    We do a case analysis on $P$: if $P=0$, then we have $\qPi{p:P}Q(p)$, and so in
    particular,
    \[\ACtype{p:P}{Q(p)}.\]
    If $P=1$, then $\qPi{p:P}\trunc{Q(p)}$is equivalent to $\trunc{Q(\star)}$, and this in turn
    is equivalent to 
    \[\trunc{\qPi{p:P}Q(p))}.\qedhere\]
\end{proof}
\index{choice!axiom of|)}

\section{Discussion}
Lemma~\ref{lemma:untruncate-decidable-predicates} has been discussed by Mart{\'i}n
Escard{\'o}, and provides real-world examples of situations where $\trunc{X}\to X$.
This is related to an interesting observation by
Nicolai Kraus~\cite{krausinvertible} that there is a dependent function $f$ such that
$f(|n|)=n$ for all~$n:\nat$, and to work by Kraus, Escard{\'o} and
others~\cite{keca2016,keca2013} on constancy and anonymous existence;
Section~\ref{section:constancy} is also related to this work.

While the notion of monad we gave is good enough for our purposes, the
universe is not a category but an
$\infty$-category~\cite{warren2011strict,lumsdaine2009weak,vdBergGarner2011},
although there is no definition of $\infty$-category in type theory at the time of writing.
Monads in higher categories are somewhat more subtle, due to the coherence issues that
arise. Some recent work by Riehl and Verity~\cite{riehl2016} is concerned with the
issue of $\infty$-monads.

The axiom of choice has a subtle relationship with constructiveness. On one hand,
you have that choice does not seem to correspond to an effective procedure in any
way, and Diaconescu's theorem that choice implies excluded middle. On the other, you
have Bishop's pronouncement that ``a choice function exists
in constructive mathematics, because a choice is implied by the very meaning of
existence,''~\cite{bishop1967} and the so-called \emph{type-theoretic axiom of
choice} which is a direct application of the definition of $\Sigma$ and~$\Pi$. A good
overview of this situation is given by Martin-L{\"o}f~\cite{martinlof2006choice}. The
subtleties discussed by Martin-L{\"o}f seem to vanish in a univalent setting: reading
the axiom of choice in the logic of structures (without truncations) gives the
version that is ``implied by the very meaning of existence'', while reading it in the
logic of propositions (with truncations) gives the classical principle.


%% file: chapters/parttwoprologue.tex
We turn now to partiality. Classically, a partial function $X\pto Y$ is taken
to be a relation $R\subseteq X\times Y$ such that for each $x\in X$, there is at most
one $y\in Y$ such that $R(x,y)$. While this approach can be used constructively,
there are a few reasons we wish to avoid this representation of partial functions.
Most importantly, expressing application of a partial function represented as a
relation is cumbersome, awkward, and non-rigorous: We express the fact that
there is at most one $y$ such that $R(x,y)$ by saying that $\qSig{y:Y}R(x,y)$
is a proposition. To apply a functional relation $R:X\times Y\to \univ$ to
$x:X$, we first demonstrate that $\qSig{y:Y}R(x,y)$ is inhabited, and then
appeal to the fact that it is a proposition to determine there is at most
one, and then apply the second projection.  In the case where $f$ is not
defined at $x$, the notation $f(x)$ does not even make sense since
$\qSig{y:Y}R(x,y)$ is empty.

In settings where functions take a
more privileged position than in (e.g.) ZF, a partial function $X\pto Y$ is often
taken to be a (total) function $X\to Y+1$, where the adjoined element represents an undefined
value. This approach only works for partial functions $f:X\pto Y$ where it is
decidable whether $f$ is defined on $x$.

We generalize the undefined value by replacing $Y+1$ with a type of \emph{partial
elements} of~$Y$. This can be expressed by taking a partial element to be a proposition (the
extent of definition) together with a map from that proposition to $Y$ (its value).
In a classical setting, this type is equivalent to $Y+1$, since a proposition is either
true or false. A partial function $X\pto Y$ is taken to be a total
function from $X$ to the type of partial elements of $Y$. We can extract the
\emph{value} of $f(x)$ (as an element of $Y$)
when we know $f(x)$ to be a \emph{total} element. However, we need to be able to lift
data about elements of $Y$ to data about partial elements of $Y$---in short, we need
the operation taking $Y$ to its type of partial elements to be a monad. This approach
is directly related to synthetic domain theory,
which we discuss in Section~\ref{section:sdt}. One point made throughout the
literature on synthetic domain theory is that the type of all partial elements (and
the type of all partial functions) is often not very interesting; we instead want to
focus on a class of particular partial functions and partial elements. Such a class
can be defined by restricting our extents of definitions. The restricted
partial elements operation succeeds in forming a monad precisely if the set of
allowable extents of definition form a \emph{dominance}
(Definition~\ref{def:dominance}). Many sets of propositions of interest arise by
truncation from \emph{structural dominances}
(Definition~\ref{def:univ-dom}), which are just $\Prop$-indexed families of
types closed under $\Sigma$ such that the type indexed by $\unittype$ is inhabited.
Unfortunately, sets of propositions arising in this way are
dominances if and only if particular (weak) forms of choice hold
(Theorem~\ref{dominance:choice-to-dominance}). Instead, we consider \emph{lifted} functions
arising by applying the same constructions directly to a structural dominance, instead
of a dominance. The type
of lifted functions arising in this way contains many representations of the same partial
function, and while they do compose, there is no guarantee that this composition is
associative. Nevertheless, we can tame these wild partial functions to get a notion
of \emph{disciplined map} (Definition~\ref{def:discipline}). These maps do indeed
compose (associatively), even in the absense of choice, and the 
information tracked explicitly by the lifted functions is now implicit.

Of particular interest are the \emph{Rosolini structures} (a structural dominance) and
the \emph{Rosolini propositions} (their truncation), which are an abstract
form of the \emph{affirmable} or \emph{semidecidable} truth values. This class of
propositions are related to the extended naturals, $\NI$ (defined by adding a point
at infinity to $\nat$), and to other approaches to partiality in type theory based on
the delay monad (Section~\ref{section:delay}) and higher inductive-inductive types
(Section~\ref{section:qiit}). These will play a crucial role in
Chapter~\ref{chapter:comp-and-partiality}, as they can be used as to define notions
we expect to satisfy the hope of a notion of partial function which can be used as a
surrogate for the computable functions.

This approach allows us to express application of partial functions directly. The
statement that $f(x)$ is defined is a proposition about an existing
partial-element, rather than a possibly non-existent element. There is another, less
important benefit of the approach via partial elements: the full univalence axiom is
required to identify functions with their graph
(Theorem~\ref{thm:functions-are-relations}), and univalence appears repeatedly when
trying to prove type families equal. The approach to partial functions as partial
elements does not need univalence.

%% file: chapters/partiality.tex
\chapter{Partiality in Type Theory and Topos Theory}\label{chapter:partiality}
Before developing our approach to partiality, we quickly review prior work on
partiality in type theory. The standard approach to partiality in type theory center
around the \emph{delay monad}, which we define concretely using the extended natural
numbers in Section~\ref{section:delay}. Kleisli maps for the delay monad can be viewed as
partial functions equipped with intensional information. This extra intensional
information means the delay monad does not work for our purposes: it makes partiality
into structure, rather than property. Chapman, Uustalu and Veltri resolve this issue
by quotienting the delay monad, but then countable choice is needed to show that the
quotiented delay monad forms a monad---to show that partial functions compose. Our
attempt to give partial functions via dominances in
Chapter~\ref{chapter:partial-functions}, leads us to the same place, so we discuss
the details from the perspective of dominances there. More recently, Altenkirch,
Danielsson and Kraus use a higher-inductive type to define the free $\omega$CPO on a
type $X$, and use Kleisli maps for this monad to give partial functions.

Our approach to partiality will center around a type-theoretic version of the notion
of \emph{dominance}.  Dominances and many of the tools we use to work with dominances
arise from synthetic domain theory and are also used in synthetic computability
theory. Synthetic domain theory aims to develop domain theory by axiomatizing a
particular topos of objects which naturally have dcpo structure. Synthetic
computability theory is analogous.
In Section~\ref{section:sdt}, we give a brief overview of this work.

\section{$\NI$ and the delay monad}\label{section:delay}
\index{monad!delay|(}
\index{type!of natural numbers!extended|(}
The \emph{extended naturals}, $\NI$, are of central
importance in other approaches to
partiality in type theory, and will be important in ours as well. The type is
coinductively generated by
\begin{itemize}
    \item[] $\zero : \NI$
    \item[] $\succop : \NI\to\NI$
\end{itemize}
That is, $\NI$ is the final coalgebra for $X\mapsto X+\unittype$. The coalgebra map
is $\pred:\NI\to\NI+\unittype$, given by
\begin{align*}
    \pred(\zero) &= \inr \star\\
    \pred(\succop x) & = x
\end{align*}
As this type is coinductively generated, there is an element $\infty:\NI$ satisfying
the equation
\[\infty = \succop(\infty).\]
We cannot use the above coinductive definition, however, as we do not have
coinduction in our type theory.  We instead must give a concrete type which satisfies
the required universal property. The type $\NI$ should be thought of as the naturals
extended by a point at infinity. As we cannot separate the extra point from $\N$, the
type $\N+\unittype$ will not suffice. Instead, we we take $\NI$ to be the type of
binary sequences with at most one $1$\cite{escardo2013omniscience}.
\begin{definition}
    The \name{Cantor space} is the type $\cantor$.
\end{definition}
Note that the cantor space can be seen both as the type of infinite binary sequences,
and as the type of decidable predicates on $\nat$.
\begin{definition}
    \index[symbol]{bracket@{$\lar{-}$}!for a binary sequence|textit}
    For $\alpha:\cantor$, define the type $\lar{\alpha}:\univ$ by
    \[\lar{\alpha} \defeq \qSig{n:\nat}{\alpha_n = 1}.\]
    Define the \nameas{extended natural numbers}{type!of natural numbers!extended} to be the type
    \definesymbolfor{$\NI$}{NI}{type of extended natural numbers}
    \[\NI\defeq \qSig{\alpha:\cantor}{\isProp(\lar{\alpha})}.\]
    The element $\infty:\NI$ is defined to be the constantly zero function,
    \[\infty\defeq \lambda n.0.\]
    Given $n:\nat$, the sequence given by the decidable predicate $\lambda k.n=k$
    only takes value 1 at input $n$, and so for each $n:\nat$ there is an element $\numeral{n}:\NI$.
\end{definition}
We could instead equivalently define $\NI$ to be the set of \emph{increasing} binary
sequences.
\begin{definition}
    A sequence $\alpha:\cantor$ is \emph{increasing} when
    \[\qPi{n,m:\nat}{(\alpha_n = 1) \to (n<m) \to (\alpha_m = 1)}.\]
\end{definition}
\begin{lemma}
    The type $\NI$ is equivalent to the type of increasing binary sequences.
\end{lemma}
\begin{proof}
    Switch all bits after the first 1\ in the sequence. Explicitly:

    Let $\alpha$ be increasing. Define $\mu:\nat\to\univ$ by
    \[\mu(n) \defeq (\alpha_n=1) \times \qPi{k<n}{\alpha_k = 0}.\]
    The right
    hand side is a decidable predicate by Lemma~\ref{lemma:untruncate-decidable-predicates},
    and so $\mu$ factors through a map $p(\alpha):\nat\to 2$. Moreover, we know that the type
    \[ \qSig{n:\nat}{(\alpha_n=1) \times \qPi{k<n}{\alpha_k = 0}}\]
    is a proposition, and then so is $\lar{p(\alpha)}$.

    Conversely, if $\mu$ is a sequence with at most one 1, we may define
    \[i(\mu)_n = 1 \Leftrightarrow \qExists{k\le n}{\mu_k = 1}\]
    Since $\mu_k = 1$ is decidable, we know that so is $\qExists{k\le n}{\mu_k = 1}$
    by bounded search.

    We have that for any increasing $\alpha:\cantor$ that
    \[i(p(\alpha))_n = 1 \Leftrightarrow \qExists{k\le n}p(\alpha)_k =1.\]
    Since $\alpha$ is an increasing sequence from $\bool$, this is the same as
    \[i(p(\alpha))_n = 1 \Leftrightarrow \alpha_n =1.\]

    Conversely, for any $\mu$ with at most one one, we have
    \[p(i(\mu))_n = 1 \Leftrightarrow (i(\mu)_n=1) \times \qPi{k<n}{i(\mu)_k = 0},\]
    which happens at the least $n$ such that $\mu_n=1$, but since $\lar{\mu}$ is a
    proposition, this is the unique value at which $\mu_n=1$, and so
    \[p(i(\mu))_n = 1 \Leftrightarrow (\mu_n=1).\qedhere\]
\end{proof}

As usual with proposition-valued components in sums, we will often suppress the
second component, proving that the required witness exists separately.
We have the map ${\pred:\NI\to\NI}$ defined by
\begin{align*}
    \pred &:\NI\to\NI \\
    \pred(\alpha) &= (\lambda n.\alpha(n+1)).
\end{align*}
That $\lar{\pred(\alpha)}$ is a proposition when $\lar{\alpha}$ is follows from the
fact that $\pred(\alpha)_n = 1 \Leftrightarrow \alpha_{n+1} = 1$.
\begin{theorem}
    The type $\NI$ satisfies the correct universal property. That is, if $p:X\to
    X+\unittype$
    for any type $X:\univ$, then there is a unique $\varphi : X\to\NI$ such that
        \[\pred(\varphi (x)) = (\varphi+\unittype)(p(x)).\]
\end{theorem}
\begin{proof}
    The idea is to count how long it takes for $p$ to give us $\inr\star$.
    To do this, first define the function $p' : X+\unittype\to X+\unittype$,
    \begin{align*}
        p'(\inl x) &= p x \\
        p'(\inr \star) &= \inr \star,
    \end{align*}
    and then define $\varphi' : X+\unittype\to\NI$
    \begin{align*}
        \varphi'(\inl x) &= 0 & \text{if $p'(\inl x) = \inl y$} \\
        \varphi'(\inl x) &= 1 & \text{if $p'(\inl x) = \inr y$} \\
        \varphi'(\inr \star) &= 0
    \end{align*}
    Then $\varphi = \varphi' \circ p$.

    It is easy to check that this is correct and unique.
\end{proof}
We have that $\NI$ is a retract of $\cantor$: Define $f:\cantor\to\NI$ by
\[f(\alpha)(n) \defeq
    \begin{cases}
        \alpha(n) & \text{ if $\alpha(k) = 0$ for all $k< n$} \\
        0 & \text{otherwise}
    \end{cases}
    \]
Note that we are doing a case analysis on $\qPi{k<n}\alpha(k) = 0$, which is
decidable by bounded search. We have that $\lar{f(\alpha)}$ is a proposition since
$f(\alpha)(k) = 1$ exactly if $k$ is the least such that $\alpha(k)=1$.
We call $f(\alpha)$ the \emph{truncation} of $\alpha$.

We immediately have,
\begin{lemma}
    The map $f:\cantor\to\NI$ above is a retraction of the inclusion~$\NI\to\cantor$.
\end{lemma}

Capretta defines, for any type $X$, the type $\Delay(X)$ \nameas{of delayed elements of
$X$}{monad!delay} , coinductively generated by
\defineopfor{$\Delay$}{Delay}{monad!delay}
\defineopfor{$\delay$}{Delay}{monad!delay}
\defineopfor{$\now$}{Delay}{monad!delay}
\begin{itemize}
    \item[] $\now : X\to\Delay(X)$
    \item[] $\delay : \Delay(X)\to\Delay(X)$.
\end{itemize}
Again, we must define this type concretely in our type theory. Classically, the
elements of $\Delay(X)$ can be written as $\delay^n(\now x)$, or $\delay^{\infty}$.
Again, constructively, we cannot do a simple case analysis to determine which case is
satisfied, but we may use the same trick as with $\NI$.
\begin{theorem}
    The following types are equivalent:
    \begin{enumerate}[label=(\roman*)]
        \item 
            $
            \qSig{\mu:\NI}{\qPi{n:\nat}{X^{\mu=\numeral{n}}}},
            $
        \item 
            $
                \qSig{\alpha:\nat\to(X+\unittype)}{\isProp{\lar{\alpha}}},
            $\\
            where $\lar{\alpha} \defeq \qSig{n:\nat}{\qSig{x:X}{\alpha(n) = \inl(x)}}$,
    \end{enumerate}
\end{theorem}
\begin{proof}
    For $(\mu,\chi):(i)$, we know that $(\mu=\numeral{n}) \simeq (\mu_n = 1)$, which is
    a decidable predicate. Hence we can take
    \[\alpha(n) = \left\{
        \begin{array}{ll}
            \inl \chi(n,-) & \text{ if $\mu=\numeral{n}$,} \\
            \inr(\star) & \text{ otherwise.}
        \end{array}\right. 
        \]
    As there is at most one $n$ with $\mu=\numeral{n}$, we know that there is at most
    one $n$ which takes the form $\inl x$.

    Conversely, let $(\alpha,-):(ii)$, and define
    \[\mu(n) = \begin{cases}
            1 & \text{ if $\alpha(n) = \inl x$,} \\
            0 & \text{ otherwise.}
        \end{cases} 
        \]
    And define $\varphi(n) = \lambda y.x$.
\end{proof}
We will tacitly switch between these two representations of $\Delay X$, but we make
(ii) official. In particular, we define
\begin{align*}
    \now(x)(0) &= \inl x\\
    \now(x)(n+1) &= \inr \star,
\end{align*}
and
\begin{align*}
    \delay(\mu)(0) &= \inr \star\\
    \delay(\mu)(n+1) &= \mu(n).
\end{align*}
We want $\Delay X$ to be the final coalgebra for $Y\mapsto Y+X$, for which we define
the coalgebra map $\tick:\Delay X\to (\Delay X + X)$, by
\begin{align*}
    \tick(\now x) &= \inr x\\
    \tick(\delay \mu) &= \inl \mu.
\end{align*}
The proof of the following theorem is identical to the case for $\NI$.
\begin{theorem}
    Given any type $Y$ and a coalgebra $t:Y\to Y+X$, there is a unique coalgebra
    homomorphism $\varphi:Y\to \Delay X$ such that for all $y:Y$
    \[ \tick(\varphi y) = (\varphi + X)(t y),\]
    where
\begin{align*}
    (\varphi+X)(\inl y) &= \inl (\varphi y)\\
    (\varphi+X)(\inr x) &= \inr x.
\end{align*}
\end{theorem}

$\now:X\to\Delay X$ gives an obvious candidate unit for $\Delay$. Given $f:X\to\Delay Y$, we have
a Kleisli extension $\klext{f}:\Delay X\to\Delay Y$ with
\begin{align*}
    \klext{f}(\now x) &= f x\\
    \klext{f}(\delay \mu) &= \delay (\klext{f} \mu).
\end{align*}
Explicitly, for $x:X$, we define $\eta:X\to\Delay X$ by
\begin{align*}
    \eta(x)(0) &= \inl x \\
    \eta(x)(n+1) &= \inr\star
\end{align*}
and if $f:X\to \Delay(Y)$, we define $\klext{f}:\Delay X\to \Delay Y$ by
\[\klext{f}(\mu)(n) = \begin{cases}
   \inr \star & \text{if }\qPi{k\le n}\mu(k) = \inr \star\\
    f(x)(n-k) & \text{if } \qSig{k\le n}(\mu(k) = \inl x) 
\end{cases}\]
Then we have that 
\[\Big(\klext{f}(\mu)(n) = \inl y\Big)\lequiv
\Big(\qSig{k\le n}\qSig{x{:}X}(\mu(n) = \inl x) \times f(x)(n-k) = \inl y\Big).\]
That is, $\klext{f}(\mu)$ takes the value $\inr\star$ for as long as $\mu$ does, and
then takes the value of $f(x)$, where $x$ is the unique value (if such exists) where
$\mu(n)=\inl x$ for some~$n$.

\begin{theorem}\label{thm:delay-monad}
    The maps $\now$ and $\klext{(-)}$ give $\Delay$ the structure of a monad.
\end{theorem}
\begin{proof}
    Rather than giving the (probably simpler) coinductive proof, we argue concretely.
    As $\eta(x)(0) = \inl x$, we have that
    \[\klext{\eta}(x)(n) = \eta(x)(n-0).\]
    Similarly, we have 
    \[\klext{f}(\eta x)(n) = f(x).\]
    Expanding the definition of Kleisli extension, we see that
    $\klext{(\klext{g}f)}(\mu)(n) = \inl z$ precisely when
    \[\qSig{k\le n}\qSig{x{:}X}(\mu(k) = \inl x) \times
        \qSig{j\le n-k}\qSig{y{:}Y} \big(f(x)(j) = \inl y\big) \times
        \big(g(y)(n-k-j)=\inl z\big).\]
    Call this type $A$.
    Pictorially, we have that $\klext{(\klext{g}f)}(\mu)$ is
    \[\underbrace{\star,\ldots,\star}_{\mu},
        \underbrace{\star,\star,\ldots, \star,}_{f(x)\text{ where }\mu_k = x}
        \underbrace{\star,\ldots,\star, z,\star,\ldots, \star}_{g(y)\text{ where }f(x)_j = y}.\]
    In words: $\klext{\klext{g}f}(\mu)$ is composed of a sequence of $\inr\star$
    with possibly one element of the form $\inl z$. The terms of the sequence come
    from $\mu$, until it takes value $\inl x$ (for some $x$), and then from $f(x)$,
    until it takes value $\inl y$ (for some $y$), and finally from $g(y)$ thereafter.

    Similarly,
    we have that $\klext{g}(\klext{f}\mu)(n)=\inl z$ precisely if
    \[\qSig{k'\le n}\qSig{y:Y}\qSig{j'\le k'}\qSig{x:X}(\mu(j') = \inl x)
    \times ((fx)(k'-j') = \inl y) \times (g(y)(n-k') = \inl z).\]
    Call this type $B$.
    Pictorially, we have that $\klext{g}(\klext{f}\mu)(n)=\inl z$ is
    \[\underbrace{\star,\ldots,\star}_{\mu},
        \underbrace{\star,\star,\ldots, \star,}_{f(x)\text{ where }\mu_j = x}
        \underbrace{\star,\ldots,\star, z,\star,\ldots, \star}_{g(y)\text{ where }f(x)_k = y}.\]
    We see that these two types are equivalent, by manipulating $A$ as follows: First
    move all quantifiers to the front and rearrange, so that we have
    \[\qSig{k\le n}\qSig{y{:}Y}\qSig{j\le n-k}\qSig{x{:}X}(\mu(k) = \inl x) \times
         \big(f(x)(j) = \inl y\big) \times
        \big(g(y)(n-k-j)=\inl z\big).\]
    Next, given $k\le n$ and $j\le n-k$, define $j'=k$ and $k' = j+k$, so that $j'\le
    k'$ and $k'\le n$. So then $j'$ and $k'$ are of the correct form to satisfy $B$.
    That is, we have the map $A\to B$ defined by
    \[(k,x,-,j,y,-,-)\mapsto (j+k,y,k,x,-,-,-),\]
    with inverse
    \[(k',y,j',x,-,-,-)\mapsto (j',x,-,k'-j',y,-,-).\]
    So then $\klext{g}(\klext{f}\mu)(n)\lequiv\klext{(\klext{g}f)}(\mu)(n)$. Applying
    function and proposition extensionality, we get the Kleisli law.
\end{proof}

The goal of the delay monad is to provide a satisfying account of partiality in type
theory; however, the delay monad tracks intensional information---the
timer given by $\delay$---which makes equality on $\Delay(X)$ too strict. The
obvious solution is to attempt to quotient $\Delay(X)$ by some equivalence relation,
which is called \emph{weak bisimilarity}. Capretta showed that the quotient of the
delay monad by weak bisimilarity gives a monad in the category of
setoids~\cite{capretta2005recursion}. Later, Chapman, Uustalu and
Veltri~\cite{Chapman2015Delay} show that this quotient gives a monad in a type
theory with quotient inductive types assuming countable choice, so it seems that
Capretta's result relies on the fact that countable choice is satisfied by the setoid
model.

Weak bisimilarity is defined via a relation $-\downarrow-:DX\to X\to\univ$ of
\emph{convergence}, defined inductively by $\now (x)\downarrow x$ and if $d\downarrow
x$, then $\delay(d)\downarrow x$. 
We then define weak bisimilarity coinductively: if $d_1\downarrow x$ and
$d_2\downarrow x$, then $d_1\bisim d_2$ and if $d_1\bisim d_2$, then
$\delay(d_1)\bisim \delay(d_2)$.

We can define both convergence and weak bisimilarity concretely on our 
representation of $\Delay X$.
\begin{definition}
    For $d:\Delay X$ and $x:X$, we say that $d$ \emph{converges to} $x$ if
    $x$ is the unique value such that $d(n) = \inl x$. That is,
    \[d\downarrow x \defeq \qSig{n:\nat}\big(d(n) = \inl x\big).\]

    We say that elements $d_1$ and $d_2$ of $\Delay X$ are \emph{weakly bisimilar} if for all
    $x:X$ we have $d_1\downarrow x$ iff $d_2\downarrow x$. That is
    \[(d_1\bisim d_2) \defeq \qPi{x:X} (d_1\downarrow x)\lequiv (d_2\downarrow x).\]
\end{definition}
Following~\cite{Chapman2015Delay}, let $\overline{\Delay}(X) \defeq
\Delay(X)/{\bisim}$. They proved the following.
\index{choice!countable|(} 
\begin{theorem}[\cite{Chapman2015Delay}, Section 7]
    Assuming countable choice, $\overline{\Delay}(X)$ is a monad.
\end{theorem}
We will see in Section~\ref{section:rosolini} that $\overline{\Delay}(X)$ is
equivalent to our \emph{Rosolini lifting}, and we will subsequently discuss
weakenings of countable choice that are sufficient to show that
$\overline{\Delay}(X)$ is a monad, culminating in a necessary and sufficient
weakening, Theorem~\ref{rosolini:choice-to-dominance}.
\index{choice!countable|)} 
\index{type!of natural numbers!extended|)}
\index{monad!delay|)}

\section{Partiality via higher-inductive types}\label{section:qiit}
Altenkirch, Danielsson and Kraus instead define a notion of partiality by giving a
higher inductive-inductive type~\cite{Altenkirch:Danielsson:Kraus}---simultaneous
defining a higher inductive type $X_\bot$ and a binary relation $\sqsubset$ on
$X_\bot$. Their goal was to give an extensional version of the delay monad which was
a monad even without countable choice. The type $X_\bot$ is essentially the free
$\omega$CPO over $X$.

\begin{definition}
    For each type $X$ define simultaneous the type $X_{\bot}:\univ$ and a binary relation
    ${\sle:X_{\bot}\to X_{\bot}\to\univ}$ inductively with constructors for~$X_{\bot}$
    \begin{itemize}
        \item a map $\eta:X\to X_{\bot}$,
        \item an element $\bot:X_{\bot}$,
        \item for each $s:\nat\to X_{\bot}$ and witness $p:\qPi{n:\nat}s_n\sle s_{n+1}$,
            an element $\sqcup(s,p):X_{\bot}$,
        \item a path constructor $(x\sle y)\to (y\sle x) \to x=y$, and
        \item a set truncation $\qPi{x,y:X_{\bot}}\qPi{p,q:x=y}p = q$;
    \end{itemize}
    and constructors for~$\sle$
    \begin{itemize}
        \item $x\sle x$ for each $x:X_{\bot}$,
        \item if $x\sle y$ and $y\sle z$ then $x\sle z$,
        \item $\bot\sle x$ for each $x:X_{\bot}$,
        \item for each $n:\nat$ we have $s_n \sle \sqcup(s,p)$,
        \item if $\qPi{n:\nat} s_n\sle x$ then $(\sqcup(s,p)\sle x)$.
    \end{itemize}
\end{definition}
In the paper, they show that $X_\bot$ behaves like the lifting in classical domain
theory, in the sense that $X_\bot$ is the initial $\omega$CPO that $X$ maps into
(Theorem 2) and
moreover $\eta$ is injective, $\eta(x)$ is maximal with respect to $\sle$ for each
$x:X$, and $\eta(x)\neq \bot$ for all $x:X$ (Corollary
8).

\section{Synthetic domain theory and synthetic computability theory}\label{section:sdt}
\index{synthetic domain theory|(}
Domain theory arose from a problem in denotational semantics: Languages with
non-termination cannot be directly modeled in the category of sets, since some
functions are partial. Since non-termination behaves in non-trivial ways, there are
limits on how functions can interact with the undefined value. To resolve this, Dana
Scott~\cite{scottdomains} introduced \emph{domains}, directed- or $\omega$-complete
partial orders with bottom. The bottom element represents an undefined value, and
maps between them are required to be monotone and preserve directed suprema. This
structure effectively captures the behavior of non-termination. However, Scott
wondered whether there was some alternative axiomatization of sets in which sets
naturally came equipped with the required structure. In his PhD
thesis~\cite{Rosolini1986}, Giuseppe Rosolini took up this question, giving a way to
access categories of partial maps from within a category. The basic idea, which we
will explore in more depth from a type-theoretic perspective in
Chapter~\ref{chapter:partial-functions}, is to isolate a subset $\Sigma\subseteq
\Omega$ of truth values in a topos, and consider the maps $X\to Y_\bot$, where
$Y_\bot$ is the set
\[ Y_\bot = \{A\in \Omega^Y\mid (\forall x,y. A(x)\to A(y)\to x=y)\wedge
(\exists x.A(x)\in \Sigma)\}.\]
In other words, $Y_\bot$ is the set of subsingleton subsets of $Y$ whose \emph{extent
of definition} is in $\Sigma$. In order for these \emph{$\Sigma$-partial functions}
to contain total functions, we need the true proposition to be in $\Sigma$. Moreover,
in order for the $\Sigma$-partial functions to form a category, we need to impose an
additional restriction that for all propositions $p,q\in\Omega$,
\[p\in\Sigma \wedge (p\to (q\in \Sigma))\to ((p\wedge q)\in \Sigma).\]
The two conditions make $\Sigma$ into a \emph{dominance}\index{dominance}.

Rosolini's approach was concrete: while he considered the internal language of a
topos, he worked with specific toposes in an informal metamathematics; in particular,
he examined the effective and recursive toposes, and considered \emph{effective
objects} in a topos. However, his work allowed synthetic axiomatizations, which were
pursued by a number of authors studying \emph{synthetic domain
theory} (See, for example~\cite{Hyland1991,VANOOSTEN2000,ReusStreicher1997}). The
idea is to pursue Scott's idea of an axiomatization of a category of domains.

Since dominances are a way to approach partial functions and effectiveness, the
notion features in Bauer's~\cite{bauer2000realizability,Bauer2006Synthetic} approach
to synthetic computability. Here, Bauer is axiomatizing a category of sets (the
\emph{modest sets}~\cite{hyland1982effective}) arising from realizability toposes.

In both synthetic domain theory and synthetic computability theory, the dominance of
interest is a dominance $\Sigma$ representing the semidecidable propositions, introduced also
by Rosolini. The dominance is given by
\[p\in\Sigma\Leftrightarrow \qExists{\alpha:\nat\to 2}p \leftrightarrow
\big(\qExists{n:\nat}\alpha(n)=1\big).\]
That is, the propositions in $\Sigma$ are those which, when true, are observably
true.

Most approaches to synthetic domain theory are topos-theoretic in nature, but it is
worth pointing to two (Reus and Streicher~\cite{ReusStreicher1997} and
Reus~\cite{Reus1999}) that use a type-theoretic axiomatization. We will briefly
describe the Approach by Reus here, and compare this approach to synthetic domain
theory to our approach to partiality in Section~\ref{section:sdt-comparison}. Reus
uses a version of the calculus of construction with an impredicative universe of
propositions, which are assumed to be subsingletons (i.e., propositions in our
sense), but note that there it is not required that all subsingletons are
propositions.  Reus also assumes a universe of \emph{sets} (which need not be sets in
our sense). Reus assumes function extensionality, but not proposition extensionality,
which is not needed in the development there. Additionally, in contrast to a
univalent setting, the principle of unique choice is not available in the calculus of
constructions, so it is assumed explicitly in Reus's development.  Rather than giving
a dominance, Reus's axiomatization of $\Sigma$ includes a number of principles that
relate $\Sigma$ to domain theory, without reference to partial functions. We explain
their axiomatization in Section~\ref{section:sdt-comparison}, where we have more
context for understanding it, but note that they include a version of Markov's
Principle (see Section~\ref{section:taboos}), as well as a version of Phoa's
Principle~\cite{taylor1991fixed}, both of which are not available in our setting.
\index{synthetic domain theory|)}

\section{Discussion}
The notion of partiality given by the delay monad is intensional. In particular, a
single partial function $f:X\pto Y$ has many representations as a function
$F:X\to\Delay(Y)$, based on how many applications of $\delay$ occur in the output of
$F$. Quotienting by bisimilarity resolves this issue at the cost of composability
of partial functions, which does not seem like a worthwhile cost.

The approach using quotient-inductive-inductive types resolves the issue without the
cost, but it doesn't seem to work for the purposes of Part~\ref{part:three}: we want it
to be possible to use our partial functions as an abstract representation of the
computable functions. Our first attempt (the Rosolini partial functions of
Chapter~\ref{chapter:partial-functions}) is too big for this---we will instead
consider a subtype of the Rosolini partial functions, the \emph{disciplined maps} of
Section~\ref{section:disciplined-maps}---but the QIIT
approach expands the set of Rosolini partial functions (compare Lemmas 13 and 14 from
\cite{Altenkirch:Danielsson:Kraus} with our Theorem~\ref{thm:qdelay-is-ros}).

The Kleisli category for the delay monad is morally an intensional version of the
category of $\Sigma$-partial maps from synthetic domain theory. The approach via
dominances puts this notion of partiality into a more general context; we transfer
the relevant parts of this context to the univalent setting in
Chapter~\ref{chapter:partial-functions}. In the synthetic approach, however, 
The core categories of interest is the effective topos, and its category of
\emph{effective objects}. The effective topos has enough structure that
we can show in its internal logic that $\Sigma$ is indeed a dominance, as well as
enough structure that we can show externally that the $\Sigma$-partial functions
capture the functions of interest (the computable or continuous functions). In the
more general situation that we consider, neither of these will be the case.
Chapter~\ref{chapter:partial-functions} is concerned with the problem of showing a
set of propositions to be a dominance, while
Chapter~\ref{chapter:comp-and-partiality} is concerned with the computability of maps
arising from our version of $\Sigma$, which we call the set of \emph{Rosolini propositions}.


%% file: chapters/partial-functions.tex
\chapter{Partiality in Univalent Mathematics}\label{chapter:partial-functions}
We now turn in earnest to the notion of partial function. Recall that a motivating
goal is to find a notion of partial function for which it is consistent that all such
partial functions are computable. We expect such partial functions to include all
total functions, and to be composable, even in the absence of countable choice and
Markov's principle. A great deal of this chapter is concerned with
attempting to satisfy those requirements.

In Sections~\ref{section:single-valued-relations}-\ref{section:dcpos}, we examine the
class of all partial functions. In
Sections~\ref{section:dominances}-\ref{section:univ-dom} we attempt to restrict the
class of partial functions using \emph{dominances}
(Definition~\ref{def:dominance})---subsets of $\Prop$ satisfying certain closure
conditions which ensure that the partial functions just defined are composable (see
Theorems~\ref{dominance:composition} and~\ref{dominance:characterization}). We examine
in particular the $\Sigma_1$ propositions,
which we call \emph{Rosolini} propositions, and relate this to prior work on
partiality. A weakening of countable choice
is required to show that the partial functions arising from Rosolini propositions
are closed under composition (see Theorems~\ref{thm:cc-weaker},
\ref{rosolini:choice-to-dominance},
and~\ref{thm:rosolini-choice}). This fact generalizes
(Theorem~\ref{dominance:choice-to-dominance}) to any dominance arising
by truncation from \emph{structural dominances} (Definition~\ref{def:univ-dom}).
The consequence is that some amount of choice is required to restrict partial
functions via dominances. In Section~\ref{section:disciplined-maps}
we give a different approach using what we call \emph{disciplined maps}. This
approach still builds off of the theory of dominances. Notably, the disciplined maps
can be shown to compose without countable choice, and in the presence of countable
choice, the disciplined maps are exactly the Rosolini partial functions.
In Part~\ref{part:two}, we will discuss the relevance of this.

\section{Single-valued relations}\label{section:single-valued-relations}
\index{relation!single-valued|(}
\index{relation|see{predicate}}
We saw in Theorem~\ref{thm:functions-are-relations} that assuming univalence we have an
equivalence
\[(X\to Y)\simeq(\qSig{R:X\to Y\to\univ}{ \qPi{x:X}{\isContr(\qSig{y:Y}{R(x,y)})}}).\]
That is, we saw that a function is relation with a unique value for each $x:X$,
matching the classical intuition. Classically, the partial functions arise
by considering instead relations for which there is instead at most one $y:Y$ with
$R(x,y)$ for each $x:X$.
\begin{definition}
A relation $R:X\to Y
    \to\univ$ is \nameas{single-valued}{relation!single-valued} when
\[\qPi{x:X}{\isProp(\qSig{y:Y}{R(x,y)})}.\]
\end{definition}
Note that being single-valued is a proposition. As with the definition of functional
relation in Section~\ref{section:predicates}, the more traditional definition doesn't
work in general. In this case, the traditional definition would say that a relation is
single-valued when
\[\qPi{x:X}\qPi{y,y':Y}R(x,y)\to R(x,y')\to y=y'.\]
Again, if $Y$ is not a set, then $y=y'$ may have many witnesses, and so the above
type may have interesting structure. Instead, we want to say that a relation $R$ is
single-valued not only when all $y:Y$ such that $R(x,y)$ are equal, but when
there is at most one pair $(y,r)$ such that $r:R(x,y)$. In the case where $Y$ is a
set, the traditional definition is equivalent to the definition we give.

There is another traditional approach to the definition of partial function, more
common in category theoretic settings. A partial function $X\pto Y$ is a function
$f:A\to Y$ for some $A$ which embeds into $X$. We can then represent partial
functions
$X\pto Y$ via two separate types in the next universe: via the type
\begin{equation}\label{rel-pto}
\qPi{x:X}{\isProp(\qSig{y:Y}{R(x,y)})},
\end{equation}
or via the type
\begin{equation}\label{emb-pto}
\qSig{A:\univ}{\qEmpty{e:A\to X}{\isEmbedding(e)\times (A\to Y)}}.
\end{equation}
We have maps between them as follows: Given a single-valued relation $R:X\to
Y\to\univ$, define
$A:\univ$ by
\[A \defeq \qSig{x:X}\qSig{y:Y} R(x,y),\]
with $e:A\to X$ the first projection and $f:A\to Y$ the second projection. We need to
see that for any $x:X$ the type
\[\qSig{(x',y,r):\qSig{x:X}\qEmpty{y':Y}R(x',y)}x'=x\]
is a proposition. By Lemma~\ref{lemma:path-families}, this type is equivalent to
\[\qSig{y:Y}R(x,y),\]
which is a proposition since $R$ is single-valued. This gives us a function
$F: (\ref{rel-pto})\to (\ref{emb-pto})$,
\[F(R,-) \defeq (\qSig{x:X}\qSig{y:Y}R(x,y),\pr_0,-,\pr_1).\]
Conversely, given an embedding $e:A\hookrightarrow X$ with a map $f:A\to Y$,
define $R:X\to Y\to\univ$ by
\[R(x,y) \defeq \qSig{a:A}(e(a) =x) \times (f(a) = y).\]
We need to see that for any $x$, the type
\[\qSig{y:Y}\qSig{a:A}(e(a)=x)\times (f(a)=y)\]
is a proposition. By reshuffling, this type is equivalent to
\[\qSig{a:A}(e(a)=x)\times \qSig{y:Y}f(a)=y,\]
and since $\qSig{y:Y}f(a)=y$ is the singleton at $f(a)$, this type is in turn
equivalent to
\[\qSig{a:A}e(a)=x,\]
which is a proposition since $e$ is an embedding. This gives us a function
$G: (\ref{emb-pto})\to (\ref{rel-pto})$,
\[G(A,e,-,f) = \Big(\big(\qSig{a:A}(e(a) =x) \times (f(a) = y)\big), -\Big).\]
We want $G$ and $F$ to be equivalences, but they are only equivalences up to
equivalence. That is, we have the following
\begin{lemma}\label{ptl-functions:classical-categorical}
    For any single-valued $R:X\to Y \to \univ$, we have that the first projection of
    $G(F(R,-))$ is equivalent to $R$. Similarly, For any $e:A\hookrightarrow X$ and
    $f:A\to Y$ we have an equivalence between the type $A' \defeq \pr_0(F(G(A)))$ and
    $A$. Hence, in the presence of univalence, the type of
    single-valued relations is equivalent to the type of functions from a subtype of
    $X$ to $Y$.
\end{lemma}
\begin{proof}
    Let $R:X\to Y\to\univ$ be single-valued. We need to see that for any $x:X$ and
    $y:Y$, the type
    \[\qSig{(x',y',r)\qSig{x':X}\qEmpty{y':Y}R(x,y)}x'=x\times y'=y\]
    is equivalent to $R(x,y)$, which follows by two applications of
    Lemma~\ref{lemma:path-families}.

    Conversely, let $e:A\hookrightarrow X$ and $f:A\to Y$. We need to see that $A$ is
    equivalent to the type
    \[\qSig{x:X}\qSig{y:Y}\qSig{a:A}\big(e(a) = x\big)\times \big(f(a) = y\big).\]
    This type reshuffles to
    \[\qSig{a:A}(\qSig{x:X}e(a) = x)\times(\qSig{y:Y} f(a) = y),\]
    which is a sum of contractible types over $A$.
\end{proof}
\index{relation!single-valued|)}

\section{Partial elements}\label{section:ptl-elems}
\index{monad!lifting|(}
We will take a third approach to partial functions, via partial elements. By analogy
with the equivalence $Y\simeq (\unittype\to Y)$, we expect a partial element of $A$
to correspond to a partial function $\unittype\pto A$. Taking partial functions
as relations, this would give us the type
\[\qSig{A:Y\to \univ}{\isProp(\qSig{y:Y}{Ay})}.\]
We could instead define partial elements directly as follows.
\begin{definition}\label{def:lifting}
    For a type $X$, the \nameas{lift}{monad!lifting}
    \index{monad!lifting|seealso {partial element; partial functions}}
    \index{lifting|see {monad, lifting}}
    of $X$, or the type of \nameas{partial elements}{partial element}
    \index{type!of partial elements|see {partial element}}
    of $X$, is the type
    \definesymbolfor{$\Lift$}{Lift}{monad!lifting}
    \[\Lift(X)\defeq \qSig{P:\univ}{\isProp(P)\times (P\to X)}.\]
    The first component of a partial element is called its \emph{extent of
    definition}
    \indexandabbr{partial element!extent of definition of}{extent of definition} and the
    third component is its
    \nameas{value}{partial element!value of}
    \index{value|see {partial element, value of}}; so we have maps
    \begin{eqnarray*}
    \extent & : & \Lift Y \to \U, \\
        \val & : & \qPi{u : \Lift Y}{\extent u \to Y}.
    \end{eqnarray*}
    A partial element $y:\Lift Y$ is \nameas{defined}{partial element!defined} or
    \emph{total}\index{partial element!total|see {partial element, defined}} if $\extent(y)$.
    For any type $Y$, there is a unique \emph{undefined} element $\bot:\Lift Y$ with extent
    of definition $\zerotype$.

    A \nameas{partial function from $X$ to $Y$}{partial functions} is a function $f:X\to \Lift Y$.
    If $f:X\to\Lift Y$, we will sometimes write $f_e\defeq\extent\comp f$ and
    $f_v\defeq\val\comp f$.
\end{definition}
To see the comparison with partial elements as partial functions from $\unittype$,
write
\[\Lift'(Y) \defeq \qSig{A:Y\to\univ}{\isProp(\qSig{y:Y}{A(y)})}.\]
We have a map
\begin{align*}
    \rel &: \Lift(Y)\to\Lift'(Y) \\
    \rel(P,w,\varphi) &\defeq \big((\lambda y. \qSig{p:P}{\varphi(p) = y})\, ,\,
    w\big)
\end{align*}
where $w:\isProp(\qSig{y:Y}{\qEmpty{p:P}{\varphi(p)=y}})$ is constructed by observing that
this is equivalent to $\isProp(\qSig{p:P}{\qEmpty{y:Y}{\varphi(p)=y}})$, and this is a sum of a
contractible type over a proposition. Similarly, we have
\begin{align*}
    \ele &: \Lift'(Y)\to\Lift(Y) \\
    \ele(A,w) &\defeq \big((\qSig{y:Y}{A(y)})\, ,\, w \, ,\, \pr_0\big).
\end{align*}
\begin{lemma}
    $\rel$ is a section of $\ele$, and for any $A:Y\to\univ$ with
    $w:\isProp(\qSig{y:Y}{A(y)})$ and $y:Y$ the type $A(y)$ is equivalent to the
    first component of $\rel(\ele(A,w))$ applied to $y$.
\end{lemma}
\begin{proof}
    Let $(A,w):\Lift'(Y)$. The relation defined by $\rel(\ele(A,w))$ is
    \[R(y) = \sum_{(y',a):\qSig{y':Y}{A(y')}}(y=y').\]
    Re-associating, this is equivalent to
    \[\qSig{y':Y}{A(y')\times (y=y')}.\]
    This in turn is equivalent to $A(y)$ by the map
    $(y',a',p) \mapsto \transport(p^{-1},a')$ with inverse mapping $a$ to $(y,a,\refl)$.

    Now let $x \defeq (P,w,\varphi):\Lift(Y)$. We need only see that the extent of
    definition of
    $\ele(\rel(x))$ is equivalent to $P$, since $P$ is a proposition. We have
    \[\extent(\ele(\rel(x))) = \qSig{y:Y}{\qEmpty{p:P}{(\varphi(p)=y)}}.\]
    Re-associating as above, this type is equivalent to $P\times 1$.
\end{proof}
Note that without univalence the equivalence 
    \[ A y \simeq \qSig{(y',a): \qSig{y':Y}Ay'}(y'=y),\]
is insufficient to prove that $\ele$ is an equivalence, even in the presence of proposition
extensionality. While $\qSig{y:Y}{Ay}$ is a proposition, we don't know in general
that $A(y)$ is.  In the case that $Y$ is a set, however, we can show that $A(y)$ must
be a proposition, by Lemma~\ref{lemma:propositional-families-have-propositional-sums}.
Then we have the following.
\begin{lemma}\label{lemma:set-lifting}
    For any set $Y$, the functions $\rel$ and $\ele$ determine an equivalence between the
    types $\Lift(Y)$ and $\Lift'(Y)$.
\end{lemma}
\begin{proof}
    We have already seen that $\rel$ is a section of $\ele$, so we need to see that
    $\ele$ is a section of $\rel$. Fix $(A,w):\Lift'(Y)$
    By function extensionality, and since the second component of $\Lift'(Y)$ is a
    proposition, it suffices to turn the equivalence $A(y) \simeq
    \pr_0(\rel(\ele(A,w)))(y)$ into an equality. As $Y$ is a set and $\qSig{y:Y}A(y)$
    is a proposition by assumption, we have by
    Lemma~\ref{lemma:propositional-families-have-propositional-sums} that $A(y)$ is a
    proposition. Then since $A(y)\simeq \pr_0(\rel(\ele(A,w)))(y)$, we may apply
    proposition extensionality to get the desired equality $A(y)=
    \pr_0(\rel(\ele(A,w)))(y)$.
\end{proof}
By considering maps into $\Lift'(Y)$ and $\Lift(Y)$, we can collect
Lemmas~\ref{lemma:set-lifting} and~\ref{ptl-functions:classical-categorical} into the
following theorem.
\begin{theorem}\label{thm:ptl-fxn-char}
    For any type $X$ and any set $Y$, the
    following types are equivalent:
    \begin{enumerate}[label=(\roman*)]
        \item $\qSig{R:X\to Y\to \univ}\qPi{x:X}{\isProp(\qSig{y:Y}{R(x,y)})}$ \label{lift:element}
        \item $\qSig{A:\univ}{\qEmpty{e:A\to X}{\isEmbedding(e)\times (A\to Y)}}$ \label{lift:subtype}
        \item $X \to \Lift(Y)$.\label{lift:partial}
    \end{enumerate}
    Moreover, in the presence of univalence, we may drop the condition that $Y$ is a set.
\end{theorem}
We have an embedding 
\index{eta@{$\eta$}!unit of the lifting monad|textit}
\[\eta:Y\to\Lift Y\]
given by 
\[\eta(y)\defeq (\unittype,-,\lambda u.y).\]
\begin{lemma}
    The map $\eta:Y\to \Lift Y$ is an embedding.
\end{lemma}
\begin{proof}
    Since
    $\unittype=\unittype$ is contractible, $\eta(x)=\eta(y)$ is equivalent to
    \[\lambda u.x=_{(\unittype\to Y)}\lambda u.y.\]
    By function extensionality and the induction principle for $\unittype$, this is
    equivalent to 
    \[(\lambda u.x)(\star)=(\lambda u.y)(\star),\]
    which reduces to $x=y$.
\end{proof}
\index{Kleisli extension!for the lifting monad|textit}
\index[symbol]{klext@{$\klext{(-)}$}|textit}
Moreover, we can extend a function $f:X\to\Lift Y$ to a function $f^\sharp:\Lift
X\to\Lift Y$ given by
\begin{eqnarray*}
  f^\sharp(P:\U,-,\varphi:P\to X) & \defeq & (Q:\U,-,\gamma:Q \to Y), \\[1ex]
                    & \text{where} & 
    Q \defeq \qSig{p:P}{\extent(f(\varphi(p)))}, \\
  & & \gamma(p,e) \defeq \val(f(\varphi(p)))(e),
\end{eqnarray*}
where the witness that $Q$ is a proposition comes from the fact that propositions are
closed under sums.

For predicativity reasons, $(\Lift,\eta,\klext{(-)})$ does not actually give a monad,
since $\Lift$ raises universe levels. We discuss these issues in~\ref{section:size},
but we still prove the monad laws here: for our purposes, we only need $\klext{(-)}$
for composition of partial function, and $\eta$ to give us the total elements.
\begin{theorem}
    The maps $\eta$ and $(-)^\sharp$ give $\Lift:\univ\to\univ[1]$ the structure of a monad.
\end{theorem}
\begin{proof}For notational convenience, we will completely
suppress the witness that the extent of definition of a partial element is a proposition.
\begin{eqnarray*}
    \eta^\sharp(P,\varphi) &=& (\qSig{p:P}{\extent(\eta(\varphi p))}\; ,\;\lambda
    (p,q). \val(\varphi p)q) \\
    &=& ((\qSig{p:P}{\unittype}),\lambda (p,q).\varphi p) \\
    &=& (P\times \unittype,\lambda (p,q).\varphi p) \\
    &=& (P,\varphi),
\end{eqnarray*}
where the last equality follows from proposition extensionality.
\begin{eqnarray*}
    f^\sharp(\eta x) &=& f^\sharp(1,\lambda p.x) \\
    &=& (\qSig{p:1}{\extent(f(x))}\;,\;\lambda (p,q). \val(f(x))q) \\
    &=& (\extent(f(x)),\val(f(x))) \\
    &=& f x.
\end{eqnarray*}
Now let $g:Y\rightarrow \Lift Z$ and $f:X\rightarrow \Lift Y$. We compute
\begin{align*}
    (g^\sharp f)^\sharp (P,\varphi ) &= \left(\qSig{p:P}{(g^\sharp f)_e(\varphi p)}\, ,\;
        \lambda (p,q).(g^\sharp f)_v(\varphi p)q\right)\\
    &= \Big(\qSig{p:P}{g^\sharp_e f(\varphi p)}\, ,\;
        \lambda (p,q).g^\sharp_v (f(\varphi p))q\Big) \\
    &= \Big(\qSig{p:P}{\qEmpty{q:f_e(\varphi p)}g_e(f_v(\varphi p)q)}\, ,\;
        \lambda(p,(q,r)).g_v(f_v(\varphi p)q)r\Big)\\
    &= \Big(\qSig{(p,q):\qSig{p:P}{f_e(\varphi p)}}{g_e(f_v(\varphi p)q)}\, ,\;
        \lambda((p,q),r).g_v(f_v(\varphi p)q)r\Big)\\
    &= g^\sharp\Big(\qSig{p:P}{f_e(\varphi p)}\, ,\;
    \lambda(p,q).f_v(\varphi p)q\Big)\\
    &= g^\sharp(f^\sharp (P,\varphi )).
\end{align*}
The first 3 equalities are just expansion of definitions; the third follows from the
    equivalence between $\qSig{a:A}{\qEmpty{b:Ba}{C(a,b)}}$ and
    $\qSig{(a,b):\qSig{a:A}{B(a)}}{C(a,b)}$,
and the last two again apply definitions.
\end{proof}
This gives us a Kleisli composition operator we denote with~``$\kcomp$''. That is, we
have
\[g\kcomp f \defeq g^\sharp \comp f.\]

We said that a partial element $p:\Lift Y$ is \emph{defined} if $\extent(p)$. It
would also be natural to say that $p$ is defined if $p$ is in the image of $\eta$.
In fact, these are equivalent. Moreover, since $\eta$ is an embedding, we know that
$\fib_\eta(p)$ is a proposition, so we have
\indexsymbolfor{$\isDefined$}{defined}{partial element, defined}
\begin{theorem}\label{thm:isdefined-char}
    The following types are equivalent for any $p:\Lift Y$.
    \begin{enumerate}
        \item $\isDefined(p)$,
        \item $\qExists{y:Y}p=\eta y,$
        \item $\qSig{y:Y}p=\eta y.$
    \end{enumerate}
\end{theorem}
\begin{proof}
    The type $\qSig{y:Y}p=\eta y$ is the fiber of $\eta$ over $p$. As $\eta$ is an
    embedding, it has propositional fibers, so $\qSig{y:Y}p=\eta y$ is a proposition.
    Hence, we have an equivalence between $(\qExists{y:Y}p=\eta y)$ and $\qSig{y:Y}p=\eta y$.

    As $\isDefined(p)$ is a proposition, we only need
    to see that $\isDefined(p)\lequiv \qExists{y:Y}p=\eta y$.

    Let $p$ be defined so that $\extent(p)=\unittype$. Then we have by definition
    that $p = \eta(\val(p)(\ast))$. Conversely, if $\qExists{y:Y}p=\eta y$, since
    $\isDefined(p)$ is a proposition, it is enough to define a map $(\qSig{y:Y}p=\eta
    y)\to \isDefined(p)$. If $p=\eta y$, then we have that $\extent(p) = 1$ by the
    characterization of equality in $\Sigma$ types.
\end{proof}
\begin{definition}
    \index{partial function!total|seealso {partial element, defined}}
    A partial function $f:X\to\Lift Y$ is \nameas{total}{partial function!total} if
    it is defined everywhere:
    \[\total(f)\defeq \qPi{x:X}\isDefined(f(x)).\]
\end{definition}
A total function is a function $X\to Y$, and we expect the above notion to align with
the notion of total function. Indeed, this is the case:
\begin{lemma}\label{lemma:total-is-defined}
    For any partial function $f:X\to \Lift Y$, the following types are equivalent:
    \begin{enumerate}
        \item the type of factorizations of $f$ through $Y$:
            \[T(f)\defeq \qSig{g:X\to Y}f = \eta\comp g.\]
        \item the type $D(f)\defeq\qPi{x:X}\isDefined(f(x))$.
    \end{enumerate}
\end{lemma}
\begin{proof}
    We have that $T(f)$ is equivalent by function extensionality to
    \[T(f)\defeq \qSig{g:X\to Y}\qPi{x:X}f(x) = \eta(g(x)).\]
    Recall that by Theorem~\ref{thm:ttchoice} we have for any type family $R$
    \[\Big(\qSig{g:X\to Y}\qPi{x:X} R(x,g(x))\Big)\simeq
    \Big(\qPi{x:X}\qSig{y:Y}R(x,y)\Big).\]
    In particular, for $R(x,y) \defeq f(x)=\eta(y)$, we have that $T(f)$ is equivalent to
    \[\qPi{x:X}\qSig{y:Y}f(x)=\eta(y).\]
    By Theorem~\ref{thm:isdefined-char}, this type is equivalent to $D(f)$.
\end{proof}
\begin{corollary}
    There is at most one factorization of a partial function $f:X\to\Lift(Y)$ through
    a total function $g:X\to Y$. That is, the type $T(f)$ above is a proposition.
\end{corollary}
The takeaway here is that from the knowledge that
$f:X\to\Lift Y$ is defined everywhere, we can treat $f$ as a function $X\to Y$. Justified
by the above results, we will call $f$ a \emph{total function} both to mean that
$f:X\to \Lift(Y)$ such that $\total(f)$ is inhabited, but also to emphasize that $f$
is an ordinary function $f:X\to Y$.

\section{Domain and range of partial functions}\label{section:dom-and-ran}
A function $f:X\to\Lift(Y)$ can be viewed either as an ordinary function from $X$ to
$\Lift(Y)$, or as a partial function from $X$ to $Y$. When we view such a function as
a partial function, we want the notion of image and domain to be
different. If we view a partial function $X\pto Y$ instead as a relation
$R:X\to Y\to\univ$, then the \emph{range} of $R$ is $\qSig{y:Y}\qExists{x:X}R(x,y)$,
and the \emph{domain} of $R$ is $\qSig{x:X}\qExists{y:Y}R(x,y)$. Similarly, if
we view a partial function $X\pto Y$ as a function $f:A\to Y$ for a subtype $A$ of
$X$, then the range of $f$ is $\qSig{y:Y}\qExists{a:A}f(a) = y$, and the domain is
$A$. These notions align, and moreover, we can translate the notion to the view of
partial functions as functions into a type of partial elements. To separate the image
of a partial function from the image of an ordinary function, we stick to the name
\emph{range}. We will be more interested in the \emph{predicates} for range and
domain than the total types of these predicates.
\begin{definition}
    The \nameas{range}{partial function!range} of a partial function $f:X\to\Lift(Y)$ is the predicate
    \[\ran_f:Y\to\Prop\]
    defined by
    \[\ran_f(y) \defeq \qExists{x:X}f(x) = \eta y.\]
    The \nameas{domain}{partial function!domain} of $f$ is the predicate 
    \[\dom_f:X\to\Prop\]
    defined by
    \[\dom_f(x) \defeq \isDefined(f(x)).\]
\end{definition}
That is, $y$ is in the range of $f$ if there exists $x:X$ such that the value of
$f(x)$ is $y$, and $x$ is in the domain of $f$ if $f$ is defined at $x$.

\section{Liftings are DCPOs}\label{section:dcpos}
Given a type $Y$, we can define a type family $-\le-:\Lift Y\to \Lift Y\to\univ$ by
\[u\le v \defeq
    \qSig{t:\extent(u)\to \extent(v)}{\qPi{p:\extent(u)}{\val(u)(p)=\val(v)(t(q))}}.\]
For a general type $Y$, there may be multiple witnesses that $u\le v$, but if $Y$ is a
set, then $-\le-$ is proposition-valued. Moreover, we have
\begin{theorem}\label{thm:dcpo}
    If $Y$ is a set, then $(\Lift Y,\le)$ is a directed-complete partial order with
    bottom $\bot$.
\end{theorem}
Here, a directed-complete partial order $X$ must have a prop-valued order relation, and
for any inhabited family $u:I\to X$ such that $\qPi{i,j:I}\qExists{k:I}u_i,u_j\le k$,
there is a least upper bound in $X$.
\begin{proof}
    Let $Y$ be a set. We need to see that $\le$ is proposition-valued, so let
    $x,y:u\le v$. Then we have
    \[(x=y)\simeq \qSig{p:\extent(u)\to\extent(v)}\qPi{p:\extent(u)}\val(u)(p) =
    \val(v)(t(q)).\]
    As $Y$ is a set, we know that $\val(u)(p) = \val(v)(t(q))$ is a proposition, and
    as $\extent(u)$ is a proposition, we have that $\qPi{p:\extent(u)}\val(u)(p) = \val(v)(t(q))$
    is as well. Then $x=y$ is equivalent to a sum of propositions over a proposition.

    If $u\le v$ and $v\le w$, then $u\le w$ by composition in the first coordinate
    and concatenation in the second. We have that $u\le u$ by $(\id,\lambda
    p.\refl)$. Finally, let $u\le v$ and $v\le u$. Then we have $t:\extent(u)\to
    \extent(v)$ and $t':\extent(v)\to\extent(u)$, and by proposition extensionality
    we have $p:\extent(u)=\extent(v)$. We wish to see that $u = v$, so
    we need to see that $\transport^{\lambda P.P\to Y}(p,\val(u)) = \val(v)$. This is equivalent
    to $\val(u) \comp \coe(p^{-1}) = \val(v)$. This in turn is equivalent to
    \[\qPi{q:\extent(v)}\val(u)(t'(q)) = \val(v)(q),\]
    and this type has an element by the assumption that $v\le u$. In other words, we have an
    element
    $p:\extent(u) = \extent(v)$ and $\val(u) =_p \val(v)$. So then $u=v$.

    Given a directed family $u_i:I\to \Lift Y$, take the extent of definition of the join
    $u_\infty$ to be $\extent(u_\infty)\defeq \trunc{\qSig{i:I}{\extent(u_i)}}$. Then by
    construction we have $\isProp(\extent(u_\infty))$. To define the value, we first
    define a function $\varphi:\big(\qSig{i:I}{\extent(u_i)}\big)\to Y$ by
    \[\varphi(i,p)\defeq \val(u_i)(p).\]
    As the family is directed, for any $i$ and $j$, there is a $k$ such that
    $u_i,u_j\le u_k$. So then we have for any $p:\extent(u_i)$ and $q:\extent(u_j)$
    we have
    \[\varphi(i,p) = \varphi_i(p) = \varphi_k(p') = \varphi_k(q') = \varphi_j(q) =
    \varphi(j,q)\]
    where $\varphi_i=\val(u_i)$, for $p',q':\extent(u_k)$ witnessing that $u_i\le
    u_k$ and $u_j\le u_k$ respectively.  That is, $\varphi$ is constant. So we
    have by Theorem~\ref{thm:constant-functions-factor}
    that $\varphi$ factors through $\extent(u_\infty)$ as $\varphi_\infty\circ |-|$.
    So take $\extent(u_\infty)\defeq \varphi_\infty$.

    We have $\extent(u_i)\to \extent(u_\infty)$ by $p\to |(i,p)|$. If $v:\Lift Y$
    with $\qPi{i:I}{u_i\le v}$, we have a map $(\qSig{i:I}{u_i})\to \extent(v)$, and this
    factors through $u_\infty$ as $\extent(v)$ is a proposition.
\end{proof}
\index{monad!lifting|)}

\section{Dominances and partial functions}\label{section:dominances}
\index{dominance|(}
\index{submonad|(}
The traditional way to examine restricted classes of partial
functions~\cite{Rosolini1986,Hyland1991,VANOOSTEN2000} is to restrict the available
extents of definition via \emph{dominances}---subsets of $\Prop$ satisfying certain
closure properties. We develop this approach here.
\begin{definition}
    A \emph{set of propositions} is a map $\dd:\univ\to\univ$ together with
    \begin{enumerate}[label=\normalfont\textrm{D\arabic*}]
        \item a map $\qPi{X:\univ}{\isProp(\dd(X))}$, \label{dominance:property}
        \item a map $\qPi{X:\univ}{\dd(X)\to\isProp(X)}$. \label{dominance:is-prop}
    \end{enumerate}
    Note that these pieces of data are propositions. We say that $\dd$ is
    \emph{proposition valued} when it satisfies~\ref{dominance:property}, and that $\dd$
    \emph{selects propositions} when it satisfies~\ref{dominance:is-prop}.

    Given any set of propositions $\dd$, the type $\Lift_{\dd}(Y)$ of
    \emph{$\dd$-partial elements} of a type $Y$ is given by
    \[\Lift_{\dd}(Y)\defeq \qSig{P:\univ}{\dd(P)\times (P\to Y)}.\]
\end{definition}
The next lemma is simple but invaluable, and we will make tacit use of it.
\begin{lemma}\label{lemma:set-of-props-is-equiv-closed}
    Any set of propositions $\dd:\univ\to\univ$ is closed under equivalence. That is,
    if $\dd(P)$ and $P\simeq Q$, then $\dd(Q)$.
\end{lemma}
\begin{proof}
    As $\dd$ selects propositions, we have $\isProp(P)$. Then since $Q\simeq P$, we
    know $\isProp(Q)$, and so $P=Q$ by proposition extensionality. Hence, $\dd(Q)$,
    by transport.
\end{proof}
We may then define the $\dd$-partial functions $X\pto_{\dd} Y$ to be functions
\index{partial functions!for a dominance|textit}
\[X\pto_{\dd} Y \defeq X\to\Lift_{\dd} Y.\]
In general, the $\dd$-partial functions might not be well-behaved: we do not
necessarily have that the $\dd$-partial functions compose, or that total functions can be
viewed as $\dd$-partial. In order to ensure these two properties, Rosolini introduced
the notion of dominance.
\begin{definition}\label{def:dominance}
    A \name{dominance} is a set of propositions $\dd:\univ\to\univ$, which we call
    \emph{dominant}, satisfying
    \begin{enumerate}[label=\normalfont\textrm{D\arabic*}]
            \setcounter{enumi}{2}
        \item The unit type is dominant: we have $\upsilon:\dd(\unittype)$, \label{dominance:unit}
        \item \label{dominance:axiom} 
            $\dd$ has \emph{conditional conjunction}: For any $P,Q:\univ$ we have
            \[\dd(P)\to(P\to \dd(Q))\to \dd(P\times Q).\]
    \end{enumerate}
    Note that both pieces of data are property by \ref{dominance:property}.
\end{definition}
In~\cite{Bauer2006Synthetic}, Bauer calls~\ref{dominance:axiom} \emph{the dominance
axiom}, but does include~\ref{dominance:unit} in the definition of dominance; since he works in
the internal language of a topos,~\ref{dominance:property}
and~\ref{dominance:is-prop} are redundant.  Reus and
Streicher~\cite{ReusStreicher1997} call~\ref{dominance:axiom} \emph{dependent
conjunction}; since they are axiomatizing a particular dominance (the \emph{Rosolini
dominance} of Section~\ref{section:rosolini}), the relevant versions
of~\ref{dominance:property},~\ref{dominance:is-prop}, and~\ref{dominance:unit} are
theorems.

There are three trivial examples of dominances:
\[
    \text{$\dd_1(X) \defeq \isContr(X)$, \quad $\dd_2(X) \defeq (X=0)+(X=1)$, \quad $\dd_\Omega \defeq \isProp$},
\] 
with liftings corresponding to
\[
    \text{$\Lift_{\dd_1}(X) \defeq X$, \quad $\Lift_{\dd_2}(X) \defeq X+1$, \quad
    $\Lift_{\dd_\Omega}(X) \defeq \Lift(X)$},
\] 
In a type theoretic context, it is usually more convenient to let the conditional
proposition $Q$ in \ref{dominance:axiom} also depend on the witness of $P$.
Generalizing \ref{dominance:axiom} to this case gives us
    \begin{enumerate}[label=\normalfont\textrm{D\arabic*'}]
            \setcounter{enumi}{3}
        \item we have map \label{dominance:closure}
            \[\sigma:\qPi{P:\univ}{\qEmpty{Q:P\to\univ}{\dd(P)\to
            (\qPi{p:P}{\dd(Q(p))})\to\dd(\qSig{p:P}{Q(p)})}}.\]
    \end{enumerate}
This condition is perhaps more immediately recognizable: it says that a set of
propositions is closed under $\Sigma$. In fact, we have this already from conditional
conjunction.
\begin{lemma}\label{dominance:axiom-is-closure}
    For a set $\dd$ of propositions,~\ref{dominance:axiom}~is equivalent
    to~\ref{dominance:closure}.
\end{lemma}
\begin{proof}
For fixed $P$ and $Q$, the constant family $\lambda(p:P).Q(p)$ is
    such that $P\times Q = \qSig{p:P}{Q(p)}$ and $P\to Q = \qPi{p:P}{Q(p)}$, so closure under
    $\Sigma$ gives~\ref{dominance:axiom}.

    Conversely, let $\dd$ satisfy~\ref{dominance:axiom} and let $P:\U$ and $R:P\to \U$ such that
    $\dd(P)$ and $\qPi{p:P}{\dd(R(p))}$. Define
\[
    Q\defeq \qSig{p:P}{R(p)}.
\]
We need to see that $\dd(Q)$. First note that for any $p:P$ we have $R(p)=Q$, by $r\mapsto (p,r)$ with inverse $(p',r')\mapsto \transport(-,r')$, where $``-''$ is the
path $p'=p$ which exists since $P$ is a proposition. But then we have
    $\qPi{p:P}{d(Q)}$, by assumption that $R$ is $\dd$-valued. By the condition, we then
have $\dd(P\times Q)$. But $Q\simeq P\times Q$ by $(p,r)\mapsto (p,(p,r))$ with
inverse given again by transport.
\end{proof}

The following result is the motivation for the notion of dominance.
\index{monad!lifting!of a dominance|textit}
\begin{theorem}\label{dominance:monad}
    A set of propositions $\dd$ is a dominance iff $\Lift_{\dd}:\univ\to\univ[1]$ is
    a submonad of $\Lift:\univ\to\univ[1]$.
\end{theorem}
\begin{proof}
    Let $\dd$ be a dominance.
    The map $\Lift_{\dd} Y\to \Lift Y$ is given by inclusion. We only need to see that
    the unit and Kleisli extension for $\Lift$ respect $\Lift_{\dd}$. As
    $\dd(\unittype)$, we have that $\eta(y)$ is a $\dd$-proposition. We need to see
    that for all $(P,-,\varphi):\Lift_{\dd}(X)$ and $f:X\to\Lift_{\dd}(Y)$, that
    the extent of $f^\sharp(P,-\varphi)$ is $\dd$-partial. But the extent of
    $f^\sharp(P,-,\varphi)$ is
    \[\qSig{p:P}{\extent(f(\varphi(p)))},\]
    but by assumption $P$ is a $\dd$-proposition and $\extent(f(\varphi(p)))$ is
    $\dd$-partial, so by condition~\ref{dominance:closure}, we are done.

    Now let $\Lift_{\dd}$ be a submonad of $\Lift$, and let $\rho:\dd(P)$ and
    $\sigma:P\to\dd(Q)$. Define
    \begin{align*}
        p & :\unittype\to\Lift_{\dd}(P)\\
        p(\star) &\defeq (P,\rho,\id)
    \end{align*}
    and
    \begin{align*}
        q & :P \to\Lift_{\dd}(\unittype)\\
        q(x) &\defeq (Q,\sigma(x),\lambda w.\star)
    \end{align*}
    As $\Lift_{\dd}$ is a monad, we have
    \[q^{\sharp}\circ p : \unittype \to \Lift_{\dd}(\unittype).\]
    Moreover, $\Lift_{\dd}$ is a submonad of $\Lift$, so we can calculate,
    \[q^{\sharp}(p(\star)) = q^{\sharp}(P,-,\id) = ((\qSig{p:P}{Q}), w , v).\]
    We do not need to calculate the explicit value of $w$ , as we have that $w$
    inhabits a proposition. We also do not need to calculate the explicit value of $v$,
    since we know there is a unique function $X\to\unittype$ for any type $X$.
\end{proof}
As with the dominance of all propositions, we can consider relations valued in a
dominance. That is, a \emph{$\dd$-valued} relation between $X$ and $Y$ is a relation
$R:X\to Y\to\univ$ such that
\[\qPi{x:X}{\dd(\qSig{y:Y}{R(x,y)})}.\]

\begin{theorem}\label{dominance:composition}
    A set of propositions $\dd$ is a dominance iff the $\dd$-valued relations are
    closed under composition and contain the identity relation.
\end{theorem}
\begin{proof}
    It is easy to check that the identity relation is $\dd$-valued iff
    $\dd(\unittype)$ since $\qSig{y:Y}{x=y}$ is contractible.

    Let $\dd$ be a dominance and consider $\dd$-valued relations $R:X\to Y\to\univ$
    and $S:Y\to Z\to\univ$. We need to see that for any $x:X$, the type
    \[\qSig{z:Z}{\big(R; S(x,z)\big)}\defeq \qSig{z:Z}{\qEmpty{y:Y}{R(x,y)\times R(y,z)}}\]
    is dominant. By the reshuffling map $(z,(y,(r,s)))\mapsto ((y,r),(z,s))$, this
    type is equivalent to
    \[\sum_{(y,r):\qSig{y:Y}{R(x,y)}}\qSig{z:Z}{S(y,z)}.\]
    By assumption, both $\qSig{y:Y}{R(x,y)}$ and $\qSig{z:Z}{S(y,z)}$ are dominant.
    As $\dd$ is closed under $\Sigma$, we have that $R; S$ is $\dd$-valued.

    Conversely, let the $\dd$-valued relations be closed under composition, and
    suppose $\dd(P)$ and $P\to \dd(Q)$. Define $R:\unittype\to P\to\univ$ and
    $S:P\to\unittype\to\univ$ by
    \[R(\star,p) \defeq P,\]
    and
    \[S(p,\star) \defeq Q.\]
    For any $p:P$ we have $\qSig{u:\unittype}{S(p,u)} \simeq S(p,\star) = Q$, so that
    \[\qPi{p:P}{\dd(\qSig{u:\unittype}{S(p,u)} )},\]
    That is, $S$ is $\dd$-valued. Similarly, we have $(\qSig{p:P}{P}) \simeq
    (P\times P)\simeq P$, so that $R$ is $\dd$-valued. Then by assumption, $R;S$ is
    $\dd$-valued. Calculating $R;S(\star,\star)$ we see
    \[R;S(\star,\star) \defeq \qSig{p:P}{P\times Q}.\]
    As $P$ is a proposition, this type is equivalent to $P\times Q.$
\end{proof}
Collecting Theorems~\ref{dominance:monad} and~\ref{dominance:composition} and
recalling Definition~\ref{def:submonad}, we have
\begin{theorem}\label{dominance:characterization}
    The following are equivalent for any set of propositions $\dd$:
    \begin{enumerate}
        \item $\dd$ is a dominance;
        \item $\Lift_{\dd}$ is a submonad of $\Lift$.
        \item There is a composition function 
            \[-\kcomp -:(Y\to \Lift_{\dd}(Z))\to (X\to\Lift_{\dd}(Y))\to
            (X\to\Lift_{\dd}(Z)),\] restricting Kleisli
            composition for the $\Lift$ monad.
        \item The $\dd$-valued relations are closed under composition.
    \end{enumerate}
\end{theorem}
Theorems \ref{dominance:monad} and \ref{dominance:composition} express essentially
the same fact: the $\dd$-partial functions are closed under composition iff
$\dd$ is a dominance. Theorem \ref{dominance:monad} expresses this fact by viewing
$\dd$-partial functions $X\pto_{\dd} Y$ as functions $X\to \Lift_{\dd}(Y)$, while
theorem \ref{dominance:composition} does so by viewing $\dd$-partial as $\dd$-valued
relations. We can relate the two views by restricting the functions $\ele$ and $\rel$
given above for the dominance of all propositions.
\begin{theorem}
    Assuming univalence, for any set of propositions $\dd$, the maps $\ele$ and
    $\rel$ lift to a section-retraction pair on the $\dd$-valued relations and
    $\dd$-partial functions.
\end{theorem}
\begin{proof}
    Letting $X\relto_{\dd}Y$ be the type of $\dd$-valued relations from $X$ to $Y$,
    the lifted functions are 
    \begin{align*}
        r&:(X\relto_{\dd}Y) \to (X\to\Lift_{\dd}(Y))\\
        r((R,\varphi)) &= \lambda x.\Big(\qSig{y:Y}R(x,y) \; ,\; \varphi,\pr_0\Big),
    \end{align*}
    and
    \begin{align*}
        s&: (X\to\Lift_{\dd}(Y)) \to(X\relto_{\dd}Y)\\
        s(f) &= \lambda x,y. (\qSig{p:f_e(x)}f_v(p) = y,w),
    \end{align*}
    where $w$ arises from the fact that $f_e(x)$ is always a $\dd$-proposition,
    and
    that 
    \[\big(\qSig{y:Y}\qSig{p:f_e(x)}f_v(p) = y\big)\simeq f_e(x),\]
    using Lemma~\ref{lemma:set-of-props-is-equiv-closed}. The
    identity $r\comp s\htpy \id$ follows directly from proposition extensionality.
\end{proof}
The pair of functions given above are the restriction of the maps in
Lemma~\ref{lemma:set-lifting}. They fail to determine an equivalence, because we
require equality between general types. Hence, they would be equivalences in the
presence of univalence.
\begin{corollary}
    Assuming univalence, for any set of propositions $\dd$ and $X,Y:\univ$, the
    type $X\to\Lift_{\dd}(Y)$ is equivalent to the type
    \[\Sigma(R:X\to Y\to \univ),\qPi{x:X}{\dd(\Sigma(y:Y),R(x,y))}.\]
\end{corollary}
Lemma~\ref{lemma:set-lifting} also restricts to an arbitrary dominance. As it relies
on the fact that when $Y$ is a set, the sum $\qSig{y:Y}A(y)$ is a proposition iff
each $A(y)$ is (Lemma~\ref{lemma:propositional-families-have-propositional-sums}), this may
be surprising: the corresponding fact does not hold for an arbitrary dominance.
However, it is not the dominance of propositions that gives us
Lemma~\ref{lemma:propositional-families-have-propositional-sums}, but the $h$-level of
propositions.
\begin{corollary}
    If $Y$ is a set, then for any set of propositions $\dd$ and any type $X:\univ$,
    the type $X\to\Lift_{\dd}(Y)$ is equivalent to the type
    \[\Sigma(R:X\to Y\to \univ),\qPi{x:X}{\dd(\Sigma(y:Y),R(x,y))}.\]
\end{corollary}
\index{submonad|)}
\index{dominance|)}

\section{Size issues}\label{section:size}
\index{universe|(}
Because the lifting raises universe levels, $\Lift$ and $\Lift_{\dd}$ do not form
monads as defined in Section~\ref{section:monads}. However, if we restrict our
attention to propositions at universe $\univ[0]$,
then the universe levels do not go past $\univ[1]$: the lifting operation is a monad
on $\univ[1]$. We use the current section to explain this.

If $X:\univ[0]$, then $\Lift(X):\univ[1]$, so $\Lift$ raises universe levels.
However, it turns out that the raising stops at
$\univ[1]$, assuming we take $\Prop$ to be $\Prop_{0}$. Our operations are of the form
$\dd:\univ[0]\to\univ[0]$, so for each
$j$ and $X:\univ[j]$, we have $\Lift^j$,
\[\Lift_{\dd}^{j}(X) \defeq \qSig{P:\univ[0]}{\dd(P)\times P\to X}.\]
Since this quantification is over $\univ[0]$, we have that
$\Lift_{\dd}^{j}(X):\univ[\max(j,1)]$. That is,
\[\Lift_{\dd}^{j}:\univ[j]\to\univ[\max(j,1)].\]
In other words, for $j>0$, we have $\Lift_{\dd}^{j}:\univ[j]\to\univ[j]$.
If we use resizing axioms, then we would have $\Prop:\univ[0]$, so that the
quantification in the definition of $\Lift(X)$ can be over a type in $\univ[0]$. Then
$\Lift^0_d:\univ[0]\to\univ[0]$, and lifting would not raise universe levels at all.

When $\dd$ satisfies~\ref{dominance:unit}, the unit $\eta$ can similarly be indexed
by universe levels, with type
\[\eta^j:\qPi{X:\univ[j]}{X\to\Lift^j(X)}.\]
Similarly, when $\dd$ satisfies~\ref{dominance:closure}, the Kleisli operator has
type
\[\qPi{X:\univ[j]}\qEmpty{Y:\univ[k]}(X\to\Lift^k(Y))\to\Lift^j(X)\to\Lift^k(Y),\]
so Kleisli composition has type
\[\qPi{X:\univ[j]}\qEmpty{Y:\univ[k]}\qEmpty{Z:\univ[l]}(Y\to\Lift^l(Z))\to(X\to\Lift^k(Y))\to
(X\to\Lift^l(Z)).\]
The point here is that even though lifting raises universe levels, it does so in a
coherent way. In particular, if $X,Y,Z:\univ[0]$, then we have that Kleisli
composition has type
\[(Y\to\Lift^0(Z))\to(X\to\Lift^0(Y))\to (X\to\Lift^0(Z)).\]
It is worth noting, moreover that the set of propositions we are most interested in, the \emph{Rosolini
propositions} (Definition~\ref{definition:rosolini}) arises in such a way that
Rosolini liftings are equivalent to a small type
(Lemma~\ref{lemma:delay-is-rosstruct}), so the lifting of interest is predicative.
\index{universe|)}

\section{Rosolini propositions}\label{section:rosolini}
\index{Rosolini proposition|(}
\index{dominance!Rosolini|see{ proposition, Rosolini}}
\index{type!of natural numbers!extended|(}
A particularly important dominance in synthetic domain theory is the dominance of
\emph{Rosolini propositions}. We approach Rosolini propositions via the
extended naturals, $\NI$. Recall the map $\lar{-}:\NI\to\univ$  which takes a sequence
(with at most one $1$) to the proposition that it takes the value $1$
(Section~\ref{section:delay}). The
\nameas{Rosolini propositions}{Rosolini proposition} are the types in the image of this map. Explicitly,
\index{proposition!Rosolini|see{Rosolini proposition}}
\begin{definition}\label{definition:rosolini}
$P:\univ$ is
Rosolini when
    \[\trunc{\qSig{u:\NI}{P=\lar{u}}}.\]
Traditionally, the family of Rosolini propositions is denoted with $\Sigma$, which we
avoid here due to notational clashes; we will use $\isRos$. 
    \defineopfor{$\isRos$}{isros}{Rosolini proposition}
\end{definition}

\index{Rosolini structure|seealso{Rosolini proposition}}
Alternatively we may define a \name{Rosolini structure} over $P$ to be a $u:\NI$ such
that $P=\lar{u}$. That is,
\[\RosStruct(P)\defeq \qSig{u:\NI}{P=\lar{u}},\]
    \defineopfor{$\RosStruct$}{rosstruct}{Rosolini structure}
so that $\isRos(P)=\trunc{\RosStruct(P)}$.

The Rosolini propositions arise from computational considerations: if $P$ is a
Rosolini proposition, we can see a sequence $\alpha:\NI$ such that $P\lequiv
\lar{\alpha}$ as a ``semi-decision procedure'' for $P$. Then, a proposition is
Rosolini if there exists a semi-decision procedure for it. However, this procedure is
abstract: we have made no assumption that there is an actual algorithm for such a
procedure. In Chapter~\ref{chapter:comp-as-prop}, we will consider semi-decidable
propositions and examine to what extent the Rosolini and semi-decidable propositions
align.
\index{type!of natural numbers!extended|)}

The Rosolini propositions also have a connection to analysis, discovered together with Auke
Booij: call a proposition $P$ \emph{Cauchy} if there exists a Cauchy real number
$r:\real$
such that
\[P\lequiv (0< r).\]
That is,
\[\isCauchy(P)\defeq \trunc{\qSig{r:\real}P\lequiv (0< r)}.\]
Here the Cauchy reals are defined following Bishop~\cite{bishop1967} as equivalence
classes of regular sequences $s:\nat\to\rat$. A sequence $s:\nat\to\rat$ is \emph{regular} when,
\[\isReg(s)\defeq \qPi{n,m:\nat}|s_n-s_m| \le \frac{1}{m+1} + \frac{1}{n+1}.\]
We will not discuss here the full construction of $\real$; what matters for us is
that for any real number $r:\real$, there exists a sequence of rationals converging
to it:
\[\qPi{r:\real}\qExists{s:\nat\to\rat}\isReg(s)\times
\left(\qPi{n:\nat}|r-s_n|<\frac{1}{n+1}\right).\]
\begin{theorem}
    A proposition is Cauchy iff it is Rosolini.
\end{theorem}
\begin{proof}
    Let $P$ be Rosolini. As being Cauchy is a proposition, we can untruncate the
    witness that $P$ is Rosolini to get $\alpha:\NI$ such that
    $P = \lar{\alpha}$. Define $s:\nat\to\rat$ by 
    \[s_n = \sum_{m\le n}\frac{\alpha_n}{n}.\]
    Then $s$ is regular. We have that the real $r:\real$ defined by $s$ is greater
    than $0$ iff $\qExists{n:\nat}s_n > 0$, which happens precisely if
    $\qExists{n:\nat}\alpha_n = 1$. Then we have
    \[P\lequiv \lar{\alpha}\lequiv (0 < r).\]

    Conversely, let $P$ be Cauchy. As being Rosolini is a proposition, we have $r:\real$
    such that $P = ( 0 < r)$, and hence
    \[\qExists{s:\nat\to\rat}\isReg(s)\times \left(\qPi{n:\nat}|r-s_n|<\frac{1}{n+1}\right),\]
    and again, since we are trying to prove a proposition, we may assume an explicit
    regular sequence $s$.

    Define $\beta:\cantor$ by $\beta_n \defeq (0\le s_n + 1/n)$. Then take
    $\alpha:\NI$ to be the truncation of $\beta$.  We have 
    \[\lar{\alpha}\lequiv \left(\qSig{n:\nat} (0\le s_n + \frac{1}{n}\right) \lequiv
    (0\le r) \lequiv P.\qedhere\]
\end{proof}
It is worth noting that countable choice is not used in the above argument.

Before we consider whether the Rosolini propositions form a dominance, let us examine
the relationship to the delay monad and its quotient by weak bisimilarity. Observe
that we may form a lifting relative to Rosolini structures by
\[\Lift_{\RS}(X) \defeq \qSig{P:\univ}\RosStruct(P)\times (P\to X).\]
Note that since $\RosStruct$ is not a proposition, the same element $x:X$ has
multiple representatives in $\Lift_{\RS}(X)$. Nevertheless we can form a
(non-canonical) inclusion $\eta:X\to \Lift_{\RS}(X)$ by
taking $\eta(x)$ to be $(\lar{\alpha},(\alpha,\refl),\lambda u.x)$ where
$\alpha=\overline{0} = \lambda n.n=0$. This makes $\Lift_{\RS}$ a reorganization of
the definition of the delay monad.
\begin{lemma}\label{lemma:delay-is-rosstruct}
    For any type $X$, there is an equivalence
    \[\Lift_{\RS}(X) \simeq \Delay(X).\]
    Moreover, $\eta:X\to \Lift_{\RS}(X)$ lifts over this equivalence to
    $\now:X\to\Delay(X)$.
\end{lemma}
\begin{proof}
    Simply manipulate the definition of $\Lift_{\RS}(X)$. We have
    \begin{align*}
        \Lift_{\RS}(X) &= \qSig{P:\univ}\RosStruct(P)\times (P\to X)\\
        &= \qSig{P:\univ}\qSig{\mu:\NI}(P=\lar{\mu})\times (P\to X) \\
        &\simeq \qSig{\mu:\NI}\lar{\mu}\to X \\
        &= \qSig{\mu:\NI}(\qSig{n:\nat}\mu_n = 1)\to X \\
        &\simeq \qSig{\mu:\NI}\qPi{n:\nat} X^{\mu=\overline{n}} \\
        &\simeq \Delay(X).\qedhere
    \end{align*}
\end{proof}
We have a map $q:\Lift_{\RS}(X)\to \Lift_{\isRos}(X)$ given by
\[q(P,d,\varphi) \defeq (P,|d|,\varphi).\]
\begin{theorem}\label{thm:qdelay-is-ros}
    The map $q:\Lift_{\RS}(X)\to \Lift_{\isRos}(X)$ is surjective. Moreover, when $X$
    is a set, there is
    an equivalence 
    \[\Lift_{\isRos}(X)\simeq D(X)/{\bisim}\]
    commuting with the quotient map $D(X)\to D(X)/{\bisim}$.

    \hspace{\mathindent}
    \begin{tikzcd}
        \Lift_{\RS}(X) \arrow[d,"q"] \ar[r,"\simeq"] & \Delay(X) \arrow[d,"q"] \\
        \Lift_{\isRos}(X) \arrow[r,"\simeq"] & \Delay(X)/{\bisim}
    \end{tikzcd}
\end{theorem}
\begin{proof}
    To show that $q$ is surjective, we need to see that for any
    $x:\Lift_{\isRos}(X)$ the fiber of $q$ over $x$ is inhabited. I.e., we want an
    inhabitant of
    \[\trunc{\qSig{y:\Lift_{\RS}(X)}q(y)=x}.\]
    Consider $x=(P,r,\varphi)$, so that we have $r:\trunc{\RS(X)}$. As we are trying to
    show a proposition, we may assume we have some $r':\RS(X)$. Then for
    $y=(P,r',\varphi)$, we have $q(y)=x$.

    Composing $q$ with the equivalence $e:\Delay(X)\to\Lift_{\RS}(X)$ determines a
    map 
    \[f:\Delay(X)\to \Lift_{\isRos}(X).\]
    It is clear that $f$ respects bisimilarity; if $x,y:\Delay(X)$ are bisimilar,
    then we must have that the sequences determining $\mu_x,\mu_y:\NI$ defining $x$
    and $y$ are such that 
    \[\lar{\mu_x} \lequiv \lar{\mu_y};\]
    when $x$ and $y$ do take a value, it must be the same value, and so $\val(f(x))
    = \val(f(y))$.

    Then $f$ factors through a map 
    \[{f/{\bisim}} : \Delay(X)/{\bisim}\to \Lift_{\isRos}(X),\]
    so long as $\Lift_{\isRos}(X)$ is a set; this happens whenever $X$ is a set.
    This map is again surjective, as $f$ is.

    It remains to show that the extension $f/{\bisim}$ has propositional fibers. It
    is enough to show that if $f(x)=f(y)$ then $x\bisim y$. We have that $f(x)=f(y)$
    precisely when both $\mu_x=\mu_y$ and if $\mu_x(n)=\mu_y(n)=1$, then $x$ and $y$
    are equal to $\delay^{n}(a)$ for some $a:X$. These two conditions are the
    definition of bisimilarity.
\end{proof}
\index{Rosolini proposition|)}

\section{Choice principles and the dominance axiom}\label{section:dominance-choice}
\index{Rosolini proposition|(}
\index{choice!axiom of|(}
The Rosolini propositions form a dominance in, for example, the effective
topos~\cite{Rosolini1986}, but in general we cannot show that they do. We give in
this section a characterization of the amount of choice needed to show that the
Rosolini propositions form a dominance, starting with a weakening of countable
choice.
\index{choice!countable|(}
\begin{theorem}\label{thm:cc-weaker}
    For any type $A$, the following are equivalent
    \begin{enumerate}
        \item Choice from $\N$ to families of the form $n\mapsto (\alpha_n=1) \to A$,
            with $\alpha:\cantor$.
        \item Choice from Rosolini propositions to $A$.
    \end{enumerate}
\end{theorem}
\begin{proof}
Consider the following propositions:
    \begin{enumerate}[label=(\arabic*)]
  \item $\lar{\alpha} \to \trunc{A}$,
  \item $\left(\qSig{n:\N}{\alpha_n = 1}\right) \to \trunc{A}$,
  \item $\qPi{n:\N}{\left((\alpha_n = 1) \to \trunc{A} \right)}$,
  \item $\qPi{n:\N}{\trunc{(\alpha_n = 1) \to A}}$,
  \item $\trunc{\qPi{n:\N}{(\alpha_n = 1) \to A}}$,
  \item $\trunc{\left(\qSig{n:\N}{\alpha_n = 1}\right) \to A}$,
  \item $\trunc{\lar{\alpha} \to A}$.
\end{enumerate}
The implication $(4) \to (5)$ is the above instance of countable
choice, and the implication $(1) \to (7)$ is the above instance of
propositional choice. Note that $(3)\to (4)$ holds because $\alpha_n=1$ is
decidable. Hence the chain of implications
$(1) \to (2) \to (3) \to (4)\to (5)\to (6)\to (7)$ gives propositional
choice from countable choice, and the chain of implications
$(4) \to (3) \to (2) \to (1) \to (7) \to (6) \to (5)$ gives countable
choice from propositional choice.
\end{proof}
The above form of countable choice, where $A$ takes the form $\RosStruct(P)$ for some
$P$ is sufficient to prove that the Rosolini propositions form a dominance. In fact,
we can do no better.
\begin{theorem}\label{rosolini:choice-to-dominance}
    The following are equivalent
    \begin{enumerate}
        \item Choice from Rosolini propositions to Rosolini structures. I.e.,
            \[\isRos(P)\to(P\to \trunc{\RosStruct{Q}})\to \trunc{P\to\RosStruct(Q)}\]
        \item The Rosolini propositions form a dominance.
    \end{enumerate}
\end{theorem}
We omit the proof as we generalize this theorem to
Theorem~\ref{dominance:choice-to-dominance} below.
\begin{corollary}
    Choice from $\nat$ to families of the form $n\mapsto (\alpha_n=1)\to A$ implies
    that the Rosolini propositions form a dominance. Hence, countable choice implies
    that the Rosolini propositions form a dominance.
\end{corollary}
\begin{proof}
    We need to see that choice from $\nat$ to families of the form $n\mapsto
    (\alpha_n=1)\to A$ gives us Rosolini choice.
    So let $P$ and $Q$ be propositions. We wish to see
    \[\trunc{\RosStruct(P)}\to(P\to \trunc{\RosStruct(Q)}) \to \trunc{P\to\RosStruct(Q)}.\]
    In fact, we may untruncate $\RosStruct(P)$, so let $\alpha:\NI$ such that
    $P=\lar{\alpha}$. Then we need
    \[(\lar{\alpha}\to \trunc{\RosStruct(Q)}) \to \trunc{\lar{\alpha}\to\RosStruct(Q)}.\]
    Note that for all $n$, we have $(\alpha_n=1)\to \lar{\alpha}$. So then setting 
    \[A \defeq \RosStruct(Q),\]
    we have $(\alpha_n=1)\to \trunc{A}$. Since $\alpha_n=1$ is
    decidable, we have choice from $\alpha_n=1$ (Theorem~\ref{thm:decidable-choice}).
    Hence, we have
    \[\qPi{n:\nat}\trunc{(\alpha_n=1)\to A}.\]
    This is exactly the assumption of our choice principle, so we conclude
    \[\trunc{\qPi{n:\nat}(\alpha_n=1)\to A}.\]
    As $(\qPi{n:\nat}(\alpha_n=1)\to A)\simeq (\lar{\alpha}\to A)$, we have
    \[\trunc{\lar{\alpha}\to A}.\]
    And this concludes the proof, since
    \[\trunc{\lar{\alpha}\to A}\simeq \trunc{P\to \RosStruct(Q)}.\qedhere\]
\end{proof}

The reliance on choice we saw in Theorem~\ref{rosolini:choice-to-dominance} is an
instance of a more general phenomenon.

\begin{theorem}\label{dominance:choice-to-dominance}
    Let $\DD:\univ\to\univ$ select propositions (\ref{dominance:is-prop}) and have
    conditional conjunction (\ref{dominance:axiom}) and define $\dd(X)\defeq
    \trunc{\DD(X)}$. Then the following are equivalent
    \begin{enumerate}
        \item Choice from $\dd$-propositions to $\DD$-structures; i.e., for all
            $X,Y:\univ$,
            \[\DD(X)\to(X\to\trunc{\DD(Y)})\to\trunc{X\to \DD(Y)}\]
        \item The $\dd$-propositions satisfy the dominance axiom.
    \end{enumerate}
\end{theorem}
\begin{proof}
    We need to show that under the assumptions of $\trunc{\DD(X)}$ and
    $X\to\trunc{\DD(Y)}$ that we have
    \[\trunc{X\to \DD(Y)}\simeq \trunc{\DD(X\times Y)}.\]
    By properties of truncation, it is enough to assume $\DD(X)$ and
    $X\to\trunc{\DD(Y)}$ and give maps
    \[(X\to \DD(Y))\leftrightarrow \DD(X\times Y).\]
    From left to right is simply the fact that $\DD$ has conditional conjunction.
    From right to left, if $(X\to \DD(Y))$ then $X\to\isProp(Y)$. Since ${X\to((X\times Y) = Y)}$ when $X$ and $Y$ are
    propositions, we have $(X\to \DD(Y)) \to (X\to\DD(X\times Y))$. By modus
    ponens, we are done.
\end{proof}

\nameas{Rosolini choice}{choice!Rosolini} (choice from Rosolini propositions to Rosolini structures) is
actually quite weak. We saw above that it follows from even a weakening of countable
choice, but it is quite a bit weaker than countable choice.
\begin{theorem}\label{thm:rosolini-choice}
    Each of the following principles alone implies Rosolini choice:
    \begin{enumerate}
        \item countable choice;
        \item the weakening of countable choice in Theorem~\ref{thm:cc-weaker};
        \item ``untruncated'' LPO: $\qPi{\alpha:\cantor}\lar{\alpha}+(\alpha=\lambda n.0)$;
        \item ``untruncated'' WLPO: $\qPi{\alpha:\cantor}\lar{\alpha} +
            \neg\lar{\alpha}$;
        \item ``truncated'' WLPO: $\qPi{\alpha:\cantor}\trunc{\lar{\alpha} +
            \neg\lar{\alpha}}$;
        \item Propositional choice: for all $P,A:\univ$ with $\isSet(A)$ and
            $\isProp(P)$,
            \[(P\to\trunc{A})\to \trunc{P\to A}.\]
    \end{enumerate}
\end{theorem}
\index{Limited Principle of Omniscience}
\index{Limited Principle of Omniscience!Weak}
\begin{proof}
    \begin{enumerate}
        \item This follows from 2;
        \item Rosolini choice is the specialization of this principle to types
            of the form $\RosStruct(Q)$ for some $Q:\univ$;
        \item LPO implies WLPO, so this follows from 4;
        \item This follows from 5, since the truncated form is weaker;
        \item Let $\isRos{P}$ and let $Q$ be a proposition such that $P\to\isRos(Q)$.
            Then \[\qExists{\alpha:\cantor}P\simeq \lar{\alpha}.\]
            By WLPO, we then have $\trunc{P+\neg P}$. However, $P$ is a
            proposition, so then so is $P+\neg P$, and we have $P+\neg P\simeq
            \trunc{P+\neg P}$. Then we may perform a case analysis on $P$: if $P$ is
            true, then so is $\trunc{\RosStruct(Q)}$, and so we have
            $\trunc{P\to\RosStruct(Q)}$. If $P$ is false, then we vacuously have
            $\trunc{P\to\RosStruct(Q)}$.
        \item If $\isRos(P)$, then $P$ is a proposition, so by propositional choice
            we have \[(P\to\isRos(Q))\to \trunc{P\to\RosStruct(Q)}.\qedhere\]
    \end{enumerate}
\end{proof}

In light of the above, it is worth wondering how far we can go towards specifying a
class of partial maps via a class of propositions---indeed we cannot prove that any
dominance is different from both $\isProp$ and $\isContr$ without violating
classicality. 

So, if we were to work with dominances, choice would be unavoidable. The first
attempt at a solution is to allow $\DD(X)$ to be structure, rather than property.
This commits us to handling structure explicitly; in particular, we would need to
use definitions and constructions which respect structure. This seems unduly
cumbersome, but we develop this line of reasoning somewhat in the next section. We
will subsequently make use of this development in finding an approach to
partial functions which works for our purposes.

Note that Rosolini choice is not enough to give a version of Theorem~\ref{thm:dcpo}
for Rosolini propositions. However, full countable choice is enough to show that the
Rosolini propositions are closed under \emph{countable} joins.
\begin{theorem}
    If countable choice holds, then the Rosolini propositions are closed under
    countable joins in the lattice of all propositions. Hence, countable choice
    implies that the Rosolini lifting is an $\omega$-CPO.
\end{theorem}
\begin{proof}
    Let $P_i:\nat\to\Prop$ be a countable chain of Rosolini propositions. We need to
    see that $\trunc{\qSig{i:\nat}P_i}$ is a Rosolini proposition. We have by
    assumption,
    \[\qPi{i:\nat}\trunc{\qSig{\alpha_i:\NI}P_i=\lar{\alpha_i}}.\]
    By countable choice we then have
    \[\trunc{\qPi{i:\nat}\qSig{\alpha_i:\NI}P_i=\lar{\alpha_i}}.\]
    Since being Rosolini is a proposition, to prove that $\trunc{\qSig{i:\nat}P_i}$
    is Rosolini, it is enough to prove this from
    \[\qPi{i:\nat}\qSig{\alpha_i:\NI}P_i=\lar{\alpha_i}.\]
    Define $\beta:\NI$ to be the sequence such that $\beta(n) = 1$ iff there is a
    lexicographically minimal $(i,j):\nat\times\nat$ with $i+j= n$ such that
    $\alpha_i(j)=1$. It is clear that $\isProp\lar{\beta}$.

    Now observe that $\trunc{\qSig{i:\nat}P_i}$ is equivalent to
    $\trunc{\qSig{i,j:\nat}\alpha_i(j) = 1}$. By construction we have that
    \[\lar{\beta}\simeq\trunc{\qSig{i,j:\nat}\alpha_i(j) = 1}.\]
    Then $\beta$ provides a witness that $\trunc{\qSig{i:\nat}P_i}$ is Rosolini.
\end{proof}
\index{choice!axiom of|)}
\index{Rosolini proposition|)}
\index{choice!countable|)}

\section{Structural Dominances}\label{section:univ-dom}
\index{dominance!structural|(}
\index{structural dominance|see {dominance, structural}}
In order to define the notion of disciplined map below in
Section~\ref{section:disciplined-maps}, we need to relax the notion of dominance to
allow $\DD(X)$ to be structured. This structural version of a dominance is as
follows.
\begin{definition}\label{def:univ-dom}
    A \nameas{structural dominance}{dominance!structural} is a map $\dd:\univ\to\univ$ together with witnesses
    of Conditions~\ref{dominance:is-prop},~\ref{dominance:unit}, and~\ref{dominance:closure}.
    That is, we have
    \begin{enumerate}[label=\normalfont\textrm{D\arabic*.}]
            \setcounter{enumi}{1}
        \item a map $\qPi{X:\univ}{\dd(X)\to\isProp(X)}$,
        \item The unit type is dominant: we have $\upsilon:\dd(\unittype)$,
            \setcounter{enumi}{4}
        \item[\normalfont\textrm{D\arabic{enumi}.'}] we have map
            \[\sigma:\qPi{P:\univ}{\qEmpty{Q:P\to\univ}{\dd(P)\to(\qPi{p:P}{\dd(Q(p))})\to\dd(\qSig{p:P}{Q(p)})}}.\]
    \end{enumerate}
\end{definition}
Note that~\ref{dominance:unit}~and~\ref{dominance:closure}~are structure rather than
property in the absence of~\ref{dominance:property}.

Lemma~\ref{dominance:axiom-is-closure} gives an equivalence between
\ref{dominance:axiom} and \ref{dominance:closure} under the assumption that $\dd$
satisfies both~\ref{dominance:property} and~\ref{dominance:is-prop}. It is enough to
show only logical equivalence because both
~\ref{dominance:axiom}~and~\ref{dominance:closure} are propositions in this case.
When we only have~\ref{dominance:is-prop}, this is not enough. It seems that in this
case, we can only show that \ref{dominance:closure} is a retract of
\ref{dominance:axiom}.

Before giving the argument in detail, a sketch is useful: let $\dd$
select propositions, and fix $P:\univ$ with $\dd(P)$. The type of closure under
conditional conjunction, when applied to a type $Q$ is the same
as the type of $\Sigma$-closure applied to $\lambda p.Q$, since we have defined
$P\times Q$ as $\qSig{p:P}Q$ and $P\to \dd(Q)$ as $\qPi{p:P}\dd(Q)$. Conversely, since $P$ is a
proposition, we can show that any type family $Y:P\to\univ$ is equal to the constant
type family $\lambda q.\qSig{p:P}Y(p)$. Then given a witness of $\Sigma$-closure, we
get a witness of closure under conditional conjunction, by treating a type as a
constant type family; given a witness of closure under conditional conjunction, we
get a witness of closure under $\Sigma$ by passing back and forth along this
equality.

We will give the required homotopy by applying the following lemma.
\begin{lemma}\label{dom:sig:hom}
    Let $X$ be a type and $A,B:X\to\univ$ with $f:\qPi{x:X}A(x)\to B(x)$. Let
    $p:x=x'$ for $x,x':X$. Then
    \[\idtofun(\ap_B(p^{-1})) \comp f_{x} \comp \idtofun(\ap_{A}(p)) = f_{x'}.\]
\end{lemma}
\begin{proof}
    By path induction on $p$. In the case that $p=\refl$, both $\ap_A(p)$ and
    $\ap_B(p)$ are reflexivity, and so we need to see $f_{x} = f_{x}$, which is
    obvious.
\end{proof}
For fixed $P:\univ$ and $d:\dd(P)$, define the type families $A,B:(P\to\univ)\to\univ$ by
\[A(Y) = \qPi{p:P}\dd(Yp)\]
and
\[B(Y) = \dd(\qSig{p:P}Yp).\]
Note that for $Q:\univ$, we have that $P\to\dd(Q) = A(\lambda p.Q)$ and $\dd(P\times
Q) = B(\lambda p.Q)$. So we can define 
\begin{align*}
    E &: \big(\qPi{Y:P\to\univ} A(Y)\to B(Y)\big)\to
    \qPi{Q:\univ}A(\lambda p.Q)\to B(\lambda p.Q)\\
    E(\theta) &\defeq \lambda Q.\theta(\lambda p.Q)
\end{align*}
Defining the candidate retraction requires some more work.
\begin{lemma}\label{homotopy-props}
    For $X:\univ$ and $Y:X\to\univ$, let $Z\defeq \qSig{x:X}{Y(x)}$. If $\isProp(X)$, then 
    \[\qPi{x:X}{(Y(x)\simeq Z)},\]
    and
    \[Z \simeq(X\times Z).\]
\end{lemma}
\begin{proof}
    For the first equivalence, fix $x:X$, we have $Y(x) \to Z$ by $y\mapsto (x,y)$,
    with inverse given by transport.

    Moreover, we have $Z\to X\times Z$ by $(x,y)\mapsto (x,x,y)$ with inverse again
    given by transport.
\end{proof}
In the case where $Y$ is also proposition valued (that is, $\qPi{x:X}\isProp(Yx)$),
these equivalences are equalities, as all types involved are propositions. Moreover,
the above gives us an equality of predicates, by function extensionality.
\begin{lemma}\label{d-equality}
    If $X$ is a proposition and $Y:X\to\univ$ is a predicate, then 
    \[Y=\lambda x.\qSig{x:X}{Y(x)}.\]
\end{lemma}
This lemma allows us to define, for our fixed $P:\univ$.
\begin{align*}
    F &: \big(\qPi{Q:\univ}A(\lambda p.Q)\to B(\lambda p.Q)\big)
    \to \qPi{Y:P\to\univ} A(Y)\to B(Y) \\
    F(\delta) &\defeq \lambda Y.\idtofun(\ap_B(w^{-1}_{Y})) \comp \delta_{\qSig{p:P}Yp}
    \comp \idtofun(\ap_A(w_{Y}))
\end{align*}
where $w_Y:Y\to \lambda p.\qSig{p:P}Y(p)$ arises from Lemma~\ref{d-equality}.
Finally, we can compute the composite $F\comp E$:
\begin{lemma}
    The map $E$ is a section of the map $F$.
\end{lemma}
\begin{proof}
    Fix $\theta:\qPi{Y:P\to \univ}A(Y)\to B(Y)$. Then
    \[F(E(\theta)) = \lambda Y.\idtofun(\ap_B(w^{-1}_{Y})) \comp \theta_{\lambda p.\qSig{p:P}Yp}
    \comp \idtofun(\ap_A(w_{Y})) \]
    Note that $\theta$ has the type of $f$ in Lemma~\ref{dom:sig:hom}, with
    $X=P\to\univ$, and so we have for each $Y$,
    \[F(E(\theta))(Y) = \theta(Y),\]
    and so $F\comp E$ is homotopic to the identity.
\end{proof}

In the above, we fixed $P:\univ$ and $d:\dd(P)$. Letting $P$ and $d$ vary again, we
see that closure under $\Sigma$ is a retract of closure under conditional conjunction
That is, we have
\begin{corollary}\label{dominance:universe-axiom-is-closure}
    If $\dd:\univ\to\univ$ selects propositions, then
    the type that $\dd$ is closed under $\Sigma$ is a retract of the type
    \[\qPi{P,Q:\univ}{\dd(P)\to(P\to \dd(Q))\to \dd(P\times Q)}.\]
\end{corollary}

The main example of a structural dominance we are interested in is the
dominance of Rosolini structures.
\index{Rosolini structure}
\begin{lemma}\label{lemma:rs-dom}
    The map $\RosStruct$ is a structural dominance.
\end{lemma}
\begin{proof}
    \ref{dominance:is-prop}~is immediate. For $\upsilon$ we take
    $\overline{0}$. To get $\sigma$, by
    Lemma~\ref{dominance:universe-axiom-is-closure} and the definition of
    $\RosStruct$, we first define a conditional addition
    $-\cadd-:\qPi{\alpha:\NI}{(\lar{\alpha}\to\NI)\to\NI}$, such that
    \[\lar{\alpha\cadd\beta} \Leftrightarrow \qExists{k:\nat}{(\alpha_k = 1)\times
    \qExists{j:\nat}{\beta(k)_j = 1}}.\]
    In fact, this specification is almost a complete definition. Define,
    \[(\alpha\cadd\beta)_n = 1 \Leftrightarrow \qExists{k\le n}{(\alpha_k = 1)\times
    \qExists{j\le n}{\beta(k)_j = 1}}.\]
    Since this is a bounded quantification of decidable predicates, this is again a
    decidable property on $\nat$. Moreover, since there can be at most one such $k$
    and $j$, we have $\isProp(\lar{\alpha\cadd\beta})$.

    Now, we define $\sigma$ as follows. Fix $P:\univ$ and $Q:P\to\univ$ with
    $\alpha:\NI$ with $w:\lar{\alpha}=P$ and for each $p:P$ a $\beta_p:\NI$ such
    that $Q(p)=\lar{\beta_p}$. We then have a map
    $\beta':\lar{\alpha}\to\NI$ with $\beta'(a) = \beta_p$ where $a$ transports over
    $w$ to $p$. Then a Rosolini structure for $\qSig{p:P}Q(p)$ is given by
    $\alpha\cadd\beta'$ together with the fact that $\alpha\cadd\beta'$ takes value 1
    precisely when there is $a:\lar{\alpha}$, and $\beta'(a)$ takes value 1.
\end{proof}
\index{dominance!structural|)}

\section{Disciplined maps}\label{section:disciplined-maps}
The takeaway from Section~\ref{section:dominance-choice} above is that the Rosolini
partial functions are not useful in the absence of choice principles: we need choice
a principle to compose Rosolini partial functions. Another approach is necessary. We
could instead work with Rosolini structures---with the delay monad---but doing this
means carrying around extra information, in particular, distinguishing between
partial elements based not on their eventual value, but on how long it takes to
compute this value. This is unsatisfactory; instead, we need a predicate on
partial functions which is closed under composition without choice.

The predicate on partial functions given by a dominance arises by restricting the
available extents of definition. Instead, we will give a predicate on partial
functions directly: Any proposition-index family $\DD:\Prop\to\univ$ gives rise to a
function from the lifting relative to that family $\Lift_{\DD}(Y)$ to the general
lifting $\Lift(Y)$, and so we get a map from $\DD$-partial functions
$X\to\Lift_{\DD}(Y)$ to general partial functions $X\to\Lift(Y)$. If $\DD$ is a
set of propositions, this map is an embedding, and so $X\to\Lift_{\DD}(Y)$ is
equivalent to its image, but this is not true if $\DD$ is not proposition-valued.
Then, instead of looking at $X\to\Lift_{\DD}(Y)$, we will look at its image in
$X\to\Lift(Y)$.

We can think of the space $X\to\Lift_{\DD}(Y)$ as being too wild for our purposes, so
we first \emph{tame} the space; the tamed maps we will call \emph{disciplined}. The
rest of the section lays this out precisely.

Given any family $\DD:\univ\to\univ$ which selects propositions and
$X:\univ$, we have a map
\begin{eqnarray*}
    e_{\DD}&:&\Lift_{\DD}(Y)\to \Lift(Y)\\
    e_{\DD}(P,d,\varphi) &=& (P,-,\varphi),
\end{eqnarray*}
where the omitted ``$-$'' follows from Condition~\ref{dominance:is-prop}. By
post-composition, this gives us a map 
\defineopfor{$\tame$}{tame}{disciplined map}
\[\tame:(X\to\Lift_{\DD}(Y))\to(X\to\Lift(Y)).\]
This map always factors through $(X\to\Lift_{\dd}(Y))$,
where $\dd(Y)=\trunc{\DD(Y)}$, and if $\DD$ satisfies
Condition~\ref{dominance:property}, then this map is an embedding. We will thus call
$e_D$ the \emph{canonical embedding}, even when $\DD$ carries structure.

We call a function in the image of $\tame_{\DD}$
\nameas{$\DD$-disciplined}{disciplined map}. That is,
\begin{definition}\label{def:discipline}
    A function $f:X\to\Lift(Y)$ is \emph{$\DD$-disciplined} if it is in the image of
    $\tame_{\DD}$.
    Formally,
\defineopfor{$\isDis$}{isdis}{disciplined map}
\defineopfor{$\Dis$}{dis}{disciplined map}
\begin{eqnarray*}
    \isDis_{\DD}&:&(X\to\Lift(Y)) \to \univ \\
    \isDis_{\DD}(f) &\defeq& \qExists{f':X\to\Lift_{\DD}(Y)}{\tame(f') = f}\\
    \Dis_{\DD}(X,Y)&:&\univ\\
    \Dis_{\DD}(X,Y) &\defeq& \qSig{f:X\to\Lift(Y)}{\isDis_{\DD}(f)}
\end{eqnarray*}
\end{definition}
\index{dominance!structural|(}
\begin{lemma}\label{partiality:closure}
    If $\DD$ is a structural dominance, for any $f:X\to\Lift Y$ and $g:Y\to \Lift Z$,
    we have
    \[\isDis_{\DD}(f)\to\isDis_{\DD}(g)\to\isDis_{\DD}(g\kcomp f).\]
\end{lemma}
\begin{proof}
    As $\isDis(g\kcomp f)$ is a proposition, we may assume
    $f':X\to\Lift_{\DD}(Y)$ and $g':Y\to\Lift_{\DD}(Z)$ such that $\tame(f')=f$ and
    $\tame(g')=g$. We need $h':X\to\Lift_{\DD}(Z)$ with $\tame(h')=g\kcomp f$. We
    claim that $h' \defeq g'\kcomp f'$ suffices. Note that tame is post-composition
    with $e$, that $e$ acts independently on the components of $z:\Lift_{\DD}(Z)$, and
    that $e$ is the identity on extent and value, so to see that $e
    (\klext{g'}(f' x)) = \klext{g}( f x)$, we need only the component witnessing that
    $\extent(x)$ satisfies~$\DD$. However, since this component takes values in a
    proposition, there is nothing to check.
\end{proof}
\begin{theorem}\label{disciplined-maps:compose}
    If $\DD$ is a structural dominance, then there is a composition operator,
    \[-\pcomp-:\Dis_{\DD}(Y,Z)\to\Dis_{\DD}(X,Y)\to\Dis_{\DD}(X,Z).\]
\end{theorem}
\begin{proof}
    Immediate from Lemma~\ref{partiality:closure}.
\end{proof}
Observe that since the truncation $\dd$ of any structural dominance is
proposition-valued, the disciplined maps are a subset of the $\dd$-partial functions.
Also notice that if $\DD$ is in fact a dominance, then for any $f:X\to\Lift Y$ the type
\[\qSig{g:X\to\Lift_{\DD}(Y)}\tame(f)=g\]
is a proposition, and so we can remove the truncation. That is, for a dominance $\DD$, we
have that the $\DD$-disciplined maps are exactly the $\DD$-partial functions. The
situation is more interesting when $\DD$ is not proposition-valued. In this case, a
lifting
\[\Dis_{\DD}(X,Y)\to (X\to\Lift_{\DD}(Y))\]
corresponds to untruncating the existential in the definition of $\isDis$. To show
this rigorously, we use the following lemma.
\begin{lemma}
    For any $f:X\to\Lift(Y)$ and any structural dominance $\DD$, the function
    \[F_f:\big(\qPi{x:X}\DD(\extent (f
    x))\big)\to\qSig{g:X\to\Lift_{\DD}(Y)}f=\tame(g)\]
    defined by
    \[F_f(d) = \big(\lambda x.(f_e(x),d(x),f_v(x))\;\, w\big)\]
    where $w$ is the equality arising from the tuple $(\refl,-,\refl)$, is an
    equivalence.
\end{lemma}
\begin{proof} 
    Letting 
    \[A(f) \defeq\qPi{x:X}\DD(\extent(f(x))),\]
    and
    \[B(f) \defeq\qSig{g:X\to\Lift_{\DD}(Y)} f=\tame(g),\]
    we have equivalences
    \begin{align*}
        \qSig{f:X\to\Lift(Y)}A(f)
        &\simeq X\to \Lift_{\DD}(Y) \\
        & \simeq \qSig{f:X\to\Lift(Y)}B(f).
    \end{align*}
    The composite equivalence $E$ is the identity on the
    first component, and $F_f$ on the second. That is, $E$ arises as $F_\Sigma$ from
    Section~\ref{section:equiv-contr}, and so by Theorem~\ref{thm:fiberwise-equivalence},
    each $F_f$ is an equivalence.
\end{proof}
\index{choice!countable|(}
\begin{theorem}\label{thm:disc-factoring}
    Fix a structural dominance $\DD$ and its associated set of propositions $\dd$.
    If countable choice holds, then for any $f:\nat\to\Lift(\nat)$, we have that $f$ is
    disciplined iff for all $n:N$, the extent of $f(n)$ is a $\dd$-proposition.
    Then, in particular, countable choice implies that $f:\nat\to\Lift(\nat)$ is
    Rosolini-disciplined iff $f$ factors through a Rosolini partial function.
\end{theorem}
\begin{proof}
    By the above lemma, we have an equivalence
    \[\isDis(f)\simeq \trunc{\qPi{n:\nat}\DD(\extent(f n))}.\]
    We know there is an implication
    \[\trunc{\qPi{n:\nat}\DD(\extent(f n))}\to\qPi{n:\nat}\dd(\extent(f n)),\]
    and the implication in the other direction is an instance of countable choice.

    Finally, every disciplined map factors through a $\dd$-partial function.
    Moreover, we have that $f:\nat\to\Lift(\nat)$ factors through a $\dd$-partial
    function iff $\qPi{n:\nat}\dd(\extent(f n))$, and by the above argument, this
    implies that $f$ is disciplined.
\end{proof}
Note that we cannot use the dominance choice principle for
Theorem~\ref{thm:disc-factoring}, since
there does not seem to be a way to write the type $\qPi{n:\nat}\trunc{\DD(\extent(f
n))}$ in the form $\qPi{a:A}\trunc{\DD(Pa)}$, where $A$ is dominant.

Note, moreover, that the above is a theorem about disciplined maps from $\nat$. Since
we will look at first-order computability theory, we are interested in partial
functions from $\N$ to $\N$. In particular, the Rosolini-disciplined maps form the
most likely candidate for a notion of partial function that can be consistently
posited to be the computable functions.
\index{choice!countable|)}
\index{dominance!structural|)}

\section{Comparison with synthetic domain theory}\label{section:sdt-comparison}
\index{synthetic domain theory|(}
In synthetic domain theory (for example, in 
\cite{Rosolini1986,Hyland1991,ReusStreicher1997,taylor1991fixed,VANOOSTEN2000}), a
dominance is defined
to be a map $\dd:\Omega\to\Omega$ with the object of \emph{all}
partial elements of $Y$ is given by
\[\{A\in \mathcal{P}Y \mid \forall x,y.A(x) \to A(y) \to x=y\},\]
and the object of $\dd$-partial elements of $Y$ by 
\[
    \{A\in \mathcal{P}Y \mid (\forall x,y.A(x) \to A(y) \to x=y)
    \wedge \dd(\exists x.A(x))\}
\]
which translates into type theory as
\[\qSig{A:Y\to \Prop}{(\qPi{x,y:Y}{A(x)\to A(y)\to x=y})} \times \dd(\qSig{x:Y}{A(x)}).\]
 
However, this does not work in a univalent setting, unless $Y$ is a set. In
particular, recall that all maps from the homotopy circle $\Sone$ into $\Prop$ are
constant (Lemma~\ref{lemma:sone-constant}). In particular, we do not
have the unit map $\Sone\to \Lift \Sone$; indeed,
\[\eta(\base) = \lambda x. (x=\base),\]
is not valued in propositions. So we must change $A$ to have type
$Y\to\univ$, but then we run into the issue that $A(x)\to A(y)\to x=y$ may have
non-trivial structure, because the type $x=y$ is not a proposition. Again, we have
that $(x=\base)\to(y=\base)\to x=y$ has numerous witnesses. In short, we must change
the total lifting in two ways:
\[\underbrace{\Sigma(A:Y\to \Prop)}_{\Sigma(A:X\to\univ)},
\underbrace{(\qPi{x,y:Y}{A(x)\to A(y)\to
x=y})}_{\isProp(\qSig{x:Y}{A(x)})}.\]
This gives the lifting of $Y$ as single-valued relations with $\unittype$. As we saw
in Theorem~\ref{thm:ptl-fxn-char} this is equivalent to the
lifting we give in Definition~\ref{def:lifting}, assuming univalence. Even without
univalence, our lifting is a retract of this lifting. The application of univalence
here occurs because we cannot establish equality between maps $Y\to\univ$ without an
extensionality principle for types. As a result, our lifting
$\qSig{p:\univ}\isProp(P)\times (P\to X)$
is more useful in the absence of univalence: we need only proposition
and function extensionality to establish equality between elements of~$\Lift(Y)$.

Typical axiomatizations of synthetic domain theory include a subset of propositions,
written $\Sigma$, which are our Rosolini propositions. In the settings of interest to synthetic
domain theorists, the dominance axiom typically holds---usually as a consequence of
countable choice. The axiomatization given by Reus~\cite{Reus1999} is of note for
not requiring $\Sigma$ to be a dominance. There, they are not interested in $\Sigma$
to define liftings, but instead are interested in developing domain theoretic ideas
abstractly. They give the following axiomatization of $\Sigma\subseteq\Prop$:
\begin{description}
    \item[(aSig)] True and false are in $\Sigma$;
    \item[(bSig)] $\Sigma$ is closed under $\vee$ and $\wedge$;
    \item[(cSig)] $\Sigma$ is closed under existential quantification over $\nat$;
    \item[(dSig)] $\Sigma$ embeds into $\Prop$;
    \item[(PHOA)] $\Sigma^\Sigma$ is equivalent to the graph of the standard ordering
        relation on $\Sigma\times\Sigma$;
    \item[(S)] For any $P:(\nat\to\Sigma)\to\Sigma$ such that $P(\lambda n.\top)$,
        there exists an $n:\nat$ such that $P(\overline{n})$, where $\overline{n} =
        \lambda k.k< n$;
    \item[(MP)] propositions in $\Sigma$ have double-negation elimination.
\end{description}
Axioms (aSig), (bSig), and (dSig) are true for the Rosolini propositions, while
(cSig) and (S) require some amount of choice. In fact, (cSig) implies Rosolini choice
in our context, so implies that the Rosolini propositions form a dominance. Note,
however, that our proof relies on the fact that the Rosolini propositions arise via
truncation from a structural dominance, in contrast to their direct axiomatization.
Both (PHOA) and (MP) are not available in our setting: (PHOA) implies that there are
propositions which are
not Rosolini propositions (the negation of Kripke's schema, and hence negating
excluded middle), and (MP) already implies we are working in a computation setting.
\index{synthetic domain theory|)}

\section{Discussion}

The Rosolini dominance plays a fundamental role in Rosolini's PhD
Thesis~\cite{Rosolini1986}, which is concerned with investigating notions of
effectiveness. It seems that the notion of effectiveness leads one inexorably to consider
partiality, since the notion of dominance was introduced by him here specifically to
deal with partiality in toposes, where (like in type theory), the notion of (total)
function is primitive. Historically, the set of Rosolini propositions has been
denoted $\Sigma$, clashing with type-theoretic notation.
The Rosolini propositions correspond to the Sierpinski space in formal topology and
to the $\Sigma_1$ propositions from the arithmetic hierarchy.
Several traditional taboos can be stated in terms of Rosolini propositions:
\index{Markov's Principle}
\index{Kripke's Schema}
\index{Limited Principle of Omniscience}
\begin{itemize}
    \item Kripke's Schema says that all propositions are Rosolini.
    \item Markov's Principle is double-negation elimination for
        Rosolini propositions.
    \item LPO says that true and false are the only Rosolini propositions.
    \item WLPO says that Rosolini propositions are either false or not false.
\end{itemize}
This view makes it clear
that Kripke's Schema and Markov's Principle together imply excluded middle.

Until the last year or so, the relationship between partiality and dominances has
only been studied in the context of synthetic domain theory and synthetic
computability theory, where countable choice holds, and so the problem with
composition is not apparent. The form of Choice in Theorem~\ref{thm:cc-weaker} was isolated by
Mart\'in Escard\'o when examining what was necessary to compose partial functions.
Theorem~\ref{dominance:choice-to-dominance} was generalized from
Theorem~\ref{rosolini:choice-to-dominance}, which we showed already
in our paper~\cite{partialelems2017}. In that paper, Theorem~\ref{thm:qdelay-is-ros}
was left as a conjecture, and Bas Spitters quickly proved it independently.
The fact that the Rosolini-structure lifting is equivalent to the delay monad means
that it does not raise universe levels. The same is true for the Rosolini lifting.
Besides the proof via Theorem~\ref{thm:qdelay-is-ros}, there is a result due to
Egbert Rijke~\cite{rijke:join} giving conditions under which an operation does not
increase universe levels.

The notion of disciplined map was proposed, but not worked out in the conclusion
to~\cite{partialelems2017}, inspired partially by comments from Mike Shulman.
Chapman, Uustalu, and Veltri~\cite{Chapman2015Delay} briefly proposed dealing with the
problem of composition by viewing the quotiented delay monad as arrow (in the sense
of~\cite{hughes2000generalising}); this is morally the same as the approach via
disciplined maps. To my knowledge, they have not pursued this line, although Uustalu
and Veltri have continued to examine partiality and
monads~\cite{uustalu2017partiality}.


%% file: chapters/setoids.tex
\chapter{Partiality in Bishop Mathematics}\label{chapter:setoids}
Our approach to partiality in Chapter~\ref{chapter:partial-functions} relies
crucially on univalent definitions. However, it is worth examining the extent
to which we can do the work contained there in a Bishop-style approach using setoids.
Since the aim of this thesis is not to work with setoids, this chapter will be short
and less detailed. Moreover, we will not use univalent definitions in this chapter. In
particular, a proposition is simply any type, a relation is any type family, we do
not use any $h$-levels, and we use the Curry-Howard interpretation of logic from
Section~\ref{section:bishop-mathematics}.

We cannot approach general partial functions from the setoid perspective with our
tools. First we require quantification over the universe, but there is no reasonable
way to turn the universe into a setoid---no universe setoid was given by Martin Hofmann in his
thesis on setoids~\cite{Hofmann1995}, and it is not clear that is
possible~\cite{Coquand2018VV}.
Second, we rely on the definition of $\Prop$, which we are avoiding in the setoid
approach. Instead, we will restrict our attention to functions arising from the
\emph{Rosolini lifting}, which can be treated as a setoid. Since the definition of
Rosolini propositions we gave relies also on $\Prop$, we will instead approach the
Rosolini dominance in this chapter via the delay monad and its quotient by
bisimilarity. Recall Lemma~\ref{lemma:delay-is-rosstruct} and
Theorem~\ref{thm:qdelay-is-ros}, which show that this presentations is sound.
The setoid approach does resolve one difficulty from the previous chapter: countable
choice holds in the setoid model. This result seems to be folklore, so we prove it
in Section~\ref{section:setoid-choice}.

\section{Setoids}
Recall that a \emph{setoid} is a type together with an equivalence relation. In this
context, an \emph{equivalence relation} on $X$ is a type family $R:X\to X\to\univ$
together with witnesses,
\begin{align*}
    r&:\qPi{x:X}R(x,x)\\
    (-)^{-1}&:\qPi{x,y:X}R(x,y)\to R(y,x)\\
    (-)\cdot(-)&:\qPi{x,y,z:X}R(y,z)\to R(x,y)\to R(x,z).
\end{align*}
We will abuse notation and refer to a setoid by its underlying type.
We will also use $\seq$ (or, when we need to be more
precise, $\seq_X$) for the relation on a setoid $X$, while leaving the
data for $\seq$ implicit.

The equivalence relation is meant to capture equality on the underlying type. This
allows us to impose extensionality principles as we wish, at the cost of
bookkeeping. For example, if we wish to have function extensionality, we equip each
function type $X\to Y$ with the equivalence relation
\[f\seq g\defeq \qPi{x:X}f(x)\seq g(x).\]
The bookkeeping required is substantial: we must ensure that any construction on
setoids respects equivalence.  In particular, if $X$ and $Y$ are (the underlying
types of) setoids, we will call an
element of the type $X\to Y$ an \emph{operation}, and we will say that an operation
$f$ is \emph{extensional} or is \emph{a function} when 
\[\qPi{x,y:X}(x\seq y)\to f(x)\seq f(y).\]
Then the type of functions from $X$ to $Y$ is given by
\[\qSig{f:X\to Y}\qPi{x,y:X}(x\seq y)\to f(x)\seq f(y).\]
Except when we state otherwise, we take $\seq$ on function types to be extensional
equality:
\[(f,w)\seq (g,u) \defequiv \qPi{x:X}f(x)\seq g(x).\]
Note that in this definition of $\seq$, we are treating the evidence that $f$ is a
function as property, regardless of the underlying type. We will consequently talk
about a function $f:X\to Y$ rather than an operation $f:X\to Y$ with a witness
$w:\qPi{x,y:X}(x\seq y)\to f(x)\seq f(y)$.

\section{Families of setoids}
The requirement that a function $f$ respects equivalence corresponds to the function~$\ap_f$. In
Part~\ref{section:homotopies}, we followed the definition of $\ap_f$ with a dependent version
$\apd_f$ which relied on transport. As a consequence of the uniform definition of
identity types, we were able to construct transport for general type families. Since
the equivalence relation on a setoid is ad-hoc, we have to give our notion of
transport in an ad-hoc way. The typical way to do this is via what
Palmgren~\cite{Palmgren2012} calls \emph{proof-irrelevant} families of
setoids, which treat the equivalence relation on the index type as property. A more nuanced approach,
called \emph{proof-relevant} families by Palmgren, treats the equivalence relation on
the index type as structure. The context is a setoid $(A,\seq)$ and a family of types
$B:A\to\univ$ where each $B(a)$ is a setoid. In order to be a family of setoids, we
need at least \emph{reindexing} bijections
\[\varphi:\qPi{x,y:A}x\seq y\to (B(x)\simeq B(y)).\]
We call $B$ a \emph{proof-irrelevant family} of setoids when
\begin{enumerate}[label=(\roman*)]
    \item $\varphi_p \seq \id_{B_x}$ whenever $p:x\seq x$ and
    \item $\varphi_q\comp \varphi_p \seq \varphi_{r}$ whenever $p:x\seq y$, $q:y\seq z$ and $r:x\seq z$.
\end{enumerate}
We call $B$ a \emph{proof-relevant family} of setoids when
\begin{enumerate}[label=(\alph*)]
    \item $\varphi_{r(x)} \seq \id_{B_x}$ for all $x:X$;
    \item $\varphi_{p^{-1}}\comp \varphi_{p} \seq \id_{B(x)}$ and 
        $\varphi_{p}\comp \varphi_{p^{-1}} \seq \id_{B(y)}$ for any $p:x\seq y$; and
    \item $\varphi_q \comp \varphi_p \seq \varphi_{q\cdot p}$ whenever $p:x\seq y$ and $q:y\seq z$.
\end{enumerate}
Note that any proof-irrelevant family of setoids is also a proof-relevant family of
setoids. In either case, we then have a version of transport, or substitution given
by the reindexing bijections. We can then say that a dependent function
$f:\qPi{x:A}B(x)$ is \emph{extensional} when $B$ is a family of setoids and
\[\qPi{x,y:A}\qPi{p:x\seq y}\varphi_p(f(x)) \seq f(y).\]
In other words, a dependent function on setoids is extensional when it satisfies a
version of $\apd$.

We are interested also in predicates $P:A\to\univ$. For these cases, we usually don't
care about any setoid structure on $P(a)$, so let us say that a \emph{extensional
predicate} on a setoid $A$ is a type family $P:A\to U$ with reindexing bijections
\[\transport:\qPi{x,y:A}x\seq y\to (P(x)\simeq P(y)).\]

\section{Truncation and quotients}
A setoid $X$ is \emph{proof-irrelevant} when $x\seq y$ for all~$x,y:X$. This notion
of proof-irrelevance is a setoid-version of being a homotopy proposition. Indeed,
given a setoid $X$, we have the setoid $\trunc{X}$ with $x \seq_{\trunc{X}} y$ for
all~$x,y:X$. Then the operation $|-|:X\to\trunc{X}$ which is identity on elements is
indeed a function. Moreover, $\trunc{X}$ has the universal property of truncation
with respect to proof-irrelevant setoids: let $P$ be any proof-irrelevant setoid, and
$f:X\to P$ be a function. Then the operation $f$ is also a function $\trunc{X}\to P$,
as for any $x,y:X$ we have that $f(x)\seq f(y)$, and so in particular, when
$x\seq_{\trunc{X}}y$, we have $f(x)\seq f(y)$. Moreover, if we take two functions
$\trunc{X}\to P$ to be equal when they are extensionally equal as operations, then
$f:\trunc{X}\to P$ is the unique factorization of $f:X\to P$ through $\trunc{X}$.

We can generalize this construction to any equivalence relation in the expected way.
Let $X$ be a setoid and let $R$ be an equivalence relation which respects $\seq$.
That is, $R$ is a family of types such that $x\seq y\to R(x,y)$. Then the quotient
$X/R$ of $X$ by $R$ is the setoid with base type $X$ and equivalence relation $R$.
Any function $X\to Y$ which respects $\seq$ will then also respect $R$, and so
extends (uniquely) to a function $X/R\to Y$.

Even more generally, Dybjer and Moeneclaey~\cite{DybjerMoeneclaey2017} showed that any
Higher-inductive type whose highest constructor is a path between points (i.e., of
type $x=y$ for some elements $x,y$) can be interpreted in the setoid model.

\section{The Delay monad}
\index{monad!delay|(}
We now turn to the delay monad. We need to ensure that the definition given in
Chapter~\ref{chapter:partiality} makes sense in the context of setoids. To do this,
we need to define $\NI$. In turn, we need setoid representations of $\nat$ and
$\bool$. Fortunately, this last task is easy: we equip $\nat$ and $\bool$ with the
equality relation given by the identity type. Note that for $\nat$ we can define this
without having the identity type using Theorem~\ref{thm:nat-is-set}. Equality on
$\bool$ can be represented in a similar way.

Then we may define $\NI$ to be the setoid whose underlying type is
\[\qSig{\alpha:\cantor}\qPi{j,k:\nat}(\alpha_j=1)\to(\alpha_k=1)\to (j=k),\]
and whose equivalence relation is extensional equality on the first component:
$(\alpha,-)\seq (\beta,-)$ when $\qPi{n:\nat}\alpha_n=\beta_n$.
Then, for a setoid $X$, we may define $\Delay(X)$ via
one of the representations in Chapter~\ref{chapter:partiality}:
\[\Delay(X) \defeq \qSig{\mu:\NI}\qPi{n:\nat}(\mu_n=1)\to X,\]
The equivalence relation that gives rise to the setoid $\Delay(X)$ is
\[(\mu,\varphi)\seq(\nu,\psi)\defeq \qSig{p:\mu\seq
\nu}\qPi{n:\nat}\transport(p,\varphi)(n) \seq_{(\nu_n=1)\to X} \psi(n).\]
In order for this to be well-defined, we need to know that the type family
\[\qPi{n:\nat}(\mu_n=1)\to X\]
is a family of setoids.

In order to show that $\Delay(X)$ forms a monoid, we need to give $\eta$ and
$\klext{(-)}$, and moreover, show that these operations are extensional. Showing
$\eta$ to be extensional is easy. Showing $\klext{(-)}$ to be extensional is a little
more work, but ultimately corresponds to the computation given in
Theorem~\ref{thm:delay-monad}.
\index{monad!delay|)}

\section{Rosolini partial elements and functions}
\index{Rosolini proposition|(}
By Theorem~\ref{thm:qdelay-is-ros}, we can represent the Rosolini lifting via the
quotient of the delay monad by weak bisimilarity. To quotient by weak bisimilarity,
we will replace the underlying type of $\Delay(X)$ with the (equivalent) uncurried form
\[\Delay(X) \defeq \qSig{\mu:\NI}\lar{\mu}\to X,\]
where $\lar{\mu}$ is again $\qSig{n:\nat}\mu_n=1$, 
and consider instead the equivalence relation
\[(\mu,\varphi)\seq (\nu,\psi)\defeq (\lar{\mu}\leftrightarrow\lar{\nu})\times
\Big(\qPi{m:\lar{\mu}}\qPi{n:\lar{\nu}}\varphi(m)\seq\psi(n)\Big).\]
Note that by using the above representation instead of the curried form 
\[\Delay(X) \defeq \qSig{\mu:\NI}\qPi{n:\nat}(\mu_n=1\to X),\]
we avoid needing to know that $\qPi{n:\nat}(\mu_n = 1)\to X$ is a family of
setoids, since we do not need to reindex. Above, where we defined the equivalence
relation on the unquotiented $\Delay(X)$, this reorganization doesn't help us, since
we also needed the sequences to agree at every point.
\index{Rosolini proposition|)}

\section{The axiom of choice}\label{section:setoid-choice}
\index{choice!countable|(}
It may come as a surprise that the monad data for the delay monad is also extensional
with respect to this coarser equivalence relation; i.e., that the quotiented delay
monad, $\Delay(-)/{\seq}$, is again a monad when we work with setoids. However,
countable choice holds in the setoid model, on account of the fact that the setoid
structure on $\nat$ is discrete---that it arises from the identity type.
\begin{theorem}
    Let $X$ be any setoid, and $R:\nat\to X\to \univ$ a relation such that for all
    $n:\nat$ there exists an $x:X$ such that $R(n,x)$. Then there is a function
    $f:A\to B$ such that $R(n,f(n))$ for all~$n:\nat$.
\end{theorem}
\begin{proof}
    The hypothesis of the theorem gives us a witness
    \[w:\qPi{n:\nat}\qSig{x:X}R(n,x).\]
    Then we get an operation $f:\nat\to X$ defined by
    \[f(n)=\pr_0(w(n)),\]
    and moreover, we have for any $n:\nat$ that $\pr_0(w(n)):R(n,f(n))$. All the
    remains is to see that $f$ is extensional. Let $n,k:\nat$ with $n\seq k$. But if
    $n\seq k$ we must have that $n=k$, and so~$f(n)=f(k)$.
\end{proof}
Unfortunately, retaining witness data is not enough for the full
axiom of choice, since we need our operations to be
extensional~\cite{martinlof2006choice}. In more detail, the
general situation is that we have a relation $R:A\to B\to\univ$ such that for all $a:A$
there exists a $b:B$ with $R(a,b)$. We wish to find a choice function $f:A\to B$. The
operation $f$ can always be found in the setoid model, by projecting from the type
$\qExists{b:B}R(x,y)$, but this operation is not guaranteed to be extensional. For
the setoid of natural numbers, the equivalence relation is such that all operations
out of $\nat$ are extensional, and so we have the above theorem.
\index{choice!countable|)}

\section{Discussion}
It is possible to construct the Rosolini lifting in the setoid approach, and
moreover, this does form a monad, as a result of the fact that countable choice holds
in the setoid model. Recent work by Coquand, Mannaa, and
Ruch~\cite{CMR2017Stacks} gives a model of univalent type theory in which
countable choice does not hold. To make the Rosolini propositions form a dominance,
we could perhaps abandon the univalent setting in favor of setoids, but this approach
won't allow us to place the Rosolini lifting in the context of a more general
lifting. To resolve this, we could follow other authors (such as \cite{Hofmann1995}
and \cite{CoC1986}), and introduce an impredicative type of propositions to have
access to a type $\Prop$ that could be used to form a general lifting. In this case,
to get the Rosolini lifting as a subtype of the lifting, we would need to ensure our
axiomatization of $\Prop$ allowed us to isolate some of those propositions
as the Rosolini propositions, and form a lifting setoid of those; whether the
resulting lifting would be equivalent to the quotient of the delay monad, and whether
this lifting would form a monad would depend on choices made about the behavior of
$\Prop$. Reus and Streicher~\cite{ReusStreicher1997} give one approach along similar
lines, although they do not use setoids. It would be worth examining the
interpretation of their type theory in the setoid model.
Unless $\Prop$ is chosen to behave much like the univalent type of
propositions, we lose the relationship between structure and property in this approach.

Nevertheless, the setoid approach is limited by the fact that there is no
clear way to interpret the universe in setoids. Beyond the fact that this makes the
univalent definition of propositions impossible, this strictly limits the
types, the logical principles~\cite{smith1988}, and more importantly, the
computable functions which can be described in our system. Indeed, work by
Aczel~\cite{Aczel1977}, Feferman~\cite{Feferman1982}, Griffor and
Rathjen~\cite{Griffor1994},  Rathjen~\cite{Rathjen2000,Rathjen2001}, and
Setzer~\cite{setzer1993,setzer1998,setzer2000} shed light on the ordinal analysis of
Martin-L{\"o}f type theory with universes. In short, adding universes to type theory
allows the construction of larger ordinals; using a hierarchy of recursive functions,
such as the L{\"o}b-Wainer hierarchy~\cite{lob1970}, these larger ordinals can be
used to define faster growing computable functions. Consequently, having universes in
our system gives a more fine-grained analysis of the limits of the partial computable
functions.  This thesis is focused on first-order computability, but we
hope to set the stage for a constructive approach to studying computability at higher
types, and computability with more interesting objects, such as metric spaces,
manifolds and algebraic structures. In this case, we need a universe not only to
define particular computable functions, but also to describe the objects we wish to
compute with.

To summarize and add to the above, the setoid approach has the following advantages over the univalent
approach:
\begin{enumerate}
    \item The setoid approach allows us to use extensionality principles in a
        simpler version of MLTT~\cite{Hofmann1995}. In particular, quotients and
        truncations can be described without extending the type theory.
    \item countable choice holds in the setoid model, so the Rosolini lifting becomes
        a monad without assuming additional principles. If we include a small type of
        propositions, we can indeed mimic the general lifting, and get the Rosolini
        lifting as a submonad of this.
    \item The setoid approach fits squarely with the traditional account of constructive
        mathematics proposed by Bishop.
\end{enumerate}
On the other hand, the univalent approach has the following advantages over the setoid approach
\begin{enumerate}
    \item The univalent approach allows us to make use of universes; in particular,
        this allows us to define more computable functions, and to develop
        computability in more interesting domains.
    \item Definitions are more uniform. As we use the inductively defined identity
        types instead of explicitly chosen equivalence relations, we can state much
        more general theorems, and all concepts are already extensional. Moreover, in
        the setoid approach, choice of representation is important: the equivalence
        relation on the quotient $\Delay(X)/{\seq}$ is more difficult to state for the
        representation \[\Delay(X) \defeq \qSig{\mu:\NI}\qPi{n:\nat}(\mu_n=1)\to X,\]
        as it requires us to make use of the fact that the type family
        $\qPi{n:\nat}(\mu_n=1)\to X$ is a family of setoids. This fact is
        not-obvious, and the result is somewhat technical. So with this
        representation, a non-trivial theorem is required to even form the
        quotient in the setoid approach.
    \item The univalent perspective fits squarely with a structural account of
        mathematics, and extends the traditional account of constructive mathematics.
    \item A type of propositions fits naturally into the
        theory, and clarifies the relationship between structure and property.
    \item There are computer proof systems with native support for formalization in
        univalent styles. No such support exists for the setoid
        approach. In particular, cubicaltt~\cite{cubicalttgit} and Cubical Agda implement cubical type
        theory. There is some support for univalent mathematics in
        Agda, Coq, and Lean, but this support is given via axioms, so is less
        satisfactory.
\end{enumerate}


%% file: chapters/computability.tex
We are now able to turn our attention to computability theory. In
Chapter~\ref{chapter:comp-as-struct}, we develop
a theory of computability \emph{as structure}, via a modest abstraction of Turing
machines (which we call \emph{recursive machines}), based on primitive recursive
combinators. Here, we follow Section~\ref{section:struct-vs-prop} in calling a type
family $X\to\univ$ a \emph{structure} and calling a type family \emph{property} when
it is valued in propositions. In particular, in
Chapter~\ref{chapter:comp-as-struct} we define a type family
$\CompStruct:(\nat\pto\nat)\to\univ$ of \emph{computation structures}, while in
Chapter~\ref{chapter:comp-as-prop}, we truncate this notion to arrive at the property
$\isComputable:(\nat\pto\nat)\to\Prop$.

Our abstraction is motivated by the fact that the
initialization and transition functions of a Turing machine are primitive recursive,
and so we can represent any Turing machine via a pair of (descriptions of) primitive
recursive functions. This approach is easier to work with in our system than Turing
machines or $\mu$-recursive functions: Turing machines contain a great deal of data
that needs to be tracked, while handling minimization takes more technical
computation than handling the other basic recursive operations. We develop some basic
results in computability
theory using recursive machines, and there are no surprises here. The point is to
give evidence that our notion of partial function is sound---that we can view the
computable partial functions as certain functions $\nat\to\Lift(\nat)$---and that
recursive machines are indeed capable of being used for a development of
computability theory. The only novelty from a computability theorist's perspective is
the language. It is worth noting, on this regard that we have no need of specifying
Kleene equality ("$f\seq g$ iff for all $x$, either both $f(x)$ and $g(x)$ are
defined and equal, or both $f(x)$ and $g(x)$ are undefined"), since Kleene equality
is simply equality in $\nat\to\Lift\nat$; our approach again allows us to dispense
with an ad hoc notion of equality. This definition also allows us to remove the
implicit use of excluded middle in the usual statement of Kleene equality.

In the first three sections of Chapter~\ref{chapter:comp-as-prop}, we truncate
notions of Chapter~\ref{chapter:comp-as-struct} to arrive at a notion of
computability \emph{as property}. Since most results have quite short proofs using
Lemma~\ref{lemma:trunc-func}, these sections are short.
After this, we prove the undecidability of the halting problem, in particular
demonstrating the existence of a function $d:\nat\to\Lift(\nat)$ (Turing's diagonal
function) which is not computable. Any total function we can define is computable,
and moreover, constructive intuition tells us that a constructive definition should
give rise to a computer program. Indeed, while it may be consistent in a
constructive framework that all total functions $\nat\to\nat$ are computable, the
partial function $d:\nat\pto\nat$ is not computable. We return to
the implications of this fact in Chapter~\ref{chapter:comp-and-partiality}, but first
isolate a subset of the Rosolini propositions which we call
\emph{semidecidable}---those propositions for which there exists a computable witness
that they are Rosolini. The semidecidable propositions are introduced and compared to
computable functions in Section~\ref{section:semidecidable}.

Chapter~\ref{chapter:comp-and-partiality} turns to the question of which functions
can be computable. We first (Section~\ref{computability:dominances}) demonstrate that
computable partial functions and semidecidable maps fit into the Rosolini partial
functions and Rosolini propositions, and then turn to the question of which functions
are computable.
As we have two notions of computability (structure and property), we can state
Church's thesis in two ways: one in the logic of structures, and one in the logic of
propositions. The former turns out to be false (Theorem~\ref{not-strong-ct}), by an
argument due to Troelstra~\cite{troelstra:1977}. The second is more plausible;
the argument given by Troelstra can be reworked slightly to give a weaker result:
Church's thesis is incompatible with the existence of an embedding from the
computable functions into $\nat$. This result can be leveraged to resolve the
conflict between the topos-theoretic and type-theoretic facts discussed in the
introduction. From there, we can consider a
version of Church's thesis for partial functions. In fact, Church's thesis tells us
that the Rosolini partial functions and the semidecidable partial functions coincide.
Countable choice tells us that each type of partial functions coincides with the
associated type of disciplined maps. Together, these results tell us
(Theorem~\ref{thm:punchline}) that under Church's thesis and countable choice,
the computable partial functions are the Rosolini partial functions. Since we know
the effective topos to be a model of countable choice and Church's thesis, this
suggests that it is consistent that all Rosolini partial functions are computable;
unfortunately, a univalent version of the
effective topos is needed to use this for a consistency result, and that requires
techiques beyond what we consider here.


%% file: chapters/comp-as-structure.tex
\chapter{Computability as structure}\label{chapter:comp-as-struct}
We develop in this chapter computability theory with computability as
\emph{structure} (C.f.~\ref{section:struct-vs-prop}) The results here are essentially
the same as classical results in computability; as we shall see, the
classical development of basic computability theory goes through with only minor
changes. In particular,
\begin{itemize}
    \item we use partial functions as in the previous chapter, as functions
        $X\to\Lift(Y)$;
    \item we make explicit reference to structure throughout.
\end{itemize}
This makes the statements of theorems slightly more cumbersome, but otherwise
presents little change. In the next Chapter~\ref{chapter:comp-as-prop}, we develop
computability as property and compare it to computability as structure. For the most
part, we get the standard results in computability theory for computability as
property from the equivalent results for computability as structure, by functoriality
of truncation. However, when we look at how computability interacts with the broader
mathematical universe, things seem to change. For example, we have an embedding from
``functions with computability structure'' to $\nat$,
but such an embedding from
``functions with the property of being computable'' to $\nat$ would imply a
decidability result that we cannot expect to hold constructively.

\section{Primitive recursion}
In order to undertake a study of computability, we need to introduce a model of
computation; Turing machines are intuitively convincing, but we would like a model
that is somewhat less cumbersome to work with, so we will abstract away many internal
details.  We will make use of primitive recursive functions for this, so we briefly
recall the definition and some basic facts. The material here can be found in any
standard reference such as Odifreddi \cite{Odifreddi1989} or Rogers
\cite{Rogers1987}.

\begin{definition}
    \index{primitive recursive!combinators|textit}
    \defineopfor{$\PR$}{PR}{primitive recursive!combinators}
    \defineopfor{$\PrimRec$}{PR}{primitive recursive!combinators}
The type family $\PR:\nat\to\univ$ of \emph{primitive recursive combinators (of arity
$n$)} is
defined inductively by
\begin{enumerate}
    \item $\succc:\PR_1$;
    \item for any $n:\nat$ and $k<n$ we have $\pc^n_k:\PR_n$;
    \item for any $n,k:\nat$ we have $\cc^n_k:\PR_n$;
    \item if $f:\PR_n$ and $g_i:\PR_m$ for each $0<i\le n$, then
        $f\langle g_1,\ldots,g_n\rangle :\PR_m$.
    \item if $f:\PR_{n+2}$ and $g:\PR_{n}$, then $\rc_{f,g}:\PR_{n+1}$.
\end{enumerate}
    We will sometimes write $\PR$ for the total type $\qSig{n:\nat}\PR_n$.
    Similarly, we define a type family $\PrimRec:\qPi{n:\nat}(\nat^n\to\nat)\to\univ$
    inductively by
\begin{enumerate}
    \item $\PrimRec_1(\succop)$;
    \item for any $n:\nat$ and $k<n$ we have $\PrimRec_n(\pr_k)$;
    \item for any $n,k:\nat$ we have $\PrimRec_n(\lambda x.k)$;
    \item if $\PrimRec_n(f)$ and $\PrimRec_m(g_i)$ for each $0<i\le n$, then
        $\PrimRec_m(\lambda x. f(g_1(x),\ldots,g_n(x)))$
    \item if $\PrimRec_{n+2}(f)$ and $\PrimRec_n(g)$, then
        $\PrimRec_{n+1}(\recop_{f,g})$, where
\begin{eqnarray*}
    \recop_{f,g}(0,x_1,\ldots,x_n) &=& g(x_1,\ldots,x_n) \\
    \recop_{f,g}(k+1,x_1,\ldots,x_n) &=&
    f(k,\recop_{f,g}(k,x_1,\ldots,x_n),x_1,\ldots,x_n),
\end{eqnarray*}
\end{enumerate}
    and where we leave the constructors for $\PrimRec_n(f)$ unnamed.

\end{definition}
\definesymbol{$\app$}{app}
We can then define a function $\app:\qPi{n:\nat}\PR_n\to\nat^n\to \nat$ in
the obvious way:
\begin{eqnarray*}
    \app(\succc) &=& \succop,\\
    \app(\pc^n_k)&=& \lambda (x_1,\ldots,x_n).  x_k,\\
    \app(\cc^n_k) &=& \lambda x.k,\\
    \app(f\langle g_1,\ldots,g_n\rangle) &=& \lambda x. f(g_1(x),\ldots,g_n(x)),\\
    \app(\rc_{f,g}) &=& \recop_{f,g}.\\
\end{eqnarray*}
It is easy to check that $\PrimRec(\app(t))$, and that $\app$ induces
an equivalence
\[\PR_n\simeq\qSig{f:\nat^n\to\nat}\PrimRec(f),\]
for all $n:\nat$. We say
that $f$ has \emph{primitive recursive structure} when $\PrimRec(f)$.
\index{primitive recursive!structure|textit}

\begin{definition}
    \index{primitive recursive!relation|textit}
    A \emph{primitive recursive relation} is a primitive recursive $R:\nat^k\to\nat$
    valued in $\{0,1\}$. We also say that a relation $R:\nat^k\to\Prop$ \emph{has primitive
    recursive structure} if it has a characteristic function (as in
    Section~\ref{section:nat})
    $\chi_R$ with $\PrimRec(\chi_R)$.
\end{definition}
The following constructions are routine.
\begin{theorem}\label{thm:pr-basics}
    The following all have primitive recursive structure.
    \begin{enumerate}[label=(\roman*)]
        \item the standard ordering $\le$ on $\nat$.
        \item equality on $\nat$.
        \item the addition function $-+-:\nat^2\to\nat$;
        \item the multiplication function $-\cdot-:\nat^2\to\nat$;
        \item the predecessor function with definition 
            \begin{align*}
                \pred(0) &= 0,\\
                \pred(n+1) &= n;
            \end{align*}
        \item truncated subtraction
            \begin{align*}
                n - 0 &= n,\\
                n - (k+1) &= \pred(n-k);
            \end{align*}
        \item bounded sums: for fixed (primitive recursive) $f:\nat^{n+1}\to\nat$, and
            $x:\nat^n$,
            \begin{align*}
                \sum_{y< 0}f(x,y) &= 0, \\
                \sum_{y< (k+1)}f(x,y) &= f(x,k) + \sum_{y< k}f(x,y);
            \end{align*}
        \item bounded products: for fixed (primitive recursive) $f:\nat^{n+1}\to\nat$, and
            $x:\nat^n$,
            \begin{align*}
                \prod_{y< 0}f(x,y) &= 1, \\
                \prod_{y< (k+1)}f(x,y) &= f(x,k) \cdot \prod_{y< k}f(x,y);
            \end{align*}
        \item $\sgn:\nat\to\nat$,
            \begin{align*}
                \sgn(0) &= 0,\\
                \sgn(n+1) &= 1;
            \end{align*}
        \item $\opsgn:\nat\to\nat$,
            \begin{align*}
                \opsgn(0) &= 1,\\
                \opsgn(n+1) &= 0;
            \end{align*}
            alternatively, $\opsgn(n) = 1-\sgn(n)$;
        \item the factorial function
            \begin{align*}
            n! &= \prod_{k<n}(k+1).
            \end{align*}
    \end{enumerate}
\end{theorem}
Using the above, we can lift primitive recursive structure over logical operations as
follows: let $P$ and $Q$ have primitive
recursive structure, then we have characteristic functions with primitive recursive
structure
\begin{align*}
    (p \wedge q)(n) &\defeq p(n)\cdot q(n), \\
    (p \vee q)(n) &\defeq \sgn(p(n) + q(n)),\\
    (p \to q)(n) &\defeq p(n) \le q(n),\\
    (\neg p)(n) &\defeq \opsgn(p(n)),\\
    \exists (k< n).p(k) &\defeq \sgn(\sum_{k< n}p(k)),\\
    \forall (k< n).p(k) &\defeq \prod_{k< n}p(k),
\end{align*}
for (respectively) $P\wedge Q$, $P\vee Q$, $P\to Q$, $\neg P$,
$\qExists{k\le n}P(k)$ and $\forall(k\le n),P(k)$.
Then we may show a relation has primitive recursive structure by defining it in terms of known
primitive recursive relations. Likewise, we can build primitive recursive structure
piecewise
\begin{lemma}
    Let $f,g:\nat^k\to\nat$ have primitive recursive structure and let $Q:\nat^k\to\Prop$
    have primitive recursive structure. Then the function $\nat^k\to\nat$ defined by
    case analysis as
    \[x\mapsto
    \begin{cases}
        f(x) & \text{if $Q(x)$;}\\
        g(x) & \text{otherwise}
    \end{cases}\]
    has primitive recursive structure
\end{lemma}
\begin{proof}
    First define $h:\nat^{k+1}\to\nat$ by $h(0,x) = g(x)$ and $h(1,x) = f(x)$. That
    is,
    \[h=\recop_{g,f'}\]
    where $f'(n,k,x) = f(x)$. Then $h$ has primitive recursive structure. Now define 
    \[r(x) = h\comp\lar{Q(x),\pr_1(x),\ldots,\pr_{n-1}(x)},\]
    which has primitive
    recursive structure. Now let $x:\nat^n$. If $Q(x)$, then $r(x) =
    h(1,x) = f(x)$ and if $\neg Q(x)$, then $r(x) = h(0,x) = g(x)$. Hence,
    $r(x)$ is the function defined by case analysis above.
\end{proof}

The last piece required before continuing the standard presentation of the primitive
recursive functions is \index{minimization!bounded}
bounded minimization. Let $P:\nat\to\Prop$ be a primitive recursive
predicate, fix an $n:\nat$ and consider the type $\min_{k<n}(P)$
\[\qSig{k:\nat}\big((k< n\times P(k)) + (k=n)\big)\times \qPi{j< k}\neg P(j)),\]
and the map $\min_{k<n}(P)\to\nat$ given by the first projection.  That is, $\mu_{k<n}P(k)$ is the
least value of $k\le n$ such that $P(k)$ holds, if such exists, and $n$ otherwise.
\begin{lemma}
    For any predicate $P:\nat\to\Prop$ with primitive recursive structure and any $n:\nat$, the type $\min_{k<n}(P)$
    is contractible.  Hence, the partial element $(\min_{k<n}P(k),\pr_0)$ is defined.
\end{lemma}
\begin{proof}
    First, show that $\min_{k<n}P(k)$ is inhabited; i.e., that we have $k:\nat$
    satisfying the predicate
    \[Q(k)\defeq \big((k< n\times P(k)) + (k=n)\big)\times \qPi{j< k}\neg P(j)).\]
    Since $P$ has primitive recursive structure, it is decidable. Hence, we can do a bounded
    search to determine
    \[(\qSig{k<n}P(k))+\qPi{k<n}\neg P(k).\]
    In the first case, take the minimum such $k$, which must exist. In the second,
    $n$ must satisfy~$Q$.

    To show that $\min_{k<n}P(k)$ is a proposition, take any $k,j$ satisfying the
    predicate $Q$ we must have $k=j$, since otherwise we have $k<j$ (or conversely),
    and then we must have both $P(k)$ and $\neg P(k)$, a contradiction, and equality
    on $\nat$ is decidable.
\end{proof}
Then we may define for any predicate $P$ with primitive recursive structure
\definesymbol{$\mu$}{min}
\[\mu_{k< -}(P):\nat\to\nat\]
\[\mu_{k< n}(P) \defeq \val(P)(\star)\]
I.e., $\mu_{k<n}(P)$ is the first projection applied to the center of contraction of
$\min_{k<n}(P)$.
\begin{theorem}
    If $P$ has primitive recursive structure, then the function $\mu_{k<-}(P)$ has
    primitive recursive structure.
\end{theorem}
\begin{proof}
    Note that the predicate $Q$ used in the definition of $\min_{k<n}P(k)$ is itself
    primitive recursive, being built up of bounded quantifiers and basic logical
    operations. So let $q$ be a primitive recursive term computing $Q$. Now consider
    the function $f:\nat^2\to\nat$ defined by primitive recursion with
       \[ f(0,n) =
      \begin{cases} 0 & \text{if }Q(0),\\
          n & \neg Q(0);
      \end{cases}\]
      \[ f(k+1,n) =
     \begin{cases}
         \min(k+1,f(k,n)) & \text{if }Q(k+1);\\
         \min(n,f(k,n))   & \text{if }\neg Q(k+1).
     \end{cases}
     \]
     Note that $f$ is defined by primitive recursion over the base function
     $g:\nat\to\nat$ which takes value $0$ if $Q(0)$ and is the identity otherwise,
     and recursive step given by
      \[ h(k,m,n) =
     \begin{cases}
         \min(k+1,m) & \text{if }Q(k+1);\\
         \min(n,m)   & \text{if }\neg Q(k+1).
     \end{cases}
     \]
     So then we see that $h$ has primitive recursive structure as a case analysis,
     and $f$ has primitive recursive structure as $h$ and $g$ do.
     Now, consider the value of $f(n-1,n)$: If $n=0$ or $n=1$, this is $0$.
     Otherwise, we know that either there is a least $k<n$ with $Q(k)$ or else there is
     no $k<n$ with $Q(k)$. If there is no such $k$, then $f(k,n) = n$ for each $k<n$,
     and so $f(n-1,n) = n$. If there is such a $k$, then we must have that $f(k,n) =
     k$ since for $j<k$, we must have $f(j,n)=n$. Then we know that
     \[\min_{k<n}P(k) = f(n-1,n),\]
     and the function $n\mapsto f(n-1,n)$ is defined by composition from functions
     with primitive recursive structure.
\end{proof}

We can continue to mimic the classical development of primitive recursion; from here,
on there will be no surprises. However, we stop to mention encodings of sequences,
since we will use them when working with general computability:
we will deal with higher arities via a
primitive recursive encoding. Specifically, for each $n$, we we fix a bijective
pairing function $\lar{-} : \nat^n\to\nat$
\index{primitive recursive!pairing function|textit} which is primitive recursive and such that
the projection functions $\pr_i:\nat\to\nat$  are primitive recursive. Defining such
functions, and indeed even a primitive recursive function
\index[symbol]{bracket@{$\lar{-}$}!encoding of PR pairing function|textit}
$\lar{-}:\sum(n:\nat),\nat^n\to\nat$ is a standard exercise---for
$\lar{-}:\sum(n:\nat),\nat^n\to\nat$, we can use G{\"o}del's $\beta$. Then we define an
$n$-ary function to be computable if its composition with $\lar{-}^{-1}$ is
computable. Note that by precomposition with $\lar{-}$, any function can be viewed as
an $n$-ary function for any $n$.

In fact, if both $\lar{-}:\nat^k\to\nat$ and $f:\nat\to\nat$ have primitive recursive
structure, then the composite function $f\comp\lar{-}:\nat^k\to\nat$ does as well. A
direct converse doesn't type check, but we can do the following: 
For any $k:\nat$ fix a pairing function $\lar{-}:\nat^k\to\nat$ and projection
functions $\pr_i:\nat\to\nat$ all of which have primitive recursive structure and
such that
\[\lar{\pr_0(n),\ldots,\pr_{k-1}(n)} = n\]
and 
\[\pr_i(\lar{x_0,\ldots,x_{k-1}}) = x_i.\]
For $f:\nat^k\to\nat$ with primitive recursive structure define $f':\nat\to\nat$ by
\[f' = f\comp\lar{\pr_0,\ldots,\pr_{k-1}}\]
so that 
\[f'(n) = f(\pr_0\lar{n},\ldots,\pr_{k-1}\lar{n}) = f(n).\]

In short, we can encode all primitive recursive functions of $k$ variables as
primitive recursive functions of $1$ variable. We will tacitly switch between a
function of a single variable and a function of $k$ variables when it suits us.

Similarly, we may encode the type $\nat+\nat$ via functions $\inl,\inr:\nat\to\nat$,
given by, for example, $\inl(n) = 2n$ and $\inr(n) = 2n+1$. Now given two functions (with
p.r.\ structure) $f,g:\nat\to\nat$ we can \emph{extend} them along $\nat+\nat$ to a
function $f+g:\nat\to\nat$ with primitive recursive structure defined by
\[(f+g)(n) = 
\begin{cases}
    f(k) & \text{if } n=2k;\\
    g(k) & \text{if } n=2k+1.\\
\end{cases}\]
This gives us a primitive recursive encoding of $\nat+\nat$, which we will use to
define a notion of computation.
\section{Recursive machines}\label{section:recursive-machines}

We abstract away the details of initializing a Turing machine, and the transition
function. The central ideas that captures how a Turing machine computes are
\begin{itemize}
    \item a Turing machine can be initialized by a simple process;
    \item checking whether a Turing computation has completed is simple;
    \item the transition function updating a Turing machine is a simple process.
\end{itemize}
We capture these properties in the following definition.
\begin{definition}
    A \nameas{recursive machine}{recursive!machine}, is a pair
    $(i,s):\PR_1\times\PR_1$. The function $i$ is called the
    \emph{initialization function}, and we treat it as a function
    $i:\nat\to\nat$, while $s$ is called the  \emph{transition function}, and we
    treat it as a function $s:\nat\to\nat+\nat$.

    We use $\RM$ for the type of recursive machines.
    \defineopfor{$\RM$}{RM}{recursive machine}
\end{definition}

We can evaluate recursive machines via a function
\definesymbolfor{$\eval$}{eval}{recursive machine}
\[ \eval:\RM\to(\nat\to\Lift \nat)\] as follows: Given $(i,s):\RM$,
let $s':\nat+\nat\to\nat+\nat$ be the function
\begin{align*}
    s'(\inl x) &= s(x),\\
    s'(\inr y) &= \inr y.
\end{align*}
As we will use this function $s'$, we will abuse notation and refer to $s'^k$ as
$s^k$.

\definesymbolfor{$\run$}{run}{recursive machine}
Now define for each $k:\nat$ the function $\run_k:\RM\to\nat\to(\nat+\nat)$ by
\[\run_k((i,s),x) \defeq s^k(\inl(i(x)))\]
And so we have for fixed $m=(i,s)$,
\[R_m(x,y)\defeq \qExists{k:\nat}\run_k(m,x) = \inr y.\]
Then $\eval(m)$ is the partial function with $R_m$ as its graph. Note that $R_m$ can
be defined without truncation using the results of Section~\ref{section:nat}.

The partial function $\eval(m)$ is the function \emph{computed by} $m$. We will
sometimes abuse notation and write $m(x)$ instead of $\eval(m,x)$. For a partial
function $f:\nat\to\Lift\nat$, let $\CompStruct(f)$ be the type of recursive machines
computing $f$:
\[\CompStruct(f) \defeq \qSig{m:\RM}f = \eval(m).\]
\index{recursive!structure|see{computability structure}}
\defineopfor{$\CompStruct$}{compstruct}{computability structure}
Say that $m$ is a \nameas{computability structure}{computability!structure} for $f$ or
that $f$ has \emph{recursive structure}, when
$m$ computes $f$, i.e., when $(m,-):\CompStruct(f)$.
Note that we have
\[\RM\simeq \qSig{f:\nat\to\Lift(\nat)}\CompStruct(f).\]
If $f:\nat\to\nat$ is an ordinary function, then $f$ has recursive structure when
$\eta\comp f:\nat\to\Lift\nat$ does.

Functions with recursive structure are closed under composition. Indeed, given machines
$m$ and $n$, we can define a composite machine $m;n$ as follows: First define
$i':\nat\to\nat+\nat$ by
\[i'(k) = \inl (i_m k);\]
and a function $s':\nat+\nat\to((\nat+\nat)+\nat)$ by
\begin{align*}
    s'(\inl k) &= \inl (\inl y) & \text{if }s_m(k)=\inl y;\\
    s'(\inl k) &= \inl (\inr (i_n y)) & \text{if }s_m(k)=\inr y;\\
    s'(\inr k) &= \inl (\inr y) & \text{if }s_n(k)=\inl y;\\
    s'(\inr k) &= \inr y & \text{if }s_n(k)=\inr y;\\
\end{align*}
Then the three components in the codomain $(\nat+\nat)+\nat$ correspond (from left to
right) to "we are still computing $m$"; "we have computed $m$ and are now computing
$n$"; and "we have finished computing the composite". By post composing with a
primitive recursive bijection $c:\nat+\nat\to\nat$, we get $(c\comp i',c\comp s')$ as
a recursive machine.

\begin{lemma}\label{thm:compstruct-composition}
If $m$ and $n$ compute $f$ and $g$ respectively, then $m;n$
computes $g\kcomp f$.
\end{lemma}
\begin{proof}
    We show that the relation $R_m;R_n$ used in the definition of $\eval$
    is the same as the relation $R_{m;n}$.

    Suppose that $(R_m;R_n)(x,z)$. That is,
    \[\qExists{y:\nat}\qExists{k:\nat}(s_m^k(i_{m}x) = \inr (y))\times
    (\qExists{j:\nat}s_n^k(i_{n}y)=z).\]
    We could give a non-truncated equivalent type, but this is not necessary here: by
    currying and properties of truncation, to get $R_{m;n}(x,z)$ it is enough to show
    \[\qPi{y:\nat}\big(\qSig{k:\nat}s_m^k(x) = \inr (y)\big)\to
    \big(\qSig{j:\nat}s_n^j(y)=z\big)
    \to R_{m;n}(x,z).\]
    The definition of $m;n$ was chosen specifically to satisfy this.

    Conversely, if $R_{m;n}(x,z)$, then take $k$ least such that $s'(\inl (i_m x)) =
    \inr y$ for some $y$, so that $s^k_m(i_m x) = \inr y'$ with $i_n(y')=y$. Then we
    must have $s_n^j(y) = z$ for some $j$, since $R_{m;n}(x,z)$.
\end{proof}

For convenience, we will often describe recursive machines in an informal way, making
reference to ``configurations'', ``continuing'', and ``halting''. As an example, we
again describe the machine for $m;n$:

To initialize $m;n$, initialize $m$. For the transition function, the input is in one
of two states:
\begin{enumerate}
    \item the input represents a partially computed output of $m$;
    \item the input represents a partially computed output of $n$.
\end{enumerate}
In each case, do the following:
\begin{enumerate}
    \item if one step of $m$ completes the computation of $m(x)$ with value $y$, then output
        the initialization of $n$ at $y$ in state (2). Otherwise, take one
        more step of $m$ and continue in state (1).
    \item if one step of $n$ completes the computation of $n(x)$ with value $y$, then
        output $y$ and halt. Otherwise, take one more step of $m$ and continue in state
        (2).
\end{enumerate}

\section{Recursive predicates and minimization}
\begin{definition}
    \index{relation!recursive|see{recursive, relation}}
    A \nameas{recursive relation}{recursive!relation} is a total function $R:\nat^k\to\nat$
    valued in $\{0,1\}$ with a computability structure. We also say that a
    relation $R:\nat^k\to\Prop$ \emph{has recursive structure} if it has a
    characteristic function $\chi_R$ with $\CompStruct(\chi_R)$.
\end{definition}

\begin{definition}
    We have a \nameas{minimization operator}{minimization} for partial functions, 
$\mu:(\nat\to\Lift \nat)\to\Lift \nat$, with
\[\extent(\mu f) \defeq \qSig{k:\nat}f(k)=\eta 0\times \qPi{j<k}\qSig{n:\nat}f(j)=
    \eta (n+1),\]
    and $\val(\mu f) \defeq \pr_1$.
    More generally, we may define the \emph{minimization $\mu y.R(y):\Lift\nat$ of a
    predicate} $R:\nat\to\Prop$
\[\extent(\mu y. R(y)) \defeq \qSig{y:\nat} R(y)\times \qPi{x<y}\neg R(x).\]
    and $\val(\mu y. R(y)) \defeq \pr_1$.
\end{definition}
When there is a least such $y$, this does indeed choose it. Moreover, the extent of
$\mu y.R(y)$ is a proposition. That is
\begin{lemma}
    For any predicate $R:\nat\to\Prop$, the type $\qSig{y:\nat} R(y)\times
    \qPi{x<y}\neg R(x)$ is a proposition.
\end{lemma}
\begin{proof}
    For any predicate $R:\nat\to\Prop$ and any $y:\nat$ we have that $R(y)\times
    \qPi{x\le y}\neg R(x)$ is a proposition, as a product of propositions, so we need
    only check the first component. That is we need
    to see that $m=n$ whenever we have witnesses $w:R(m)\times \qPi{x< m}\neg R(x)$ and
    $v:R(n)\times \qPi{x< n}\neg R(x)$. As the ordering on $\nat$ is
    decidable, we have either $m<n$, $m=n$ or $n<m$. If $m<n$, then we have $\neg R(m)$
    by $v$ and $\R(m)$ by $w$, a contradiction. Symmetrically, we cannot have $n<m$,
    and so $n=m$.
\end{proof}
Note that if there is a unique $y:\nat$ satisfying a predicate, then we can find $y$ by
minimization. Then expanding the definition of $\eval$ and $\mu$, we see that for any
recursive machine $m$, we have
\[\eval(m)(x) = \mu y.\qExists{k:\nat}\run_k(m,x) = \inr y.\]
Using this fact we can prove that the minimization of a relation with primitive
recursive structure has recursive structure.
\begin{theorem}\label{lemma:pr-minimization}
    If $R:\nat^2\to\Prop$ has primitive recursive structure, then
    $\lambda x. \mu y.R(x,y)$ has recursive structure.
\end{theorem}
\begin{proof}
    Define
    \[i(x) \defeq \langle x,0\rangle;\]
    and
    \[s(\langle x,n\rangle) \defeq
    \begin{cases}
        \inr(n) & \text{if }R(x,n);\\
        \inl \langle x,n+1\rangle &\text{otherwise.}
    \end{cases}
        \]
    Pairing has primitive recursive structure, and case analysis on a primitive recursive
    predicate has a primitive recursive structure, so we know $i$ and $s$ have primitive
    recursive structure. Hence $m=(i,s)$ is a recursive machine. We need to see that
    $\eval(m)(x) = \mu y.R(x,y)$.
    We have that $\run_k(m,x) = \inl \langle x, k+1\rangle$ until $R(x,k)$. Then
    $\run_k(m,x) =\inr k$ precisely when $k$ is least such that $R(x,k)$.
\end{proof}
It is tempting to use the same argument to show that this works for any
\emph{recursive} predicate, but we need to be more careful: the function $s$ defined
in the proof is only primitive recursive because $R$ is. For a general recursive
predicate, more work needs to be done. The tools required are discussed in the next
section.

\section{The Normal Form Theorem}\label{section:normal-form}
In this section we prove a version of the normal form theorem for recursive machines.
That is, we define encodings of recursive machines, such that we may define Kleene's
predicate $\KleeneT:\nat^3\to\Prop$ and function $\KleeneU:\nat\to\nat$.  The
meaning of $\KleeneT(m,x,y)$ is that $y$ represents a full computation trace of the
application of the function encoded by $m$ to $x$---such a $y$ exists only when
$m(x)$ produces a value. The function $\KleeneU$ then extracts the value produced by~$m(x)$.
We will then prove that the minimization of an arbitrary predicate with computability
structure defines a function with computability structure.

In order to prove this, we first prove a version for primitive recursive combinators.
\begin{theorem}[Kleene's normal form theorem for Primitive recursive
    functions]\label{knf-pr}
    \definesymbol{$\encode$}{encode}
    There is an injection $\encode:\PR\to\nat$, a predicate $\KleeneT':\nat^3\to\Prop$ with
    primitive recursive structure,
    and a function $\KleeneU':\nat\to\nat$ with primitive recursive structure such
    that for any primitive recursive combinator $t$ and any $x:\nat$ we have that
    $\qSig{y:\nat}\KleeneT'(\encode(t),x,y)$ is contractible, and moreover 
    \[\app(t,x) = \KleeneU'(y),\]
    whenever $y$ is such that $\KleeneT'(\encode(t),x,y)$. 
\end{theorem}
\begin{theorem}[Kleene's normal form theorem for recursive machines]\label{knf}
    There is an injection $\encode:\RM\to\nat$, a predicate $\KleeneT:\nat^3\to\univ$
    with primitive recursive structure, and a function $\KleeneU:\nat\to\nat$ with primitive
    recursive structure such that for any recursive machine $m$ and any $x:\nat$ we
    have that $\qSig{y:\nat}\KleeneT(\encode(m),x,y)$ is a proposition, and moreover
    \[\eval(m,x) = ((\qSig{y:\nat}\KleeneT(\encode(m),x,y)),\KleeneU\circ\pr_1).\]
\end{theorem}

Recall that we have a primitive recursive onto coding of finite sequences, along with
projection and length functions. We use $\langle a_0,\ldots,a_n\rangle$ for the
encoding of the sequence $(a_0,\ldots,a_n)$, and $\concat$ for concatenation. If $x$
is a number representing a sequence, then $(x)_i$ is the $i$-th coordinate of~$x$.

Now, we define an interpretation
\begin{align*}
    \encode &: \PR\to\nat \\
    \encode(\succc) &= \langle 0\rangle, \\
    \encode(\pc^n_k)&= \langle 1, n, k\rangle, \\
    \encode(\cc^n_k)&= \langle 2, n ,k\rangle, \\
    \encode(f\langle g_1,\ldots,g_n\rangle) &= \langle 3\rangle \concat \encode(f)
        \concat \encode(g_1) \concat \cdots \concat \encode(g_n), \\
    \encode(\rc_{f,g}) &= \langle 4\rangle \concat \encode(f) \concat \encode(g).
\end{align*}

To arrive at $\KleeneT'$ and $\KleeneU'$, we first define a type of computation
trees. Each node of the tree will have 3 labels:
\begin{itemize}
    \item (The code of) a primitive recursive combinator,
    \item an integer $x$, representing a sequence of inputs,
    \item an integer $z$ representing an output.
\end{itemize}

We define
\begin{enumerate}
    \item $[\succc,x,y]$ is a leaf computation tree;
    \item $[\pc^n_k,x,y]$ is a leaf computation tree;
    \item $[\cc^n_k,x,y]$ is a leaf computation tree;
    \item $[f\langle g_1,\ldots,g_n\rangle,x,y]$ is the root of a computation tree with
        $n+1$ branches;
    \item $[\rc_{f,g},\langle 0,x\rangle,y]$ is the root of a computation tree with
        1 branch;
    \item $[\rc_{f,g},\langle n+1,x\rangle,y]$ is the root of a computation tree with
        2 branches;
\end{enumerate}
A computation tree $t$ is \emph{correct} when
\begin{enumerate}
    \item the root of $t$ is $[\succc,x,x+1]$ for some $x:\nat$;
    \item the root of $t$ is $[\pc^n_k,x,k]$ for some $x:\nat$;
    \item the root of $t$ is $[\cc^n_k,x,(x)_0]$, for some $x:\nat$;
    \item the root of $t$ is $[f\langle g_1,\ldots,g_n\rangle,x,y]$, the branches of $t$ are
        all correct, and
        \begin{itemize}
            \item the $0$-th branch has root $[f,\langle z_1,\ldots,z_n\rangle,y]$,
                and
            \item for $0<k\le n$, the $k$-th branch has root $[g_k,x,z_k]$;
        \end{itemize}
    \item the root of $t$ is $[\rc_{f,g},\langle 0\rangle\concat x,y]$, the branch is
        correct and has label $[g,x,y]$.
    \item the root of $t$ is $[\rc_{f,g},\langle n+1\rangle\concat x,y]$, both branches are
        correct, and
        \begin{itemize}
            \item the $0$-th branch has root $[\rc_{f,g}\langle n\rangle\concat x, z]$,
                and
            \item the first branch has root $[f,\langle n\rangle\concat
                x\concat\langle z\rangle,y]$;
        \end{itemize}
\end{enumerate}
In short, a computation tree with root $[t,x,y]$ is correct precisely when
$\app(t,x)=y$.

Now we can encode a tree $t$ with root label $[c,x,y]$ and branches $\{t_i\}_{i<n}$ as
\[\hat{t} = \langle \langle \encode(c),x,y\rangle, \hat{t}_0,\ldots,\hat{t}_{n-1}\rangle.\]

By expanding the definition of the coding of sequences, and using the definition of
``correct'' above, we get a primitive recursive
predicate $\correct(x)$ expressing "$x$ encodes a correct computation tree." Now we may
define 
\[\KleeneT'(e,x,y) \defeq \correct(y)\wedge \big(((y)_0)_0 = e\big) \wedge
\big(((y)_0)_1 = x\big).\]
and
\[\KleeneU'(y) \defeq ((y)_0)_2,\]
so that we can define $\{-\}_{\PR}$ to be
\[\{n\}_{\PR}(x) = \KleeneU'(\mu_y.\KleeneT'(n,x,y)).\]

We need to see that for any $t:\PR$ and $x:\nat$
\[\app(t,x) = \{\encode(t)\}_{\PR}(x),\]
i.e., that
\[\app(t,x) = \KleeneU'(\mu_y.\KleeneT'(n,x,y)).\]
But we've constructed $\KleeneT'$ such that
\[\KleeneT'(n,x,y)\Leftrightarrow \eval(t,x) = y.\]

Note that for a fixed $k$ and recursive machine $m$, $\run_k(m):\nat\to\nat$ has
primitive recursive structure. In fact, we can do better.
\begin{lemma}
The function \[ \source:\nat\to\nat\to\nat\] taking a number $k$ and
$\langle \encode(i),\encode(s)\rangle$ to $\encode(\run_k(i,s))$ has 
primitive recursive structure.
\end{lemma}
\begin{proof}
    Straightforward induction on $k$ and the combinators for $i$ and $s$.
\end{proof}
 Then we may use $\KleeneT'$ and $\KleeneU'$ to define
$\KleeneT$ and $\KleeneU$ for recursive machines:
\[\KleeneT(e,x,\langle k,y\rangle) = \KleeneT'(\source(k,e),x,y).\]
and
\[\KleeneU(\langle k,y\rangle) = \KleeneU'(y).\]

This result gives us the Kleene-bracket function
\definesymbol{$\{-\}$}{bracket-kleene}
$\{-\}:\nat\to(\nat\to\Lift \nat)$, defined by
\[\{e\}(x) \defeq \Lift \KleeneU(\mu y.\KleeneT(e,x,y)).\]
In order to be able to apply $\{-\}$ to a partial natural number, we may Kleisli
extend $\{-\}$ in the first component. We will abuse notation and write also
$\{-\}:\Lift \nat\to(\nat\to\Lift \nat)$ for this map.

\begin{theorem}\label{thm:rec-cases}
    If $R:\nat\to\Prop$ has recursive structure, then case analysis on $R$ has
    recursive structure. Explicitly, if $R$ is a recursive predicate and $e,e':\RM$
    then the function $f:\nat\to\Lift\nat$ defined by
    \[f(x) = \begin{cases}
        e(x) & \text{if }R(x)\\
        e'(x) & \text{if }\neg R(x)
    \end{cases}
     \]
    has recursive structure.
\end{theorem}
\begin{proof}
    The initialization function initializes $R$ for the transition function with 3
    states:
    \begin{enumerate}
        \item the input represents a pair $(x,y)$, where $x$ is the input to $f$, and
            $y$ is a partially computed value of $R(x)$;
        \item the input represents a partially computed value of $e(x)$;
        \item the input represents a partially computed value of $e'(x)$.
    \end{enumerate}
    The transition function does the following on $y$ in each case
    \begin{enumerate}
        \item if one step of $s_R(\pr_1 y)$ completes the computation of $R(x)$, then check
            if $R(x)$ returns $1$ or $0$, and return the initialization of $e$ or
            $e'$ at $x$ accordingly, and continue in the corresponding sates. If not, then
            return $(x,s_R(\pr_1 y))$, continuing in state (1);
        \item if one step of $s_e(y)$ completes the computation of $e(x)$, then
            return this and halt, otherwise take one step and continue in state (2);
        \item if one step of $s_{e'}(y)$ completes the computation of $e'(x)$, then
            return this and halt, otherwise take one step and continue in state (3).
    \end{enumerate}
    As this returns $e(x)$ when $R(x)=1$ and $e'(x)$ when $R(x)=0$, this machine is
    correct for $f$.
\end{proof}
\begin{theorem}\label{thm:mu-recursion}
    If $R:\nat^k\to\Prop$ has recursive structure, then $\lambda x.\mu y.R(x,y)$ has
    recursive structure.
\end{theorem}
\begin{proof}
    The initialization function initializes $x$ to $\langle x,0,i_R(x,0)\rangle$. The transition
    function does the following on input $\langle x,y,c\rangle$:

    If one step of $R$ applied to $c$ finishes the computation of $R(x,y)$ with
    output $1$, then return $y$ and halt. Otherwise, if one step of $R$ applied to
    $c$ finishes the computation of $(x,y)$ with output~$0$, then return ${\langle
    x,y+1,i_R(x,y+1)\rangle}$ and continue. Finally, if neither of these cases hold,
    then we must continue computing $R$ at $c$, so output $\langle x,y,s_R(c)\rangle$
    and continue.
\end{proof}
\begin{theorem}\label{thm:pr-is-recursive}
    If $f:\nat^{n+2}\to\nat$ and $g:\nat^{n}\to\nat$ have recursive structure, then
    $\recop_{f,g}$ has recursive structure.
\end{theorem}
Thanks to Theorem~\ref{thm:mu-recursion}, the classical proof of
elimination of primitive recursion via G\"{o}del's $\beta$ function also gives this
result.
\begin{proof}
    Define
    \[t(x,y) = \mu z.\big(\beta(z,0) = g(x) \wedge (\forall i<y. \beta(z,i+1) =
    f(x,i,\beta(z,i)))\big).\]
    Then we have
    \[\recop_{f,g}(x,y) = \beta(t(x,y),y).\qedhere\]
\end{proof}
Using the above, we get the most difficult cases of the classical characterization of
Turing computability in terms of $\mu$-recursiveness.
\begin{theorem}
    \begin{enumerate}
        \item All constant functions $\nat^k\to\nat$ have recursive structure;\label{comp:const}
        \item the successor function has recursive structure;\label{comp:succ}
        \item all projections have recursive structure;\label{comp:proj}
        \item if $f:\nat^k\to\nat$ has recursive structure and $g_i:\nat^n\to\nat$ has
            recursive structure for each~${i<k}$, then
            $\lambda x.f(g_0(x),\ldots,g_{k-1}(x))$ has recursive structure \label{comp:comp}
        \item if $f:\nat\pto \nat$ and $g:\nat\pto\nat$ have recursive structure,
            then $g\kcomp f$ has recursive structure;\label{comp:pr}
        \item If $R:\nat^2\to\Prop$ has recursive structure, then $\lambda x.\mu
            y.R(x,y)$ has recursive structure.\label{comp:mu}
    \end{enumerate}
\end{theorem}
\begin{proof}
    Parts (1), (2) and (3) are immediate; part (5) is
    Theorem~\ref{thm:pr-is-recursive}; part (6) is Theorem~\ref{thm:mu-recursion}.
    For (4), we use a similar trick as for part (5): let
    \[t(x) = \mu z.\big(\beta(z,0) = g_0(x) \wedge \ldots \wedge \beta(z,k-1) =
    g_{k-1}(x)\big),\]
    and then we have $(t;f)(x) = \lambda x.f(g_0(x),\ldots,g_{k-1}(x)))$.
\end{proof}

There is one last result to prove using the techniques of this section, before we can
develop recursion theory without worrying about the details of how to encode
computable functions
\begin{theorem}[$S^m_n$ Theorem]
    There is a primitive recursive $\smn:\nat^2\to\nat$ such that for all $e,x,y:\nat$
    we have
    $\{\smn(e,x)\}(y) = \{e\}\langle x,y\rangle$.
\end{theorem}
\begin{proof}
    Let $e$ encode a recursive machine computing $f:\nat\pto\nat$, let $x:\nat$ and
    define $f':\nat\pto\nat$ by $f'(y) = f(\langle x,y\rangle).$
    By manipulating the tree defining $e$, we can get a machine $m$ computing $f'$,
    which we can encode as $e'$. This manipulation is primitive recursive in $e$ and
    $x$, since it amounts to simple operations on finite sequences. Take
    $\smn(e,x) = e'$. Then 
    \[\{\smn(e,x)\}(y) = \{e'\}(y) = \{e\}\langle x,y\rangle.\qedhere\]
\end{proof}
\section{Basic Recursion Theory}\label{computability:recursion-theory}
We are now able to show how to develop basic recursion theory via recursive machines.
We roughly follow the presentation in Odifreddi \cite{Odifreddi1989}; in fact,
the proofs given by Odifreddi for the results stated here translate to our framework
with only minor adjustment.

Rogers' Fixed Point Theorem and Kleene's Second Recursion Theorem can
be proved in the standard way.

\index{recursive!structure|(}
\index{recursive!relation|(}
\begin{theorem}[Rogers' fixed point theorem]
    For any total $f:\nat\to\nat$ with recursive structure, we have $n:\nat$ with
    $ 
    \{n\}=\{f(n)\}.
    $ 
\end{theorem}
\begin{proof}
    Let $g:\nat\to\nat$ be defined by
    \[g\langle x,y\rangle \defeq \{\{x\}(x)\}(y).\]
    Note that we are Kleisli-extending along the first argument of $\{-\}(-)$ to make
    sense of $\{\{x\}(x)\}(y)$.
    As this has recursive structure, it has a code $e_g:\nat$. Define $h(x) =
    \smn(e_g,x)$, and consider $e=\encode(f\comp h)$ and $n=h(e)$. Then we have
    \[\{n\}y = \{h(e)\}y = \{\smn(e_g,e)\} = \{e_g\}\langle e,y\rangle =
    \{\{e\}(e)\}(y) =
    \{f(h(e))\}(y)=\{f(n)\}(y).\qedhere\]
\end{proof}
\begin{theorem}[Kleene's Second Recursion Theorem]
    For any function $f$ with recursive structure we have $p:\nat$ such that
    \[\qPi{y:\nat}\left(\{p\}(y) = f(\langle p,y\rangle)\right).\]
\end{theorem}
\begin{proof}
    Let $g:\nat\to\nat$ be defined by 
    \[g(x) \defeq \smn\langle e_f,x\rangle,\]
    where $e_f:\nat$ is a code of the recursive structure for $f$.
    This function $g$ has primitive recursive structure, so $g$ is a total function
    with recursive structure. By Rogers' fixed point theorem, there is then a $p$ such that
    $\{p\} = \{f(p)\}$. Then,
    \[\{p\}(y) = \{f(p)\}(y) = \{g(p)\}(y) = \{\smn (e_f,p)\}(y) = f(\langle
    p,y\rangle).\qedhere\]
\end{proof}

\begin{definition}
    A  \emph{recursive enumeration} of a subset $A:\nat\to\Prop$ is a recursive machine
    $m$ such that $\ran(\eval(m)) = A$.
\end{definition}

\begin{theorem}\label{thm:rec-re-co-re-struct}
A subset $A$ of $\nat$ has recursive structure if and only if it is complemented and both $A$ and
    its complement have recursive enumerations.
\end{theorem}
\begin{proof}
    It is immediate that recursive sets are complemented, r.e.\ and co-r.e.\ so consider $A$ complemented
    with $f$ witnessing that $A$ is r.e.\ and $g$ witnessing that $A$ is co-r.e.\
    according to the characterization in Theorem~\ref{thm:re-image-struct}. To see that $A$ is
    recursive, consider the ATM with initialization function $i(n) = \langle
    i_f(n),i_g(n)\rangle$ and step function
    \[s(\langle x, y\rangle) \defeq
        \begin{cases}
            \inl \langle n, m\rangle & \text{if }s_f(x) = \inl n, s_g(x)=\inl
            m,\\
            \inr 1 & \text{if }s_f(x) = \inr n,\\
            \inr 0 & \text{if }s_g(x) = \inr m.
        \end{cases}\]
    It is clear that this computes the characteristic function of $A$. We
    require $A$ to be complemented for this case analysis to be exhaustive.
\end{proof}
Equivalence at the type level, rather than logical equivalence, seems impossible, as
we would need to ensure that we respect changes in recursive structure.

Likewise we have the classical characterization of recursively enumerable sets.
\begin{theorem}\label{thm:re-image-struct}
    A subset $A$ of $\nat$ has a recursive enumeration iff it is the domain of a function
    with recursive structure.
\end{theorem}
\begin{proof}
    Let $f$ have recursive structure such that $\isDefined{f(x)}\Leftrightarrow
    A(x)$. Define $f':\nat\pto\nat$ by
    \[f'(x) \defeq (\eta\comp c_x)^\sharp(f x).\]
    $f'$ has recursive structure as the composition of functions with computable
    structure. Moreover we have
    \[f'(x) = (\extent(f(x)),\lambda y.x),\]
    so that $\isDefined{f'(x)} = \isDefined{f(x)}$.
    Then we have that $A(x)\Leftrightarrow \isDefined{f'(x)}$.

    For the other direction, let $g$ have recursive structure $m$ and
    $\big(\qExists{x:\nat}g(x)=\eta n\big) = A(n)$ for all $n$.  Define
    $h:\nat\to\nat+\nat$ by
    $h(\langle k,y\rangle) = \run_k(m,y).$
    We know that this has primitive recursive structure, so define $j:\nat\to\nat$ to be
    \[j(x) \defeq \mu z. h(z) = \inr x.\]
    It is then easy to check that $j$ is defined exactly on the image of $g$.
    Let $\isDefined(j(x))$. Then there are $(k,y)$ such that $\run_k(m,y) = \inr j(x)$.
    So we have $m(y) = j(x)$.
    Conversely, let $x$ be in the image of $g$. That is,
    \[\qExists{y:\nat}\qExists{k:\nat}\run_k(m,y) = x.\]
    But then $\qExists{z:\nat}g(z) = x$, and from this we see $j(x)$ is defined.
\end{proof}
A slightly better classical characterization requires some modification to be true
constructively.

\begin{theorem}
    If $A$ is a subset of $\N$ with $a:A$, then the following are logically equivalent.
    \begin{enumerate}
        \item There is a partial function $f:\nat\pto\nat$ with recursive structure
            whose domain is $A$.
        \item There is a partial function $f:\nat\pto\nat$ with recursive structure
            whose range is $A$.
        \item There is a total function $f:\nat\to\nat$ with recursive structure
            whose range is $A$.
        \item There is a total function $f:\nat\to\nat$ with primitive recursive
            structure whose range is $A$.
    \end{enumerate}
\end{theorem}
\begin{proof}
    $(1\Leftrightarrow 2)$ is Theorem~\ref{thm:re-image-struct}.

    $(2\Rightarrow 3)$ Let $f:\nat\pto\nat$ be computable with image $A$, and let
    $a:A$. Define
    \[g(\langle y,k\rangle) = \begin{cases}
            x & \text{if }\run_k(f,y) = \inr x,\\
            a & \text{otherwise.}
        \end{cases}
        \]

    $(2\Rightarrow 4)$ In fact, the function $g$ above is primitive recursive, as
    a case analysis over a primitive recursive predicate.

    $(4\Rightarrow 3)$ and $(3\Rightarrow 2)$ are obvious.
\end{proof}
\index{recursive!relation|)}
\index{recursive!structure|)}
\section{Discussion}
Except for nuances in the way the results are stated, the results here are the classical ones,
and better discussion of their history and importance can be found in more complete
texts on computability theory. Instead, let's focus on the goals of this chapter. As
stated in the introduction, the goal was not to give a complete constructive account
of computability theory, but to examine the notion of partiality in a univalent
setting, and to compare computability as structure with computability as property in a
setting where we can formally distinguish them. On the first point, we hope the above
gives a satisfying enough account to believe that our notion of partiality is robust
enough to handle computability theory. It is worth remarking on the fact that when
taking partial functions $\nat\pto\nat$ to be functions $\nat\to\Lift(\nat)$, Kleene
equality between partial functions is simply equality.

The distinction between computability as structure and computability as property
justifies the nonstandard and sometimes awkward statement of results: we need to keep
computable functions separate from functions equipped with computability structure.

The only remaining comment is on the use of recursive machines. The aim of
introducing recursive machines was to minimize the number of ``moving parts'', while
capturing the intuitive clarity of Turing machines. Moreover, since minimization is
not a total operation, and the primitive recursive combinators all give total
functions, it felt natural to separate minimization. Once this separate is made,
there is no reason besides history to insist on a minimization operator as the way to
introduce fixed-point recursion.


%% file: chapters/comp-as-prop.tex
\chapter{Computability as property}\label{chapter:comp-as-prop}

In Chapter~\ref{chapter:comp-as-struct}, we developed a notion of functions
\emph{with} computable structure, as those functions which are computed by given
recursive machines. Here we turn to computability as property. Since most of the
results of this section follow by functoriality of truncation from corresponding
results about computability as property, this chapter is quite short.
Sections~\ref{section:pr-functions},~\ref{section:comp-functions},
and~\ref{section:rec-theory} summarize these results.
Section~\ref{section:halting-problem} contains the proof of the undecidability of the
halting problem, and leverages this to define a non-computable partial function,
Turing's diagonal function. As all total functions which we
can define in the empty context (i.e., without additional assumptions) are
computable and constructive intuition suggests that our constructively definable
objects ought to be computable; it is not obvious
that non-computable partial functions ought to be definable. Their existence means
that we cannot simply extend Church's thesis to partial functions in a consistent
way. Instead we must weaken the statement that all partial functions are computable.
We will chose to do so by restricting the class of computable functions using machinery
from Chapter~\ref{chapter:partial-functions}, but we will put that off until
Chapter~\ref{chapter:comp-and-partiality}. In the meantime we give
restriction of the partial functions, which is
explicitly tied to computation (Section~\ref{section:semidecidable}). This
restriction is more concrete than the Rosolini partial functions, since it is tied
explicitly to the definition of computability, but this explicit connection to
computability makes it less suited for our purposes, since one of our goals is an
abstract notion of partial function which could be a surrogate for the computable
partial functions.

\section{Primitive recursive functions}\label{section:pr-functions}
\begin{definition}
    \index{primitive recursive!function|seealso{primitive recursive, structure}}
    \defineopfor{$\isPR$}{isPR}{primitive recursive function}
    The predicate $\isPR_n:(\nat^n\to\nat)\to\univ$ is 
    \[\isPR_n(f)\defeq\trunc{\PrimRec_n(f)}.\]
    We say that $f$ \nameas{is primitive recursive}{primitive recursive!function} when $\isPR(f)$.
\end{definition}
Then the primitive recursive functions are the image of the primitive recursive
combinators under $\app$.
\begin{definition}
    \index{primitive recursive!predicate}
    A predicate $R:\nat^n\to\Prop$ is \emph{primitive recursive} if it has a
    primitive recursive characteristic function, and write $\isPR(R)$ for the type of
    primitive recursive characteristic functions of~$R$.
    \[\isPR(R)\defeq \qSig{f:\nat^n\to\nat}\isPR(f)\times (f(n) =0 \lequiv \neg
    R(n))\times (f(n) = 1\lequiv R(n))\]
\end{definition}
As characteristic functions for predicates are unique when they exist, 
$\isPR(R)$ is a proposition.

As the basic functions of Theorem~\ref{thm:pr-basics} have primitive recursive
structure, we have
\begin{theorem}
    The following are all primitive recursive
    \begin{enumerate}[label=(\roman*)]
        \item the standard ordering $\le$ on $\nat$,
        \item equality on $\nat$,
        \item the addition function $-+-:\nat^2\to\nat$,
        \item the multiplication function $-\cdot-:\nat^2\to\nat$,
        \item the predecessor function,
        \item truncated subtraction,
        \item bounded sums,
        \item bounded products,
        \item $\sgn$,
        \item $\opsgn$,
        \item the factorial function.
    \end{enumerate}
\end{theorem}
Similarly, the primitive recursive predicates are closed under the logical
operations.
\begin{theorem}\label{thm:pr-logic}
    If $P$ and $Q$ are primitive recursive, then so are
$P\wedge Q$, $P\vee Q$, $P\to Q$, $\neg P$,
$\qExists{k\le n}P(k)$ and $\forall(k\le n),P(k)$.
\end{theorem}

\section{Computable functions}\label{section:comp-functions}
\begin{definition}
    \defineopfor{$\isComputable$}{iscomputable}{computable function}
    \index{computable function|textit}
    \index{computable function|seealso{recursive structure}}
    The predicate $\isComputable:(\nat\to\Lift\nat)\to\univ$ is defined by
    \[\isComputable(f)\defeq \trunc{\CompStruct(f)}.\]
    The partial function $f$ is \emph{computable} when $\isComputable(f)$.

    \definesymbolfor{$\Comp$}{comp}{computable function}
    The type $\Comp$ is the type of all computable partial functions:
    \[\Comp\defeq \qSig{f:\nat\to\Lift\nat}\isComputable(f)\]
\end{definition}
Then $\Comp$ is the image of $\RM$ under $\eval$.

The distinction between $\CompStruct(f)$ and $\isComputable(f)$ is necessary. By the
enumeration theorem for primitive recursive functions, we have a map
\[\left(\qSig{f:\nat\pto\nat}\CompStruct(f)\right)\to\nat,\]
given by
$(i,s)\mapsto \langle e_i,e_s\rangle$,
which is easily seen to be an embedding.

\index{Limited Principle of Omniscience!Weak}
On the other hand, we cannot expect there to be an en embedding $F:\Comp\to\nat$.
\begin{theorem}\label{embed-to-wlpo}
    If there is an embedding $F:\Comp\to\nat$, then equality between computable
    functions is decidable. In particular, an embedding $F:\Comp\to\nat$ implies WLPO for
    computable functions.
\end{theorem}
\begin{proof}
    Let $F:\Comp\to\nat$ be an embedding, and let $f,g:\Comp$, so 
    $(f=g)\simeq (F(f)=F(g))$. As $\nat$ has decidable equality, $F(f)=F(g)$
    is decidable.

    Then taking $g$ to be $\lambda x.0$, which we know to be computable, we have that
    it is decidable whether $f(x) = \lambda x.0$.
\end{proof}
We do not expect to be able to decide whether two computable functions are equal, as
then the halting problem would be a decidable predicate. Decidable means complemented
rather than recursive here, but we expect any predicate which can be proved to be
complemented in our system to be recursive. We will re-examine this situation more
closely in Chapter~\ref{chapter:comp-and-partiality}.

A consequence of Theorem~\ref{embed-to-wlpo} is that we also cannot even expect to
have an embedding
\[ \Comp\to\qSig{f:\nat\pto\nat}\CompStruct(f).\]
That is, we have no way of finding a program for a function just by knowing one exists.

However, many basic facts about functions with computable structure also apply to
computable functions.
\begin{theorem}
    If $f$ and $g$ are computable, then so is $g\kcomp f$.
\end{theorem}
\begin{proof}
    By functoriality of truncation (Lemma~\ref{lemma:trunc-func}), we know
    \begin{multline*}
    \big(\CompStruct(f)\to\CompStruct(g)\to\CompStruct(g\kcomp f)\big)\to \\
    \isComputable(f) \to \isComputable(g)\to\isComputable(g\kcomp f),
    \end{multline*}
    and the antecedent of this implication is 
    Theorem~\ref{thm:compstruct-composition}. 
\end{proof}
\begin{theorem}
    If $R:\nat^2\to\Prop$ is primitive recursive, then
    \[\isComputable(\lambda x. \mu y.R(x,y)).\]
\end{theorem}
\begin{proof}
    Functoriality of truncation and Theorem~\ref{lemma:pr-minimization}.
\end{proof}

\begin{definition}
    A predicate $R:\nat^k\to\univ$ is recursive when it has a computable
    characteristic function.
\end{definition}
\begin{theorem}
    If $R:\nat^k\to\univ$ is recursive, and $e,e'$ are computable, then  the function
    $f:\nat^k\to\nat$ defined by
    \[f(x) = \begin{cases}
        e(x) & \text{if }R(x)\\
        e'(x) & \text{if }\neg R(x)
    \end{cases}
     \]
     is also recursive
\end{theorem}
\begin{proof}
Functoriality of truncation from Theorem~\ref{thm:rec-cases}.
\end{proof}
\begin{theorem}
    If $R:\nat^k\to\Prop$ is recursive, then $\lambda x.\mu y.R(x,y)$ is
    computable.
\end{theorem}
\begin{proof}
    Functoriality of truncation from Theorem~\ref{thm:mu-recursion}.
\end{proof}
\begin{theorem}
    If $f:\nat^{n+2}\to\nat$ and $g:\nat^{n}\to\nat$ are computable, then
    $\recop_{f,g}$ is computable.
\end{theorem}
\begin{proof}
    Functoriality of truncation from Theorem~\ref{thm:pr-is-recursive}.
\end{proof}

\section{Recursion theory}\label{section:rec-theory}

Most of the results of Section~\ref{section:normal-form} are results that deal explicitly
with machines, as such, they do not lift directly to computable functions. However,
the theory the results allow can be lifted to computable functions. In particular, we
have $\KleeneT$, $\KleeneU$ and $\{-\}$, which we saw have recursive structure, and
so these functions are computable. From this, we can prove Roger's fixed point
theorem and Kleene's second recursion theorem for computable functions.

\begin{theorem}
    For any total computable function $f:\nat\to\nat$,
    \[\qExists{n:\nat}\{n\}=\{f(n)\}.\]
\end{theorem}
\begin{theorem}
    For any computable function $f:\nat\to\Lift \nat$, 
    \[\qExists{p:\nat}\qPi{y:\nat}\left(\{p\}(y) = f(\langle p,y\rangle)\right).\]
\end{theorem}
These follow by functoriality of truncation from the corresponding version for
recursive machines.

Similarly, we have the characterization of recursive and recursively enumerable sets.
\begin{definition}
    A subset $A$ of $\nat$ is \emph{recursively enumerable} when
    \[\qExists{f:\nat\to\Lift\nat}\isComputable(f)\times \ran(f)=A.\]
  \end{definition}
Notice that being recursively enumerable is a proposition, and it is the truncation
of having recursively enumerable structure.

\begin{theorem}\label{thm:rec-re-co-re-prop}
A subset $A$ of $\nat$ is recursive if and only it is complemented and both $A$ and
    its complement are recursively enumerable.
\end{theorem}
\begin{theorem}\label{thm:re-image-prop}
    A subset of $\nat$ is r.e.\ iff it is the domain of a computable
    function.
  \end{theorem}

\begin{theorem}
    If $A$ is an inhabited subset of $\nat$, then the following are equivalent.
    \begin{enumerate}
        \item $A$ is the domain of a computable partial function.
        \item $A$ is the range of a computable partial function.
        \item $A$ is the range of a computable total function.
        \item $A$ is the image of a primitive recursive function.
    \end{enumerate}
\end{theorem}

\section{The halting problem}\label{section:halting-problem}
Consider the predicate $H:\nat\to\nat\to\univ$ defined by
\[H(e,x) = \isDefined(\{e\}(x)).\]
One of the first results in computability theory is that this predicate is not
recursive. We repeat the proof here. First, let us define the diagonal function

    \begin{eqnarray*}
        d &:& \nat \to \Lift\nat\\
    \extent d(x) & \defeq & (\{x\}(x) = \bot), \\
    \val(d(x))(p) & \defeq & 0.
    \end{eqnarray*}
\begin{lemma}\label{lemma:diagonal}
    The function $d$ is not computable.
\end{lemma}
\begin{proof}
    Since we are trying to prove a proposition, we may assume we have a machine $m$
    computing $d$ with code $e$. Then we have
    \[\{e\}(e) = \eval(m,e) = d(e).\]
    Then we have $\extent \{e\}(e) = \neg(\extent \{e\}(e))$, which is
    impossible.
\end{proof}

\begin{theorem}[Undecidability of the halting problem]
    The predicate $H$ is not recursive.
\end{theorem}
\begin{proof}
    We wish to derive a contradiction from the assumption $\isRec(H)$. As $\emptyset$
    is a proposition, we may assume we have a computable total function
    $f:\nat\to\nat\to\nat$ such that $f(e,x) = 1$ iff $\{e\}(x)$ is defined and
    $f(e,x) = 0$ iff $\{e\}(x)$ is not defined.

    Then $d(x)$ is pointwise equal to the function
    \[
        d(x) = 
       \begin{cases}
           0 & \text{if $f(x,x) = 0$}\\
           \bot & \text{if $f(x,x) = 1$},
       \end{cases}
    \]
    so that $d(x)$ is defined iff $\neg H(x,x)$ and $d(x)$ is undefined iff $H(x,x)$.

    We claim that $d(x)$ is computable. By assumption, $f$ is a total computable
    function. Moreover, we have that $\lambda x.\bot$ is computable: we know that the
    constant functions $\lambda x,y.0$ has computable structure, and this is the
    characteristic function for the relation $\lambda x,y.\ff$. Hence, $\lambda x.\mu y.
    \ff$ has computable structure. But as there is no $y$ such that $\ff$, we must have that
    $\mu y.\ff = \bot$. So then $d$ is a computable case analysis, and has computable
    structure. But this contradicts Lemma~\ref{lemma:diagonal}.
\end{proof}

\section{Semidecidable propositions}\label{section:semidecidable}

The Rosolini partial functions provide an abstract notion of semidecidable
proposition. A more concrete notion is given by restricting our attention to
computable sequences.
\begin{definition}
    \index{semidecision procedure|see{proposition, semidecidable}}
    A \emph{semidecision procedure} for a proposition $P$ is a function 
    $f:\cantor$ with computable structure and at most one $1$ such that $P\simeq
    \lar{f}$. That is,
    \defineopfor{$\SemiDecision$}{semidecision}{proposition!semidecidable}
    \[\SemiDecision(P)\defeq\qSig{f:\cantor}\isProp\lar{f}\times
    \isComputable(f)\times (P\simeq\lar{f}).\]
    A proposition $P$ \nameas{is semidecidable}{proposition!semidecidable}
    when $\SemiDecision(P)$ is inhabited:
    \defineopfor{$\isSemiDecidable$}{isSemiDecidable}{proposition!semidecidable}
    \[\isSemiDecidable(P)\defeq \trunc{\SemiDecision(P)}.\]
\end{definition}

The semidecidable propositions are exactly those which arise as the value
of a program. That is,
\begin{theorem}\label{thm:semidecidable-is-computable}
    For all $A:\univ$, there exists a computable partial function $f:\nat\pto\nat$
    such that $A=\defined{f(0)}$ if and only if $A$ is semidecidable.
    That is,
    \[\isSemiDecidable(A)\simeq \big(\qExists{f:\nat\pto\nat}\isComputable(f)\times
    (A=\defined{f(0)})\big).\]
\end{theorem}
\begin{proof}
    Suppose we are given a function $f:\nat\pto \nat$ with computability structure
    $t$ such that such that $A=\defined{f(0)}$. We may define
    \[\alpha(n)=
        \begin{cases}
            0 & \text{if $\run_n(t,0)=\inr k$ for some $k$, and $\alpha(m)=1$ for
            $m<n$,}\\
            1 & \text{otherwise.}
        \end{cases}\]
    It is easy to see that $\alpha$ is recursive, that $\langle \alpha\rangle$
    is a proposition and that $A=\langle \alpha\rangle$.

    Conversely, suppose we are given $\alpha$ with recursive structure. Consider the constant function 
    $ 
    f(x) = \mu k.(\alpha(k)=0).
    $ 
    which also has recursive structure. Moreover, it is clear that
    $\defined{f(0)} = A$.

    The desired equivalence follows by functoriality of truncation.
\end{proof}

Unfortunately, the semidecidable propositions do not form a dominance without
additional assumptions (recall Section~\ref{section:dominance-choice}, and see
Section~\ref{computability:which}), and moreover,
the semidecision procedures do not even form a structural dominance (recall
Section~\ref{section:univ-dom}). We say more about
this in the discussion below, but the key obstacle is the fact that we use a
computable function rather than a function with recursive structure.
Regardless, we can talk about disciplined maps for the semidecision
procedures, but since the semidecision procedures are not closed under $\Sigma$, we
cannot apply Theorem~\ref{disciplined-maps:compose} to determine that the
\emph{semidecidably-disciplined maps} compose.

\section{Discussion}
Chapter~\ref{chapter:comp-as-struct} contains the proofs one sees in a classical
course on computability theory, while this chapter contains the results. This split
is a result of handling structure and property separately. The careful distinction, and the
resulting awkward phrasing in Chapter~\ref{chapter:comp-as-struct} may seem like
unnecessary work. The discussion of decidable equality for computable functions in
Section~\ref{section:comp-functions} is the first hint that this distinction is
actually meaningful. Chapter~\ref{chapter:comp-and-partiality} provides more
evidence that there is technical value in using a framework that allows a distinction
between computability as structure and computability as property.

The statement that every partial function
$\N \to \Lift \N$ is computable is false, as the diagonal function
$d:\nat\to\Lift\nat$ is not computable. This is the technical reason why we
must restrict the available partial functions. If we want to make any sense of the
claim ``all partial functions are computable'', then we must discount partial
functions like the above from consideration.

The Rosolini propositions can be seen as an abstract version of the semidecidable
propositions, ignoring computability. The restriction to
semidecidable propositions not only fails to form a dominance, but does not even give
a structural dominance without countable choice. The immediate barrier is that a
semidecision procedure comes with a
computable function, not a recursive machine. Therefore, we cannot lift the
dependent composition used in Lemma~\ref{lemma:rs-dom} to
the semidecision procedures, unless we use countable choice.

The obvious solution is to use a recursive machine (or a function with recursive structure)
in place of a computable function. Indeed, this resolves the issue of choice.
However, we will see in Chapter~\ref{chapter:comp-and-partiality}, and we saw already
in Theorem~\ref{embed-to-wlpo} that computable functions and functions with recursive
structure are not interchangeable. In particular, the developments in
Section~\ref{computability:which} are more natural with the above notion of
semidecision procedure. Since for our purposes the notion of semidecision procedure
is an intermediate notion between Rosolini structure and computability, the above
notion is more useful. A more focused analysis of semidecidable propositions may find
tracking explicit structure in semidecision procedures to be useful.


%% file: chapters/comp-and-partiality.tex
\chapter{Computability and partiality}\label{chapter:comp-and-partiality}
There are three basic threads running through this thesis: restricted notions of
partiality; the distinction between structure and property; and the notion of
computability. Here we tie these three threads together.
Section~\ref{computability:dominances} relates partiality and computability by
showing that all computable functions must be Rosolini partial functions.
Section~\ref{computability:church} expands on a remark in
Chapter~\ref{chapter:comp-as-prop} concerning the embedding from recursive machines
to $\nat$, and resolves the two paradoxes set forth in the introduction. Finally,
Section~\ref{computability:which} gives an explicit conjecture concerning which partial
functions can be computable, tying together the material of Parts~\ref{part:two}
and~\ref{part:three}.
\section{Semidecidable propositions and disciplined maps}\label{computability:dominances}
\index{proposition!semidecidable|(}
\index{Rosolini proposition|(}
\index{disciplined map|(}

We saw in Section~\ref{section:halting-problem} that there exists non-computable
partial functions $\nat\pto\nat$. We spent Chapter~\ref{chapter:partial-functions}
foreshadowing this by looking at restrictions of the set of all partial functions. In
particular, we examined the set of Rosolini propositions. In computational contexts,
the Rosolini propositions correspond to the semidecidable propositions.

Indeed, all computable functions are valued in Rosolini propositions. Recall the relation $R_m$ from Section~\ref{section:recursive-machines} used in the
definition of $\eval(m)$.
\begin{theorem}\label{thm:kleene-rosolini}
    For any machine $m:\RM$ and any $x,y:\nat$, the types $R_{m}(x,y)$ and
    $\qSig{y:\nat}R_{m}(x,y)$ are Rosolini propositions.
\end{theorem}
\begin{proof}
    For the first type, let $g:\nat+\nat\to 2$ be the function
    which takes value $1$ on $\inr y$ and $0$ otherwise.
    Define $\alpha_k \defeq g(\run_k(m,x))$. We have
    \[(\alpha_k = 1)\simeq (g(\run_k(m,x)) = 1);\]
    i.e., $\alpha$ takes the value $1$ exactly when $\run_k(m,x)=\inr y$,
    so $R_{m}(x,y)\simeq \qExists{k:\nat}\alpha_k=1$.
    Finally, we truncate $\alpha$, to get the sequence which takes value $1$ only at
    the first location that $\alpha$ does.

    For the second type, let $h:\nat+\nat\to 2$ be the function
    which is $0$ on $\inl k$ and $1$ on $\inr k$, 
    and again take $\alpha_k = h(\run_k(m, x))$. We have that 
    \[(\alpha_k=1) \simeq (\qSig{y:\nat}\run_k(m, x) = \inr y).\]
    Summing over all $k:\nat$, rearranging and truncating takes us to
    \[(\qExists{k:\nat}\alpha_k=1) \simeq
    (\qExists{y:\nat}\qSig{k:\nat}\run_k(m,x) = \inr y).\]
    By Lemma~\ref{lemma:sigmatrunc-is-exists}, we then have
    \[(\qExists{k:\nat}\alpha_k=1) \simeq
    \qSig{y:\nat}\qExists{k:\nat}\run_k(m,x) = \inr y.\qedhere\]
\end{proof}
\begin{corollary}
    If $f:\nat\pto\nat$ is computable, then $f(n)$ is a Rosolini partial element for
    all~$n:\nat$.
\end{corollary}
\begin{proof}
    Since we are trying to prove a proposition, we may assume we have some $m:\RM$
    computing~$f$. We know that the extent of $f(n)$ is equivalent to
    $\qSig{y:\nat}R_m(x,y)$, so by Theorem~\ref{thm:kleene-rosolini}, $f(n)$ is
    Rosolini.
\end{proof}
Taking this one step farther, we have
\begin{theorem}
    Every computable partial function $f:\nat\pto\nat$ is disciplined with respect to
    the Rosolini propositions.
    \[\isComputable(f)\to\isDis(f).\]
\end{theorem}
\begin{proof}
    By functoriality of truncation, it is enough to show that for any recursive
    machine ${m:\RM}$, there is a function $g:\nat\to\Lift_{\RS}(\nat)$ such
    that $\tame(g)=\eval(m)$. This function $g$ is given by $\run^k(m)$. We have the
    predicate
    \[p(n)\defeq \lambda k. \text{$k$ is least such that }\defined{\run^k(n)}.\]
    Then $p(n)$ defines a Rosolini structure, and if $w:(p(n)(k)=1)$, then $\run^k(n) =
    \inr m_w$ for some $m_w:\nat$, and so we can take $g:\nat\to\Lift_{\RS}(\nat)$ to
    be the function
    \[g(n) \defeq \big(p(n),-,\lambda w.m_w\big).\qedhere\]
\end{proof}
In fact, the functions $g$ and $h$ used in the proof of
Theorem~\ref{thm:kleene-rosolini} are computable, and so we may strengthen the above
results.
\begin{theorem}
    For any machine $m:\RM$ and any $x,y:\nat$, the types $R_m(x,y)$ and
    $\qSig{y:\nat}R_m(x,y)$ are semidecidable propositions.
\end{theorem}
\begin{corollary}\label{computable-is-disciplined}
    If $f:\nat\pto\nat$ is computable, then $f$ is disciplined with respect to the
    semidecidable propositions.
\end{corollary}
We examine the converse of this last result later in the chapter.
\index{proposition!semidecidable|)}
\index{Rosolini proposition|)}
\index{disciplined map|)}

\section{Which total functions can be computable?}\label{computability:church}
\index{Church's Thesis|(}
Consider the following two versions of ``Church's Thesis'':
\[\qPi{f:\nat\to\nat}\isComputable(\eta\comp f),\]
and
\[\qPi{f:\nat\to\nat}\CompStruct(\eta\comp f).\]
The first is consistent, with the effective topos as a model. 
The second, however, is false~\cite{troelstra:1977,Beeson1980}. Let us call the first
\emph{Church's thesis} and
the second \emph{strong Church's thesis}. We will follow the structure of Troelstra's
proof to show that Strong Church's thesis is false. 
    The proof proceeds as follows:
   \begin{itemize}
       \item If strong Church's thesis holds,
        then $\nat\to\nat$ embeds into $\nat$, and so has decidable
           equality.
       \item if $\nat\to\nat$ has decidable equality, then there is a function
           $H:\nat\to\bool$ deciding the complement of the halting problem.
       \item if strong Church's thesis holds, then $H$ is computable. Hence, we can recursively
           decide the complement of the halting problem.
       \item the complement of the halting problem is not recursively decidable.
   \end{itemize}

    We break the proof into lemmas below.
   \begin{lemma}
       Assuming strong Church's Thesis, there is an embedding
       $\epsilon:(\nat\to\nat)\to\nat$ such that for any $f$,
       \[\{\epsilon(f)\} = \eta\comp f.\]
   \end{lemma}
  \begin{proof}
      Given a witness $c$ of Church's thesis we get a $c':(\nat\to\nat)\to\RM$ as
      $c'\defeq\pi_0\comp c$. By definition, we have
      \[\eval(c'(f))(n) = (\eta\comp f)(n),\]
      so that $\{\encode(c'(f))\} = \eta\comp f$. So we may take
      \[\epsilon(f) \defeq \encode(c'(f)).\]
      Now if $\epsilon(f) = \epsilon(g)$, then
      \[\eta\comp f = \{\epsilon(f)\} = \{\epsilon(g)\} = \eta\comp g.\]
      Finally, $\eta$ is an embedding, so $f=g$.
  \end{proof}
  \index{Limited Principle of Omniscience!Weak}
 \begin{corollary}\label{ct-to-dec-eq}
     If strong Church's thesis holds, then $\nat\to\nat$ has decidable equality.
     That is, strong Church's thesis implies WLPO.
 \end{corollary}
\begin{proof}
    Assuming strong Church's thesis, $(\nat\to\nat)$ embeds into a type with decidable equality.
\end{proof}
\begin{lemma}\label{dec-eq-to-halting}
    If $\nat\to\nat$ has decidable equality, there is a function $H:\nat\to\nat$
    such that $H(x) = 0$ iff $\{x\}(x)$ is
    undefined.
\end{lemma}
\begin{proof}
    Fix $x:\nat$. Define $r_x:\nat\to\nat$ as
    \[r_x(n) = \left\{
        \begin{array}{ll}
            1 & \text{if $\run_n(\{x\},x) = \inr(y)$ for some $y$}\\
            0 & \text{if $\run_n(\{x\},x) = \inl(k)$ for some $k$}.
        \end{array}\right.\]
    That is, $r_x(n)=1$ iff $\{x\}(x)$ returns in under $n$ steps.

    By our assumption that $\nat\to\nat$ has decidable equality, we have 
    \[(r_x = \lambda k.0) + \neg(r_x = \lambda k.0)\]
    Note that $r_x = \lambda k.0 \Leftrightarrow \{x\}(x) = \bot$.

    Define $H(x):\nat\to\nat$ by
    \[H(x) = \left\{
        \begin{array}{ll}
            1 & \text{if $r_x \neq \lambda n.0$}\\
            0 & \text{if $r_x = \lambda n.0$}.
        \end{array}\right.\]
    We have 
    \[H(x) = 0\Leftrightarrow r_x = \lambda k.0 \Leftrightarrow \{x\}(x) =
    \bot.\qedhere\]
\end{proof}
\begin{theorem}[Troelstra~\cite{troelstra:1977}]\label{not-strong-ct}
    Strong Church's thesis is false.
\end{theorem}
\begin{proof}
    If strong Church's Thesis holds, by Lemma~\ref{ct-to-dec-eq} we have that
    $\nat\to\nat$ has decidable equality, and so by Lemma~\ref{dec-eq-to-halting} we
    have a computable  $H:\nat\to\nat$ such that $H(x)=0$ iff $\{x\}(x)$ is
    undefined. But we saw in Section~\ref{section:halting-problem} that this
    function cannot be computable.
\end{proof}
As any reasonable notion of partial function includes all total functions, we cannot
have a version of strong Church's thesis for partial functions. Then in particular,
if $\dd$ is any set of propositions containing $\unittype$, we have
    \[\neg\qPi{f:\nat\to\Lift_{\dd}(\nat)}\CompStruct(f).\]
However, the weaker form $\qPi{f:\nat\to\Lift_{\dd}(\nat)}\isComputable(f)$
is consistent for $\dd=\isContr$. The above proof hinges on the fact that there is an
embedding from the type of functions with recursive structure to $\nat$. In the
above, we extend this to get an embedding from $(\nat\to\nat)$ to $\nat$. Recall that
in Section~\ref{section:comp-functions} we argued that we cannot expect an embedding
from the type of all computable functions to $\nat$. Indeed, while this is true
classically, it is contradicted by Church's thesis.
\begin{theorem}\label{weak-ct-embedding}
    If Church's thesis holds, then there is no embedding
    \[\big(\qSig{f:\nat\pto\nat}\isComputable(f) \big)\to\nat.\]
\end{theorem}
\begin{proof}
    We essentially follow the argument for Theorem~\ref{not-strong-ct}.
    Church's thesis gives us an embedding $(\nat\to\nat)\to
    \qSig{f:\nat\pto\nat}\isComputable(f)$. Extending this along the hypothetical embedding $e$ gives
    an embedding $(\nat\to\nat)\to\nat$. As a subtype of a type with decidable
    equality, $\nat\to\nat$ has decidable equality. Then, by
    Lemma~\ref{dec-eq-to-halting}, we have a function $H:\nat\to\nat$ deciding the halting
    problem. By Church's thesis, $H$ is computable, which is impossible.
\end{proof}
In the introduction, we mentioned four facts about computability theory arising from
topos theoretic and type theoretic intuition.
\begin{enumerate}[label=C\arabic*]
    \item It is \emph{consistent} that all total functions $\nat\to \nat$ are computable.\label{c:wct}
    \item It is \emph{false} that all total functions $\nat\to\nat$ are computable.\label{c:sct}
    \item There is an embedding from the class of computable functions to the
        natural numbers.\label{c:we}
    \item It cannot be proved that there is an embedding from the class of computable
        functions to the natural numbers.\label{c:se}
\end{enumerate}
\ref{c:sct} is Theorem~\ref{not-strong-ct}, while \ref{c:we} is an easy consequence
of Theorem~\ref{knf}. \ref{c:wct} and \ref{c:se} correspond to the incompatible principles in
Theorem~\ref{weak-ct-embedding}. We know that Church's thesis holds in the
effective topos, and so is consistent with pure MLTT (but see the discussion below
about univalent mathematics), and so we have strong reason to expect \ref{c:wct} and
\ref{c:se} to hold. The conflict between the topos-theoretic and type-theoretic facts
has evaporated: \ref{c:sct} and \ref{c:we}
refer to computability as structure, while \ref{c:wct} and \ref{c:se} refer to
computability as property.

This is not an artefact of univalent mathematics: Indeed, Troelstra's argument was
originally used to show that function extensionality, Church's thesis and the axiom
of choice are mutually incompatible over $\mathrm{HA}^\omega$. Moreover, in the
effective topos, with the Mitchell-Benabou language, all total functions are
computable, and there is no embedding from the class of computable functions to
$\nat$, but in the effective topos as a model of MLTT, it is false that all total
functions are computable, and there is an embedding from the class of computable
functions to $\nat$.

\section{Which partial functions can be computable}\label{computability:which}
\index{Rosolini proposition|(}
\index{proposition!semidecidable|(}
\index{choice!countable|(}
We have spent several chapters now circling around claims about which partial
functions can be computable. It is time to approach the question directly. We have so
far presented two sets of propositions: the Rosolini propositions, which for this
section we call $\RR$, and the semidecidable propositions which for this section we
call $\SD$. For each of these, two notions of partial function:
the $\RR$- and $\SD$-partial functions, and the $\RR$- and $\SD$-disciplined maps.
One may expect that the $\SD$-disciplined maps are exactly the computable
functions---that the converse of Corollary~\ref{computable-is-disciplined} holds, but
it seems that this result needs both countable choice and Church's thesis to prove.
Before presenting the result explicitly (Theorem~\ref{thm:punchline}), let's
examine the effect of Church's thesis on Rosolini structures.

\begin{theorem}\label{ros-is-sd}
    Church's thesis implies that semidecision functions and Rosolini structures
    coincide. That is, assuming Church's thesis we have
    \[\qPi{P:\univ}\SemiDecision(P)\simeq \RosStruct(P)\]
\end{theorem}
\begin{proof}
    We have that
    \[\SemiDecision(P)\defeq\qSig{f:\cantor}\isProp\lar{f}\times
    \isComputable(f)\times (P\simeq\lar{f}).\]
    By Church's thesis, we have that $\isComputable(f)$ is true for all
    $f:\nat\to\bool$, and so we have
    \[\SemiDecision(P)\simeq \qSig{f:\cantor}\isProp\lar{f}\times (P\simeq\lar{f}),\]
    and the latter is the definition of Rosolini structures on $P$.
\end{proof}
That is, Church's thesis collapses $\RR$ and $\SD$. So we have
\begin{corollary}\label{cor:semidec-char}
    Assuming Church's thesis,
    \begin{enumerate}[label=(\alph*)]
        \item Rosolini propositions and semidecidable propositions coincide;
        \item The Rosolini structure lifting and the semidecision function lifting
            coincide;
        \item Rosolini partial functions and semidecidable partial functions coincide;
        \item Rosolini-disciplined maps and semidecidable-disciplined maps coincide.
    \end{enumerate}
\end{corollary}
Recall Theorem~\ref{thm:disc-factoring}, which says that under countable choice, the
disciplined maps coincide with the partial functions~$\nat\to\Lift_{\isRos}\nat$.
\begin{corollary}\label{church-identification}
    Assuming both countable choice and Church's thesis, the following all
    coincide
    \begin{enumerate}[label=(\alph*)]
        \item the Rosolini partial functions, $\nat\to\Lift_{\RR}\nat$;
        \item the Rosolini-disciplined maps, $\Dis_{\RR}(\nat,\nat)$;
        \item the semidecidable partial functions, $\nat\to\Lift_{\SD}\nat$;
        \item the semidecidable-disciplined maps, $\Dis_{\SD}(\nat,\nat)$;
    \end{enumerate}
\end{corollary}
In short, Church's thesis and countable choice together unite all the restricted
notions of partial function we have considered so far. There is one notion still
missing: the computable partial functions. Indeed, countable choice and Church's
thesis together imply that the Rosolini and computable partial functions coincide.
\begin{theorem}\label{thm:punchline}
    Assuming countable choice and Church's thesis, any function $\nat\pto\nat$
    which is valued in Rosolini propositions is computable.
\end{theorem}
\begin{proof}
    By Corollary~\ref{church-identification}, it is enough to show (under the two assumptions) that any
    semidecidable-disciplined map is computable.
    Since
    we are trying to prove a proposition, we may assume we are given an explicit
    $f:\nat\to\Lift_{\SD}(\nat)$, and show that $\tame(f)$ is computable.

    Decomposing $f$ and applying Theorem~\ref{thm:semidecidable-is-computable}, we
    have $F:\nat\to\univ$ such that
    \[\qPi{n:\nat}\Big(\qExists{g:\nat\to\nat}\isComputable(g)\times
    (\defined{F(n)}=\defined{g(0)})\Big)\times \big(F(n)\to \nat\big).\]
    By countable choice, then, we have
    \[\trunc{\qPi{n:\nat}\Big(\qSig{g:\nat\to\nat}\CompStruct(g)\times
    (\defined{F(n)}=\defined{g(0)})\Big)\times \big(F(n)\to \nat\big)}.\]
    Again, we are trying to prove a proposition, so we may assume we are given some
    $H$ of type
    \[H: \qPi{n:\nat}\Big(\qSig{g:\nat\to\nat}\CompStruct(g)\times
    (\defined{F(n)}=\defined{g(0)})\Big)\times \big(F(n)\to \nat\big).\]
    We can again reorganize our type, to get a function $M:\nat\to\nat$ such that
    \[\qPi{n:\nat}\defined{F(n)} = \defined{\{M(n)\}(0)} \times F(n)\to \nat\]
    Now we can consider $h:\nat\pto\nat$ whose extent is given by $F(n)$ and whose
    value is given by the map $f_n:F(n)\to\nat$ which is the last component of $H(n)$. By
    definition, this is pointwise equal to $\tame(f)$, so we need a recursive machine
    computing $h$. For each $n:\nat$,  let $m_n$ be the machine for $g$ where
    $\defined{F(n)}=\defined{g(0)}$. Define the machine $m$ which acts as follows:

    To initialize $m$ on input $n$, initialize $m_n$ on input 0, and ensure we are in state $n$.
    In state $n$, to transition, use the transition function of $m_n$ and stay in
    state $n$. If we have halted in state $n$, then we know $\defined{g(0)}$, and so
    we have $w:\defined{F(n)}$, and we can return $f_n(w)$.

    The only question is whether we can computably access the machines $m_n$, but
    each $m_n$ is coded by $M(n)$, and by Church's thesis $M$ is computable.
\end{proof}
This gives us an extension of Corollary~\ref{cor:semidec-char}.
\begin{theorem}
    Assuming both countable choice and Church's thesis, the following all
    coincide
    \begin{enumerate}[label=(\alph*)]
        \item the Rosolini partial functions, $\nat\to\Lift_{\RR}\nat$;
        \item the Rosolini-disciplined maps, $\Dis_{\RR}(\nat,\nat)$;
        \item the semidecidable partial functions, $\nat\to\Lift_{\SD}\nat$;
        \item the semidecidable-disciplined maps, $\Dis_{\SD}(\nat,\nat)$;
        \item the computable partial function $\qSig{f:\nat\pto\nat}\isComputable(f)$.
    \end{enumerate}
\end{theorem}
\index{Rosolini proposition|)}
\index{proposition!semidecidable|)}
\index{choice!countable|)}
\index{Church's Thesis|)}

\section{Discussion}\label{section:comp-part-discussion}
We know that countable choice and Church's thesis hold in the effective topos. While
the effective topos satisfies proposition extensionality, propositions are
handled differently in univalent mathematics than in topos logic. Nevertheless,
we have reason to believe the following
\begin{conjecture}
    It is consistent that all Rosolini-disciplined maps are computable.
\end{conjecture}
To prove this, it is enough to give a model of MLTT with proposition and function
extensionality in the univalent style which validates countable choice and Church's
thesis. While such a model would actually give the stronger result that it is
consistent that Rosolini
partial functions are computable, we have intentionally chosen to state the
weaker conjecture. The Rosolini-disciplined maps can be shown to compose even in
the absence of choice principles, and so can be used as a surrogate version of the
computable partial functions, even in settings where we do not have countable choice,
so long as such settings are compatible with Church's thesis.

The distinction between the stronger and weaker form of the above conjecture speak
towards our goal of understanding how the computable functions fit into a
constructive framework. Our results on this point can be summarized by looking at
possible ways of collapsing the following diagram:

\index{Kripke's Schema}
\begin{center}
\begin{tikzpicture}
    \node (center) {};
    \node (ros) [above=of center] {$\nat\pto_{\RR}\nat$};
    \node (all) [above=of ros, rectangle,draw,solid]    {$\nat\pto\nat$};
    \node (sd)  [right=of center]  {$\nat\pto_{\SD}\nat$};
    \node (dis) [left=of center,rectangle,draw,solid]   {$\Dis_{\RR}(\nat,\nat)$};
    \node (disS) [below=of center]   {$\Dis_{\SD}(\nat,\nat)$};
    \node (comp) [below=of disS,rectangle,draw,solid]     {$\Comp$};
    \draw [thick,dash dot]   (comp.north) -- (disS.south) node
                [midway, label=right:{\tiny$\CC+\CT$}] {};
    \draw [thick,dashed]   ([xshift=-3pt]disS.north) -- (dis.south east) node
                [midway, label=left:{\tiny$\CT$}] {};
    \draw [thick,dotted]   ([xshift=3pt]disS.north) -- (sd.south west) node
                [midway, label=right:{\tiny$\CC$}] {};
    \draw [thick,dotted]   (dis.north east) -- ([xshift=-3pt] ros.south) node
                [midway, label=left:{\tiny$\CC$}] {};
    \draw [thick,dashed]    (sd.north west)  -- ([xshift=3pt] ros.south) node
                [midway, label=right:{\tiny$\CT$}] {};
    \draw [solid]   (ros.north) -- (all.south) node
                [midway, label=right:{\tiny$\KS$}] {};
\end{tikzpicture}
\end{center}
Lower types embed into higher types and the boxed types can be shown to be closed
under composition with no choice principles; the dashed lines collapse under Church's thesis,
while the dotted lines collapse under countable choice. Moreover, the solid line
collapses under Kripke's schema.

Since we know not all partial functions are
computable, a corollary of Theorem~\ref{thm:punchline} is that countable Choice, Church's thesis
and Kripke's schema are incompatible. The stronger result that Kripke's Schema and
Church's thesis are incompatible is already
known, but the proof in~\cite{troelstra1988constructivism} goes via a different
route.

The argument used in Theorem~\ref{not-strong-ct} was first given to show that
function extensionality, Church's thesis and the axiom of choice are mutually
inconsistent over $\mathrm{HA}^\omega$. Note that under choice, Church's thesis
is the truncation of strong Church's thesis.


%% file: chapters/future.tex
\chapter*{Further work}\label{chapter:future}
\addcontentsline{toc}{chapter}{Further work}  
There are two clear lines of research from here: The first is to resolve the
model-theoretic aspects in order to prove the conjecture, and the second is to
develop constructive computability theory at higher types.

\section*{A model validating Church's thesis}
\addcontentsline{toc}{section}{A model validating Church's thesis}  

A naive sketch of the proof goes as follows: We've already shown that
Church's Thesis and Countable choice imply that the computable partial functions are
exactly the Rosolini disciplined functions. Moreover, we know that the effective topos
models both countable choice and Church's Thesis, and so in the effective topos the
Rosolini disciplined maps coincide with the partial computable functions.

Unfortunately, the notion of proposition as element of the subobject classifier, and
the notion of proposition as subsingleton type do not quite line up: in fact, it is
not the case that all subsingletons of the effective topos are equal. In other words,
proposition extensionality as stated in
Definition~\ref{ax:propext} is not meaningful in the effective topos. There seem to be two
ideas one can use to get around this: one is essentially semantic, and the other syntactic.

On the semantic side, some work has been done By Ian Orton and Andy
Pitts~\cite{orton2017cubical} on using partial elements to give a version of the
interval and the Kan filling operation used in cubical type theory. While we don't
expect arbitrary toposes to satisfy the univalence axiom (or even proposition
extensionality), these ideas could be used to rework the effective topos into a
univalent form. Work has already been done on a version of \emph{cubical
realizability}, in work by Harper and others on computational higher type
theory~\cite{computationalttI,computationalttIV}; ideas from here, or the stack
models~\cite{coquand2017stack} will likely be of value.

On the syntactic side, it may be possible to interpret our type theory in one
that is modeled by the effective topos. Since we do not use full univalence in the
above, and function extensionality is true in the effective topos, the sticking point
is proposition extensionality. However, since propositions in a topos are
subobjects of the (specified) terminal object, they do satisfy proposition
extensionality. That is, the only obstacle is that the notions of proposition do not
coincide. It should be possible to give an intermediate
type theory with a type $\Prop$ of propositions, along with a decomposition of the
truncation operator $\trunc{-}:\univ\to\univ$ through this. It is known that
truncation forms an idempotent monad~\cite{hottbook,rijke2017modalities}, so we could
add to a type theory such as CiC rules corresponding to operators $|-|:\univ\to\Prop$
and $\Elem:\Prop\to\univ$ decomposing this monad as an adjunction. In particular, if
$\Prop$ is extensional (that is, $P\leftrightarrow Q$ then $P=Q$ for $P,Q:\Prop$),
then $\Elem$ would pick a canonical subsingleton for each equivalence class of
subsingletons. The difficulties in this approach arise in two places:
\begin{enumerate}
    \item First, some care would be required to ensure we correctly capture the
        relationship between types that truncation represents. For example, when
        $P(x)$ is a subsingleton for all $x:X$ we want
        $\Elem(\qAll{x:X}|P(x)|)$ to be equivalent to $\qPi{x:X}P(x)$. Effectively
        characterizing such properties may be non-trivial.
    \item Second, we need some sort of conservativity result, but it's not clear how
        to state such a conservativity result for this situation.
\end{enumerate}
There's an additional benefit of the second approach: For mathematicians working in a
univalent setting, there's a general
impression that the univalent approach restricted to the level of sets and
propositions is ultimately little more than a more convenient internal language for a
topos. A type theory bridging the gap between univalent type theory and a type theory
modeled by general toposes would give technical support to this impression.

\section*{Higher-type computation}
\addcontentsline{toc}{section}{Higher-type computation}  
Despite the close ties between the development of constructive theories and
computation at higher types, there does not seem to be a coherent constructive
development of higher-type computability theory.

The starting point of higher-order computability is PCF, and indeed it is not
difficult to give the operational semantics for PCF constructively: 
Recall that PCF is the simply type lambda calculus with base type $\iota$
(representing the naturals, with a zero-test $\ifz$ and predecessor $\predtt$) and a fixed point combinator
$\YY_\sigma:(\sigma\to\sigma)\to\sigma$ for each type
$\sigma$.
Specifically, we have the term grammar
\begin{align*}
    M\mathrel{{::}{=}}&\;\; x \mid \lambda x. M \mid M(M) \mid \YY_{\sigma}(M) \mid \\
    \qquad& \numeral{m} \mid \succtt(M) \mid \predtt(M) \mid \ifz(M,M,M)
\end{align*}
with the expected typing rules where $\ifz$ is first-order. Note that we take the
constants to be constructors rather than combinators. We have

\vspace{10pt}
\begin{minipage}{0.75\textwidth}

    $\numeral{n}\Downarrow \numeral{n}$\qquad $\lambda x.M \Downarrow \lambda
    x.M$ \qquad $x \Downarrow x$

    \vspace{0.8em}

    \AxiomC{$M \Downarrow \numeral{n}$}\UnaryInfC{$\succtt M \Downarrow
    \numeral{n+1}$}\DisplayProof

    \AxiomC{$M \Downarrow \numeral{0}$}\UnaryInfC{$\predtt M \Downarrow
    \numeral{0}$}\DisplayProof\qquad
    \AxiomC{$M \Downarrow \numeral{n+1}$}\UnaryInfC{$\predtt M \Downarrow
    \numeral{n}$}\DisplayProof

    \vspace{0.8em}

    \AxiomC{$L \Downarrow \overline{0}$}
    \AxiomC{$M \Downarrow v$}
    \BinaryInfC{$\ifz LMN \Downarrow v$}\DisplayProof
    \qquad
    \AxiomC{$L \Downarrow \overline{n+1}$}
    \AxiomC{$N \Downarrow v$}
    \BinaryInfC{$\ifz LMN \Downarrow v$}\DisplayProof

    \vspace{0.8em}

    \AxiomC{$M (\YY M) \Downarrow v$}\UnaryInfC{$\YY M \Downarrow v$}
    \DisplayProof\qquad

    \vspace{0.8em}

    \AxiomC{$L \Downarrow \lambda x.N$}
    \AxiomC{$N[M/X] \Downarrow v$}
    \BinaryInfC{$ LM \Downarrow v$}\DisplayProof

\end{minipage}
\vspace{10pt}

The Scott model has interpretation at base type $\interpret{\iota} = \nat_\bot$, the
dcpo of the lifted naturals. Constructively, we may replace this
with our lifting. Note that we cannot constructively use the poset given by adding an
element $\bot$ below every element of $\nat$ (with the discrete ordering), since this
type $\nat_\bot$ can not be shown to be a dcpo constructively. In fact, the
two-element poset with $0\le 1$  is a dcpo precisely if excluded middle holds: fix a
proposition $P$ and consider the family $u:P+\unittype\to\bool$ defined by $u(\inl p)
= 1$ and $u(\inr \star) = 0$. Then this is directed as the only possible non-equal
pairs of elements are $\inl p$ and $\inr \star$, for some $p:P$. In this case, we have $u(\inr
\star) < u(\inl p)$. Then if the supremum of $u$ is $1$, then $p$ holds, and if it is
$0$, then $p$ does not hold, but as $u:\bool$ one of these cases must hold.

To see an example of how a constructive development of higher-type computability
differs from the classical development, we sketch an example based
on the metric model of PCF in~\cite{escardo1998metric}. There, Escard\'o presents a
model of PCF in complete ultrametric spaces where we use the metric to count
recursive unfoldings---two points are closer if we need to unfold more applications
of the fixed-point combinator to distinguish them. We can define a logical relation
(in fact, a partial equivalence relation at each type)
of extensional equality between elements of this model, and then we can use a logical
relation between the metric model and the Scott model to show that the Scott model is
a subquotient (the \emph{extensional collapse}) of the ultrametric model.

The situation is somehow reversed in the classical case, since we must use
$\Lift(\nat)$ instead of $\nat_\bot$ to interpret the base type.

\section*{PCF and ultrametric spaces}
\addcontentsline{toc}{section}{PCF and ultrametric spaces}  
An \emph{ultrametric space} $M$ is a metric space where the distance function $d$
satisfies the \emph{strong triangle inequality}, 
\[d(x,z) \le \max\{d(x,y),d(y,z)\}.\]
Real numbers are not essential in the development. We instead avoid the technical
issues that arise from working with the reals by presenting
ultrametric spaces via a sequence of equivalence relations.
\begin{definition}
    An \emph{ultrametric space} is a set $M$ equipped with countably many equivalence
    relations $\{=_n\}_{n:\nat}$ such that
    \begin{enumerate}
        \item $=_n$ is an equivalence relation for all $n$;
        \item $=_0$ is the total relation;
        \item if $x=_{n+1}y$, then $x=_{n}y$;
        \item if $x=_{n}y$ for all $n$, then $x=y$.
    \end{enumerate}
\end{definition}
The computational intuition is that $x=_n y$ when it takes at least $n$ steps to
distinguish $x$ and $y$. By setting $d(x,y)=\inf\{2^{-k}\mid x=_k y\}$, we get an
ultrametric space in the classical sense from an ultrametric space as above.

We can define Cauchy sequences, limits and Cauchy completeness in the expected way,
and take a function between ultrametric spaces to to be \emph{non-expansive}
if $x=_n y\Rightarrow f(x) =_n f(y)$, and \emph{contractive} if $x=_n y\Rightarrow f(x)
    =_{n+1} f(y)$.

Then we can prove that the category of complete ultrametric spaces with
non-expansive maps is cartesian closed, and take $X\to Y$ be the (complete)
ultrametric space of non-expansive maps.  We can also show that $\Delay(X)$ is a
complete ultrametric space for any type $X$, and give a version of the Banach
fixed-point theorem. From this, we can model PCF in complete ultrametric spaces by
taking $\Delay(\nat)$ to interpret the base type, and using the fixed-point theorem
to interpret fixed-point recursion. Let us use $D$ to denote this model. We can also
give an operational semantics which tracks recursive unfoldings, and show that $D$ is
adequate with respect to this semantics.

The logical relation of \emph{extensional equality} on $D$ is given by saying that
$\delay^k(n)\seq \delay^l(n)$ for each $k,l,n:\nat$, and then lifting this to higher
types in the usual way. The quotient of the base type $D_{\iota}$ by this relation is
$\Lift_{\isRos}(\nat)$. Classically, this is equivalent to $\Lift(\nat)$, and this
quotient lifts to higher types. Constructively, however, we do not have this
equivalence. Worse, we do not know that $\Lift_{\isRos}(\nat)$ is a dcpo, so we
cannot use a version of the Scott model replacing $\Lift(\nat)$ with
$\Lift_{\isRos}(\nat)$. In other words, the interpretation of a type $\sigma$ in the
Scott model is not a subquotient of $D_\sigma$.

How much, then, can we say constructively about the relationship between the ultrametric and Scott
models?


%% file: chapters/conclusion.tex
\chapter*{Conclusion and summary}
\addcontentsline{toc}{chapter}{Conclusion and summary}
Here we summarize the material in the chapter discussions:
a key point in the univalent perspective is the careful distinction between structure
and property, while taking the former to subsume the latter. There is a distinction
between property and structure in logical systems (with propositions as formulas),
and some type theories (such as the calculus of inductive
constructions~\cite{CoC1986,CoC1988,Huet1987}, which has a type of propositions),
but there is usually no syntax for relating structure and property. It is worth
pointing out that propositions are not exactly
\emph{proof-irrelevant} in our framework (as suggested in work leading up to the notion such
as~\cite{pfenning2001,mendler,Palmgren2012,Hofmann1995}), because information can be
extracted from them (Lemma~\ref{lemma:untruncate-decidable-predicates} and
\cite{krausinvertible,keca2013,keca2016}).

The difference in treatment of propositions makes direct translation of known
results (especially model-theoretic results) somewhat non-trivial. Internally, a
choice between structure and property is necessary at all points (but the correct
choice often presents itself quickly); externally, we have difficulty using many
standard models of MLTT, since they do not satisfy our version of proposition
extensionality. On the other hand, this treatment of propositions resolves some
ambiguity surrounding choice and extensionality (consider Section~\ref{section:choice}
and~\cite{martinlof2006choice}.

Besides the treatment of propositions, there are 3 novelties of univalent mathematics
\begin{itemize}
    \item The namesake due to Voevodsky, the \emph{univalence axiom} which gives a
        general extensionality principle, from which proposition and function
        extensionality follow.
    \item The stratification of types into levels determined by the structure of
        identity types---contractible types; propositions whose identity types are
        contractible; sets whose identity types are propositions, (1-)groupoids whose
        identity types are sets, 2-groupoids whose identity types are groupoids, and
        so on.
    \item Higher-inductive types (HITs), which allow us to generate a type using not only
        constructors of that type, but also \emph{path constructors} which inhabit
        identity types.
\end{itemize}
The status of HITs and the relationship between univalence and computational meaning
are active areas of
investigation~\cite{ls2017semantics,chm2018cubical,DybjerMoeneclaey2017,kaposi2018hiit,computationalttIV,ssetmodel,huber2016,bezem2014model,cubicaltt,huber2016canonicity}.
In the above, we restricted most of our attention to sets, limiting our use of these
new notions to quotient types, (propositional) truncations, proposition
extensionality, and function extensionality.

Our approach to partiality builds off of past approaches to partiality in type theory
and in topos theory. On the type-theoretic side, there are versions of the \emph{delay
monad}~\cite{BoveCapretta2008,capretta2005recursion}.
At its most basic, the delay monad is an intensional approach to partial functions
with a clock. Attempts have been made to make it more extensional by
quotienting~\cite{Chapman2015Delay}, but this approach requires countable choice for
composability; or generating it as a higher-inductive
type~\cite{Altenkirch:Danielsson:Kraus}, but this approach
expands the type of Rosolini partial functions, where we want to restrict the
Rosolini partial functions: we are trying to represent the computable partial
functions, and this type embed into the Rosolini partial functions. On the topos-theoretic
side, there is the notion of dominance~\cite{Rosolini1986} which is used to restrict
the type of partial elements. Of particular interest is the set of \emph{Rosolini
propositions}, which is used as a form of semidecidable propositions in synthetic domain
theory~\cite{Rosolini1986,Hyland1991,VANOOSTEN2000,ReusStreicher1997} and synthetic
computability theory~\cite{Bauer2006Synthetic}.

Several traditional taboos can be stated in terms of Rosolini propositions:
\begin{itemize}
    \item Kripke's Schema says that all propositions are Rosolini.
    \item Markov's Principle is double-negation elimination for
        Rosolini propositions.
    \item LPO says that true and false are the only Rosolini propositions.
    \item WLPO says that Rosolini propositions are either false or not false.
\end{itemize}
However, in order to show that the Rosolini propositions form a dominance, some
version of countable choice is needed. In fact, we need exactly that we can recover
Rosolini structures from Rosolini propositions
(Theorem~\ref{dominance:choice-to-dominance}). In order for the partial functions
arising from a set of propositions to be composable, they must form a dominance.
Since partial functions ought to compose, we need countable choice to see the
Rosolini partial functions as a good notion of partial function.

Instead, we consider \emph{disciplined maps}---those partial functions which respect
information from an intensional version of the Rosolini dominance. This intensional
Rosolini dominance (the \emph{Rosolini structures}) gives a lifting which is
equivalent to the delay monad, so we can say that a disciplined map from $X$ to $Y$
is a partial function $f:X\pto Y$ for which there exists (as a proposition) a Kleisli
function for the delay monad which $f$ factors through.

Since countable choice holds in the setoid model, and does not in univalent type
theory~\cite{coquand2017stack}, it is worth examining how things play out in the
setoid approach. We can indeed construct the delay monad, and then quotient it to
arrive at the Rosolini lifting. In general, the setoid approach has the following
advantages over the univalent approach:
\begin{enumerate}
    \item The setoid approach allows us to use extensionality principles in a
        simpler version of MLTT.
    \item countable choice holds in the setoid model.
    \item The setoid approach fits squarely with the traditional account of constructive
        mathematics proposed by Bishop.
\end{enumerate}
However, the univalent approach has the following advantages over the setoid approach
\begin{enumerate}
    \item The univalent approach allows us to make use of universes; in particular,
        this allows us to define more computable functions, and to develop
        computability in more interesting domains.
    \item Definitions are more uniform, since we use an inductively defined identity
        type and all operations are automatically extensional.
    \item The univalent perspective fits squarely with a structural account of
        mathematics, and extends the traditional account of constructive mathematics.
    \item A type of propositions fits naturally into the
        theory, and clarifies the relationship between structure and property.
    \item There are computer proof systems with native support for formalization in
        univalent styles. No such support exists for the setoid
        approach.
\end{enumerate}

Once we have a notion of partial function, we can turn to computability theory. Much
of the classical theory goes through with little modification, however we introduced
a modest abstraction from Turing machines (the \emph{recursive machines}) to simplify
notation. We developed computability theory in two forms: with computability as
structure ($\CompStruct(f)$; Chapter~\ref{chapter:comp-as-struct}), and with
computability as property ($\isComputable(f)$; Chapter~\ref{chapter:comp-as-prop}).
Many foundational results in a first course on recursion theory go through for either
notion, but a few (such as the Normal Form Theorem) can only be stated in terms of
computability as structure.

Although every definable total function is computable, the statement that every
partial function $\N \to \Lift \N$ is computable is false, as the diagonal function
$d:\nat\to\Lift\nat$ is not computable. As one goal is to find a good notion of partial
function (in particular, we should be able to compose them) which can consistently be
assumed to be the computable partial functions, we need to restrict the set of all
partial functions in some way.

The Rosolini propositions give the ideas for a first attempt, and we can restrict these to
those which can be witnessed to be Rosolini by a computable function to arrive at the
\emph{semidecidable propositions}. In any case, neither the Rosolini partial
functions nor the semidecidable partial functions are composable, so we must pass
instead to the associated disciplined maps. And in fact, countable choice and
Church's thesis together imply that the four notions (Rosolini and semidecidable
partial functions, Rosolini disciplined maps and Semidecidable disciplined maps)
coincide, and in fact are exactly the computable functions. We summarize this in the following diagram:
\begin{center}
\begin{tikzpicture}
    \node (center) {};
    \node (ros) [above=of center] {$\nat\pto_{\RR}\nat$};
    \node (all) [above=of ros, rectangle,draw,solid]    {$\nat\pto\nat$};
    \node (sd)  [right=of center]  {$\nat\pto_{\SD}\nat$};
    \node (dis) [left=of center,rectangle,draw,solid]   {$\Dis_{\RR}(\nat,\nat)$};
    \node (disS) [below=of center]   {$\Dis_{\SD}(\nat,\nat)$};
    \node (comp) [below=of disS,rectangle,draw,solid]     {$\Comp$};
    \draw [thick,dash dot]   (comp.north) -- (disS.south) node
                [midway, label=right:{\tiny$\CC+\CT$}] {};
    \draw [thick,dashed]   ([xshift=-3pt]disS.north) -- (dis.south east) node
                [midway, label=left:{\tiny$\CT$}] {};
    \draw [thick,dotted]   ([xshift=3pt]disS.north) -- (sd.south west) node
                [midway, label=right:{\tiny$\CC$}] {};
    \draw [thick,dotted]   (dis.north east) -- ([xshift=-3pt] ros.south) node
                [midway, label=left:{\tiny$\CC$}] {};
    \draw [thick,dashed]    (sd.north west)  -- ([xshift=3pt] ros.south) node
                [midway, label=right:{\tiny$\CT$}] {};
    \draw [solid]   (ros.north) -- (all.south) node
                [midway, label=right:{\tiny$\KS$}] {};
\end{tikzpicture}
\end{center}
Lower types embed into higher types and the boxed types can be shown to be closed
under composition with no choice principles; the dashed lines collapse under Church's thesis,
while the dotted lines collapse under countable choice. Moreover, the solid line
collapses under Kripke's schema.

We know that countable choice and Church's thesis hold in the effective topos.
Although the handling of propositions in a topos is slightly different than in
univalent mathematics, we therefore have reason to believe the following
\begin{conjecture}
    It is consistent that all Rosolini-disciplined maps are computable.
\end{conjecture}
While the obvious model validating this (some univalent version of the effective
topos) should also validate the logically stronger statement that all Rosolini partial
functions are computable, the Rosolini-disciplined maps can be shown to compose even in
the absence of choice principles, and so can be used as a surrogate version of the
computable partial functions, even in settings where we do not have countable choice,
so long as such settings are compatible with Church's thesis.
The fact that the effective topos does not model our version of proposition
extensionality leaves us with a direction for future work: find either a way to
interpret univalent type theory in the effective topos (or toposes more
generally), or come up with a univalent version of the effective topos.

As our primary aim was not to develop computability theory, but to fit it into a
larger context, we have not done a great deal of computability theory. This presents
another area for future work. As higher-typed computability seems particularly
well-suited to study from a constructive view-point (relying on many of the same
tools), and there are non-trivial questions raised almost immediately when we begin
the attempt (for example, what does the Scott model look like constructively?), this
seems to be a fruitful area of study.

An important point in the above development is that we can view a notion as either
structure or as property, and the notion will behave differently depending on the
view we take. We illustrated this with Church's thesis (which is false for computability as
structure), and with our ``facts'' \ref{c:wct}-\ref{c:se} from the introduction and
Section~\ref{computability:church}. We end by emphasizing this point with another anecdote
about the development of Part~\ref{part:three}: Originally, I had stated
Theorem~\ref{thm:semidecidable-is-computable} in a stronger form,
about semidecision procedures, giving the weaker statement as a corollary. The proof
looked the same, except truncation was applied differently. In fact, the proof was
wrong, because I used truncations incorrectly. If I had untruncated all structure in
the theorem (using computability structure everywhere) the corresponding result would
hold using the same proof, but Theorem~\ref{ros-is-sd} would not hold, and
Corollary~\ref{cor:semidec-char} would need to be modified accordingly. In short,
each decision about whether to view computability as property or structure leads to a
different outcome. The development of ideas in Sections~\ref{section:semidecidable}
and~\ref{computability:church} is highly dependent on proper handling of the
distinction between structure and property.


%% file: ms.bbl
\begin{thebibliography}{10}

\bibitem{Aczel1977}
Peter Aczel.
\newblock The strength of {M}artin-{L}{\"o}f's intuitionistic type theory with
  one universe.
\newblock In Miettinen Seppo and V{\"a}n{\"a}nen Jouko, editors, {\em
  Proceedings of the symposiums on mathematical logic in Oulu 1974 and in
  Helsinki 1975}, 1977.
\newblock Report No. 2.

\bibitem{Altenkirch:Danielsson:Kraus}
Thorsten Altenkirch, Nils~Anders Danielsson, and Nicolai. Kraus.
\newblock Partiality, revisited: The partiality monad as a quotient
  inductive-inductive type.
\newblock In {\em FoSSaCS Proceedings}, 2017.

\bibitem{computationalttI}
Carlo Angiuli, Robert Harper, and Todd Wilson.
\newblock Computational higher type theory {I:} abstract cubical realizability.
\newblock {\em CoRR}, abs/1604.08873, 2016.

\bibitem{AwodeyBauer2004Bracket}
Steven Awodey and Andrej Bauer.
\newblock Propositions as [types].
\newblock {\em J. Log. and Comput.}, 14(4):447--471, August 2004.

\bibitem{bauer2000realizability}
Andrej Bauer.
\newblock {\em The realizability approach to computable analysis and topology}.
\newblock PhD thesis, Carnegie Mellon University Pittsburgh, PA, USA, 2000.

\bibitem{Bauer2006Synthetic}
Andrej Bauer.
\newblock First steps in synthetic computability theory.
\newblock {\em Electronic Notes in Theoretical Computer Science}, 155:5 -- 31,
  2006.
\newblock Proceedings of the 21st Annual Conference on Mathematical Foundations
  of Programming Semantics (MFPS XXI).

\bibitem{Beeson1980}
Michael~J. Beeson.
\newblock {\em Foundations of Constructive Mathematics}.
\newblock Springer Verlag, 1980.

\bibitem{bezem2014model}
Marc Bezem, Thierry Coquand, and Simon Huber.
\newblock A model of type theory in cubical sets.
\newblock In {\em 19th International Conference on Types for Proofs and
  Programs (TYPES 2013)}, volume~26, pages 107--128, 2014.

\bibitem{bishop1967}
Errett Bishop.
\newblock {\em Foundations of constructive analysis}.
\newblock McGraw-Hill series in higher mathematics. McGraw-Hill, 1967.

\bibitem{Bishop1970Numerical}
Errett Bishop.
\newblock Mathematics as a numerical language.
\newblock In A.~Kino, J.~Myhill, and R.E. Vesley, editors, {\em Intuitionism
  and Proof Theory: Proceedings of the Summer Conference at Buffalo N.Y. 1968},
  volume~60 of {\em Studies in Logic and the Foundations of Mathematics}, pages
  53 -- 71. Elsevier, 1970.

\bibitem{bishop1973schizophrenia}
Errett Bishop.
\newblock {\em Schizophrenia in Contemporary Mathematics}.
\newblock American Mathematical Society, 1973.

\bibitem{BoveCapretta2008}
Ana Bove and Venanzio Capretta.
\newblock {\em A Type of Partial Recursive Functions}, pages 102--117.
\newblock Springer, Berlin, Heidelberg, 2008.

\bibitem{Bridges1987Varieties}
Douglas Bridges and Fred Richman.
\newblock {\em Varieties of constructive mathematics}, volume~97 of {\em London
  Mathematical Society Lecture Note Series}.
\newblock Cambridge University Press, 1987.

\bibitem{capretta2005recursion}
Venancio Capretta.
\newblock General recursion via coinductive types.
\newblock {\em Logical Methods in Computer Science}, 15:1--28, 2005.

\bibitem{computationalttIV}
E.~{Cavallo} and R.~{Harper}.
\newblock {Computational Higher Type Theory {IV:} Inductive Types}.
\newblock {\em ArXiv e-prints}, January 2018.

\bibitem{Chapman2015Delay}
James Chapman, Tarmo Uustalu, and Niccol{\`o} Veltri.
\newblock {\em Quotienting the Delay Monad by Weak Bisimilarity}, pages
  110--125.
\newblock Springer, 2015.

\bibitem{cubicalttgit}
Cyril Cohen, Thierry Coquand, Simon Huber, and Anders M{\"o}rtberg.
\newblock Cubical type theory.
\newblock \url{https://agda.readthedocs.io/en/latest/language/cubical.html}.

\bibitem{cubicaltt}
Cyril Cohen, Thierry Coquand, Simon Huber, and Anders M{\"{o}}rtberg.
\newblock Cubical type theory: a constructive interpretation of the univalence
  axiom.
\newblock {\em CoRR}, abs/1611.02108, 2016.

\bibitem{coq}
ADT Coq.
\newblock The coq proof assistant reference manual, 2009.

\bibitem{chm2018cubical}
T.~{Coquand}, S.~{Huber}, and A.~{M{\"o}rtberg}.
\newblock {On Higher Inductive Types in Cubical Type Theory}.
\newblock {\em ArXiv e-prints}, February 2018.

\bibitem{CoC1986}
T.~Coquand and G{\'e}rard Huet.
\newblock {The calculus of constructions}.
\newblock Technical Report RR-0530, {INRIA}, May 1986.

\bibitem{CMR2017Stacks}
T.~Coquand, B.~Mannaa, and F.~Ruch.
\newblock Stack semantics of type theory.
\newblock In {\em 2017 32nd Annual ACM/IEEE Symposium on Logic in Computer
  Science (LICS)}, pages 1--11, June 2017.

\bibitem{Coquand2018VV}
Thierry Coquand.
\newblock Some contributions of {V}ladimir {V}oevodsky to dependent type
  theory.
\newblock Presentation at EUTypes 2018 Working Meeting, 2018.

\bibitem{coquand2013isomorphism}
Thierry Coquand and Nils~Anders Danielsson.
\newblock Isomorphism is equality.
\newblock {\em Indagationes Mathematicae}, 24(4):1105--1120, 2013.

\bibitem{CoC1988}
Thierry Coquand and G\'{e}rard Huet.
\newblock The calculus of constructions.
\newblock {\em Information and Computation}, 76(2):95 -- 120, 1988.

\bibitem{coquand2017stack}
Thierry Coquand, Bassel Mannaa, and Fabian Ruch.
\newblock Stack semantics of type theory.
\newblock In {\em Logic in Computer Science (LICS), 2017 32nd Annual ACM/IEEE
  Symposium on}, pages 1--11. IEEE, 2017.

\bibitem{DybjerMoeneclaey2017}
Peter Dybjer and Hugo Moeneclaey.
\newblock Finitary higher inductive types in the groupoid model.
\newblock In {\em MFPS XXXIII}, Electronic Notes in Theoretical Computer
  Science, Berlin, Heidelberg, 2017. Springer Berlin Heidelberg.
\newblock 33rd Conference on Mathematical Foundations of Programming Semantics
  (to appear).

\bibitem{escardo1998metric}
Mart{\'i}n Escard{\'o}.
\newblock A metric model of {PCF}, 1998.

\bibitem{typetopology}
Mart\'{i}n Escard\'{o} et~al.
\newblock Typetopology.
\newblock Available at \url{https://github.com/martinescardo/TypeTopology}.

\bibitem{escardo2013omniscience}
Mart{\'i}n~H. Escard{\'o}.
\newblock Infinite sets that satisfy the principle of omniscience in any
  variety of constructive mathematics.
\newblock {\em The Journal of Symbolic Logic}, 78(3):764--784, 2013.

\bibitem{partialelems2017}
Mart{\'i}n~H{\"o}tzel Escard{\'o} and Cory~M. Knapp.
\newblock Partial elements and recursion via dominances in univalent type
  theory.
\newblock In {\em LIPIcs-Leibniz International Proceedings in Informatics},
  volume~82. Schloss Dagstuhl-Leibniz-Zentrum fuer Informatik, 2017.

\bibitem{Feferman1982}
Solomon Feferman.
\newblock Iterated inductive fixed-point theories: Application to hancock's
  conjecture.
\newblock In George Metakides, editor, {\em PATRAS LOGIC SYMPOSION}, volume 109
  of {\em Studies in Logic and the Foundations of Mathematics}, pages 171 --
  196. Elsevier, 1982.

\bibitem{gambino2011univalence}
Nicola Gambino, Notes by~C. Kapulkin, and P.~L. Lumsdaine.
\newblock The univalence axiom and functional extensionality.
\newblock Lecture notes, 2011.

\bibitem{Griffor1994}
Edward Griffor and Michael Rathjen.
\newblock The strength of some {M}artin-{L}{\"o}f type theories.
\newblock {\em Archive for Mathematical Logic}, 33(5):347--385, Oct 1994.

\bibitem{Hofmann1995}
Martin Hofmann.
\newblock {\em Extensional concepts in intensional type theory}.
\newblock PhD thesis, University of Edinburgh, 1995.

\bibitem{hofmann1995lccc}
Martin Hofmann.
\newblock On the interpretation of type theory in locally cartesian closed
  categories.
\newblock In Leszek Pacholski and Jerzy Tiuryn, editors, {\em Computer Science
  Logic}, pages 427--441, Berlin, Heidelberg, 1995. Springer Berlin Heidelberg.

\bibitem{gpdmodel}
Martin Hofmann and Thomas Streicher.
\newblock The groupoid interpretation of type theory.
\newblock In {\em Twenty-five years of constructive type theory ({V}enice,
  1995)}, volume~36 of {\em Oxford Logic Guides}, pages 83--111. Oxford Univ.
  Press, New York, 1998.

\bibitem{escardoxu2015}
Mart{\'\i}n H{\"o}tzel~Escard{\'o} and Chuangjie Xu.
\newblock The inconsistency of a brouwerian continuity principle with the
  curry--howard interpretation.
\newblock In {\em LIPIcs-Leibniz International Proceedings in Informatics},
  volume~38. Schloss Dagstuhl-Leibniz-Zentrum fuer Informatik, 2015.

\bibitem{huber2016canonicity}
S.~{Huber}.
\newblock {Canonicity for Cubical Type Theory}.
\newblock {\em ArXiv e-prints}, July 2016.

\bibitem{huber2016}
Simon Huber.
\newblock {\em Cubical interpretations of type theory}.
\newblock PhD thesis, University of Gothenburg, 2016.

\bibitem{Huet1987}
G{\'e}rard Huet.
\newblock Induction principles formalized in the calculus of constructions.
\newblock In Hartmut Ehrig, Robert Kowalski, Giorgio Levi, and Ugo Montanari,
  editors, {\em TAPSOFT '87}, pages 276--286, Berlin, Heidelberg, 1987.
  Springer Berlin Heidelberg.

\bibitem{hughes2000generalising}
John Hughes.
\newblock Generalising monads to arrows.
\newblock {\em Science of computer programming}, 37(1-3):67--111, 2000.

\bibitem{hyland1982effective}
J.~M.~E. Hyland.
\newblock The effective topos.
\newblock {\em Studies in Logic and the Foundations of Mathematics},
  110:165--216, 1982.

\bibitem{Hyland1991}
J.~M.~E. Hyland.
\newblock {\em First steps in synthetic domain theory}, pages 131--156.
\newblock Springer, 1991.

\bibitem{kaposi2018hiit}
Ambrus Kaposi and Andr{\'a}s Kova{\'a}cs.
\newblock A syntax for higher-inductive inductive types, Feb 2018.
\newblock Submitted to FSCD.

\bibitem{ssetmodel}
C.~{Kapulkin} and P.~L. Lumsdaine.
\newblock {The Simplicial Model of Univalent Foundations (after {V}oevodsky)}.
\newblock {\em ArXiv e-prints}, November 2012.

\bibitem{keca2016}
N.~Kraus, M.H. Escard\'o, T.~Coquand, and T.~Altenkirch.
\newblock Notions of anonymous existence in {M}artin-{L}\"of type theory.
\newblock {\em Logical Methods in Computer Science}, 2016.

\bibitem{krausinvertible}
Nicolai Kraus.
\newblock The truncation map $|\_|:\nat\to\trunc{\nat}$ is nearly invertible.
\newblock
  \url{https://homotopytypetheory.org/2013/10/28/the-truncation-map-_-%E2%84%95-%E2%80%96%E2%84%95%E2%80%96-is-nearly-invertible/},
  2013.

\bibitem{keca2013}
Nicolai Kraus, Mart{\'i}n Escard{\'o}, Thierry Coquand, and Thorsten
  Altenkirch.
\newblock Generalizations of {H}edberg's theorem.
\newblock In Masahito Hasegawa, editor, {\em Typed Lambda Calculi and
  Applications}, pages 173--188, Berlin, Heidelberg, 2013. Springer Berlin
  Heidelberg.

\bibitem{Ladyman218Foundation}
James Ladyman and Stuart Presnell.
\newblock Does homotopy type theory provide a foundation for mathematics?
\newblock {\em The British Journal for the Philosophy of Science},
  69(2):377--420, 2018.

\bibitem{lambek1986higher}
J.~Lambek and P.~J. Scott.
\newblock {\em Introduction to Higher Order Categorical Logic}.
\newblock Cambridge University Press, New York, NY, USA, 1986.

\bibitem{lob1970}
M.~H. L{\"o}b and S.~S. Wainer.
\newblock Hierarchies of number-theoretic functions. {I}.
\newblock {\em Archiv f{\"u}r mathematische Logik und Grundlagenforschung},
  13(1):39--51, Mar 1970.

\bibitem{lumsdaine2009weak}
P.~L. Lumsdaine.
\newblock Weak $\omega$-categories from intensional type theory.
\newblock {\em Typed Lambda Calculi and Applications}, pages 172--187, 2009.

\bibitem{ls2017semantics}
P.~L. Lumsdaine and M.~{Shulman}.
\newblock {Semantics of higher inductive types}.
\newblock {\em ArXiv e-prints}, May 2017.

\bibitem{maclane1994sheaves}
S.~MacLane and I.~Moerdijk.
\newblock {\em Sheaves in Geometry and Logic: A First Introduction to Topos
  Theory}.
\newblock Universitext. Springer New York, 1994.

\bibitem{markov1962}
A.~A. Markov.
\newblock On constructive mathematics.
\newblock In {\em Problems of the constructive direction in mathematics.
  Part~2. Constructive mathematical analysis}, volume~67 of {\em Trudy Mat.
  Inst. Steklov.}, pages 8--14. Acad. Sci. USSR, Moscow--Leningrad, 1962.

\bibitem{MartinLof1975Predicative}
Per Martin-L{\"o}f.
\newblock An intuitionistic theory of types: Predicative part.
\newblock In H.E. Rose and J.C. Shepherdson, editors, {\em Logic Colloquium
  '73}, volume~80 of {\em Studies in Logic and the Foundations of Mathematics},
  pages 73 -- 118. Elsevier, 1975.

\bibitem{martinlof2006choice}
Per Martin-L\"{o}f.
\newblock 100 years of zermelo's axiom of choice: What was the problem with it?
\newblock {\em Comput. J.}, 49(3):345--350, May 2006.

\bibitem{mclarty1992elementary}
C.~McLarty.
\newblock {\em Elementary Categories, Elementary Toposes}.
\newblock Clarendon Press, 1992.

\bibitem{mendler}
N.~P. Mendler.
\newblock Quotient types via coequalizers in martin-l\"of type theory.
\newblock In G. Huet and G. Plotkin, editors, Informal Proceedings of the First
  Workshop on Logical Frameworks, Antibes, May 1990, pages 349–360, 1990, May
  1990.

\bibitem{AgdaWiki}
Ulf Norell, Nils~Anders Danielsson, and Andreas Abel.
\newblock The agda wiki.
\newblock \url{http://wiki.portal.chalmers.se/agda/pmwiki.php}.

\bibitem{Odifreddi1989}
Piergiorgio Odifreddi.
\newblock {\em Classical Recursion Theory}.
\newblock Elsevier, 1989.

\bibitem{orton2017cubical}
I.~{Orton} and A.~M. {Pitts}.
\newblock {Axioms for Modelling Cubical Type Theory in a Topos}.
\newblock {\em ArXiv e-prints}, December 2017.

\bibitem{Palmgren2012}
Erik Palmgren.
\newblock Proof-relevance of families of setoids and identity in type theory.
\newblock {\em Archive for Mathematical Logic}, 51(1):35--47, Feb 2012.

\bibitem{pfenning2001}
F.~Pfenning.
\newblock Intensionality, extensionality, and proof irrelevance in modal type
  theory.
\newblock In {\em Proceedings 16th Annual IEEE Symposium on Logic in Computer
  Science}, pages 221--230, 2001.

\bibitem{Rathjen2000}
Michael Rathjen.
\newblock The strength of {M}artin-{L}{\"o}f type theory with a superuniverse.
  part i.
\newblock {\em Archive for Mathematical Logic}, 39(1):1--39, Jan 2000.

\bibitem{Rathjen2001}
Michael Rathjen.
\newblock The strength of martin-l{\"o}f type theory with a superuniverse. part
  ii.
\newblock {\em Archive for Mathematical Logic}, 40(3):207--233, Apr 2001.

\bibitem{ReusStreicher1997}
B.~Reus and Th. Streicher.
\newblock {\em General synthetic domain theory --- A logical approach (extended
  abstract)}, pages 293--313.
\newblock Springer Berlin Heidelberg, Berlin, Heidelberg, 1997.

\bibitem{Reus1999}
Bernhard Reus.
\newblock Formalizing synthetic domain theory.
\newblock {\em Journal of Automated Reasoning}, 23(3):411--444, Nov 1999.

\bibitem{riehl2016}
Emily Riehl and Dominic Verity.
\newblock Homotopy coherent adjunctions and the formal theory of monads.
\newblock {\em Advances in Mathematics}, 286:802 -- 888, 2016.

\bibitem{rijke:join}
E.~{Rijke}.
\newblock {The join construction}.
\newblock {\em ArXiv e-prints}, January 2017.

\bibitem{rijke2017modalities}
E.~{Rijke}, M.~{Shulman}, and B.~{Spitters}.
\newblock {Modalities in homotopy type theory}.
\newblock {\em ArXiv e-prints}, June 2017.

\bibitem{rijke2012homotopy}
Egbert Rijke.
\newblock Homotopy type theory.
\newblock Master's thesis, Utrecht University, 2012.

\bibitem{Rogers1987}
H.~Rogers, Jr.
\newblock {\em Theory of Recursive Functions and Effective Computability}.
\newblock MIT Press, 1987.

\bibitem{Rosolini1986}
G.~Rosolini.
\newblock {\em Continuity and Effectiveness in Topoi}.
\newblock PhD thesis, University of Oxford, 1986.

\bibitem{scottdomains}
Dana~S. Scott.
\newblock A type-theoretical alternative to iswim, cuch, owhy.
\newblock {\em Theor. Comput. Sci.}, 121(1-2):411--440, December 1993.

\bibitem{seely1984lccc}
Robert~AG Seely.
\newblock Locally cartesian closed categories and type theory.
\newblock In {\em Mathematical proceedings of the Cambridge philosophical
  society}, volume~95, pages 33--48. Cambridge University Press, 1984.

\bibitem{setzer1993}
Anton Setzer.
\newblock {\em Proof theoretical strength of {M}artin-{L}{\"o}f Type Theory
  with {W}-type and one universe}.
\newblock PhD thesis, Ludwig-Maximilians-Universit{\"a}t M{\"u}nchen, 1993.

\bibitem{setzer1998}
Anton Setzer.
\newblock Well-ordering proofs for {M}artin-{L}{\"o}f type theory.
\newblock {\em Annals of Pure and Applied Logic}, 92(2):113 -- 159, 1998.

\bibitem{setzer2000}
Anton Setzer.
\newblock Extending {M}artin-{L}{\"o}f type theory by one {M}ahlo-universe.
\newblock {\em Archive for Mathematical Logic}, 39(3):155--181, Apr 2000.

\bibitem{smith1988}
Jan~M. Smith.
\newblock The independence of peano's fourth axiom from {M}artin-{L}{\"o}f's
  type theory without universes.
\newblock {\em Journal of Symbolic Logic}, 53(3):840--845, 1988.

\bibitem{Streicher2013How}
Thomas Streicher.
\newblock How intensional is homotopy type theory?
\newblock In Maria del~Mar Gonz{\'a}lez, Paul~C. Yang, Nicola Gambino, and
  Joachim Kock, editors, {\em Extended Abstracts Fall 2013}, pages 105--110,
  Cham, 2015. Springer International Publishing.

\bibitem{taylor1991fixed}
Paul Taylor.
\newblock The fixed point property in synthetic domain theory.
\newblock In {\em Logic in Computer Science, 1991. LICS'91., Proceedings of
  Sixth Annual IEEE Symposium on}, pages 152--160. IEEE, 1991.

\bibitem{troelstra:1977}
A.~S. Troelstra.
\newblock A note on non-extensional operations in connection with continuity
  and recursiveness.
\newblock {\em Indag. Math.}, 39(5):455--462, 1977.

\bibitem{troelstra1988constructivism}
A.S. Troelstra and D.~van Dalen.
\newblock {\em Constructivism in Mathematics}.
\newblock Number v. 1 in Studies in Logic and the Foundations of Mathematics.
  Elsevier Science, 1988.

\bibitem{hottbook}
The {Univalent Foundations Program}.
\newblock {\em Homotopy Type Theory: Univalent Foundations of Mathematics}.
\newblock \url{https://homotopytypetheory.org/book}, Institute for Advanced
  Study, 2013.

\bibitem{uustalu2017partiality}
Tarmo Uustalu and Niccolo Veltri.
\newblock Partiality and container monads.
\newblock In {\em Asian Symposium on Programming Languages and Systems}, pages
  406--425. Springer, 2017.

\bibitem{vdBergGarner2011}
Benno van~den Berg and Richard Garner.
\newblock Types are weak $\omega$-groupoids.
\newblock {\em Proceedings of the London Mathematical Society},
  102(2):370--394, 2011.

\bibitem{VANOOSTEN2000}
Jaap van Oosten and Alex~K. Simpson.
\newblock Axioms and (counter)examples in synthetic domain theory.
\newblock {\em Annals of Pure and Applied Logic}, 104(1):233 -- 278, 2000.

\bibitem{VoevodskyBergen}
Vladimir Voevodsky.
\newblock Resizing rules.
\newblock Plenary Lecture at TYPES2011, 2011.

\bibitem{UniMath}
Vladimir Voevodsky, Benedikt Ahrens, Daniel Grayson, et~al.
\newblock {\em UniMath}: {Univalent} {Mathematics}.
\newblock Available at \url{https://github.com/UniMath}.

\bibitem{warren2011strict}
Michael~A Warren.
\newblock The strict $\omega$-groupoid interpretation of type theory.
\newblock {\em Models, Logics, and Higher-dimensional Categories: A Tribute to
  the Work of Mih{\'a}ly Makkai}, 53:291--340, 2011.

\end{thebibliography}
